\documentclass{report}
\usepackage{sty}

\title{\huge {\bf 
Probabilistic Principles for\vspace{-5pt}\\ Biophysics and Neuroscience}\vspace{20pt}\\
\Large { \bf Entropy production, \vspace{-5pt}\\Bayesian mechanics\vspace{-5pt}\\ \& \vspace{-5pt}\\the free-energy principle.}}
\author{\vspace{42pt}\\\vspace{0pt} A thesis presented for the degree of\\
Doctor of Philosophy of Imperial College London\\
\vspace{20pt}by,\\
\vspace{110pt}\large{\bf Lancelot Da Costa}\\
Department of Mathematics \\
Imperial College London\\
\vspace{20pt}180 Queen's Gate, London SW7 2AZ}
\date{October 2023}

\begin{document}

\maketitle

\addcontentsline{toc}{chapter}{Abstract}

\begin{abstract}

This thesis focuses on three fundamental aspects of biological systems; namely, entropy production, Bayesian mechanics, and the free-energy principle. The contributions are threefold: 1) We compute the entropy production for a greater class of systems than before, including almost any stationary diffusion process, such as degenerate diffusions where the driving noise does not act on all coordinates of the system. Importantly, this class of systems encompasses Markovian approximations of stochastic differential equations driven by colored noise, which is significant since biological systems at the macro- and meso-scale are generally subject to colored fluctuations. 2) We develop a Bayesian mechanics for biological and physical entities that interact with their environment in which we give sufficient and necessary conditions for the internal states of something to infer its external states, consistently with variational Bayesian inference in statistics and theoretical neuroscience. 3) We refine the constraints on Bayesian mechanics to obtain a description that is more specific to biological systems, called the free-energy principle. This says that active and internal states of biological systems unfold as minimising a quantity known as free energy. The mathematical foundation to the free-energy principle, presented here, unlocks a first principles approach to modeling and simulating behavior in neurobiology and artificial intelligence, by minimising free energy given a generative model of external and sensory states.
\end{abstract}

\clearpage

\addcontentsline{toc}{chapter}{Acknowledgements}

\begin{acknowledgements}


I dedicate this thesis to my loving partner Pauline Chatellard who continuously and restlessly supports me through highs and lows, and so for the past many years. I could not have done this without you. 

My heartfelt gratitude goes to my awesome team of supervisors, Grigorios A. Pavliotis and Karl Friston, who have been incredibly supportive, helpful, and insightful throughout my research journey. Through your mentorship you have both given me a rich and complementary perspective, and a unique technical expertise. It is in large part to you that I owe my present and future academic success.

I thank my family for their support; especially my brother for inspiring me to be dedicated, resilient, and never give up, my father for encouraging me to pursue mathematics and neuroscience, and stimulating my curiosity since childhood, my mother for supporting me whatever happened, and teaching me to be kind and generous in all areas of life, and my aunt for being there in the most difficult moments. I am grateful to my longtime friends, particularly Nicolas Ward, Mathieu Binder, Arnaud Pedrazzani, Baptiste Pesanti, Ritwick Sundar, Stéphane Nordin, and Wojtek Reise.

Thank you to my amazing colleagues and collaborators. Your perspectives have been mind expanding, and many of you have become dear friends over the years. You are the reason why the PhD has been such a fun, interesting and rewarding experience. At the risk of omitting many, I particularly thank Conor Heins, Maxwell Ramstead, Thomas Parr, Kai Ueltzhöffer, Noor Sajid, Alessandro Barp, Victorita Neacsu, Laura Convertino, Anjali Bhat, Ryan Smith, Biswa Sengupta, Pablo Lanillos, Tim Verbelen, Dalton Sakthivadivel,  Théophile Champion, Guillermo B. Morales, Guilherme França, Aswin Paul, Magnus T. Koudahl, Beren Millidge, Peter Zeidman, Anil Seth, Christopher Buckley, Zafeirios Fountas, Alexey Zakharov, Umais Zahid, Sergio Rubin, Cyrus Mostajeran, Come Annicchiarico, Jeremie Mattout, Lars Sandved-Smith, Antoine Lutz, Danijar Hafner, Mahault Albarracin, Riddhi Pitliya Jain and Jonas Mago. Not least, I am deeply grateful to Michael I. Jordan, Joshua B. Tenenbaum, Philipp Hennig, and Samuel Gershman for invaluable career advice. Lastly, I thank my friends and classmates from the Centre for Doctoral Training in Mathematics of Random Systems, and particularly Rémy Messadene, Felix Prenzel, Alain Rossier, Victoria Klein, Julian Sieber, Benedikt Petko, Jonathan Tam, and Alessandro Micheli.

I thank the staff of the Centre for Doctoral Training in Mathematics of Random Systems, particularly our centre manager Lydia Noa, for always being available, incredibly helpful, and supportive; and the directors, who took a chance on me in spite of my atypical academic background.

Last but not least, I will always be grateful too the Fonds National de la Recherche of Luxembourg (Project code: 13568875) for funding my PhD research. Your generous funding has allowed me to pursue my PhD free of many administrative burdens, with unparalleled intellectual freedom, and with the ability to attend and contribute to many scientific events.

\end{acknowledgements}
\clearpage
\addcontentsline{toc}{chapter}{Integrity}

\begin{integrity}
  I certify that this thesis, and the research to which it refers, are the product of my own work,
and that any ideas or quotations from the work of other people, published or otherwise, are
fully acknowledged in accordance with the standard referencing practices of the discipline. In addition, I certify that I have obtained the necessary copyright allowances for reproducing some of my own academic papers in this thesis.
\end{integrity}
\vspace{25pt}

\addcontentsline{toc}{chapter}{Copyright}

\begin{copyr}
The copyright of this thesis rests with the author. Unless otherwise indicated, its contents are licensed under a Creative Commons Attribution-Non Commercial 4.0 International Licence (CC BY-NC).

Under this licence, you may copy and redistribute the material in any medium or format. You may also create and distribute modified versions of the work. This is on the condition that: you credit the author and do not use it, or any derivative works, for a commercial purpose.

When reusing or sharing this work, ensure you make the licence terms clear to others by naming the licence and linking to the licence text. Where a work has been adapted, you should indicate that the work has been changed and describe those changes.

Please seek permission from the copyright holder for uses of this work that are not included in this licence or permitted under UK Copyright Law.
\end{copyr}

\tableofcontents

\listoffigures

\chapter{Introduction}






\section{Background and motivation}

The ambition for this thesis started many years ago as I was finishing Part III of the Mathematical Tripos at Cambridge. I already knew that I wanted to work in neuroscience – deriving a fundamental theory of brain function that could also be used for principled artificial intelligence. A friend of the time, Sebastiano Cultrera, told me about the free-energy principle. This is a principle postulated by the British neuroscientist Karl Friston promising a powerful and unifying account of brain function in terms of minimising one single objective – variational free-energy \cite{fristonFreeenergyPrincipleUnified2010}. I did not understand the theory at the time but I thought it – given my idealistic pure mathematician's background – extremely exciting and a promising starting point  towards the end goal of a complete mathematical theory of the brain. Following my masters in mathematics, I undertook a masters in brain science at UCL where I was very fortunate to work with Karl Friston himself on applications of the free-energy principle for modelling neural dynamics in the brain. During that year, I realised that the free-energy principle was a physicist's intuition that did not yet admit rigorous mathematical and physical underpinnings – although attempts led by Karl Friston over the last decade had made considerable progress towards this. I realised that there was a significant opportunity for a mathematician like me to develop this fundamental theory. To this end, I organised an inter-university co-supervision with Professor Pavliotis, a world-renowned expert in stochastic differential equations – the type of mathematics that seemed to be needed to formalise the principle – and Karl Friston.

This collaboration and co-supervision has exceeded my expectations for learning and contributing to mathematical neuroscience and various related fields. For instance, I learned that the free-energy principle had been postulated to hold in biological systems and some physical systems, and a burgeoning line of work – which started at about the same time as I started my PhD – applies the free-energy principle to generate intelligent agents for artificial intelligence.

My PhD work has been devoted to developing the mathematical foundation of the free-energy principle for neuroscience and biophysics more broadly, as a description of human and biological cognition, and establishing principled and rigorous methodologies for its application in artificial intelligence. I also explored its potential for generating intelligent agents that are robust, whose decisions are explainable, and that are safe to interact with, which are current main challenges in artificial intelligence \cite{dacostaHowActiveInference2022a}. The overarching goal was to obtain an unbroken narrative from mathematical physics, to cognitive neuroscience, and finally to artificial intelligence – and a coherent framework to think about these fields.

This thesis focuses on three highlights from my PhD's work which pertain to fundamental principles in biophysics and neuroscience, and practically ignores all of the applications to artificial intelligence. The reason for this is twofold: to keep the length of this thesis within reasonable bounds, and to focus on some of the most mathematically significant aspects of my PhD project.

\section{Overview of thesis}


This thesis starts with the fundamental assumption any fundamental theory for biophysics and neuroscience has to be consistent with the rest of physics. Its starting point must therefore be a stochastic differential equation (a.k.a. a Langevin equation)

\begin{equation}
\label{eq: sde}
    \dot x_t = f(x_t)  + w_t.
\end{equation}
This equation decomposes the motion of a system into a deterministic part $f(x_t)$ governed by a vector field $f$ , which summarises what we know about the system (this part is usually called the \textit{drift} or the \textit{flow}), and a stochastic part $w_t$ which summarises what we don't know about the system (this part is usually called the \textit{noise} or the \textit{random fluctuations}). So why start here? Because this description is consistent with the rest of physics. For instance, stochastic differential equations form 
the basis of statistical mechanics \cite{pavliotisStochasticProcessesApplications2014}. In particular classical mechanics is the limit of statistical mechanics as particles become large and random fluctuations on the motion become infinitesimally small. Note that we did not commit to any specific interpretation of the stochastic differential equation at this stage – e.g. Stratonovich or Ito – and we did not commit to any specific assumptions about the functional form of the random fluctuations. This will come in later.

This thesis is the collection of three papers from my PhD's work which study three fundamental aspects of biological systems; namely, entropy production, Bayesian mechanics, and the free-energy principle.

\subsection{Chapter 2: Entropy production\footnote{\normalsize \textbf{Adapted from:} L Da Costa, GA Pavliotis. The entropy production of stationary diffusions. \textit{Journal of Physics A: Mathematical
and Theoretical}. 2023.}}

One of the fundamental traits of living systems is time-reversibility and entropy production. Entropy production measures the amount of time-irreversibility present in a system, that is to what extent the process as time goes forward differs from the process as time goes backwards. The entropy production also measures the minimum amount of energy consumed by the system for the period of time over which it is measured, or equivalently the minimal amount of heat dissipated by the system during that period of time.

In this chapter we develop a comprehensive and general theory for the entropy production of systems described by stochastic differential equations, and contextualise it within the broader literatures of probability and statistical physics.

From a mathematical perspective we are able to give a computable expression for the entropy production in a wider range of equations than was previously available, allowing for so-called degenerate diffusions which are SDEs where the noise process need not act on all coordinates. Although we focus on diffusion processes, that is SDEs driven by white noise, the fact that we allow degenerate diffusions means that the results apply to Markovian approximations of systems driven by coloured noise. This is fundamentally important because any biological system that is not modelled at a molecular level will plausibly be subject to coloured fluctuations. Although we focus on Euclidean state spaces, our results can easily be generalised to manifold state spaces.

Technically, we show that any stationary diffusion can be decomposed into a time reversible part and a time irreversible part, which generalises the classical Helmholtz-Hodge decomposition from vector calculus. We show that this decomposition is equivalent to writing the Fokker-Planck equation of the system in pre-GENERIC form, where GENERIC stands for General Equations for Non-Equilibrium Reversible-Irreversible Coupling – a prominent framework for analysing Dynamics in non-equilibrium statistical physics pioneered by Ottinger \cite{ottingerEquilibriumThermodynamics2005}. This is also equivalent to a decomposition of the backward Kolmogorov equation of the system considered in hypocoercivity theory, a theory developed by C. Villani for quantifying the rate of convergence of dynamics to their non-equilibrium steady state \cite{villaniHypocoercivity2009}.

All in all, this paper reviews essential concepts related to entropy production and time irreversibility and provides a wide range of tools for statistical physicists and biological physicists to study non-equilibrium phenomena.

\subsection{Chapter 3: Bayesian mechanics\footnote{\normalsize \textbf{Adapted from:} L Da Costa, K Friston, C Heins, GA Pavliotis. Bayesian mechanics for stationary processes. \textit{Proceedings of the
Royal Society A}. 2021.}}

The main idea that underwrites this chapter, is that if a 'thing' (e.g. a cell) and its surrounding environment, together defining a 'system' $x_t$, have an attracting stationary state, then the thing can be interpreted as being adaptive in the sense that its interactions on the environment look as if the thing is bringing itself – and the environment – towards a preferred set of states (i.e. homeostatic) following any kind of perturbation.

The implications of this simple observation are quite rich. For instance, the existence of a thing and what is not the thing entails the existence of boundary states $b_t$ between the internal states of the thing $\mu_t$ and the states of every thing else $\eta_t$. We formalise the idea of a boundary between a thing and its environment as a Markov blanket of the stationary density $p$ (i.e. conditional independence given the blanket states)

\begin{equation}
\label{eq: Markov blanket}
\mu_t \perp \eta_t \mid b_t, \text{ where } x_t = (\mu_t,b_t,\eta_t)\sim p.
\end{equation}
that exists over the period of time that the thing exists.

The existence of a Markov blanket means that internal states synchronise with external states via vicarious interactions through the blanket. We can go beyond this and show that internal states encode probabilistic beliefs about external states and update this continuously in the face of new blanket states (e.g. new sensory information) consistently with variational Bayesian inference in theoretical neuroscience and statistics.

If we further decompose blanket states into sensory and active states, then we can show that active states control the system back to its stationary state defining its range of preferred states following any perturbation, consistently with active inference in theoretical neuroscience and generalized forms of PID (Proportional-Integral-Derivative) control in engineering.

Since the publication of this article, it has become clear that the proper way to define a boundary between a thing and everything else is through a Markov blanket over trajectories, that is the fulfilment of \eqref{eq: Markov blanket} where states are replaced by paths over some period of time, as explored in the next Chapter of this thesis. Indeed, although articulated in a different language, this coincides with the notion of a multipartite process in statistical physics \cite{kardesThermodynamicUncertaintyRelations2021,wolpertMinimalEntropyProduction2020a}. Fortunately for this paper, a Markov blanket over paths often entails a Markov blanket over states. This was first realized empirically in simulations of a primordial soup \cite{fristonLifeWeKnow2013}, and then in further simulations involving sparsely coupled stochastic Lorenz systems \cite{fristonStochasticChaosMarkov2021}. The sparse coupling conjecture claims that for sufficiently high dimensional and nonlinear random dynamical systems, a Markov blanket over paths implies the existence of a Markov blanket over states (partial demonstrations are given in \cite{heinsSparseCouplingMarkov2022,sakthivadivelWeakMarkovBlankets2022}). In particular, this conjecture always holds in SDEs driven by coloured noise that can be expanded in generalized coordinates \cite{fristonFreeEnergyPrinciple2023a,fristonPathIntegralsParticular2023}.

\subsection{Chapter 4: the free-energy principle\footnote{\normalsize \textbf{Adapted from:} K Friston, L Da Costa, N Sajid, C Heins, K Ueltzhoeffer, GA Pavliotis, T Parr. The free energy principle made simpler but not too simple. \textit{Physics Reports}. 2023}}

Chapter 3 goes one step further in defining the type of constraints that a biological organism must comply with, and derives stronger conclusions. In this chapter, the boundary between an organism and its environment is defined as a Markov blanket over paths. As before, the boundary is composed of sensory and active states. Further, sensory states are permitted to influence internal states directly but internal states may not directly influence sensory states. Similarly for active states: they may influence external states directly, however external states may not directly influence active states. These influences are summarised by sparse coupling in the flow/drift vector field in the stochastic differential equation defining the whole system.

With this special set up, the conclusions of chapter 2 still hold and then some. We can show that the internal and active parts of the organism have paths of least action (i.e. most likely path) that minimise a quantity known as variational free-energy – the objective that used in variational Bayesian inference and theoretical neuroscience to describe the purpose of brain function. When we consider stochastic differential equations driven by coloured noise, we can show that internal dynamics correspond to generalized (Bayesian) filtering \cite{fristonGeneralisedFiltering2010} – perhaps the most generic filtering algorithm in the literature.

When we consider organisms at a sufficiently high-level of coarse-graining such that the dynamics of the organism itself (and not necessarily its environment) are governed by classical mechanics, we can see that the most likely path of internal and active states satisfy a principle of least action by minimizing quantity known as expected free-energy. The expected free-energy is the analog of the variational free-energy for trajectories. What is interesting is that by examining the expected free-energy carefully we recover as special cases several objective functions that are used to describe optimal behavior in various fields of science and engineering, e.g. economics, machine learning, and psychology. This suggests that several of these theories that have been arrived at independently may be unified and contextualized within a generic free-energy framework.

Pragmatically, the free-energy principle on offer gives us equations of motion, in terms of minimizing variational and expected free-energy, for describing and simulating intelligent biological behavior. This is shown by rehearsing some simulations from the literature such as a simulation of handwriting and a simulation of saccadic eye movements.

Much remains to be done in terms of analyzing the precise relationship between the free-energy principle and other theories of purposeful biological behavior that exist in the literature. Much also remains in refining the principle to obtain a description that is exclusive to the human brain, with more specific implications for neuroscience and artificial intelligence. However, the vast body of work utilizing the free-energy principle for modeling and generating intelligent biological behavior both in neuroscience and artificial intelligence, and its mathematical foundation – in terms of statistical physics – presented in this thesis suggests that the free-energy principle provides a solid foundation for a physics of sentient and cognitive beings.

\section{Limitations of thesis}

Perhaps the main limitation of this thesis is the reliance on white noise assumptions for the fluctuations in the SDE \eqref{eq: sde} in chapter 1 (Entropy production) and chapter 3 (the free-energy principle). Clearly white noise may be a good model for the random fluctuations driving systems at the molecular level (e.g. molecular biology) but is not fit for purpose for describing complex biological systems at a higher level of coarse graining, because in that case random fluctuations are usually the output of other dynamical systems, and thus have a nontrivial autocorrelation structure. In these two chapters the only place where coloured noise is treated explicitly is in chapter 3. (Note that there exists a refinement of chapter 3 that is specific to the case of coloured noise \cite{fristonPathIntegralsParticular2023}, however it is less mathematically rigorous than the current state of chapter 3).

\textit{Practically all of my fourth year of PhD has been devoted developing rigorous analysis, numerical simulation and filtering techniques for SDEs driven by coloured noise. My original plan was to include this in the thesis as an additional chapter placed between current chapters 1 and 2. However, the administration of Imperial College has been very inflexible in refusing to grant me an unpaid two month extension to the submission deadline, the time I need to complete this work. As a result I am not including this in my thesis submission, but would be happy to add it during revision stage as, crucially, it 1) nicely complements chapter 1, addressing its main limitations, 2) underpins chapters 2 and 3, and 3) offers a systematic treatment of a fundamentally important topic in the modelling and simulation of biological and neural systems.}

\section{A note on notations and terminologies}

While I have done my best to keep notation and terminology consistent across papers, i.e. chapters, there remain discrepancies which reflect the fact that each of the papers was originally written with a different audience in mind. For instance, probabilists use the letter $b$ for the drift of a stochastic differential equation, while physicists use the letter $f$ and call it the flow. Here I mostly chose to maintain the notation and terminology of the community that each paper was originally addressed to. This highlights the parallels between different and converging fields of research, which is perhaps the main implicit thread linking the various chapters.

\section{Publications}
My privileged position as one of the only mathematicians working on the free-energy principle, a topic that is becoming widely adopted in neuroscience and other application domains, means that I have been able and fortunate to collaborate on a wide range of projects, with collaborators of many different disciplines, which has been an incredibly mind expanding and enthralling experience.

\subsection{List of papers (32)}
\small{\textit{$^*$ Corresponding author, $^\dagger$ Equal contribution, $\star$ Highlighted work, $\diamond$ Cover article.}}

\subsubsection{Manuscripts Under Review (6)}
N Da Costa, M Pförtner, \textbf{L Da Costa}, P Hennig. Sample Path Regularity of Gaussian Processes from the Covariance Kernel. \textit{Under review in JMLR. arXiv:2312.14886}. 2023.

K Friston, \textbf{L Da Costa}, et al. Supervised structure learning. \textit{Under review in Neural Computation. arXiv:2311.10300} 2023.

C Heins, B Millidge, \textbf{L Da Costa}, RP Mann, K Friston, ID Couzin. Collective behavior from surprise minimization. \textit{Under review in PNAS.} 2023.

A Paul, N Sajid, \textbf{L Da Costa}, A Razi. On efficient computation in active inference. \textit{Under review in IEEE Transactions on Neural Networks and Learning Systems. arXiv:2307.00504}. 2023.

S Rubin, C Heins, T Mitsui, \textbf{L Da Costa}, K Friston. Climate homeostasis by and for active inference in the biosphere. \textit{Under review in Nature Communications.} 2023.
    
N Sajid, E Holmes, \textbf{L Da Costa}, C Price, K Friston. A mixed generative model of auditory word repetition. \textit{Under review in Frontiers in Neuroscience. BioRxiv}. 2022.

\subsubsection{Book Chapters (2)}

$\hspace{-8pt}\star$ A Barp$^\dagger$, \textbf{L Da Costa}$^\dagger$, G Franca$^\dagger$, K Friston, M Girolami, MI Jordan, GA Pavliotis. Geometric Methods for Sampling, Optimisation, Inference and Adaptive Agents. \textit{Geometry and Statistics, Handbook of Statistics Series vol 46 (Eds. F Nielsen, AS Rao, CR Rao)}. 2022.

N Sajid, \textbf{L Da Costa}, T Parr, K Friston. Active inference, Bayesian optimal design, and expected utility. \textit{The Drive for Knowledge (Eds. IC Dezza, E Schulz, CM Wu)}. 2022.


\subsubsection{Workshop Papers (2)}

N Sajid, P Tigas, Z Fountas, Q Guo, A Zakharov, \textbf{L Da Costa}. Modelling non-reinforced preferences using selective attention. \textit{First Conference on Lifelong Learning Agents: Workshop Track. arXiv:2207.13699}. 2022.
    
\textbf{L Da Costa}$^*$, N Sajid, D Zhao, S Tenka. Active inference as a model of agency. \textit{RLDM 2022 Workshop on Reinforcement Learning as a Model of Agency}. 2022.


\subsubsection{Journal Articles (22)}
\textbf{L Da Costa}, L Sandved-Smith. Towards a Bayesian mechanics of metacognitive particles: A commentary on "Path integrals, particular kinds, and strange things" by Friston, Da Costa, Sakthivadivel, Heins, Pavliotis, Ramstead, and Parr. \textit{Physics of Life Reviews.} 2023.

    K Friston, \textbf{L Da Costa}$^*$, DAR Sakthivadivel, C Heins, GA Pavliotis, MJD Ramstead, T Parr. Path integrals, particular kinds, and strange things. \textit{Physics of Life Reviews. arXiv:2210.12761}. 2023.
    
    $\hspace{-8pt}\star$
    K Friston, \textbf{L Da Costa}$^*$, N Sajid, C Heins, K Ueltzhoeffer, GA Pavliotis, T Parr. The free-energy principle made simpler but not too simple. \textit{Physics Reports}. 2023.

    $\hspace{-8pt}\star$
    \textbf{L Da Costa}$^*$, GA Pavliotis. The entropy production of stationary diffusions. \textit{Journal of Physics A: Mathematical and Theoretical}. 2023.

    $\hspace{-8pt}\star$ {\textbf{L Da Costa}}$^*$, N Sajid, T Parr, K Friston, R Smith. Reward Maximisation through Discrete Active inference. \textit{Neural Computation}. 2023.

    MJD Ramstead, DAR Sakthivadivel, C Heins, M Koudahl, B Millidge, \textbf{L Da Costa}, B Klein, K Friston. On Bayesian Mechanics: A Physics of and by Beliefs. \textit{Journal of the Royal Society Interface}. 2023.

    C Heins, \textbf{L Da Costa}. Sparse coupling and Markov blankets: A comment on "How particular is the physics of the free-energy Principle?" by Aguilera, Millidge, Tschantz and Buckley. \textit{Physics of Life Reviews}. 2022.

    T Champion, \textbf{L Da Costa}, H Bowman, M Grzes. Branching Time Active Inference: the theory and its generality. \textit{Neural Networks}. 2022.
    
    N Sajid, F Faccio, {\textbf{L Da Costa}}, T Parr, J Schmidhuber, K Friston. Bayesian brains and the R\'enyi divergence. \textit{Neural Computation}. 2022.
    
    $\hspace{-8pt}\diamond$ \textbf{L Da Costa}$^\dagger$$^*$, P Lanillos$^\dagger$, N Sajid, K Friston, S Khan. How active inference could help revolutionise robotics. \textit{Entropy}. 2022.

    $\hspace{-8pt}\star$
    \textbf{L Da Costa}$^*$, K Friston, C Heins, GA Pavliotis. Bayesian mechanics for stationary processes. \textit{Proceedings of the Royal Society A}. 2021.
    
    K Friston, C Heins, K Ueltzhoeffer, {\textbf{L Da Costa}}, T Parr. Stochastic Chaos and Markov Blankets. \textit{Entropy}. 2021.
    
    K Ueltzhoeffer, {\textbf{L Da Costa}}, D Cialfi, K Friston. A Drive towards Thermodynamic Efficiency for Dissipative Structures in Chemical Reaction Networks. \textit{Entropy}. 2021.
    
    T Parr, {\textbf{L Da Costa}}, C Heins, MJD Ramstead, K Friston. Memory and Markov Blankets. \textit{Entropy}. 2021.
    
    K Friston, {\textbf{L Da Costa}}, T Parr. Some interesting observations on the free-energy principle. \textit{Entropy}. 2021.
    
    K Ueltzhoeffer, {\textbf{L Da Costa}}, K Friston. Variational free-energy, individual fitness, and population dynamics under acute stress. Comment on "Dynamic and thermodynamic models of adaptation" by Alexander N. Gorban et al. \textit{Physics of Life Reviews}. 2021.
    
    T Parr, N Sajid, {\textbf{L Da Costa}}, MB Mirza, K Friston. Generative models for active vision. \textit{Frontiers in Neurobotics}. 2021.
    
    $\hspace{-8pt}\star$ {\textbf{L Da Costa}}$^*$, T Parr, B Sengupta, K Friston. Neural dynamics under active inference: plausibility and efficiency of information processing. \textit{Entropy}. 2021.
    
    K Friston, {\textbf{L Da Costa}}, D Hafner, C Hesp, T Parr. Sophisticated inference. \textit{Neural Computation}. 2021.
    
    $\hspace{-8pt}\star$ {\textbf{L Da Costa}}$^*$, T Parr, N Sajid, S Veselic, V Neacsu, K Friston. Active inference on discrete state-spaces: a synthesis. \textit{Journal of Mathematical Psychology}. 2020.
    
    S Rubin, T Parr, {\textbf{L Da Costa}}, K Friston. Future climates: Markov blankets and active inference in the biosphere. \textit{Journal of the Royal Society Interface}. 2020.
    
    T Parr, {\textbf{L Da Costa}}, K Friston. Markov blankets, information geometry and stochastic thermodynamics. \textit{Philosophical Transactions of the Royal Society A}. 2020.

\section{Awards for PhD work}

My PhD's work has been honoured by several awards including a Doris Chen Award by Imperial College London (2023), a Best Paper Award by the journal Entropy (2023), an Excellence Grant by G-Research (2022), and a Distinguished Student Award by the American Physical Society at the APS March meeting held in Chicago, IL, USA (2022).

\Xchapter{Entropy production}{The entropy production of stationary diffusions}{By Lancelot Da Costa and Grigorios A. Pavliotis\blfootnote{\normalsize \textbf{Adapted from:} L Da Costa, GA Pavliotis. The entropy production of stationary diffusions. \textit{Journal of Physics A: Mathematical
and Theoretical}. 2023}}

\newpage

\section{Abstract}
The entropy production rate is a central quantity in non-equilibrium statistical physics, scoring how far a stochastic process is from being time-reversible. In this chapter, we compute the entropy production of diffusion processes at non-equilibrium steady-state under the condition that the time-reversal of the diffusion remains a diffusion. We start by characterising the entropy production of both discrete and continuous-time Markov processes. We investigate the time-reversal of time-homogeneous stationary diffusions and recall the most general conditions for the reversibility of the diffusion property, which includes hypoelliptic and degenerate diffusions, and locally Lipschitz vector fields. We decompose the drift into its time-reversible and irreversible parts, or equivalently, the generator into symmetric and antisymmetric operators. We show the equivalence with a decomposition of the backward Kolmogorov equation considered in hypocoercivity theory, and a decomposition of the Fokker-Planck equation in GENERIC form. The main result shows that when the time-irreversible part of the drift is in the range of the volatility matrix (almost everywhere) the forward and time-reversed path space measures of the process are mutually equivalent, and evaluates the entropy production. When this does not hold, the measures are mutually singular and the entropy production is infinite. We verify these results using exact numerical simulations of linear diffusions. We illustrate the discrepancy between the entropy production of non-linear diffusions and their numerical simulations in several examples and illustrate how the entropy production can be used for accurate numerical simulation. Finally, we discuss the relationship between time-irreversibility and sampling efficiency, and how we can modify the definition of entropy production to score how far a process is from being generalised reversible.

\textbf{Keywords:} measuring irreversibility; time-reversal; hypoelliptic; degenerate diffusion; Helmholtz decomposition; numerical simulation; Langevin equation; stochastic differential equation; entropy production rate.



\section{Introduction}

The entropy production rate $e_p$ is a central concept in statistical physics. In a nutshell, it is a measure of the time-irreversibility of a stochastic process, that is how much random motion differs statistically speaking as one plays it forward, or backward, in time.

At non-equilibrium steady-state, the $e_p$ is quantified by the relative entropy $\H$ (a.k.a Kullback-Leibler divergence) 
\begin{align*}
    e_p= \frac 1 T \H[\p_{[0,T]} \mid \bar \p_{[0,T]}]
\end{align*}
between the path-wise distributions of the forward and time-reversed processes in some time interval $[0,T]$, denoted by $\p_{[0,T]},\bar \p_{[0,T]}$, respectively.

Physically, the $e_p$ measures the minimal amount of energy needed, per unit of time, to maintain a system at non-equilibrium steady-state. Equivalently, it quantifies the heat dissipated by a physical system at non-equilibrium steady-state per unit of time \cite[p. 86]{jiangMathematicalTheoryNonequilibrium2004}. The second law of thermodynamics for open systems is the non-negativity of the entropy production. 

The $e_p$ plays a crucial role in stochastic thermodynamics. It is the central quantity in the so-called Gallavotti-Cohen fluctuation theorem, which quantifies the probability of entropy decrease along stochastic trajectories \cite{seifertStochasticThermodynamicsFluctuation2012,gallavottiDynamicalEnsemblesNonequilibrium1995a,lebowitzGallavottiCohenTypeSymmetry1999}. 
More recently, entropy production is at the heart of the so-called thermodynamic uncertainty relations, which provide estimates of $e_p$ from observations of the system \cite{dechantMultidimensionalThermodynamicUncertainty2018}.
The $e_p$ is also an important tool in biophysics, as a measure of the metabolic cost of molecular processes, such as molecular motors \cite{seifertStochasticThermodynamicsFluctuation2012}, 
and it was shown empirically that brain states associated with effortful activity correlated with a higher entropy production from neural activity \cite{lynnBrokenDetailedBalance2021}.

In this chapter, we will be primarily concerned with the entropy production of diffusion processes, that is, solutions of stochastic differential equations. Consider an Itô stochastic differential equation (SDE)
\begin{equation*}
    \d x_t=b\left(x_t\right) \d t+\sigma\left( x_t\right) \d w_{t}
\end{equation*}
with drift $b: \R^d \to \R^d$ and volatility $\sigma :\R^d\to \R^{d\times m}$, and $w_t$ a standard Brownian motion on $\R^m$, whose solution is stationary at a probability measure $\mu$ with density $\rho$.
A result known as the Helmholtz decomposition, which is central in non-equilibrium statistical physics \cite{grahamCovariantFormulationNonequilibrium1977,eyinkHydrodynamicsFluctuationsOutside1996,aoPotentialStochasticDifferential2004,pavliotisStochasticProcessesApplications2014} but also in statistical sampling \cite{barpUnifyingCanonicalDescription2021,maCompleteRecipeStochastic2015,barpGeometricMethodsSampling2022,chaudhariStochasticGradientDescent2018}, tells us that we can decompose the drift into time-reversible and time-irreversible parts: $b= b\rev+b\irr$. In particular, $b\irr \rho$ is the stationary probability current, as considered by Nelson \cite{nelsonDynamicalTheoriesBrownian1967}. Jiang and colleagues derived the entropy production rate for such systems under the constraint that the coefficients of the SDE $b, \sigma$ are globally Lipschitz, and the solution uniformly elliptic (i.e., the diffusion matrix field $D = \frac 1 2 \sigma \sigma^\top$ is uniformly positive definite). This takes the form of \cite[Chapter 4]{jiangMathematicalTheoryNonequilibrium2004}:
\begin{align*}
    e_p = \intr b\irr^\top D^{-1} b\irr \rho(x) \d x.
\end{align*}
In this chapter, we extend their work by computing the entropy production for a greater range of diffusion processes, which includes non-elliptic, hypoelliptic and degenerate diffusions, and SDEs driven by locally Lipschitz coefficients. Non-elliptic diffusions are solutions to SDEs whose diffusion matrix field $D= \frac 1 2 \sigma \sigma^\top$ is not positive definite everywhere; this means that there are regions of space in which the random fluctuations cannot drive the process in every possible direction. Depending on how the volatility interacts with the drift (i.e., Hörmander's theorem \cite[Theorem 1.3]{hairerMalliavinProofOrmander2011}), solutions initialised at a point may still have a density---the hypoelliptic case---or not---the degenerate case; prominent examples are underdamped Langevin dynamics and deterministic dynamics, respectively. In our treatment, we only assume that the time-reversal of a diffusion is a diffusion (the necessary and sufficient conditions for which were first established by Millet, Nualart and Sanz \cite{milletIntegrationPartsTime1989a}) and sufficient regularity to apply Girsanov's theorem or the Stroock-Varadhan support theorem.

This extension has become important since many processes that are commonplace in non-equilibrium statistical physics or statistical machine learning are hypoelliptic. For instance, the underdamped and generalised Langevin equations in phase space \cite{pavliotisStochasticProcessesApplications2014}, which model the motion of a particle interacting with a heat bath, and which form the basis of efficient sampling schemes such as Hamiltonian Monte-Carlo \cite{barpGeometricMethodsSampling2022}, or, stochastic gradient descent in deep neural networks \cite{chaudhariStochasticGradientDescent2018}, or, the linear diffusion process with the fastest convergence to stationary state \cite{guillinOptimalLinearDrift2021}, which informs us of the properties of efficient samplers. Much research in statistical sampling has drawn the connection between time-irreversibility and sampling efficiency \cite{ottobreMarkovChainMonte2016,barpGeometricMethodsSampling2022,rey-belletIrreversibleLangevinSamplers2015,hwangAcceleratingDiffusions2005}, so it is informative to understand the amount of time-irreversibility associated with the most efficient samplers. Lastly, it is known that numerically integrating a diffusion processes can modify the amount of irreversibility present in the original dynamic \cite{katsoulakisMeasuringIrreversibilityNumerical2014}. The entropy production rate is thus an important indicator of the fidelity of a numerical simulation, and can serve as a guide to developing sampling schemes that preserve the statistical properties of efficient (hypoelliptic) samplers.

The outline of the chapter and our contribution are detailed below.





\subsection{Chapter outline and contribution}

\paragraph{Section \ref{sec: ep stationary Markov process}:} We give various characterisations and formulas for the $e_p$ of stationary Markov processes in discrete and continuous-time, and recall a crude but general recipe for numerical estimation. 

\paragraph{Section \ref{sec: time reversal of stationary diffusions}:}
We investigate the time-reversal of time-homogeneous diffusions. We give the general conditions under which the time-reversal of a time-homogeneous diffusion remains a diffusion, based on the results of Millet, Nualart and Sanz \cite{milletIntegrationPartsTime1989a}. 
We then recall how the drift vector field of an SDE can be decomposed into time-reversible and irreversible parts. We show that this decomposition is equivalent to a decomposition of the generator of the process into symmetric and antisymmetric operators on a suitable function space. Then we show that this decomposition is equivalent to two other fundamental decompositions in the study of far from equilibrium systems: the decomposition of the backward Kolmogorov equation considered in hypocoercivity theory \cite{duongNonreversibleProcessesGENERIC2021,villaniHypocoercivity2009} and the decomposition of the Fokker-Planck equation considered in the GENERIC formalism \cite{duongNonreversibleProcessesGENERIC2021,ottingerEquilibriumThermodynamics2005}. 

\paragraph{Section \ref{sec: the ep of stationary diffusions}:}
We compute the $e_p$ of stationary diffusion processes under the condition that the time-reversal of the diffusion remains a diffusion. We show that:
\begin{itemize}
    \item \textbf{Section \ref{sec: epr regularity}:} If $b_{\mathrm{irr}}(x) \in \im \sigma(x)$, for $\mu$-almost every $x \in \R^d$\footnote{$\mu$-almost everywhere: This means that the statement holds with probability $1$ when $x$ is distributed according to the probability measure $\mu$.} (in particular for elliptic or time-reversible diffusions). Then the forward and backward path-space measures are mutually equivalent, in other words, the sets of possible trajectories by forward and backward processes are equal
    ---
    and the entropy production equals $e_p = \int b_{\mathrm{irr}}^\top D^- b_{\mathrm{irr}} \d \mu$, where $\cdot^-$ denotes the Moore-Penrose matrix pseudo-inverse. See Theorems \ref{thm: epr regular case simple}, \ref{thm: epr regular case general} for details.
    \item \textbf{Section \ref{sec: ep singularity}:} When the above does not hold, the forward and backward path-space measures are mutually singular
    , in other words, there are trajectories that are taken by the forward process that cannot be taken by the backward process---and vice-versa\footnote{Precisely, two measures are mutual singular if and only if they are not mutually equivalent.}. In particular, the entropy production rate is infinite $e_p = +\infty$. See Theorem \ref{thm: epr singular} for details.
\end{itemize}

\paragraph{Section \ref{sec: illustrating results epr}:} We compute the $e_p$ of various models such as the multivariate Ornstein-Uhlenbeck and the underdamped Langevin process. We numerically simulate and verify the value of $e_p$ when the coefficients are linear. We then discuss how numerical discretisation can influence the value of $e_p$. As examples, we compute and compare the $e_p$ of Euler-Maruyama and BBK \cite{brungerStochasticBoundaryConditions1984,roussetFreeEnergyComputations2010} discretisations of the underdamped Langevin process. We summarise the usefulness of $e_p$ as a measure of the accuracy of numerical schemes in preserving the time-irreversibility properties of the underlying process, and give guidelines, in terms of $e_p$, for developing accurate simulations of underdamped Langevin dynamics.

\paragraph{Section \ref{sec: discussion}:} We give a geometric interpretation of our main results and discuss future perspectives: what this suggests about the relationship between time-irreversibility and mixing in the context of sampling and optimisation, and how we could modify the definition---and computation---of $e_p$ to quantify how a process is far from being time-reversible up to a one-to-one transformation of its phase-space.

\section{The $e_p$ of stationary Markov processes}
\label{sec: ep stationary Markov process}

In this section, $(x_t)_{t \in [0,T]}, T>0$ is a time-homogeneous Markov process on a Polish space $\mathcal X$ with almost surely (a.s.) continuous trajectories.

\begin{definition}[Time reversed process]
\label{def: Time reversed process}
The time reversed process $(\bar x_t)_{t \in [0,T]}$ is defined as $\bar x_t= x_{T-t}$.
\end{definition}
 Note that since the final state of the forward process is the initial state of the backward process this definition makes sense only on finite time intervals.

We define $\left(C([0,T], \mathcal X), d_\infty\right)$ as the space of $\mathcal X$-valued continuous paths endowed with the supremum distance $d_\infty$, defined as
    $d_\infty(f,g)= \sup_{t\in [0,T]}d(f(t),g(t))$,
where $d$ is a choice of distance on the Polish space $\mathcal X$. Naturally, when we later specialise to $\mathcal X= \mathbb R^d$, the supremum distance will be given by the $L^\infty$-norm $\|\cdot\|_\infty$.

\begin{definition}[Path space measure]
\label{def: path space measure}
Each Markov process $(x_t)_{0\leq t \leq T}$ with a.s. continuous trajectories defines probability measure $\p_{[0,T]}$ on the canonical path space $\left(C([0,T], \mathcal X), \mathscr B\right)$, where $\mathscr B$ is the Borel sigma-algebra associated with the supremum distance. This probability measure determines the probability of the process to take any (Borel set of) paths. 
\end{definition}


\begin{remark}[Cadlag Markov processes]
All definitions and results in this section hold more generally for Markov processes with \emph{cadlag} paths (i.e., right continuous with left limits), simply replacing the canonical path-space $\left(C([0,T], \mathcal X), d_\infty\right)$ with the Skorokhod space. We restrict ourselves to processes with continuous paths for simplicity.
\end{remark}

\begin{definition}[Restriction to a sub-interval of time]
Given a path space measure $\p_{[0,T]}$ and times $a<b \in [0,T]$, we define $\p_{[a,b]}$ to be the path space measure describing $\left(x_t\right)_{t \in [ a,b]}$. This is the restriction of $\p_{[0,T]}$ to the sub sigma-algebra
    $\mathscr B_{[a, b]}:= \left\{A \in \mathscr B : \left.A\right|_{[a, b]} \in \text{Borel sigma-algebra on} \left(C([a, b], \mathcal X), d_\infty\right)\right\}$.
\end{definition}

Let $\p$ be the path space measure of the Markov process $x_t$ and $\bar \p$ be that of its time-reversal $\bar x_t$, respectively. We can measure the statistical difference between the forward and time-reversed processes at time $\tau \in [0,T]$ with the \emph{entropy production rate} 
    \begin{align}
    \label{eq: time dependent EPR}
    \lim _{\e \downarrow 0} \frac{1}{\e} \H\left[\p_{[\tau, \tau+\e]}, \bar \p_{[\tau, \tau+\e]}\right],
    \end{align}
where $\H$ is the relative entropy (a.k.a. Kullback-Leibler divergence). This measures the rate at which the forward and backward path space measures differ in the relative entropy sense at time $\tau$.

The following result \cite[Theorem 2.2.4]{jiangMathematicalTheoryNonequilibrium2004} (see also \cite[Theorem 10.4]{varadhanLargeDeviationsApplications1988}) 
shows that the limit exists in stationary and time-homogeneous Markov processes. Obviously, the limit is independent of $\tau$ in this case.

\begin{theorem} 
\label{thm: ep is a rate over time}
Suppose that $\left(x_t\right)_{t \in [ 0,T]}$ is a stationary time-homogeneous Markov process on a Polish space $\mathcal X$ with continuous sample paths. Stationarity implies that we can set the time-horizon $T>0$ of the process to arbitrarily large values. Then the quantity
\begin{align*}
\frac 1 t\H\left[\p_{[\tau,\tau+t]} \mid   \bar \p_{[\tau,\tau+t]}\right] \text{ for all } \tau \in [0, +\infty), t \in (0, +\infty)
\end{align*}
is a constant $\in [0, +\infty]$.
\end{theorem}

This yields the following general definition of entropy production rate for stationary Markov processes: 

\begin{definition}[Entropy production rate of a stationary Markov process]
\label{def: epr}
Let $\left(x_t\right)_{t \in [ 0,T]}$ be a stationary time-homogeneous Markov process. Stationarity implies that we can set the time-horizon $T>0$ to be arbitrarily large. For such processes, the entropy production rate is a constant $e_p \in [0, +\infty]$ defined as
\begin{equation}
\label{eq: def epr}
    e_p := \frac{1}{t}\H\left[\p_{[0,t]} \mid  \bar \p_{[0,t]} \right]
\end{equation}
for any $t \in (0,+\infty)$. $e_p$ scores the amount to which the forward and time-reversed processes differ per unit of time. In particular, $Te_p$ is the total entropy production in a time interval of length $T$. Note that, in the literature, the $e_p$ is often defined as $e_p = \lim _{t \to +\infty}\frac{1}{t}\H\left[\p_{[0,t]} \mid  \bar \p_{[0,t]} \right]$, e.g., \cite[Definition 4.1.1]{jiangMathematicalTheoryNonequilibrium2004}; this is just \eqref{eq: def epr} in the limit of large $t$. However, Theorem \ref{thm: ep is a rate over time} showed us that \eqref{eq: def epr} is constant w.r.t. $t\in (0,+\infty)$ so we do not need to restrict ourselves to defining the $e_p$ as \eqref{eq: def epr} in the limit of large $t$. This added generality will be very helpful to compute $e_p$ later, by exploiting the fact that \eqref{eq: def epr} is often more easily analysed in the regime of finite or small $t$.
\end{definition}

\begin{remark}[Physical relevance of Definition \ref{def: epr}]
\label{rem: physical relevance}
    In some stationary processes (e.g., Hamiltonian systems), physicists define entropy production as Definition \ref{def: epr} with an additional operator applied to the path space measure of the time-reversed process; that is,
    \begin{equation}
    \label{eq: rem def gen epr}
    e_p^{\mathrm{gen}, \theta} := \lim _{\e \downarrow 0} \frac{1}{\e} \H\left[\p_{[0, \e]}, \theta_\#\bar \p_{[0, \e]}\right],
    \end{equation}
    where $\theta_\#$ is the pushforward operator associated to an involution of phase-space $\theta$ (e.g., the momentum flip \cite{vanvuUncertaintyRelationsUnderdamped2019,eckmannEntropyProductionNonlinear1999,spinneyNonequilibriumThermodynamicsStochastic2012}) that leaves the stationary distribution invariant\footnote{This generalised definition of entropy production is taken as a limit of $\e \downarrow 0$ analogously to \eqref{eq: time dependent EPR} to capture the fact that we are modelling a \emph{rate}. We cannot, a priori state that the expression is constant for any $\e \in (0,+\infty)$, as in Definition \ref{def: epr}, since we do not know whether a result analogous to Theorem \ref{thm: ep is a rate over time} holds in this case.}. In this article, we will refer to \eqref{eq: def epr} as \emph{entropy production} and to \eqref{eq: rem def gen epr} as \emph{generalised entropy production}, and proceed to derive general results for \eqref{eq: def epr}. The results we derive are informative of the process and applicable independently of whether \eqref{eq: def epr} is the physically meaningful definition of entropy production; yet, physicists looking to interpret these results should bear in mind that they are physically informative about entropy production insofar as Definition \ref{def: epr} is physically meaningful for the system at hand. We will briefly revisit generalised entropy production in the discussion (Section~\ref{sec: discussion generalised non-reversibility}).
\end{remark}


\begin{proposition}
\label{prop: aggregating local ep}
Let $\left(x_t\right)_{t \in [ 0,T]}$ be a time-homogeneous Markov process on a Polish space $\mathcal X$, stationary at the probability measure $\mu$. Then the entropy production rate equals
\begin{equation*}
    e_p = \frac 1 t \E_{x\sim \mu}\left[\H\left[\p^{x}_{[0,t]}\mid \bar  \p^{x}_{[0,t]}\right]\right]
\end{equation*}
for any $t\in (0, +\infty)$.
\end{proposition}

A proof is provided in Section \ref{app: aggregating local ep}.

\begin{notation}
\label{notation: path space measure deterministic initial condition}
By $\p^{x}$ we mean the path space measure of the 
process initialised (possibly out of stationarity) at a deterministic initial condition $x_0=x \in \mathcal X$.  
\end{notation}

\begin{proposition}[$e_p$ in terms of transition kernels]
\label{prop: epr transition kernels}
Let $\left(x_t\right)_{t \in [ 0,T]}$ be a time-homogeneous Markov process on a Polish space $\mathcal X$, stationary at the probability measure $\mu$. Denote by $p_t(dy,x)$ the transition kernels of the Markov semigroup, and by $\bar p_t(dy,x)$ those of the time-reversed process. Then the entropy production rate equals
\begin{equation*}
    e_p = \lim_{\e \downarrow 0} \frac{1}{\e} \E_{x\sim \mu}\left[\H\left[p_\e(\cdot,x)\mid\bar p_\e(\cdot,x)\right]\right].
\end{equation*}
\end{proposition}

The fact that the time-reversed process possesses transition kernels holds as it is also a stationary Markov process \cite[p. 113]{jiangMathematicalTheoryNonequilibrium2004}. A proof of Proposition \ref{prop: epr transition kernels} is provided in Section \ref{app: epr transition kernels}.

\subsection{The $e_p$ of numerical simulations}
\label{sec: ep of numerical schemes theory}

Proposition \ref{prop: epr transition kernels} entails a formula for the $e_p$ of Markov processes in discrete time:
\begin{definition}
\label{def: EPR numerical scheme}
The entropy production rate of a discrete-time Markov process with time-step $\e$ equals
\begin{equation}
\label{eq: def EPR numerical scheme}
    e_p^{\text{NS}}(\e) = \frac{1}{\e} \E_{x\sim \tilde \mu}\int p_\e(y,x) \log \frac{p_\e(y,x)}{p_\e(x,y)} \d y,
\end{equation}
where $\tilde \mu$ is the invariant measure of the process.
\end{definition}

This definition is useful, for example, to quantify the entropy production of numerical simulations of a stochastic process \cite{katsoulakisMeasuringIrreversibilityNumerical2014}.
In particular, it suggests a simple numerical estimator of the entropy production rate for numerical simulations (at stationarity). Consider a small $\delta$ (e.g., $\delta$ is the time-step of the numerical discretisation). Given samples from the process at $\delta$ time intervals, discretise the state-space into a finite partition $U_1, \ldots, U_n$, and approximate $\p_{[0,\delta]}$ and $\bar \p_{[0,\delta]}$ by the empirical transition probabilities $p_{i \to j}$ between $U_i, U_j$ from time $0$ to $\delta$.

\begin{align*}
    e_p^{\text{NS}}= \lim_{\e \to 0} \frac{1}{\e} \H\left[\p_{[0,\e]}\mid\bar \p_{[0,\e]}\right] \approx \frac{1}{\delta} \H\left[\p_{[0,\delta]}\mid\bar \p_{[0,\delta]}\right] \approx \frac{1}{\delta} \sum_{i,j}p_{i \to j} \log \frac{p_{i \to j}}{p_{j \to i}}. 
\end{align*}

Note that this method measures the entropy production rate of the numerical discretisation as opposed to that of the continuous process. This typically produces results close to $e_p$, but does not necessarily converge to $e_p$ in the continuum limit $\delta \to 0$ of the numerical discretisation. Indeed \cite{katsoulakisMeasuringIrreversibilityNumerical2014} showed that numerical discretisations can break detailed balance, so that the continuum limit of the numerical discretisation can differ from the initial process. Thus one should choose numerical schemes carefully when preserving the entropy production rate of a process is important. We will return to this in Section \ref{sec: illustrating results epr}.


\section{Time reversal of stationary diffusions}
\label{sec: time reversal of stationary diffusions}


We now specialise to diffusion processes in $\R^d$. These are Markov processes with an infinitessimal generator that is a second order linear operator without a constant part \cite[Definition 1.11.1]{bakryAnalysisGeometryMarkov2014}, 
which entails almost surely (a.s.) continuous sample paths. 
Conveniently, diffusion processes are usually expressible as solutions to stochastic differential equations.
From now on, we consider an Itô stochastic differential equation
\begin{equation}
\label{eq: Ito SDE}
    \d x_t=b\left(x_t\right) \d t+\sigma\left( x_t\right) \d w_{t}
\end{equation}
with drift $b:\mathbb{R}^{d} \rightarrow \mathbb{R}^{d}$ and volatility $\sigma: \mathbb{R}^{d}  \rightarrow \R^{d \times m}$, and $w_t$ a standard Brownian motion on $\R^m$.

\begin{notation}
Let $D=\sigma\sigma^\top/2 \in \mathbb{R}^{{d \times d}}$ be the \textit{diffusion tensor}. Denote by $\|\cdot\|$ the Euclidean distance, and, for matrices
$$
\|\sigma\|^2:=\sum_{i=1}^{d} \sum_{j=1}^{n}\left|\sigma_{i j}\right|^{2}.
$$
Throughout, $\nabla$ and $\nabla \cdot$ are the gradient and the divergence in the distributional sense. We operationally define the divergence of a matrix field $Q : \R^d \to \mathbb{R}^{{d \times d}}$ by $(\nabla \cdot Q)_i:=\sum_{j=1}^d \partial_j Q_{i j}$ for $0\leq i \leq d$.
We will denote by $\mu$ the stationary probability measure of the process $x_t$ and by $\rho$ its density with respect to the Lebesgue measure, i.e., $\mu(\d x)= \rho(x) \d x$ (assuming they exist).
\end{notation}

\subsection{On the time-reversibility of the diffusion property}

There is a substantial literature studying the time-reversal of diffusion processes. In general, the time-reversal of a diffusion need not be a diffusion \cite{milletIntegrationPartsTime1989a}, 
but Haussman and Pardoux showed that the diffusion property is preserved under some mild regularity conditions on the diffusion process \cite{haussmannTimeReversalDiffusions1986}. A few years later Millet, Nualart, Sanz derived necessary and sufficient conditions for the time-reversal of a diffusion to be a diffusion \cite[Theorem 2.2 \& p. 220]{milletIntegrationPartsTime1989a}. We provide these conditions here, with a proof of a different nature that exploits the existence of a stationary distribution.

\begin{lemma}[Conditions for the reversibility of the diffusion property]
\label{lemma: reversibility of the diffusion property}
Let an Itô SDE \eqref{eq: Ito SDE} with locally bounded, Lebesgue measurable coefficients $b: \mathbb{R}^{d} \rightarrow \mathbb{R}^{d}, \sigma : \mathbb{R}^{d}  \rightarrow \R^{d \times m}$. Consider a strong solution $(x_t)_{t\in [0,T]}$, i.e., a process satisfying
\begin{align*}
    x_t=x_0+\int_0^t b\left(x_s\right) d s+\int_0^t \sigma\left(x_s\right) d w_s,
\end{align*}
and assume that it is stationary with respect to a probability measure $\mu$ with density $\rho$. Consider the time-reversed stationary process $(\bar x_t)_{t\in [0,T]}$. Then, the following are equivalent:
\begin{itemize}
    \item $(\bar x_t)_{t\in [0,T]}$ is a Markov diffusion process. 
    \item The distributional derivative $\nabla \cdot (D \rho)$ is a function, which is then necessarily in $L^1\loc(\R^d, \R^d)$. 
\end{itemize}
\end{lemma}
A proof is provided in Section \ref{app: reversibility of the diffusion property}.

\subsection{Setup for the time-reversal of diffusions}

From now on, we will work under the assumption that the time-reversal of the diffusion is a diffusion. We assume that:
\begin{assumption}
\label{ass: coefficients time reversal}
\begin{enumerate}
    \item \label{ass: locally bounded and Lipschitz} The coefficients of the SDE \eqref{eq: Ito SDE} $b, \sigma$ are 
    \emph{locally Lipschitz continuous}. In other words, $\forall x \in \R^d, \exists r >0, k > 0 \st \forall y\in \R^d:$
    \begin{align*}
         \|x-y \|<r \Rightarrow \|b( x)-b( y)\|+\|\sigma(x)-\sigma(y)\|\leq k\|x-y\|,
    \end{align*}
    \item \label{ass: global solution} The solution $x_t$ to \eqref{eq: Ito SDE} is defined globally up to time $T>0$. Sufficient conditions in terms of the drift and volatility for Itô SDEs are given in Theorem \cite[Theorem 3.1.1]{prevotConciseCourseStochastic2008}. 
\end{enumerate}
\end{assumption}
Assumption \ref{ass: coefficients time reversal}.\ref{ass: locally bounded and Lipschitz} ensures the existence and uniqueness of strong solutions locally in time \cite[Chapter IV Theorem 3.1]{watanabeStochasticDifferentialEquations1981}, while Assumption \ref{ass: coefficients time reversal}.\ref{ass: global solution} ensures that this solution exists globally in time (i.e., non-explosiveness). Altogether, Assumption \ref{ass: coefficients time reversal} ensures that the SDE \eqref{eq: Ito SDE} unambiguously defines a diffusion process. 


Furthermore, we assume some regularity on the stationary distribution of the process.

\begin{assumption}
\label{ass: steady-state time reversal}
\begin{enumerate}
    \item $(x_t)_{t\in [0,T]}$ is stationary at a probability distribution $\mu$, with density $\rho$ with respect to the Lebesgue measure, i.e., $\mu(\d x)= \rho(x) \d x$.
\end{enumerate}
Then, $\rho \in L^1(\R^d)$ and, under local boundedness of the diffusion tensor (e.g., Assumption \ref{ass: coefficients time reversal}), $D\rho \in L^1\loc(\R^d, \R^{d\times d})$. Thus, we can define the distributional derivative $\nabla \cdot (D\rho)$. We assume that:
\begin{enumerate}
\setcounter{enumi}{1}
    \item $\nabla \cdot (D\rho) \in L^1\loc(\R^d, \R^{d})$, i.e., the distributional derivative $\nabla \cdot (D\rho)$ is a function.
\end{enumerate}
\end{assumption}

Assumption \ref{ass: steady-state time reversal} ensures that the time-reversal of the diffusion process remains a diffusion process, as demonstrated in Lemma \ref{lemma: reversibility of the diffusion property}.

\subsection{The time reversed diffusion}

Now that we know sufficient and necessary conditions for the time-reversibility of the diffusion property, we proceed to identify the drift and volatility of the time-reversed diffusion. This was originally done by Hausmann and Pardoux \cite[Theorem 2.1]{haussmannTimeReversalDiffusions1986}, and then by Millet, Nualart, Sanz under slightly different conditions \cite[Theorems 2.3 or 3.3]{milletIntegrationPartsTime1989a}. Inspired by these, we provide a different proof, which applies to stationary diffusions with locally Lipschitz coefficients.

\begin{theorem}[Characterisation of time-reversal of diffusion]
\label{thm: time reversal of diffusions}
Let an Itô SDE \eqref{eq: Ito SDE} with coefficients satisfying Assumption \ref{ass: coefficients time reversal}. Assume that the solution $(x_t)_{t\in [0,T]}$ is stationary with respect to a density $\rho$ satisfying Assumption \ref{ass: steady-state time reversal}. 
Then, the time-reversed process $(\bar x_t)_{t\in [0,T]}$ is a Markov diffusion process, stationary at the density $\rho$, with drift
\begin{align}
\label{eq: time reversed drift}
\bar b(x)=& \begin{cases}
        -b(x) +2\rho^{-1} \nabla \cdot \left(D \rho\right)(x)\quad \text{when } \rho(x) > 0,\\
        -b(x) \quad \text{when } \rho(x) = 0.
    \end{cases}
\end{align}
and diffusion $\bar D = D$. Furthermore, any such stationary diffusion process induces the path space measure of the time-reversed process $\bar \p_{[0,T]}$. 
\end{theorem}

A proof is provided in Section \ref{app: time reversed diffusion}. 
Similar time-reversal theorems exist in various settings: for more singular coefficients on the torus \cite{quastelTimeReversalDegenerate2002}, under (entropic) regularity conditions on the forward path space measure \cite{follmerTimeReversalWiener1986,cattiauxTimeReversalDiffusion2021}, for infinite-dimensional diffusions \cite{belopolskayaTimeReversalDiffusion2001,milletTimeReversalInfinitedimensional1989,follmerTimeReversalInfinitedimensional1986}, for diffusions on open time-intervals \cite{nagasawaDiffusionProcessesOpen1996}, or with boundary conditions \cite{cattiauxTimeReversalDiffusion1988}. Furthermore, we did not specify the Brownian motion driving the time-reversed diffusion but this one was identified in \cite[Remark 2.5]{pardouxTimereversalDiffusionProcesses1985}.

We illustrate the time-reversal of diffusions with a well-known example:

\begin{example}[Time reversal of underdamped Langevin dynamics]
\label{eg: Time reversal of underdamped Langevin dynamics}
Underdamped Langevin dynamics is an important model in statistical physics and sampling \cite{pavliotisStochasticProcessesApplications2014,roussetFreeEnergyComputations2010,maThereAnalogNesterov2021}. Consider a Hamiltonian $H(q,p)$ function of positions $q\in \R^n$ and momenta $ p \in \R^n$. We assume that the Hamiltonian has the form
\begin{align*}
   H(q,p) = V(q) + \frac 1 2 p^\top M^{-1} p,
\end{align*}
for some smooth potential function $V: \R^d \to \R$ and diagonal mass matrix $M \in \R^{d\times d}$. The underdamped Langevin process is given by the solution to the SDE \cite[eq 2.41]{roussetFreeEnergyComputations2010}
\begin{align}
\label{eq: underdamped Langevin dynamics}
\begin{cases}
\d q_{t} = M^{-1}p_{t} \d t \\
\d p_{t} =-\nabla V\left(q_{t}\right) \d t-\gamma M^{-1} p_{t} \d t+\sqrt{2 \gamma \beta^{-1}} \d w_{t}
\end{cases}
\end{align}
for some friction, and inverse temperature coefficients $\gamma, \beta >0$. The stationary density, assuming it exists, is the canonical density \cite[Section 2.2.3.1]{roussetFreeEnergyComputations2010}
\begin{align}
\label{eq: canonical density underdamped}
    \rho(q,p)\propto e^{-\beta H(q,p)}= e^{-\beta V(q)-\frac{\beta}{2} p^\top M^{-1} p}.
\end{align}
Thus, the time-reversal of the stationary process (Theorem \ref{thm: time reversal of diffusions}) 
is a weak solution to the SDE
\begin{align*}
&\begin{cases}
\d \bar q_{t} = -M^{-1}\bar p_{t} \d t \\
\d \bar p_{t} =\nabla V\left(\bar q_{t}\right) \d t-\gamma M^{-1} \bar p_{t} \d t+\sqrt{2 \gamma \beta^{-1}} \d  w_{t}.
\end{cases}
\end{align*}
Letting $\hat p_{t} = -\bar p_{t}$, 
the tuple $(\bar q_t,\hat p_t)$ solves the same SDE as $(q_t,p_t)$ but with a different Brownian motion $\hat w_t$
    \begin{equation*}
    \begin{cases}
    \d \bar q_{t} = M^{-1}\hat p_{t} \d t \\
    \d \hat p_{t} =-\nabla V\left(\bar q_{t}\right) \d t-\gamma M^{-1} \hat p_{t} \d t+\sqrt{2 \gamma \beta^{-1}} \d \hat w_{t}.
    \end{cases}
    \end{equation*}
Since path space measures are agnostic to changes in the Brownian motion, 
this leads to the statement that time-reversal equals momentum reversal in underdamped Langevin dynamics (with equality in law, i.e., in the sense of path space measures)
    \begin{align*}
        (\bar q_t, \bar p_t)_{t\in [0,T]}=(\bar q_t,-\hat p_t)_{t\in [0,T]}\stackrel{\ell}{=} (q_t, -p_t)_{t\in [0,T]}.
    \end{align*}
In other words, we have $\p_{[0,T]} = \theta_\# \bar \p_{[0,T]}$ where $\theta(q,p)=(q,-p)$ is the momentum flip transformation in phase space. 
\end{example}

\subsection{The Helmholtz decomposition}
\label{sec: helmholtz decomposition}

Armed with the time-reversal of diffusions we proceed to decompose the SDE into its time-reversible and time-irreversible components.
This decomposition is called the Helmholtz decomposition because it can be obtained geometrically by decomposing the drift vector field $b$ into horizontal $b_{\mathrm{irr}}$ (time-irreversible, conservative) and vertical $b_{\mathrm{rev}}$ (time-reversible, non-conservative) components with respect to the stationary density \cite{barpUnifyingCanonicalDescription2021}. These vector fields are called horizontal and vertical, respectively, because the first flows along the contours of the stationary density, while the second ascends the landscape of the stationary density (see the schematic in the upper-left panel of Figure \ref{fig: Helmholtz decomposition}). For our purposes, we provide a self-contained, non-geometric proof of the Helmholtz decomposition in Section \ref{app: helmholtz decomposition}.

\begin{proposition}[Helmholtz decomposition]
\label{prop: helmholtz decomposition}
Consider the solution $(x_t)_{t\in [0,T]}$ of the Itô SDE \eqref{eq: Ito SDE} with coefficients satisfying Assumption \ref{ass: coefficients time reversal}. Let a probability density $\rho$ satisfying $\nabla \cdot(D\rho) \in L^1\loc (\R^d, \R^d)$. Then, the following are equivalent:
\begin{enumerate}
    \item The density $\rho$ is stationary for $(x_t)_{t\in [0,T]}$.
    \item We can write the drift as
    \begin{equation}
    \label{eq: Helmholtz decomposition of drift}
    \begin{split}
        b &= b_{\mathrm{rev}}+ b_{\mathrm{irr}}\\
        b_{\mathrm{rev}} &=
        \begin{cases}
        D\nabla \log \rho + \nabla \cdot D \quad \text{if } \rho(x) > 0\\
        0 \quad \text{if } \rho(x) = 0
        \end{cases}\\
        \nabla\cdot (b_{\mathrm{irr}} \rho)&= 0. 
    \end{split}
    \end{equation}
\end{enumerate}
Furthermore, $b_{\mathrm{rev}}$ is \emph{time-reversible}, while $b_{\mathrm{irr}}$ is \emph{time-irreversible}, i.e.,
    \begin{align*}
        b = b_{\mathrm{rev}}+ b_{\mathrm{irr}},\quad  \bar b= b_{\mathrm{rev}}- b_{\mathrm{irr}}.
    \end{align*}
\end{proposition}


The fundamental importance of the Helmholtz decomposition was originally recognised in the context of non-equilibrium thermodynamics by Graham in 1977 \cite{grahamCovariantFormulationNonequilibrium1977}, but its inception in this field dates from much earlier: for instance, the divergence free vector field $b_{\mathrm{irr}} \rho$ is precisely the stationary probability current or flux introduced by Nelson in 1967 \cite{nelsonDynamicalTheoriesBrownian1967}. More recently, the decomposition has recurrently been used in non-equilibrium statistical physics \cite{eyinkHydrodynamicsFluctuationsOutside1996,aoPotentialStochasticDifferential2004,qianDecompositionIrreversibleDiffusion2013,pavliotisStochasticProcessesApplications2014,fristonStochasticChaosMarkov2021,dacostaBayesianMechanicsStationary2021a,yangPotentialsContinuousMarkov2021}, and in statistical machine learning as the basis of Monte-Carlo sampling schemes \cite{maCompleteRecipeStochastic2015,barpUnifyingCanonicalDescription2021,lelievreOptimalNonreversibleLinear2013,pavliotisStochasticProcessesApplications2014}.

\begin{figure}[t!]
    \centering
    \includegraphics[width= 0.45\textwidth]{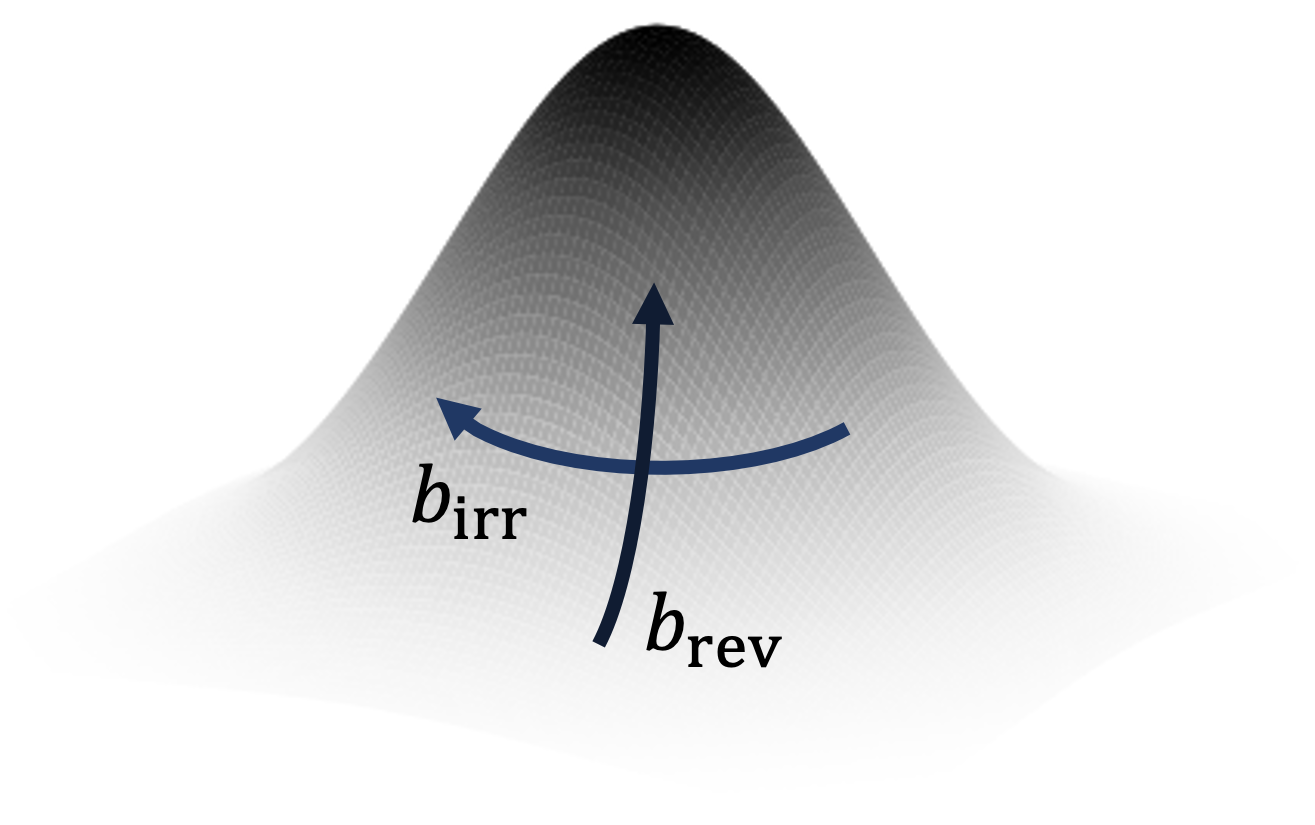}
    \includegraphics[width= 0.45\textwidth]{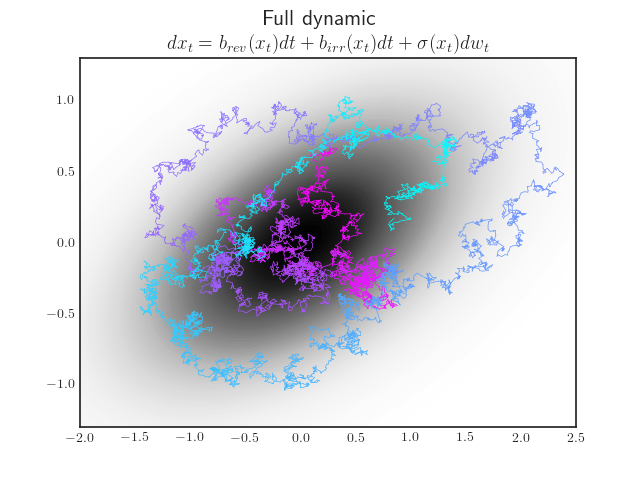}
    \includegraphics[width= 0.45\textwidth]{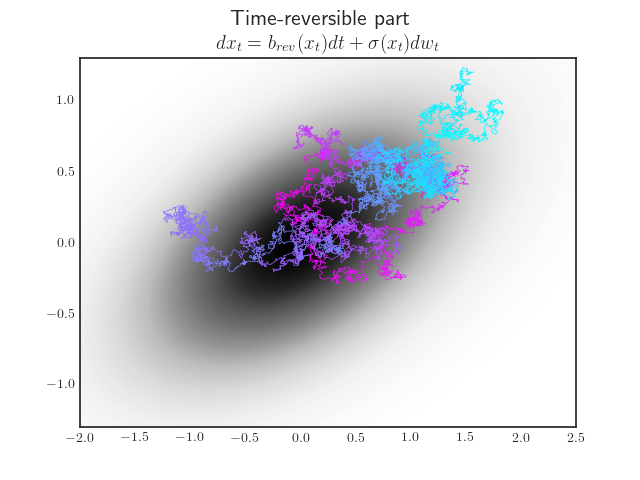}
    \includegraphics[width= 0.45\textwidth]{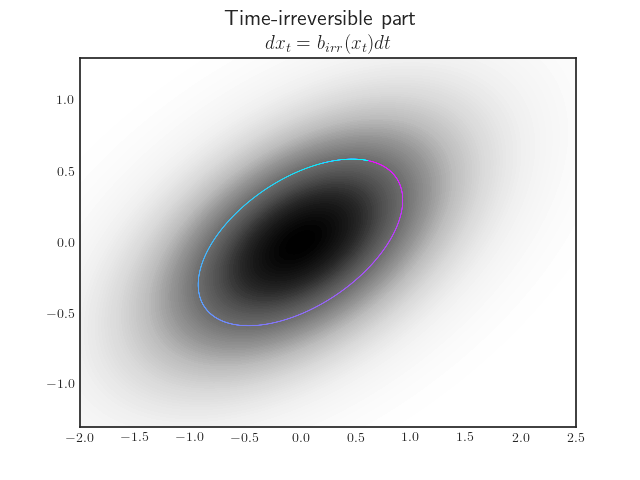}
    \caption[Helmholtz decomposition]{\footnotesize\textbf{Helmholtz decomposition.} The upper left panel illustrates the Helmholtz decomposition of the drift into time-reversible and time-irreversible parts: the time-reversible part of the drift flows towards the peak of the stationary density, while the time-irreversible part flows along its contours. The upper right panel shows a sample trajectory of a two-dimensional diffusion process stationary at a Gaussian distribution. The lower panels plot sample paths of the time-reversible (lower left) and time-irreversible (lower right) parts of the dynamic. Purely conservative dynamics (lower right) are reminiscent of the trajectories of massive bodies (e.g., planets) whose random fluctuations are negligible, as in Newtonian mechanics. Together, the lower panels illustrate time-irreversibility: If we were to reverse time, 
    the trajectories of the time-reversible process would be statistically identical, while the trajectories of the time-irreversible process be distinguishable by flow, say, clockwise instead of counterclockwise. The full process (upper right) is a combination of both time-reversible and time-irreversible dynamics. The time-irreversible part defines a non-equilibrium steady-state and induces its characteristic wandering, cyclic behaviour.}
    \label{fig: Helmholtz decomposition}
\end{figure}

\begin{remark}[Probabilistic reversibility]
Here, time-reversible means reversibility in a probabilistic sense; that is, invariance under time reversal, also known as detailed balance \cite[Proposition 3.3.4]{jiangMathematicalTheoryNonequilibrium2004}. 
Probabilistic reversibility often leads to the non-conversation of quantities like the potential $-\log \rho(x)$. 
For example, the identity $\nabla\cdot (b_{\mathrm{irr}} \rho)= 0$ implies that the time-irreversible drift $b_{\mathrm{irr}}$ flows along the contours of the probability density; in other words, the probability density and the potential are conserved along the time-irreversible vector field. In contrast, none of them are conserved when flowing along the time-reversible vector field $b_{\mathrm{rev}}$. See Figure \ref{fig: Helmholtz decomposition} for an illustration.


\end{remark}

\begin{remark}[Stratonovich formulation of Helmholtz decomposition]
\label{rem: stratonovich Helmholtz}
There exists an equivalent decomposition of the drift of Stratonovich SDEs into time-reversible and irreversible parts. Assuming that $\sigma$ is differentiable, we can rewrite the Itô SDE \eqref{eq: Ito SDE} into its equivalent Stratonovich SDE 
\begin{align*}
    dx_t = b^s(x_t) dt + \sigma(x_t) \circ dw_t.
\end{align*}
where $b^s =b- \iota$ and $\iota$ is the Itô to Stratonovich correction \cite[eq. 3.31]{pavliotisStochasticProcessesApplications2014}. Note that the correction is time-reversible. It follows that
\begin{align}
\label{eq: irr drift equals ito and stratonovich}
    b^s_{\mathrm{irr}}&=b_{\mathrm{irr}},
\end{align}
and for $x \st \rho(x) >0$,
\begin{align*}
    b^s_{\mathrm{rev}}(x)&= (b_{\mathrm{rev}}-\iota)(x)=  D\nabla \log \rho(x) + \frac{1}{2} \sigma \nabla \cdot \sigma^\top(x).
\end{align*}
In particular, for $x \st \rho(x) >0$,
\begin{align}
\label{eq: stratonovich drift in image of vol equiv irr drift in image of vol}
    b^s(x) \in \im\sigma(x) &\iff b^s_{\mathrm{irr}}(x) \in \im\sigma(x)
\end{align}
as $b^s_{\mathrm{rev}}(x)\in \im\sigma(x)$. For diffusions driven by additive noise, Itô and Stratonovich formulations coincide $b^s=b$. Thus, we conclude
\begin{align}
    \sigma \text{ is constant} &\Rightarrow b(x) = \begin{cases}
        D\nabla \log \rho(x) + b_{\mathrm{irr}}(x) \quad \text{if } \rho(x) > 0\\
        b_{\mathrm{irr}}(x) \quad \text{if } \rho(x) = 0
        \end{cases} \nonumber\\
        &\Rightarrow \left( b(x) \in \im\sigma(x) \iff b_{\mathrm{irr}}(x) \in \im\sigma(x)\right) \label{eq: drift image of volatility additive noise}
\end{align}
These identities will be useful to compute the entropy production rate later on.
\end{remark}

The time-irreversible part of the drift often takes a simple form:

\begin{proposition}[Characterisation of time-irreversible drift]
\label{prop: characterisation of irreversible drift}
Consider a smooth, strictly positive probability density $\rho$ and an arbitrary smooth vector field $b_{\mathrm{irr}}$. Then
\begin{align*}
    \nabla \cdot( b_{\mathrm{irr}}\rho) = 0 \iff 
    b_{\mathrm{irr}} = Q \nabla \log \rho + \nabla \cdot Q
\end{align*}
where $Q=-Q^\top $ is a smooth antisymmetric matrix field.
\end{proposition}

A proof is provided in Section \ref{app: helmholtz decomposition}. We conclude this section by unpacking the Helmholtz decomposition of underdamped Langevin dynamics.

\begin{example}[Helmholtz decomposition of underdamped Langevin]
\label{eg: helmholtz decomp Langevin}
Following Example \ref{eg: Time reversal of underdamped Langevin dynamics}, it is straightforward to decompose underdamped Langevin dynamics into its time-irreversible and time-reversible parts. Indeed we just need to identify the parts of the drift whose sign changes, and remains invariant, under time reversal:

\begin{align*}
b_{\mathrm{rev}}(q,p) = \begin{bmatrix}
    0 \\
    -\gamma M^{-1} p
\end{bmatrix}, \quad b_{\mathrm{irr}}(q,p) &= \begin{bmatrix}
    M^{-1}p  \\
    -\nabla V\left(q\right)
\end{bmatrix}.   
\end{align*}
We can rewrite these in canonical form recalling the gradient of the stationary density \eqref{eq: canonical density underdamped}
\begin{align*}
    &b_{\mathrm{rev}}(q,p)= D \nabla \log \rho(q,p) , \quad b_{\mathrm{irr}}(q,p)= Q \nabla \log \rho(q,p)\\
    \nabla \log  \rho(q,p)&= -\beta \begin{bmatrix}
         \nabla V(q)\\
         M^{-1}p
    \end{bmatrix},
    \quad D=\begin{bmatrix}
        0 &0\\
        0& \gamma \beta^{-1}\id_n
    \end{bmatrix}
 , \quad Q=\beta^{-1}\begin{bmatrix}
        0 &-\id_n\\
        \id_n & 0
    \end{bmatrix}.
\end{align*}
Clearly, the time-irreversible part of the process $\d [q_t, p_t] = b_{\mathrm{irr}}(q_t,p_t)\d t$ is a Hamiltonian dynamic that preserves the energy (i.e., the Hamiltonian), while the time-reversible part is a reversible Ornstein-Uhlenbeck process. Example trajectories of the time-irreversible trajectory are exemplified in Figure \ref{fig: Helmholtz decomposition} (bottom right).
\end{example}



\subsection{Multiple perspectives on the Helmholtz decomposition}

The Helmholtz decomposition is a cornerstone of the theory of diffusion processes. In addition to being a geometric decomposition of the drift \cite{barpUnifyingCanonicalDescription2021}, it is, equivalently, a time-reversible and irreversible decomposition of the SDE \eqref{eq: Helmholtz decomp SDE}, of the generator and the (backward and forward) Kolmogorov PDEs describing the process. Briefly, the Helmholtz decomposition is equivalent to a functional analytic decomposition of the generator into symmetric and antisymmetric operators in a suitable function space. This corresponds to a decomposition of the backward Kolmogorov equation---which determines the evolution of (macroscopic) observables under the process---into a conservative and a dissipative flow. This decomposition can be used as a starting point to quantify the speed of convergence of the process to its stationary state from arbitrary initial conditions using hypocoercivity theory \cite{villaniHypocoercivity2009}. The same goes for the Fokker-Planck equation, which can be decomposed into a dissipative gradient flow, and a flow that is conservative in virtue of being orthogonal to the gradient flow in a suitable function space. This casts the Fokker-Planck equation in GENERIC form (General Equations for Non-Equilibrium Reversible-Irreversible Coupling), a general framework for analysing dynamical systems arising in non-equilibrium statistical physics  \cite{ottingerEquilibriumThermodynamics2005,duongNonreversibleProcessesGENERIC2021}.

Below we outline these different equivalent perspectives. 
This section is provided for independent interest but will not be used to derive our main results on entropy production; \emph{you may conveniently skip it on a first reading}.

\subsubsection{Helmholtz decomposition of the SDE}

Proposition \ref{prop: helmholtz decomposition} is equivalent to a Helmholtz decomposition of the SDE into its time-reversible and time-irreversible parts, noting that the volatility is invariant under time-reversal (Theorem \ref{thm: time reversal of diffusions})
\begin{align}
\label{eq: Helmholtz decomp SDE}
\mathrm{d} x_t= \underbrace{b_{\mathrm{irr}}\left(x_t\right) \mathrm{d} t}_{\text{Time-irreversible}} +\underbrace{b_{\mathrm{rev}}\left(x_t\right) \mathrm{d} t+\sigma\left(x_t\right) \mathrm{d} w_t}_{\text{Time-reversible}}.
\end{align}
Figure \ref{fig: Helmholtz decomposition} illustrates this decomposition with simulations.

\subsubsection{Helmholtz decomposition of the infinitesimal generator}
\label{sec: helmholtz decomp inf gen}

Following the differential geometric viewpoint, a deterministic flow---namely, a vector field $b$---is given by a first order differential operator $b \cdot \nabla$. Similarly, a stochastic flow given by a diffusion---namely, a vector field $b$ and a diffusion tensor $D$---is characterised by a second order differential operator
\begin{align}
\label{eq: generic generator}
    \L= b \cdot \nabla + D\nabla \cdot \nabla,
\end{align}
known as the generator. Note that the first order part is the deterministic flow given by the drift while the second order part is the stochastic flow determined by the diffusion. More precisely, the generator of a diffusion process solving the SDE \eqref{eq: Ito SDE} under Assumptions \ref{ass: coefficients time reversal} and \ref{ass: steady-state time reversal} is a linear, unbounded operator defined as
\begin{align}
&\L: C_c^\infty(\R^d) \subset  \dom \L \subset  L^p_\mu(\R^d) \to L^p_\mu(\R^d), 1 \leq p \leq \infty, \quad 
     \L f(y) := \lim_{t\downarrow 0}\frac 1 t\E[f(x_t)-f(y) \mid x_0=y], \label{eq: def generator}\\
     &f \in \dom \L = \left\{f \in L^p_\mu(\R^d) \mid \exists g \in L^p_\mu(\R^d) \st \frac 1 t\E[f(x_t)-f(y) \mid x_0=y] \xrightarrow{t \downarrow 0} g(y) \text{ in }  L^p_\mu(\R^d)\right\} \nonumber
\end{align}
Diffusions are among the simplest and most canonical Markov processes because they are characterised by generators that are second order differential operators (with no constant part). 
Indeed, starting from \eqref{eq: def generator}, a quick computation using Itô's formula yields \eqref{eq: generic generator}.

Recall that we have a duality pairing $\langle \cdot, \cdot\rangle_\mu : L^{p'}_\mu(\R^d) \otimes L^p_\mu(\R^d) \to \R$ defined by
    $\langle f,g \rangle_\mu= \intr fg\: \d \mu$,
where $\frac{1}{p'}+\frac{1}{p}=1$.

A well-known fact is that the generator $\bar \L$ of the time-reversed diffusion is the adjoint of the generator under the above duality pairing \cite{yosidaFunctionalAnalysis1995,pazySemigroupsLinearOperators2011}, \cite[Thm 4.3.2]{jiangMathematicalTheoryNonequilibrium2004}.
The adjoint $\bar \L$ is implicitly defined by the relation
\begin{align*}
    &\langle f,\L g \rangle_\mu = \langle \bar \L f, g \rangle_\mu, \forall f \in \dom \bar \L, g \in \dom \L,\\
    \dom \bar \L=\:&\{f \in L^1_\mu(\R^d) \mid  \exists h \in L^1_\mu(\R^d),  \forall g\in \dom \L: \langle f,\L g \rangle_\mu =\langle h, g \rangle_\mu   \}.
\end{align*}
(The concept of adjoint generalises the transpose of a matrix in linear algebra to operators on function spaces).
The proof of Lemma \ref{lemma: reversibility of the diffusion property} explicitly computes the adjoint and shows that it is a linear operator $\bar \L: L^1_\mu(\R^d) \to L^1_\mu(\R^d)$ which equals
\begin{align*}
    \bar \L f&= -b \cdot \nabla f + 2 \rho^{-1}\nabla \cdot (D\rho )\cdot \nabla f+ D\nabla \cdot \nabla f.
\end{align*}
Notice how the first order part of the adjoint generator is the drift of the time-reversed diffusion, while the second order part is its diffusion, as expected (cf. Theorem~\ref{thm: time reversal of diffusions}).

Much like we derived the Helmholtz decomposition of the drift by identifying the time-reversible and irreversible parts (see the proof of Proposition \ref{prop: helmholtz decomposition}), we proceed analogously at the level of the generator. Indeed, just as any matrix can be decomposed into a sum of antisymmetric and symmetric matrices, we may decompose the generator into a sum of antisymmetric and symmetric operators
\begin{equation}
\label{eq: gen symmetric and antisymmetric decomp}
\begin{split}
    \L &= \A + \S ,\quad \A := \left(\L-\bar \L\right)/2, \quad \S:= \left(\L+\bar \L\right)/2, \quad \dom \A= \dom\S= \dom\L\cap \dom\bar \L.
\end{split}
\end{equation}
By its analogous construction, this decomposition coincides with the Helmholtz decomposition; indeed, the symmetric operator recovers the time-reversible part of the dynamic while the antisymmetric operator recovers the time-irreversible part. In a nutshell, the Helmholtz decomposition of the generator is as follows
\begin{align*}
\label{eq: symmetric and antisymmetric unpacked}
   \L &= \A + \S , \quad \underbrace{\A= b_{\mathrm{irr}} \cdot \nabla}_{\text{Time-irreversible}} \quad \underbrace{ \S= b_{\mathrm{rev}} \cdot \nabla + D\nabla \cdot \nabla}_{\text{Time-reversible}},
\end{align*}
where the summands are symmetric and antisymmetric operators because they behave accordingly under the duality pairing:
\begin{align*}
    \underbrace{\langle \A f, g\rangle_\mu= -\langle f, \A g\rangle_\mu}_{\text{Antisymmetric}}, \quad \forall f,g \in \dom \A,\quad  \underbrace{\langle \S f, g\rangle_\mu= \langle f, \S g\rangle_\mu}_{\text{Symmetric}}, \quad \forall f,g \in \dom \S.
\end{align*}
Noting that $-\S$ is a positive semi-definite operator, we can go slightly further and decompose it into its square roots. To summarise:

\begin{proposition}
\label{proposition: hypocoercive decomposition of the generator}
We can rewrite the generator of the diffusion process as
$\operatorname L = \operatorname A-\operatorname \Sigma^*\operatorname \Sigma$ where $\A$ is the antisymmetric part of the generator, and $-\operatorname \Sigma^*\operatorname \Sigma$ is the symmetric part, as defined in \eqref{eq: gen symmetric and antisymmetric decomp}. Here $\cdot^*$ denotes the adjoint with respect to the duality pairing $\langle\cdot, \cdot\rangle_\mu$. The operators have the following functional forms: $\operatorname A f= b_{\mathrm{irr}}\cdot \nabla f$, $\sqrt 2 \operatorname \Sigma f =\sigma^\top \nabla f$, $\sqrt 2 \operatorname \Sigma^*g=- \nabla \log \rho \cdot \sigma g- \nabla \cdot(\sigma g) $.
\end{proposition}
A proof is provided in Section \ref{app: multiple perspectives on Helmholtz}.

\subsubsection{Helmholtz decomposition of the backward Kolmogorov equation}

We say that a real-valued function over the state-space of the process $f: \R^d \to \R$ is an \textit{observable}. Intuitively, this is a macroscopic quantity that can be measured or observed in a physical process (e.g., energy or pressure) when the (microscopic) process is not easily accessible. The evolution of an observable $f$ given that the process is prepared at a deterministic initial condition is given by $f_t(x) = \E [f(x_t)|x_0=x]$.

The \textit{backward Kolmogorov equation} is a fundamental equation describing a Markov process, as it encodes the motion of observables
\begin{equation*}
\label{eq: def backward Kolmogorov}
    \partial_t f_t  = \L f_t, \quad f_0 = f \in \dom \L.
\end{equation*}
In other words, $f_t= \E [f(x_t)|x_0=x]$ solves the equation. This highlights the central importance of the generator as providing a concise summary of the process.

The Helmholtz decomposition entails a decomposition of the backward Kolmogorov equation
\begin{equation}
\label{eq: helmholtz BKE}
    \partial_t u_t  = \A f_t + \S f_t = (\operatorname A - \operatorname \Sigma^*\operatorname \Sigma)f_t, \quad f_0 = f \in \dom \L.
\end{equation}
This decomposition is appealing, as it allows us to further characterise the contributions of the time-reversible and irreversible parts of the dynamic. Along the time-irreversible part of the backward Kolmogorov equation $\partial_t f_t=\operatorname A f_t$, the $L^2_\mu(\R^d)$-norm $\|\cdot\|_\mu$ is conserved. Indeed, since $\operatorname A$ is antisymmetric, $\langle \operatorname A f, f\rangle_\mu=0$ for every $f \in \dom \operatorname A$, and hence
$$
\partial_t\left\|f_t\right\|_\mu^2=2\left\langle \operatorname A f_t, f_t\right\rangle_\mu=0 .
$$
On the other hand, along the time-reversible part of the backward Kolmogorov equation generated by $-\operatorname \Sigma^*\operatorname \Sigma$, the $L^2_\mu(\R^d)$-norm is dissipated:
$$
\partial_t\left\|f_t\right\|_\mu^2=-2\left\langle \operatorname \Sigma^*\operatorname \Sigma f_t,  f_t\right\rangle_\mu=-2\left\|\operatorname \Sigma f_t\right\|_\mu^2 \leq 0.
$$
This offers another perspective on the fact that the time-irreversible part of the dynamic is conservative, while the time-reversible part is dissipative---of the $L_\mu^2(\R^d)$-norm.

Beyond this, hypocoercivity theory allows us to analyse the backward Kolmogorov equation, once one has written its Helmholtz decomposition. Hypocoercivity is a functional analytic theory developed to analyse abstract evolution equations of the form \eqref{eq: helmholtz BKE}, originally devised to systematically study the speed of convergence to stationary state of kinetic diffusion processes like the underdamped Langevin dynamics and the Boltzmann equation. As an important result, the theory provides sufficient conditions on the operators $\operatorname A,\operatorname \Sigma$ to ensure an exponentially fast convergence of the backward Kolmogorov equation to a fixed point \cite[Theorems 18 \& 24]{villaniHypocoercivity2009}. Dually, these convergence rates quantify the speed of convergence of the process to its stationary density from a given initial condition.

\subsubsection{GENERIC decomposition of the Fokker-Planck equation}


This perspective can also be examined directly from the Fokker-Planck equation. The Fokker-Planck equation is another fundamental equation describing a diffusion process: it encodes the evolution of the density of the process over time (when it exists). 
The Fokker-Planck equation is the $L^2(\R^d)$-dual to the backward Kolmogorov equation. It reads
\begin{align*}
    \partial_t \rho_t =\L'\rho_t= \nabla \cdot (-b \rho_t + \nabla \cdot(D \rho_t)),
\end{align*}
where $\L'$ is the adjoint of the generator with respect to the standard duality pairing $\langle\cdot, \cdot\rangle$; in other words $\langle \L'f , g \rangle =\langle f , \L g \rangle$ where $\langle f, g \rangle = \intr fg(x) \: \d x$.

The Helmholtz decomposition implies a decomposition of the Fokker-Planck equation into two terms: assuming for now that $\rho_t,\rho >0$ (e.g., if the diffusion is elliptic)
\begin{equation}
\label{eq: Fokker-Planck equation GENERIC form}
\begin{split}
    \partial_t \rho_t &= \nabla \cdot (-b\irr \rho_t) + \nabla \cdot(-b\rev \rho_t +\nabla \cdot(D \rho_t))\\
    &= \nabla \cdot (-b\irr \rho_t) + \nabla \cdot(-\rho_t \rho^{-1}\nabla \cdot(D \rho) +\nabla \cdot(D \rho_t))\\
    &= \nabla \cdot (-b\irr \rho_t) + \nabla \cdot\left(\rho_t D \nabla \log \frac{\rho_t}{\rho} \right).
\end{split}
\end{equation}
We will see that this decomposition casts the Fokker-Planck equation in pre-GENERIC form.

GENERIC (General Equations for Non-Equilibrium Reversible-Irreversible Coupling) is an important theory for analysing dynamical systems arising in non-equilibrium statistical physics like the Fokker-Planck equation. The framework rose to prominence through the seminal work of Ottinger \cite{ottingerEquilibriumThermodynamics2005} and was later developed by the applied physics and engineering communities. Only recently, the framework developed into a rigorous mathematical theory. We refer to \cite{duongNonreversibleProcessesGENERIC2021} for mathematical details. The following Proposition shows how we can rewrite the Fokker-Planck equation in pre-GENERIC form:

\begin{align}
\label{eq: pre-GENERIC equation}
    &\partial_t \rho_t = \underbrace{\operatorname W(\rho_t)}_{\text{time-irreversible}} - \underbrace{\operatorname M_{\rho_t}\left(\d \operatorname{H}[\rho_t \mid \rho ]\right)}_{\text{time-reversible}},
\end{align}

\begin{proposition}[GENERIC decomposition of the Fokker-Planck equation]
\label{prop: GENERIC decomposition of the Fokker-Planck equation}
The Fokker-Planck equation \eqref{eq: Fokker-Planck equation GENERIC form} is in pre-GENERIC form \eqref{eq: pre-GENERIC equation}, with
\begin{align*}
    &W(\rho_t) = \nabla \cdot (-b\irr \rho_t), \quad -M_{\rho_t}\left( \d \operatorname{H}[\rho_t \mid \rho ]\right)=\nabla \cdot\left(\rho_t D \nabla \log \frac{\rho_t}{\rho} \right),\\
     &M_{\rho_t}(\xi) = \Sigma'(\rho_t \Sigma \xi) = - \nabla \cdot (\rho_t D \nabla \xi),\quad \d \operatorname{H}[\rho_t \mid \rho ]= \log \frac{\rho_t}{\rho}+1,
\end{align*}
where $\H[\cdot \mid \cdot]$ is the relative entropy, $\cdot '$ denotes the adjoint under the standard duality pairing $\langle\cdot, \cdot\rangle$, $\d$ is the Fréchet derivative in $L^2(\R^d)$, and $W, M_{\rho_t}$ satisfy the following relations:
\begin{itemize}
    \item \emph{Orthogonality}: $\langle \operatorname W(\rho_t), \d \operatorname{H}[\rho_t \mid \rho ] \rangle =0$,
    \item \emph{Semi-positive definiteness}: $\langle \operatorname M_{\rho_t}(h), g\rangle=\langle h, \operatorname M_{\rho_t}(g)\rangle$, and $\langle \operatorname M_{\rho_t}(g), g\rangle \geq 0$.
\end{itemize}
\end{proposition}

A proof is provided in Section \ref{app: multiple perspectives on Helmholtz}.

Writing the Fokker-Planck equation in pre-GENERIC form \eqref{eq: pre-GENERIC equation} explicits the contributions of the time-reversible and time-irreversible parts at the level of density dynamics. Indeed, the relative entropy functional $\operatorname{H}[\rho_t \mid \rho ]$ is conserved along the time-irreversible part of the Fokker-Planck equation $\partial_t \rho_t = \operatorname W(\rho_t)$
\begin{align*}
    \frac{d \operatorname{H}[\rho_t \mid \rho ]}{d t}=\left\langle \partial_t \rho_t, \d \operatorname{H}[\rho_t \mid \rho ]\right\rangle=\langle \operatorname W(\rho_t), \d \operatorname{H}[\rho_t \mid \rho ]\rangle =0.
\end{align*}
Contrariwise, the relative entropy is dissipated along the time-reversible part of the equation
\begin{align*}
    \frac{d \operatorname{H}[\rho_t \mid \rho ]}{d t}=\left\langle \partial_t \rho_t, \d \operatorname{H}[\rho_t \mid \rho ]\right\rangle=-\left\langle\operatorname M_{\rho_t}\left( \d \operatorname{H}[\rho_t \mid \rho ]\right), \d \operatorname{H}[\rho_t \mid \rho ] \right\rangle \leq 0.
\end{align*}
Aggregating these results, we obtain the well-known fact that the relative entropy with respect to the stationary density is a Lyapunov function of the Fokker-Planck equation; a result sometimes known as de Bruijn's identity or Boltzmann's H-theorem \cite[Proposition 1.1]{chafaiEntropiesConvexityFunctional2004}.

\section{The $e_p$ of stationary diffusions}
\label{sec: the ep of stationary diffusions}

We are now ready to investigate the entropy production of stationary diffusions. First, we give sufficient conditions guaranteeing the mutual absolute continuity of the forward and time-reversed path space measures and compute the entropy production rate. Second, we demonstrate that when these conditions fail the entropy production is infinite.

\subsection{Regular case}
\label{sec: epr regularity}

\begin{figure}[t!]
    \centering
    \includegraphics[width= 0.45\textwidth]{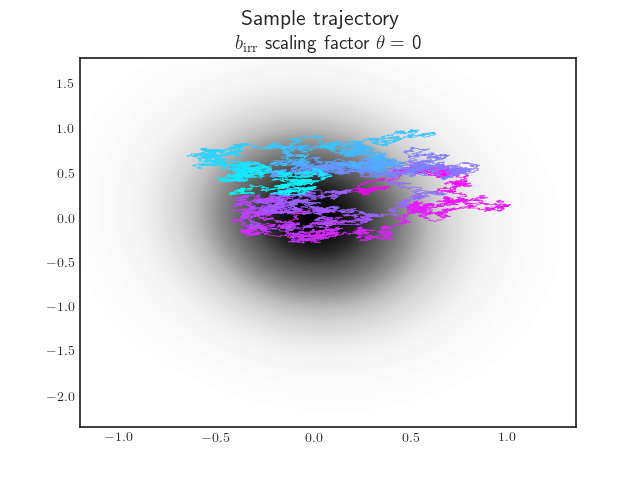}
    \includegraphics[width= 0.45\textwidth]{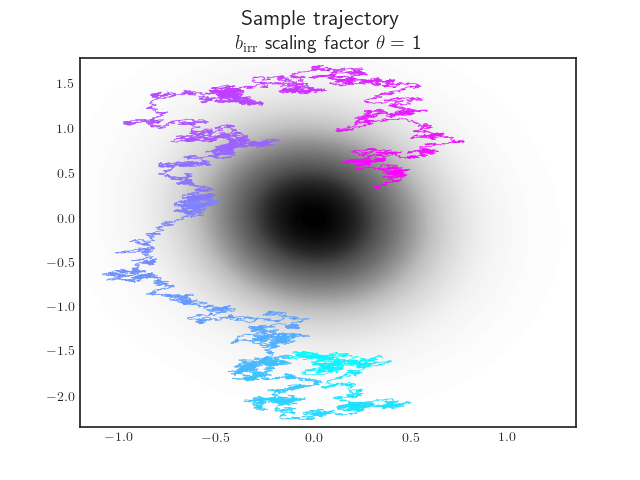}
    \includegraphics[width= 0.45\textwidth]{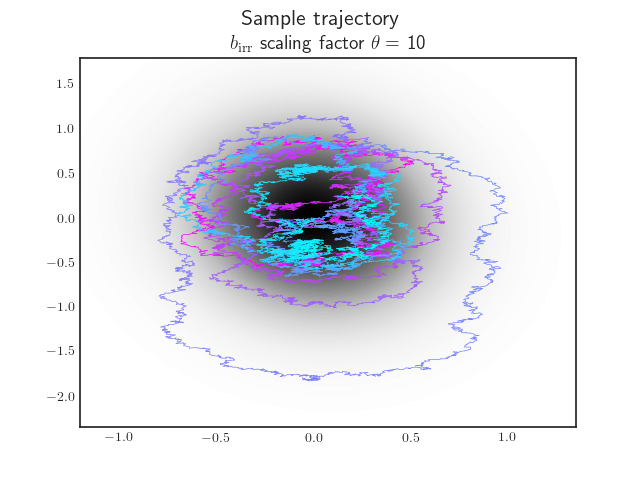}
    \includegraphics[width= 0.45\textwidth]{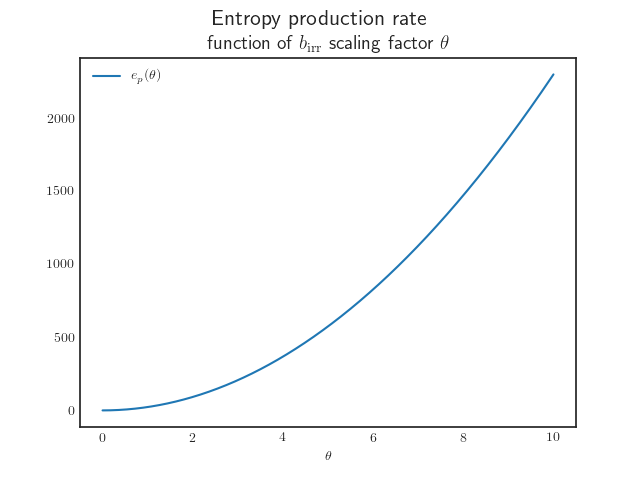}
    \caption[Entropy production as a function of time-irreversible drift]{\footnotesize\textbf{Entropy production as a function of time-irreversible drift.} This figure illustrates the behaviour of sample paths and the entropy production rate as one scales the irreversible drift $b_{\mathrm{irr}}$ by a parameter $\theta$. The underlying process is a two-dimensional Ornstein-Uhlenbeck process, for which exact sample paths and entropy production rate are available (Section \ref{sec: OU process}). The heat map represents the density of the associated Gaussian steady-state. One sees that a non-zero irreversible drift induces circular, wandering behaviour around the contours of the steady-state, characteristic of a non-equilibrium steady-state (top right and bottom left). This is accentuated by increasing the strength of the irreversible drift. The entropy production rate measures the amount of irreversibility of the stationary process. It grows quadratically as a function of the irreversible scaling factor $\theta$ (bottom right). When there is no irreversibility (top left), we witness an equilibrium steady-state. This is characterised by a vanishing entropy production (bottom right).}
    \label{fig: EPR as a function of gamma}
\end{figure}

\begin{theorem}
\label{thm: epr regular case simple}
Let an Itô SDE \eqref{eq: Ito SDE} with coefficients satisfying Assumption \ref{ass: coefficients time reversal}. Assume that the solution $(x_t)_{t\in [0,T]}$ is stationary with respect to a density $\rho$ satisfying Assumption \ref{ass: steady-state time reversal}. Denote by $b_{\mathrm{irr}}$ the time-irreversible part of the drift (Proposition \ref{prop: helmholtz decomposition}), and by $\cdot^-$ the Moore-Penrose matrix pseudo-inverse. Suppose that:
\begin{enumerate}
    \item \label{item: epr regular drift image of volatility} For $\rho$-almost every $x \in \R^d$,
        $b_{\mathrm{irr}}(x) \in \im\sigma(x)$, and
    \item \label{item: continuity} The product $\sigma^- b_{\mathrm{irr}}: \R^d \to \R^m$ is Borel measurable (e.g., if $\sigma^- b_{\mathrm{irr}}$ is continuous), and
    \item 
         \label{item: epr bounded}
        $\int_{\R^d} b_{\mathrm{irr}}^\top D^- b_{\mathrm{irr}} \rho(x)\d x <+\infty$.
\end{enumerate}
Denote by $\p_{[0, T]}, \bar \p_{[0, T]}$ the path space measures of the forward and time-reversed diffusions, respectively, on $C([0,T], \R^d)$ (Definition \ref{def: path space measure}). Then,
\begin{enumerate}
    \item The path-space measures are equivalent $\p_{[0, T]} \sim \bar \p_{[0, T]}$, and 
    \item $e_p = \int_{\R^d} b_{\mathrm{irr}}^\top D^- b_{\mathrm{irr}} \rho(x)\d x$.
\end{enumerate} 
\end{theorem}

Under the assumptions of Theorem \ref{thm: epr regular case simple}, the $e_p$ is a quadratic form of the time-irreversible drift, see Figure~\ref{fig: EPR as a function of gamma}. 

A proof of Theorem \ref{thm: epr regular case simple} is provided in Section \ref{app: epr regularity}. The idea of the proof is simple: in the elliptic case, the approach follows \cite[Chapter 4]{jiangMathematicalTheoryNonequilibrium2004} with some generalisations. In the non-elliptic case, the condition $b_{\mathrm{irr}}(x) \in \im\sigma(x)$ intuitively ensures that the solution to the SDE, when initialised at any point, evolves on a sub-manifold of $\R^d$ and is elliptic on this manifold (e.g., Figure \ref{fig: OU process b in Im sigma}). The pseudo-inverse of the diffusion tensor is essentially the inverse on this sub-manifold. Thus, a proof analogous to the elliptic case, but on the sub-manifold (essentially replacing all matrix inverses by pseudo-inverses and making sure everything still holds), shows that the path space measures of the forward and backward processes initialised at a given point are equivalent---and Girsanov's theorem gives us their Radon-Nykodym derivative. Finally, Proposition \ref{prop: aggregating local ep} gives us the usual formula for the entropy production rate but with the matrix inverse replaced by the pseudo-inverse. Please see Section \ref{sec: geom} for a geometric discussion of this proof.

Suppose either of assumptions \ref{item: continuity}, \ref{item: epr bounded} of Theorem \ref{thm: epr regular case simple} do not hold. Then we have the following more general (and technical) result:

\begin{theorem}
\label{thm: epr regular case general}
Let $(\Omega, \mathscr{F}, P)$ be a probability space and $\left(w_t\right)_{t \geqslant 0}$ a standard Wiener process on $\mathbb{R}^{m}$, with respect to the filtration $(\mathscr{F}_{t})_{t\geq 0}$ \cite[Definition 2.1.12]{prevotConciseCourseStochastic2008}. 
Consider the Itô SDE \eqref{eq: Ito SDE} with coefficients satisfying Assumption \ref{ass: coefficients time reversal}. Consider its unique strong solution $(x_t)_{t\in [0,T]}$ with respect to the given Brownian motion on the filtered probability space $(\Omega, \mathscr{F}, (\mathscr{F}_{t})_{t\geq 0}, P)$. Assume that the solution is stationary with respect to a density $\rho$ satisfying Assumption \ref{ass: steady-state time reversal}. Denote by $b_{\mathrm{irr}}$ the time-irreversible part of the drift (Proposition \ref{prop: helmholtz decomposition}), and by $\cdot^-$ the Moore-Penrose matrix pseudo-inverse. Suppose that:
\begin{enumerate}
    \item \label{item: epr regular general drift image of volatility} For $\rho$-almost every $x \in \R^d$, $b_{\mathrm{irr}}(x) \in \im\sigma(x)$, and
    \item \label{item: epr regular general adapted} $\sigma^- b_{\mathrm{irr}}(x_t)$ is an $\mathscr F_t$-adapted process (e.g., $\sigma^- b_{\mathrm{irr}}: \R^d \to \R^m$ is Borel measurable), and
    \item \label{item: epr regular general expectation} The following holds 
    \begin{equation}
    \label{eq: exponential condition}
        \E_P\left[Z_T \right]=1,  Z_T :=\exp \left( -2 \int_0^T \langle \sigma^- b_{\mathrm{irr}}(x_t), dw_t \rangle + |\sigma^- b_{\mathrm{irr}}(x_t) |^2 dt\right).
    \end{equation}
\end{enumerate}
Denote by $\p_{[0, T]}, \bar \p_{[0, T]}$ the path space measures on $C([0,T], \R^d)$ of the forward and time-reversed diffusions, respectively (Definition \ref{def: path space measure}). Then,
\begin{enumerate}
    \item The path-space measures are equivalent $\p_{[0, T]} \sim \bar \p_{[0, T]}$, and 
    \item $e_p = \int_{\R^d} b_{\mathrm{irr}}^\top D^- b_{\mathrm{irr}} \rho(x)\d x$.
\end{enumerate}
\end{theorem}

A proof is provided in Section \ref{app: epr regularity}. The proof is similar to that of Theorem \ref{thm: epr regular case simple}, but much shorter, since \eqref{eq: exponential condition} allows one to apply Girsanov's theorem directly; in contrast, a large part of the proof of Theorem \ref{thm: epr regular case simple} is dedicated to showing that indeed, a version of Girsanov's theorem can be applied.




In relation to assumption \ref{item: epr regular general adapted} of Theorem \ref{thm: epr regular case general}, note that if a matrix field $\sigma : \R^d \to \R^{d\times m}$ is Borel measurable, then its pseudo-inverse $\sigma^- : \R^d \to \R^{m \times d}$ is also Borel measurable. We now give sufficient conditions for the exponential condition \eqref{eq: exponential condition}. 

\begin{proposition}
\label{prop: suff conds exp cond}
Consider a stochastic process $(x_t)_{t\in [0,T]}$ on the probability space $(\Omega, \mathscr{F}, (\mathscr{F}_{t})_{t\geq 0}, P)$, which is stationary at the density $\rho$. Assume that $\sigma^-b_{\mathrm{irr}}(x_t)$ is $\mathscr F_t$-adapted. Then, either of the following conditions implies \eqref{eq: exponential condition}:
\begin{enumerate}
    \item \label{item: suff cond martingale} $Z_t=\exp \left(-2{\int_{0}^{t}\langle \sigma^- b_{\mathrm{irr}}(x_s), d w_s\rangle+|\sigma^- b_{\mathrm{irr}}(x_s)|^{2} d s}\right), t\in [0,T]$ is a martingale on the probability space $(\Omega, \mathscr{F},\{\mathscr{F}_{t}\}_{t\geq 0}, P)$.
    \item \label{item: suff cond novikov} $\mathbb{E}_{P}\left(e^{2\int_{0}^{T}|\sigma^- b_{\mathrm{irr}}(x_t)|^{2} d t}\right)<+\infty$ (Novikov's condition).
    \item \label{item: suff cond delta exp condition} There exists $\delta>0$ such that $ \mathbb{E}_{\rho}\left(e^{\delta|\sigma^- b_{\mathrm{irr}}(x)|^{2}}\right)<+\infty$.
    \item \label{item: suff cond Kazamaki} $\sup_{t \in [0,T]} \mathbb{E}_P\left[\exp \left(- \int_0^t \langle\sigma^- b_{\mathrm{irr}}(x_s), d w_s\rangle\right)\right]<+\infty$ (Kazamaki's criterion).
    \item \label{item: suff cond Dellacherie meyer} There exists $K<1$ s.t. for all $t \in [0,T]$
    \begin{equation*}
        \E_P\left[\left. 2\int_t^T\left|\sigma^- b_{irr}(x_t)\right|^2 ds \:\right| \mathscr F_t\right]\leq K.
    \end{equation*}
    \item \label{item: suff cond tail} The tail of $|\sigma^- b_{\mathrm{irr}}(x)|^{2}, x\sim \rho$ decays exponentially fast, i.e., there exists positive constants $c,C, R > 0$ such that for all $r> R$ 
    \begin{align}
    \label{eq: exponential boundedness of tail}
         P(|\sigma^- b_{\mathrm{irr}}(x)|^{2}> r)\leq Ce^{-c r}.
    \end{align}
\end{enumerate}
Furthermore, (\ref{item: suff cond novikov} or \ref{item: suff cond Kazamaki}) $\Rightarrow$ \ref{item: suff cond martingale}; \ref{item: suff cond Dellacherie meyer} $\Rightarrow$ \ref{item: suff cond novikov}; and \ref{item: suff cond tail} $\Rightarrow$ \ref{item: suff cond delta exp condition}.
\end{proposition}

A proof is provided in Section \ref{app: epr regularity}. 

\subsection{Singular case}
\label{sec: ep singularity}

When the time-irreversible part of the drift is not always in the range of the volatility matrix field, we have a different result.

\begin{theorem}
\label{thm: epr singular}
Suppose that the Itô SDE \eqref{eq: Ito SDE} satisfies Assumption \ref{ass: coefficients time reversal} and that the volatility is twice continuously differentiable $\sigma\in C^2(\R^d, \R^{d\times m})$. Furthermore suppose that the solution $(x_t)_{t\in [0,T]}$ is stationary with respect to a density $\rho$ satisfying Assumption \ref{ass: steady-state time reversal}. Denote by $\p_{[0, T]}, \bar \p_{[0, T]}$ the path space measures on $C([0,T], \R^d)$ of the forward and time-reversed processes, respectively (Definition \ref{def: path space measure}). If $b_{\mathrm{irr}}(x) \in \im\sigma(x)$ does not hold for $\rho$-a.e. $x \in \R^d$, then
    \begin{align*}
        \p_{[0, T]} \perp \bar \p_{[0, T]} \text{ and } e_p = +\infty.
    \end{align*}
\end{theorem}

A proof is provided in Section \ref{app: epr singular}. The proof uses a version of the Stroock-Varadhan support theorem to show that there are paths that can be taken by the forward diffusion process that cannot be taken by the backward diffusion process---and vice-versa. Specifically, when considering the two processes initialised at a point $x \in \R^d$ where $b_{\mathrm{irr}}(x) \notin \im\sigma(x)$, we can see that the derivatives of their respective possible paths at time $0$ span different tangent sub-spaces at $x$. Thus the path space measures $\p^{x}_{[0,T]} , \bar \p^{x}_{[0,T]}$ are mutually singular. Since such $x$ occur with positive probability under the stationary density, it follows that the path space measures of the forward and time-reversed stationary processes $\p^{x}_{[0,T]} , \bar \p^{x}_{[0,T]}$ are also mutually singular. Finally, the relative entropy between two mutually singular measures is infinity, hence the $e_p$ must be infinite.

By Theorem \ref{thm: epr singular} and \eqref{eq: drift image of volatility additive noise} we can readily see that the underdamped \eqref{eq: underdamped Langevin dynamics}
and generalised Langevin \cite[eq. 8.33]{pavliotisStochasticProcessesApplications2014} processes in phase-space have infinite entropy production. Expert statistical physicists will note that this contrasts with previous results in the literature stating that these diffusions have finite entropy production. There is no contradiction as physicists usually add an additional operator to the definition of the entropy production in these systems (Remark \ref{rem: physical relevance}). While obviously informative of the underlying process, statistical physicists should take the results of Theorems \ref{thm: epr regular case simple}, \ref{thm: epr regular case general} and \ref{thm: epr singular} to be physically relevant to entropy production insofar as the definition of entropy production we adopted (Definition \ref{def: epr}) is physically meaningful for the system at hand. What if it is not? We will return to this in the discussion (Section \ref{sec: discussion generalised non-reversibility}).

In contrast to the underdamped and generalised Langevin equations, there exist hypoelliptic, non-elliptic diffusions with finite entropy production. 
For example, consider the following volatility matrix field
\begin{align*}
    \sigma(x,y,z) = \begin{bmatrix} x & 1  \\1&1\\ 0&1 
    \end{bmatrix}.
\end{align*}
By Hörmander's theorem \cite[Theorem 1.3]{hairerMalliavinProofOrmander2011}, for any smooth, confining (e.g., quadratic) potential $V: \R^3 \to \R$, the process solving the SDE
\begin{equation*}
   dx_t =- D\nabla V(x_t)dt + \nabla \cdot D(x_t)dt + \sigma(x_t) dw_t 
\end{equation*}
is hypoelliptic and non-elliptic. Furthermore, it is stationary and time-reversible at the Gibbs density $\rho(x)\propto \exp(-V(x))$.

\section{Examples and $e_p$ of numerical simulations}
\label{sec: illustrating results epr}

We illustrate these results for linear diffusion processes, underdamped Langevin dynamics and their numerical simulations.

\subsection{Linear diffusion processes}
\label{sec: OU process}

Given matrices $B\in \R^{d \times d}, \sigma \in \R^{d \times m}$, and a standard Brownian motion $(w_t)_{t\in [0, +\infty)}$ on $\R^m$, consider a linear diffusion process (i.e., a multivariate Ornstein-Uhlenbeck process)
\begin{equation}
\label{eq: OU process}
\begin{split}
    dx_t &= b(x_t)dt +\sigma(x_t) dw_t, \quad b(x)= -Bx, \quad \sigma(x) \equiv \sigma.
\end{split}
\end{equation}
This process arises, for example, in statistical physics as a model of the velocity of a massive Brownian particle subject to friction \cite{uhlenbeckTheoryBrownianMotion1930}; it covers the case of interacting particle systems when the interactions are linear in the states (e.g., the one dimensional ferromagnetic Gaussian spin model \cite{godrecheCharacterisingNonequilibriumStationary2019}); or when one linearises the equations of generic diffusion processes near the stationary density. 

By solving the linear diffusion process (e.g., \cite[Section 2.2]{godrecheCharacterisingNonequilibriumStationary2019}) one sees that the solution can be expressed as a linear operation on Brownian motion---a Gaussian process---thus the process must itself be Gaussian, and its stationary density as well (when it exists). Consider a Gaussian density $\rho$
\begin{align*}
    \rho (x) = \mathcal N (x;0, \Pi^{-1}), \quad -\log \rho(x) = \frac 1 2 x^\top \Pi x,
\end{align*}
where $\Pi \in \R^{d\times d}$ is the symmetric positive definite \emph{precision} matrix. By the Helmholtz decomposition (Propositions \ref{prop: helmholtz decomposition} \& \ref{prop: characterisation of irreversible drift}), $\rho$ is a stationary density if and only if we can decompose the drift as follows:
\begin{align*}
    b = b_{\mathrm{rev}} + b_{\mathrm{irr}}, \quad b_{\mathrm{rev}}(x)= -D\Pi x, \quad b_{\mathrm{irr}}(x)= -Q\Pi x,
\end{align*}
where $Q=-Q^\top \in \R^{d\times d}$ is an arbitrary antisymmetric matrix, and, recall $D= \sigma \sigma^\top/2\in \R^{d\times d}$ is the diffusion tensor. In particular, the drift of the forward and the time-reversed dynamic, are, respectively,
\begin{align*}
    b(x) = -Bx, \quad B = (D+ Q)\Pi,\quad 
    \bar b(x) = -Cx, \quad C:= (D-Q)\Pi.
\end{align*}



Suppose that $b_{\mathrm{irr}}(x)\in \im \sigma $ for any $x\in\R^d$. By definiteness of $\Pi$ this is equivalent to $\im Q \subseteq \im \sigma$. By Theorem \ref{thm: epr regular case simple}, and applying the trace trick to compute the Gaussian expectations of a bilinear form, we obtain\footnote{To obtain the last equality we used $\tr(D^- D \Pi Q)=\tr(\Pi Q D^- D)= - \tr(D D^- Q \Pi) $. By standard properties of the pseudo-inverse $D D^-$ is the orthogonal projector onto $\im D= \im \sigma$. Thus, $\im Q \subseteq \im \sigma$ implies $D D^-Q=Q$. Finally, the trace of a symmetric matrix $\Pi$ times an antisymmetric matrix $Q$ vanishes.}
\begin{equation}
\label{eq: ep OU regular}
\begin{split}
     e_p &= \int_{\R^d} b_{\mathrm{irr}}^\top D^- b_{\mathrm{irr}} \rho(x)\d x = -\int_{\R^d} x^\top \Pi Q D^- Q\Pi x \,\rho(x)\d x \\
     &= -\tr(\Pi Q D^- Q\Pi\Pi^{-1})
     = -\tr(D^- Q \Pi Q) \\
    &=-\tr(D^- B Q)+\underbrace{\tr(D^- D \Pi Q)}_{=0}=-\tr(D^- B Q).   
\end{split}
\end{equation}
This expression for the entropy production is nice as it generalises the usual formula to linear diffusion processes to degenerate noise, simply by replacing inverses with pseudo-inverses, cf. \cite[eqs. 2.28-2.29]{godrecheCharacterisingNonequilibriumStationary2019} and \cite{mazzoloNonequilibriumDiffusionProcesses2023,buissonDynamicalLargeDeviations2022}.

Contrariwise, suppose that $\im Q \not \subseteq \im \sigma$. Then by Theorem \ref{thm: epr singular},
\begin{align}
\label{eq: ep OU singular}
    e_p = +\infty.
\end{align}

\subsubsection{Exact numerical simulation and entropy production rate}

Linear diffusion processes can be simulated exactly as their transition kernels are known. Indeed, by solving the process, one obtains the transition kernels of the Markov semigroup as a function of the drift and volatility \cite[Theorem 9.1.1]{lorenziAnalyticalMethodsMarkov2006}. The forward and time-reserved transition kernels are the following:
\begin{equation}
\label{eq: OU transition kernels}
\begin{split}
    p_\e(\cdot,x) &= \mathcal N(e^{-\e B}x, S_\e),\quad S_\e:= \int_0^\e e^{-tB}\sigma \sigma ^\top e^{-tB^\top} \d t, \\
    \bar p_\e(\cdot,x) 
    &= \mathcal N (e^{-\e C}x, \bar S_\e),\quad \bar S_\e
    := \int_0^\e e^{-t C}\sigma \sigma^\top e^{-tC^\top}\d t.
\end{split}
\end{equation}
Sampling from the transition kernels allows one to simulate the process exactly, and offers an alternative way to express the entropy production rate. Recall from Proposition \ref{prop: epr transition kernels} that the $e_p$ is the infinitesimal limit of the entropy production rate of an exact numerical simulation $e_p(\e)$ with time-step $\e$
    \begin{align*}
        e_p &= \lim_{\e\downarrow 0} e_p(\e), \quad e_p(\e) = \frac{1}{\e} \E_{x\sim \rho}[\H[p_\e(\cdot,x)\mid \bar p_\e(\cdot,x)]].
    \end{align*}
We can leverage the Gaussianity of the transition kernels to compute the relative entropy and obtain an alternative formula for the entropy production rate:

\begin{lemma}
\label{lemma: limit ep formula OU process}
The entropy production rate of the stationary linear diffusion process can also be expressed as 
\begin{equation}
\label{eq: epr hypo OU process lim}
\begin{split}
    e_p &= \lim_{\e \downarrow 0} e_p(\e),\\
    e_p(\e)&= \frac{1}{2\e} \left[\tr(\bar S_\e^{-}S_\e)-\rank \sigma + \log \frac{\det^*(\bar S_\e)}{\det^*(S_\e)}\right. \\
    &\left.+\tr \left( \Pi^{-1} (e^{-\e C}-e^{-\e B})^\top \bar S_\e^{-}(e^{ -\e C}-e^{-\e B})\right)\right],
\end{split}
\end{equation}
where $\cdot^-$ is the Moore-Penrose pseudo-inverse and $\det^*$ is the pseudo-determinant.
\end{lemma}
A proof is provided in Section \ref{app: exact simul OU process}. Computing the limit \eqref{eq: epr hypo OU process lim} analytically, gives us back \eqref{eq: ep OU regular}, \eqref{eq: ep OU singular}, however, we will omit those details here. For our purposes, this gives us a way to numerically verify the value of $e_p$ that was obtained from theory. See Figures \ref{fig: OU process b in Im sigma} and \ref{fig: OU process b not in Im sigma} for illustrations.

\begin{figure}[t!] 
    \centering
    \includegraphics[width= 0.45\textwidth]{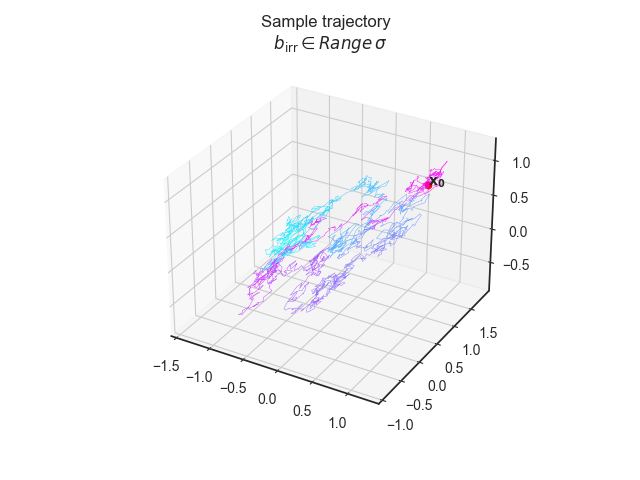}
    \includegraphics[width= 0.45\textwidth]{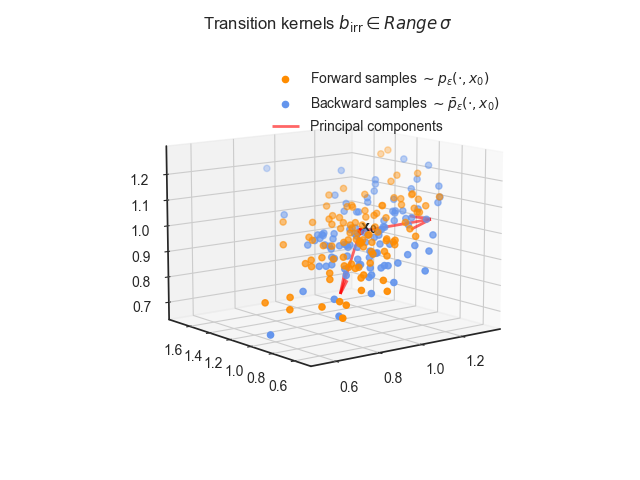}
    \includegraphics[width= 0.5\textwidth]{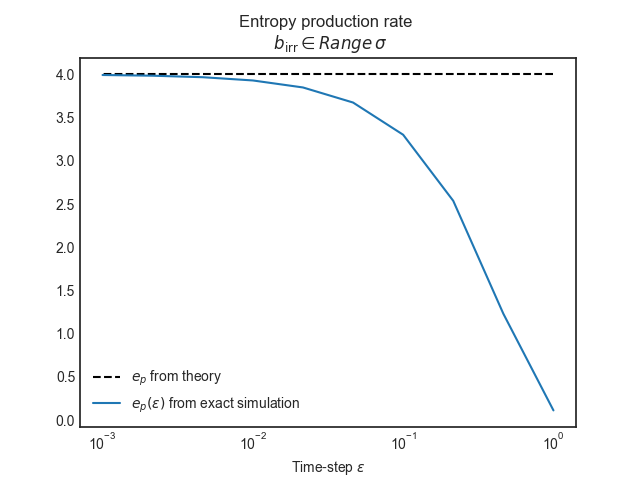}
    \caption[Exact simulation of linear diffusion process with $b_{\mathrm{irr}}(x) \in \im\sigma$]{\footnotesize\textbf{Exact simulation of linear diffusion process with $b_{\mathrm{irr}}(x) \in \im\sigma$.} This figure considers an OU process in 3d space driven by degenerate noise, i.e., $\rank \sigma <3$. The coefficients are such that $\sigma=Q$ are of rank $2$. In particular, $b_{\mathrm{irr}}(x) \in \im\sigma$ holds for every $x$. The process is not elliptic nor hypoelliptic, but it is elliptic over the subspace in which it evolves. The upper-left panel shows a sample trajectory starting from $x_0=(1,1,1)$. The upper-right panel shows samples from different trajectories after a time-step $\e$. There are only two principal components to this point cloud as the process evolves on a two dimensional subspace. In the bottom panel, we verify the theoretically predicted value of $e_p$ by evaluating the entropy production of an exact simulation $e_p(\e)$ with time-step $\e$. As predicted, we recover the true $e_p$ in the infinitesimal limit as the time-step of the exact simulation tends to zero $\e \to 0$. Furthermore, since the process is elliptic in its subspace, the entropy production is finite.}
    \label{fig: OU process b in Im sigma}
\end{figure}

\begin{figure}[t!]
    \centering
    \includegraphics[width= 0.45\textwidth]{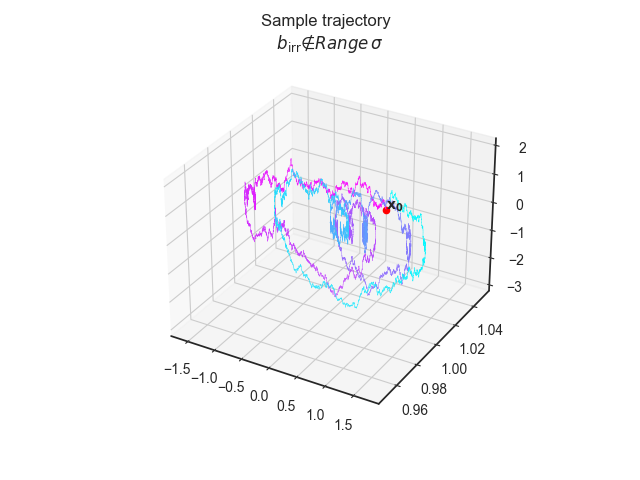}
    \includegraphics[width= 0.45\textwidth]{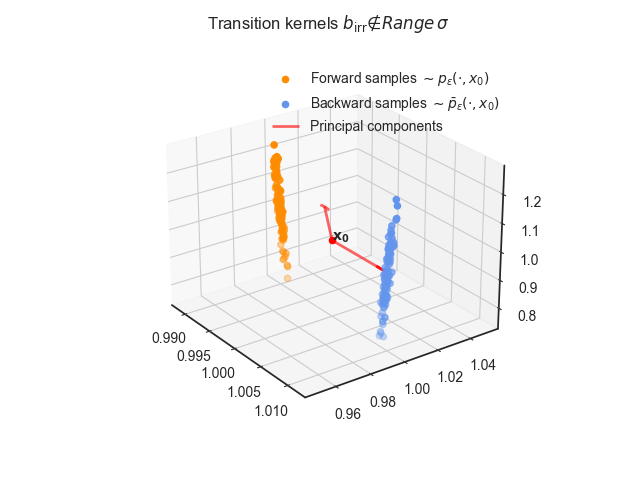}
    \includegraphics[width= 0.5\textwidth]{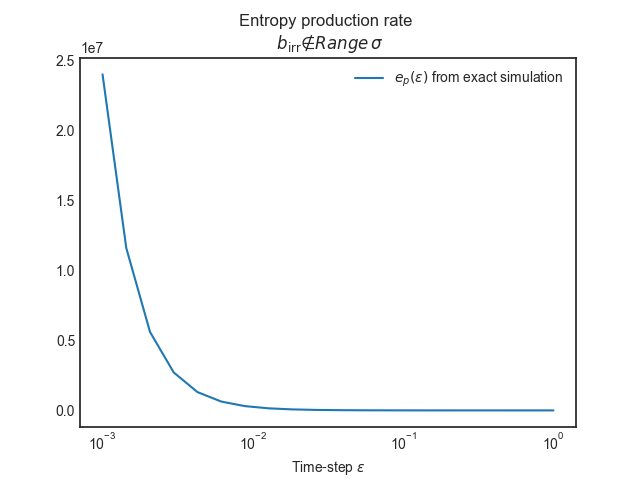}
    \caption[Exact simulation of linear diffusion process with $b_{\mathrm{irr}}(x) \not\in \im\sigma$]{\footnotesize\textbf{Exact simulation of linear diffusion process with $b_{\mathrm{irr}}(x) \not\in \im\sigma$.} This figure considers an OU process in 3d space driven by degenerate noise. The coefficients are such that $\im b_{\mathrm{irr}}$ is two-dimensional while $\im\sigma$ is one-dimensional, and such that the process does not satisfy Hörmander's hypoellipticity condition. As such the process is hypoelliptic on a two-dimensional subspace; see a sample trajectory in the upper-left panel. By hypoellipticity its transition kernels are equivalent in the sense of measures, although far removed: On the upper right panel we show samples from different trajectories after a time-step $\e$. There are only two principal components to this data-cloud as the process evolves on a two dimensional subspace. In the bottom panel, we verify the theoretically predicted $e_p$ by evaluating the entropy production of an exact simulation $e_p(\e)$ with time-step $\e$. As predicted, we recover $e_p=+\infty$ in the infinitessimal limit as the time-step of the exact simulation tends to zero $\e \downarrow 0$.  This turns out to be as the transition kernels of the forward and time-reversed processes become more and more mutually singular as the time-step decreases.}
    \label{fig: OU process b not in Im sigma}
\end{figure}

\subsection{Underdamped Langevin dynamics}
\label{sec: underdamped numerical simulations}

In this sub-section, we consider the entropy production rate of underdamped Langevin dynamics and its numerical simulations. Recall that the $e_p$ we compute here is defined \textit{without} an additional momentum flip operator on the path space measure of the time-reversed process (i.e., \eqref{eq: def epr} and not \eqref{eq: rem def gen epr}), and may be a distinct quantity from the entropy production that physicists usually consider in such systems (see the discussion in Section \ref{sec: discussion generalised non-reversibility}).

Consider a Hamiltonian $H(q,p)$ function of positions $q\in \R^n$ and momenta $ p \in \R^n$ of the form
\begin{equation}
H(q, p)=V(q)+\frac{1}{2} p^\top M^{-1} p
\end{equation}
for some smooth potential function $V: \R^n \to \R$ and diagonal mass matrix $M \in \R^{n\times n}$.

The underdamped Langevin process is the solution to the SDE \cite[eq 2.41]{roussetFreeEnergyComputations2010}
\begin{equation}
\begin{cases}
d q_{t} = M^{-1}p_{t} d t \\
d p_{t} =-\nabla V\left(q_{t}\right) d t-\gamma M^{-1} p_{t} d t+\sqrt{2 \gamma \beta^{-1}} d w_{t}
\end{cases}
\end{equation}
for some friction coefficient $\gamma >0$. This process arises in statistical physics, as a model of a particle coupled to a heat bath \cite{rey-belletOpenClassicalSystems2006}, \cite[Chapter 8]{pavliotisStochasticProcessesApplications2014}; in Markov chain Monte-Carlo as an accelerated sampling scheme \cite{maThereAnalogNesterov2021,barpGeometricMethodsSampling2022}; and also as a model of interacting kinetic particles.

The invariant density, assuming it exists, is \cite[Section 2.2.3.1]{roussetFreeEnergyComputations2010}
\begin{align*}
    \rho(q,p)=\frac 1 Z e^{-\beta H(q,p)}=\frac 1 Z e^{-\beta V(q)}e^{-\frac{\beta}{2} p^\top M^{-1} p}.
\end{align*}

Since the noise is additive the Itô interpretation of the SDE coincides with the Stratonovich interpretation, thus the irreversible drift is in the range of the volatility if and only if the drift is in the range of the volatility \eqref{eq: drift image of volatility additive noise}. Observe that when the momentum is non-zero $p\neq 0$ the drift is not in the image of the volatility: in the $q$ components the drift is non-zero while the volatility vanishes. Since $p \neq 0$ has full measure under the stationary density $\rho(q,p)$ we obtain, from Theorem \ref{thm: epr singular}, 
\begin{align}
\label{eq: ep underdamped}
    e_p = +\infty.
\end{align}

Note that the entropy production of an exact numerical simulation with time-step $\e >0$ is usually finite by hypoellipticity\footnote{\eqref{eq: finite ep for exact simulation of underdamped} always holds in a quadratic potential, whence the process is a linear diffusion and the results from Section \ref{sec: OU process} apply. We conjecture this to hold in the non-linear case as well, as hypoellipticity guarantees that the transition kernels are mutually equivalent in the sense of measures.}
\begin{align}
\label{eq: finite ep for exact simulation of underdamped}
    e_p(\e) <+\infty.
\end{align}
Figure \ref{fig: exact underdamped} illustrates this with an exact simulation of underdamped Langevin dynamics in a quadratic potential.

\begin{figure}[t!]
    \centering
    \includegraphics[width= 0.45\textwidth]{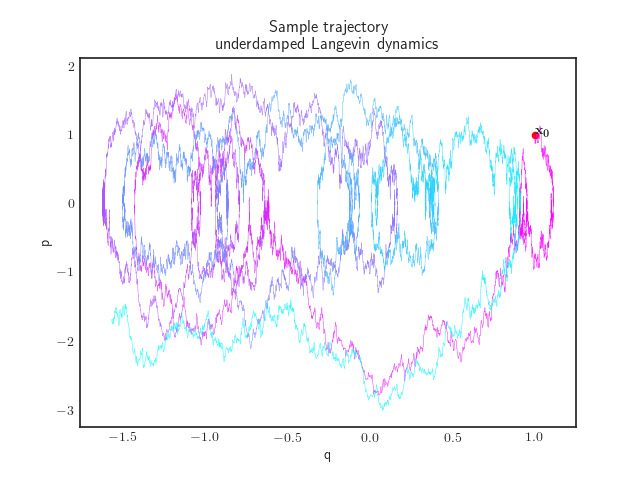}
    \includegraphics[width= 0.45\textwidth]{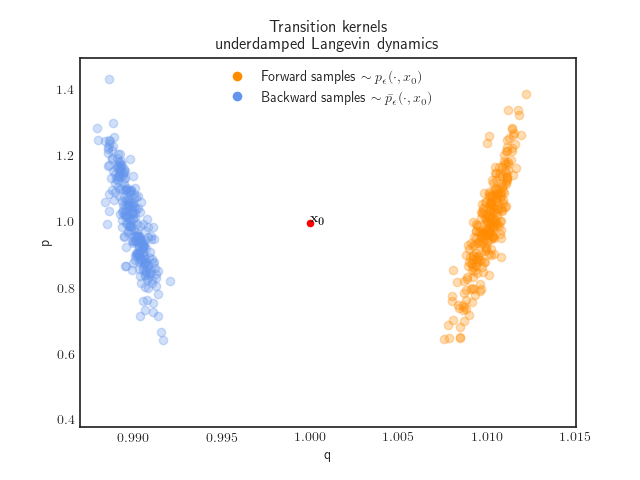}
    \includegraphics[width= 0.5\textwidth]{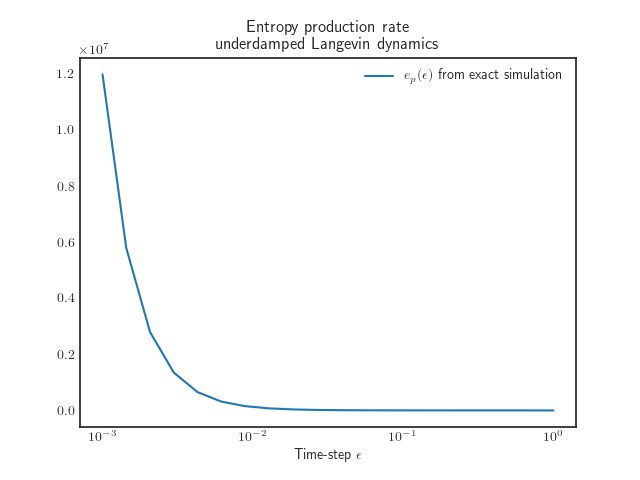}
    \caption[Exact simulation of underdamped Langevin dynamics]{\footnotesize\textbf{Exact simulation of underdamped Langevin dynamics.} This figure plots underdamped Langevin dynamics in a quadratic potential. Here, the process is two dimensional, i.e., positions and momenta evolve on the real line. We exploit the fact that underdamped Langevin in a quadratic potential is an Ornstein-Uhlenbeck process to simulate sample paths exactly. The choice of parameters was: $V(q)=q^2/2, M=\gamma =1$. The upper left panel plots a sample trajectory. One observes that the process is hypoelliptic: it is not confined to a prespecified region of space, cf. Figures \ref{fig: OU process b in Im sigma}, \ref{fig: OU process b not in Im sigma}, even though random fluctuations affect the momenta only. The upper right panel plots samples of the forward and time-reversed processes after a time-step of $\e$. In the bottom panel, we verify the theoretically predicted $e_p$ by evaluating the entropy production of an exact simulation $e_p(\e)$ with time-step $\e$. As predicted, we recover $e_p=+\infty$ in the infinitessimal limit as the time-step of the exact simulation tends to zero $\e \downarrow 0$. This turns out to be because the transition kernels of the forward and time-reversed processes become more and more mutually singular as the time-step decreases.}
    \label{fig: exact underdamped}
\end{figure}

When the potential is non-quadratic, the underdamped process is a non-linear diffusion and one is usually unable to simulate it exactly. Instead, one resolves to numerical approximations to the solution of the process. We now turn to two common numerical discretisations of underdamped: the Euler-Maruyama and BBK discretisations. 
We will examine whether these discretisations are good approximations to the true process by computing their entropy production rate.

\subsubsection{Euler-Maruyama discretisation}

In this section, we show that an Euler-Maruyama (E-M) discretisation of underdamped Langevin dynamics at any time-step $\e>0$ has infinite entropy production
\begin{align}
\label{eq: ep Euler}
    e_p^{\text{E-M}}(\e)=+\infty.
\end{align}

To see this, we take a step back and consider an arbitrary Itô SDE in $\R^d$
\begin{align*}
    dx_t = b(x_t)dt + \sigma(x_t)dw_t
\end{align*}
The Euler-Maruyama discretisation for some time-step $\e >0$ is
\begin{align*}
    x_{i+1} = x_i + b(x_i)\e + \sigma(x_i)\omega_i, \quad \omega_i \sim \mathcal N(0, \e \operatorname{Id}_d).
\end{align*}
This is a Markov chain with the following transition kernels
\begin{equation}
\label{eq: E-M scheme definition}
\begin{split}
    p^{\text{E-M}}_\e\left(x_{i+1}, x_{i}\right) &=\mathcal N(x_{i+1}; x_i+\e b(x_i), 2\e D(x_i)),\\
    \bar p^{\text{E-M}}_\e\left(x_{i+1}, x_{i}\right)&:=p^{\text{E-M}}_\e\left(x_{i},x_{i+1}\right),
\end{split}
\end{equation}
where $\bar p^{\text{E-M}}_\e$ denotes the transition kernel of the backward chain\footnote{Caution: this is different from the E-M discretisation of the time-reversed process.}.

It turns out that when the SDE is not elliptic the transition kernels $p^{\text{E-M}}_\e\left(\cdot,x \right), \bar p^{\text{E-M}}_\e\left(\cdot,x \right)$ tend to have different supports:

\begin{lemma}
\label{lemma: support of transition kernels Euler}
For any $x\in \R^d$
\begin{align*}
    \operatorname{supp}p^{\mathrm{E}\text{-}\mathrm{M}}_\e\left( \cdot,x \right)&= \overline{\{y: y \in x + \e b(x) + \im D(x)\}} \\
  \operatorname{supp}\bar p^{\mathrm{E}\text{-}\mathrm{M}}_\e\left( \cdot,x \right)&=\overline{\{y:  x  \in y +\e b(y)+\im D(y)\}}
\end{align*}
\end{lemma}
Lemma \ref{lemma: support of transition kernels Euler} is immediate from \eqref{eq: E-M scheme definition} by noting that the support of $\bar p^{\text{E-M}}_\e\left( \cdot,x \right)$ is the closure of those elements whose successor by the forward process can be $x$.

Unpacking the result of Lemma \ref{lemma: support of transition kernels Euler} in the case of underdamped Langevin dynamics yields
\begin{align*}
   \operatorname{supp}p^{\text{E-M}}_\e\left( \cdot,x \right)=\overline{\{y : y_q = x_q +\e x_p\}}, \quad
   \operatorname{supp}\bar p^{\text{E-M}}_\e\left(\cdot,x \right) = \overline{\{y: y_q + \e y_p = x_q\}},
\end{align*}
where $x:=(x_q,x_p)$ respectively denote position and momenta. One can see that $p^{\text{E-M}}_\e\left( \cdot,x \right) \perp \bar p^{\text{E-M}}_\e\left( \cdot,x \right), \forall x\in \R^d$. From Definition \ref{def: EPR numerical scheme}, we deduce that the entropy production rate of E-M applied to the underdamped process is infinite for any time-step $\e >0$.


\begin{figure}[ht!]
    \centering
    \includegraphics[width= 0.45\textwidth]{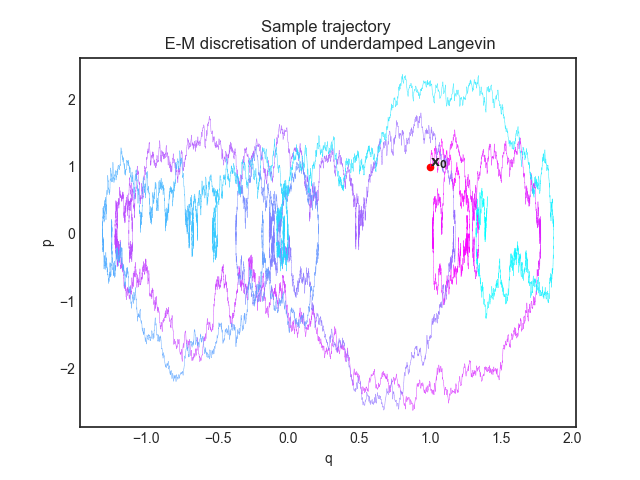}
    \includegraphics[width= 0.45\textwidth]{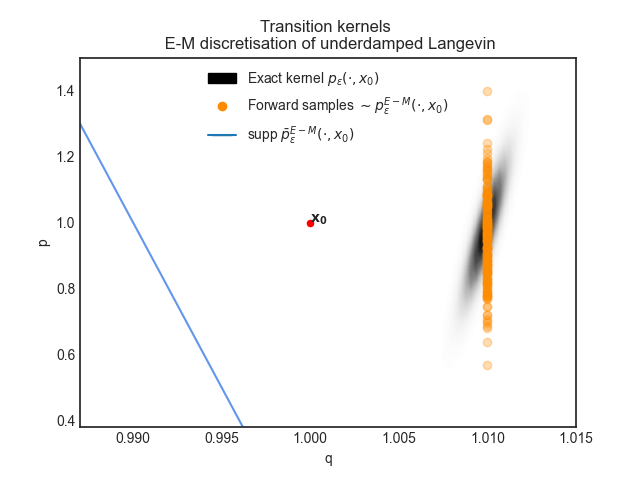}
    \caption[Euler-Maruyama simulation of underdamped Langevin dynamics]{\footnotesize\textbf{Euler-Maruyama simulation of underdamped Langevin dynamics.} This figure compares the Euler-Maruyama simulation of underdamped Langevin dynamics with the exact simulation available in Figure \ref{fig: exact underdamped}. The choice of parameters was the same: $V(q)=q^2/2, M=\gamma =1$. The upper left panel plots a sample trajectory of the numerical scheme. One observes that the numerical scheme is not confined to a prespecified region of space like the true process. The upper right panel plots samples of the numerical scheme after a time-step of $\e$ (in orange) given an initial condition at $x_0$ (in red). This is superimposed onto a heat map of the true transition kernel (in black). We see that samples from the numerical scheme are in the right region of space, but are confined to a subspace which is not aligned with the heat map of the true transition kernel. The support of the transition kernel of the time-reversed scheme is shown in blue. One sees that the supports of forward and reverse transition kernels are mutually singular, thus the entropy production of the numerical discretisation is infinite for any time-step, which differs from the true process which has finite entropy production for any positive time-step.
    }
    \label{fig: Euler underdamped}
\end{figure}

\subsubsection{BBK discretisation}

Contrariwise to Euler, the BBK integrator \cite{brungerStochasticBoundaryConditions1984,roussetFreeEnergyComputations2010}
is a splitting scheme that when applied to underdamped Langevin yields absolutely continuous transition kernels. The numerical scheme consists of three intermediate steps
$$
\begin{aligned}
&p_{i+\frac{1}{2}}=p_{i}-\nabla V\left(q_{i}\right) \frac{\e}{2}-\gamma M^{-1} p_{i} \frac{\e}{2}+\sqrt{2 \gamma \beta^{-1}} \omega_{i} \\
&q_{i+1}=q_{i}+M^{-1} p_{i+\frac{1}{2}} \e \\
&p_{i+1}=p_{i+\frac{1}{2}}-\nabla V\left(q_{i+1}\right) \frac{\e}{2}-\gamma M^{-1} p_{i+\frac{1}{2}} \frac{\e}{2}+\sqrt{2 \gamma \beta^{-1}} \omega_{i+\frac{1}{2}}
\end{aligned}
$$
with $\omega_{i}, \omega_{i+\frac{1}{2}} \sim N\left(0, \frac{\e}{2} \operatorname{Id}\right)$. Its stability and convergence properties were studied in \cite{brungerStochasticBoundaryConditions1984,roussetFreeEnergyComputations2010} and its ergodic properties in \cite{mattinglyConvergenceNumericalTimeAveraging2010,mattinglyErgodicitySDEsApproximations2002,talayStochasticHamiltonianSystems2002}.

It was shown in \cite[Theorem 4.3]{katsoulakisMeasuringIrreversibilityNumerical2014} that the BBK discretisation of the underdamped Langevin process is quasi time-reversible, so that
\begin{align}
\label{eq: ep BBK}
    e_p^{\mathrm{BBK}}\leq O(\e).
\end{align}

One sees (cf. Figures \ref{fig: Euler underdamped} and \ref{fig: BBK underdamped} vs Figure \ref{fig: exact underdamped}) that the BBK integrator better approximates the transition kernels than E-M does, however BBK still is largely inaccurate from the point of view of the entropy production rate as the simulation becomes reversible when the time-step tends to zero.

\begin{figure}[ht!]
    \centering
    \includegraphics[width= 0.45\textwidth]{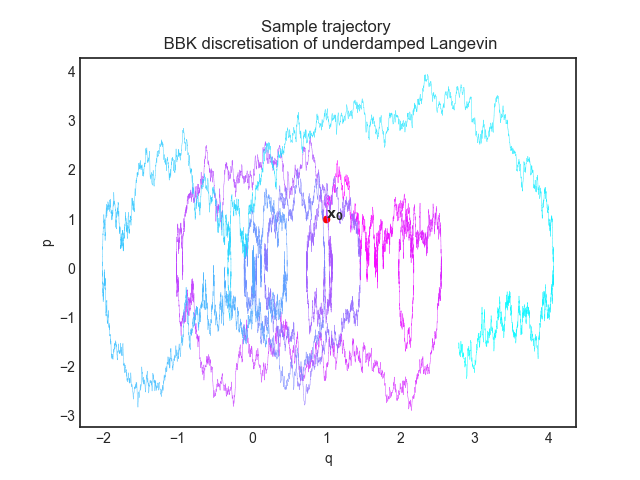}
    \includegraphics[width= 0.45\textwidth]{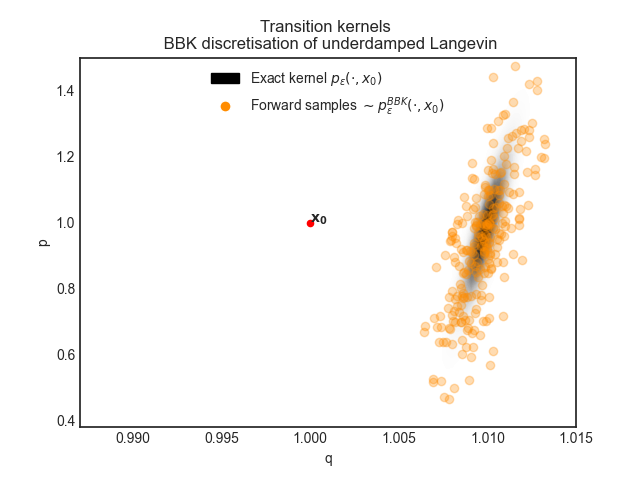}
    \includegraphics[width= 0.5\textwidth]{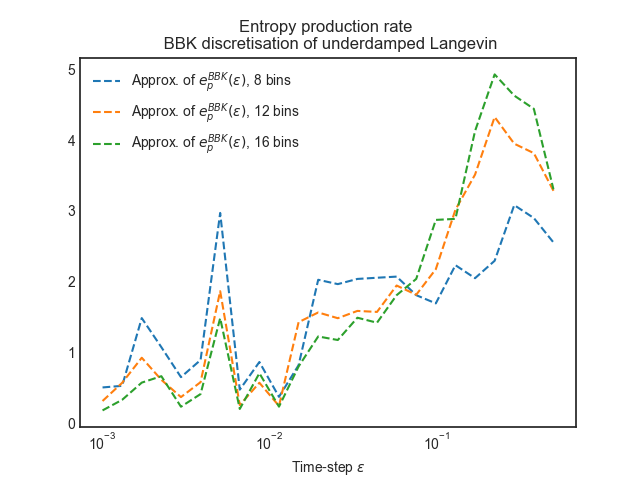}
    \caption[BBK simulation of underdamped Langevin dynamics]{\footnotesize\textbf{BBK simulation of underdamped Langevin dynamics.} This figure compares the BBK simulation of underdamped Langevin dynamics with the exact simulation available in Figure \ref{fig: exact underdamped}. The choice of parameters was the same: $V(q)=q^2/2, M=\gamma =1$. The upper left panel plots a sample trajectory of the numerical scheme. One observes that the numerical scheme is not confined to a prespecified region of space like the true process. The upper right panel plots samples of the numerical scheme after a time-step of $\e$ (in orange) given an initial condition at $x_0$ (in red). This is superimposed onto a heat map of the true transition kernel (in black). We see that samples from the numerical scheme fit the true transition kernel relatively well, but have a higher variance. The bottom panel estimates the entropy production rate of the numerical scheme for several choices of time-step $\e$. This is done by discretising the state-space into a number of bins and numerically evaluating the entropy production of the resulting Markov chain using samples, see Section \ref{sec: ep of numerical schemes theory} for details. The numerical values are consistent with the theoretical result \eqref{eq: ep BBK}.
    }
    \label{fig: BBK underdamped}
\end{figure}

\subsubsection{Summary}

In summary, the underdamped Langevin process has infinite entropy production rate in phase space\footnote{Recall that the $e_p$ we computed here is defined \textit{without} an additional momentum flip operator on the path space measure of the time-reversed process, i.e., \eqref{eq: def epr} and not \eqref{eq: rem def gen epr}. See also the discussion in Section \ref{sec: discussion generalised non-reversibility},}, but finite entropy production rate for any exact simulation with a positive time-step. When the potential is non-quadratic, the process is a non-linear diffusion that one usually cannot simulate exactly. To simulate it as accurately as possible, one should seek an approximating numerical scheme that has finite entropy production for any time-step, and whose entropy production tends to infinity for infinitesimally small time-steps.

Two well-known choices of numerical discretisation are the Euler-Maruyama and BBK schemes. By comparing their transition kernels with an exact simulation, we saw that the BBK scheme is a much better approximation to the true process than Euler-Maruyama. Analysis of the entropy production rate shows how these discretisations still far short in capturing important statistical properties of the process: the E-M discretisation has infinite entropy production for any time-step; while the BBK discretisation has finite entropy production for any time-step, and vanishing entropy production for infinitesimally small time-steps. Whenever possible, a good way to choose a time-step size for the BBK integrator might be matching its entropy production rate with that of an exact simulation. These results indicate that employing a BBK scheme with very small time-steps might be inadequate. Luckily, large step-sizes are usually preferred in practice.

In conclusion, the entropy production rate is a useful statistic of stochastic processes that can be used as a tool to devise accurate numerical schemes, particularly in a non-equilibrium statistical physics or sampling context where preserving the amount of time-irreversibility is important. Future development of numerical schemes should take entropy production into account; for example, in developing numerical schemes for underdamped Langevin, one should seek a finite entropy production rate for any positive time-step, which tends to infinity when time-steps become infinitesimally small. Other numerical schemes should be analysed in future work, such as those based on the lexicon for the approximation of the underdamped process developed by Leimkühler and Matthews \cite[p. 269 \& 271]{leimkuhlerMolecularDynamicsDeterministic2015}.

\section{Discussion}
\label{sec: discussion}

Briefly, we unpack a couple of observations and possible extensions of this work.

\subsection{$e_p$ and sampling efficiency}

A well-known criterion for efficient sampling is time-irreversibility~\cite{hwangAcceleratingDiffusions2005,rey-belletIrreversibleLangevinSamplers2015,duncanVarianceReductionUsing2016,barpGeometricMethodsSampling2022}. Intuitively, non-reversible processes backtrack less often and thus furnish more diverse samples~\cite{nealImprovingAsymptoticVariance2004}. Furthermore, the time-irreversible part of the drift flows along the contours of the stationary probability density which yields mixing and accelerates convergence to the target measure. It is well known that removing non-reversibility worsens the spectral gap and the asymptotic variance of the MCMC estimator~\cite{rey-belletIrreversibleLangevinSamplers2015,duncanVarianceReductionUsing2016,hwangAcceleratingDiffusions2005}, which are two main indicators of the speed of convergence to stationary state \cite{barpGeometricMethodsSampling2022}. Thus efficient samplers at non-equilibrium steady-state have positive entropy production.

In elliptic linear diffusions, one can construct the optimal time-irreversible drift to optimise the spectral gap~\cite{lelievreOptimalNonreversibleLinear2013,wuAttainingOptimalGaussian2014}
or the asymptotic variance~\cite{duncanVarianceReductionUsing2016}. 
This indicates that one cannot optimise elliptic samplers by simply increasing their entropy production at steady-state without any other constraints, as, we recall, $e_p$ is a quadratic form of the strength of the time-irreversible drift (Figure \ref{fig: EPR as a function of gamma}).


Beyond this, the entropy production rate of general diffusions (Theorems \ref{thm: epr regular case simple}, \ref{thm: epr regular case general}) 
bears a formal resemblance to the Donsker-Varadhan functional \cite[Theorem 2.2]{rey-belletIrreversibleLangevinSamplers2015}, %
from which the asymptotic variance of MCMC estimators is derived \cite{rey-belletIrreversibleLangevinSamplers2015}. It is entirely possible that one might be able to relate the non-stationary entropy production rate (\eqref{eq: time dependent EPR} or \cite[eq. 3.19]{jiangMathematicalTheoryNonequilibrium2004}) 
to the Donsker-Varadhan functional, and thus give a more complete characterisation of sampling efficiency in terms of entropy production.

Many diffusion models of efficient sampling (the underdamped \eqref{eq: underdamped Langevin dynamics} and generalised \cite[eq. 8.33]{pavliotisStochasticProcessesApplications2014} Langevin dynamics, the fastest converging linear diffusion \cite{guillinOptimalLinearDrift2021}), and stochastic optimisation (stochastic gradient descent in deep neural networks \cite{chaudhariStochasticGradientDescent2018}) are not elliptic; that is, they are driven by less Brownian drivers than there are dimensions to their phase space. In particular, these processes have their forward and backward path space measures which are mutually singular, and infinite entropy production\footnote{\cite[Section 5]{chaudhariStochasticGradientDescent2018} shows that stochastic gradient descent is out of equilibrium. Furthermore, it shows empirically that the rank of the diffusion matrix is about 1\% of its dimension in deep neural networks. The sparsity of the noise with respect to the highly out-of-equilibrium behaviour they observe conjectures $b\irr(x) \not \in \operatorname{Range} \sigma(x)$ and thus, mutual singularity of forward and backward path space measures.}. In light of this, we conjecture that mutual singularity of the forward and backward path space measures is an important facet of sampling efficiency (provided the process is ergodic). Mutual singularity apparently exacerbates the mixing effect that time-irreversibility introduces in the elliptic case. Heuristically, if some paths can be taken by the forward process and not by the backward process, these trajectories cannot be reversed, thus the process is constantly forced to visit new regions of phase space, which contributes to the (non-reversible) convergence to steady-state. 

If the above intuition holds, a useful statistic of sampling efficiency might be the probability that the forward process takes paths that cannot be taken by the backward process. By the Lebesgue decomposition theorem we can decompose the forward path space measure $\p$ into $\p_{\textrm{reg}}+\p_{\textrm{sing}}$ such that $\p_{\textrm{reg}} \ll \bar \p$ and $\p_{\textrm{sing}} \perp \bar \p$.  
 This statistic is the non-negative real number 
 \begin{align*}
    \p(\{\gamma \in C([0,T], \R^d): d\p/d \bar \p(\gamma)= +\infty\})= \p_{\textrm{sing}}\left(C([0,T], \R^d)\right),
 \end{align*}
where $d\p/d \bar \p$ is the Lebesgue derivative between forward and backward path space measures. 
Note that the linear diffusion that converges fastest to steady-state maximises the latter (under the constraint that the process remains ergodic) since it has only one Brownian driver \cite{guillinOptimalLinearDrift2021}. However, this statistic does not tell us all since the direction of the Brownian driver with respect to the drift and the stationary density is important to determine sampling efficiency. Yet, these observations indicate that employing diffusions with less Brownian drivers might be an advantage for sampling and optimisation (provided ergodicity is maintained). 
A careful investigation of these relationships is left to future work.

\subsection{Generalised non-reversibility and entropy production rate}
\label{sec: discussion generalised non-reversibility}

Many diffusions studied in statistical physics are not time-reversible but they are generalised reversible; that is, they are time-reversible up to a one-to-one transformation $\theta$ of phase-space which leaves the stationary measure invariant \cite[Section 5.1]{duongNonreversibleProcessesGENERIC2021}, \cite[eq. 2.32]{roussetFreeEnergyComputations2010}. For example, the underdamped langevin equation is generalised reversible---it is reversible up to momentum reversal (Example \ref{eg: Time reversal of underdamped Langevin dynamics}); the generalised Langevin equation is also generalised reversible. 


The entropy production, as defined in Definition \ref{def: epr}, measures time-irreversibility as opposed to generalised non-reversibility. However, as pointed out in Remark \ref{rem: physical relevance}, the physically meaningful definition of entropy production rate sometimes comprises additional operators applied to the path-space measure of the time-reversed process. This modified notion of $e_p$, which we refer to as generalised entropy production, usually takes the form of

\begin{equation}
\label{eq: def gen epr}
    e_p^{\mathrm{gen}, \theta} := \lim _{\e \downarrow 0} \frac{1}{\e} \H\left[\p_{[0, \e]}, \theta_\#\bar \p_{[0, \e]}\right],
\end{equation}
where $\theta_\#$ is the pushforward operator associated to an involution of phase-space $\theta$ that leaves the stationary distribution invariant. The generalised entropy production rate measures the generalised non-reversibility of the process; that is, the extent to which the process is time-irreversible up to the one-to-one transformation $\theta$.
Of course, generalised entropy production reduces to entropy production, as defined in Definition \ref{def: epr}, when $\theta \equiv \operatorname{Id}$.

Since generalised entropy production can sometimes be more physically meaningful, we spend the rest of this section computing it in a couple of examples.

It seems to be a general consensus in statistical physics that the physically relevant notion of entropy production for the underdamped Langevin process is the generalised entropy production when $\theta$ is the momentum reversal \cite{vanvuUncertaintyRelationsUnderdamped2019,eckmannEntropyProductionNonlinear1999,spinneyNonequilibriumThermodynamicsStochastic2012,luposchainskyEntropyProductionContinuous2013}. It is then a by-product of Example \eqref{eg: Time reversal of underdamped Langevin dynamics} that underdamped Langevin dynamics has zero (generalised) entropy production $e_p^{\mathrm{gen}, \theta}=0$, which contrasts with the infinite entropy production one obtains in the non-generalised case (Section \ref{sec: ep singularity}) when one sets $\theta \equiv \operatorname{Id}$.

Beyond this, generalised entropy production could be a useful construct to quantify how far certain diffusion processes are from being generalised reversible. 
For example we can quantify to what extent certain time-irreversible perturbations of underdamped Langevin dynamics are far from being generalised reversible up to momentum reversal.

\begin{example}[$e_p^{\mathrm{gen}, \theta}$ of perturbed underdamped Langevin dynamics]
Consider the following perturbations of underdamped Langevin dynamics \cite[eq. 8]{duncanUsingPerturbedUnderdamped2017}
\begin{equation}
\label{eq: perturbation of underdamped}
    \begin{cases}
        \mathrm{d} q_t =M^{-1} p_t \mathrm{~d} t- Q_1 \nabla V\left(q_t\right) \mathrm{d} t \\
\mathrm{d} p_t =-\nabla V\left(q_t\right) \mathrm{d} t- Q_2 M^{-1} p_t \mathrm{~d} t-\gamma M^{-1} p_t \mathrm{~d} t+\sqrt{2 \gamma \beta^{-1}} \mathrm{d} w_t,
    \end{cases}
\end{equation}
where $Q_1, Q_2 \in \mathbb{R}^{d \times d}$ are constant antisymmetric matrices. By inspection this equation has a Helmholtz decomposition that is similar to underdamped Langevin dynamics (cf. Example \ref{eg: helmholtz decomp Langevin})
\begin{align*}
    &b_{\mathrm{rev}}(q,p)= D \nabla \log \rho(q,p) , \quad b_{\mathrm{irr}}(q,p)= Q \nabla \log \rho(q,p)\\
    \nabla \log  \rho(q,p)&= -\beta \begin{bmatrix}
         \nabla V(q)\\
         M^{-1}p
    \end{bmatrix},
    \quad D=\begin{bmatrix}
        0 &0\\
        0& \gamma \beta^{-1}\id_n
    \end{bmatrix}
 , \quad Q=\beta^{-1}\begin{bmatrix}
        Q_1 &-\id_n\\
        \id_n & Q_2
    \end{bmatrix}.
\end{align*}
The time-reversed process solves the following SDE (Section \ref{sec: helmholtz decomposition})
\begin{align*}
\begin{cases}
\d \bar q_{t} = -M^{-1}\bar p_{t} \d t+ Q_1 \nabla V\left(\bar q_t\right)\mathrm{d} t \\
\d \bar p_{t} =\nabla V\left(\bar q_t\right) \mathrm{d} t+ Q_2 M^{-1} \bar p_t \mathrm{~d} t-\gamma M^{-1} \bar p_{t} \d t+\sqrt{2 \gamma \beta^{-1}} \d  w_{t}.
\end{cases}
\end{align*}
Define $\theta(q,p)=(q,-p)$ to be the momentum reversal transformation of phase space (that leaves underdamped Langevin dynamics invariant as shown in Example \ref{eg: Time reversal of underdamped Langevin dynamics}). Letting $\hat p_{t} = -\bar p_{t}$, 
the time-reversed momentum-flipped equation looks like
\begin{align}
\label{eq: time-reversed momentum flipped}
\begin{cases}
\d \bar q_{t} = M^{-1}\hat p_{t} \d t+ Q_1 \nabla V\left(\bar q_t\right)\mathrm{d} t \\
\d \hat p_{t} =-\nabla V\left(\bar q_t\right) \mathrm{d} t+ Q_2 M^{-1} \hat p_t \mathrm{~d} t-\gamma M^{-1} \hat p_{t} \d t+\sqrt{2 \gamma \beta^{-1}} \d \hat w_{t}.
\end{cases}
\end{align}
Denote by $b^{\mathrm{gen}, \theta}_{\mathrm{irr}}$ the vector field whose sign changes after successively applying these two transformations:
\begin{align*}
   b^{\mathrm{gen}, \theta}_{\mathrm{irr}}
        (q,p)
  = \begin{bmatrix}
        -Q_1 \nabla V\left( q\right)\\
        - Q_2 M^{-1} p
    \end{bmatrix}.
\end{align*}
It follows that the time-reversed, momentum-flipped equation \eqref{eq: time-reversed momentum flipped} does not induce the same path space measure as the initial equation \eqref{eq: perturbation of underdamped} unless $Q_1=Q_2=0$. To see this,  
we follow the proofs of Theorems \ref{thm: epr singular} and \ref{thm: epr regular case simple} to compute the generalised entropy production rate
\begin{align*}
     Q_1 \neq 0 \Rightarrow \p \perp \theta_\# \bar \p \Rightarrow e_p^{\mathrm{gen}, \theta}&= +\infty,\\
    Q_1 =0 \Rightarrow \p \sim \theta_\# \bar\p \Rightarrow e_p^{\mathrm{gen}, \theta}&= \iint_{\R^n}b^{\mathrm{gen}, \theta}_{\mathrm{irr}}\cdot  D^- b^{\mathrm{gen}, \theta}_{\mathrm{irr}} \rho(q,p) \,\d p \,\d q\\
    &= \gamma^{-1}\beta \int_{\R^n} (Q_2 M^{-1}p)^2\rho(p)\d p \\
    &=-\gamma^{-1}\beta \tr\left(Q_2 M^{-1}Q_2\right)< +\infty.
\end{align*}
The last line equality follows from a standard result about expectations of bilinear forms under Gaussian distributions, since $\rho(p)$ is Gaussian with covariance matrix $M$. As usual, the generalised entropy production rate is a quadratic form of the (generalised) irreversible drift.
\end{example}

\subsection{Geometric interpretation of results}
\label{sec: geom}

Our main results concerning the value of entropy production have a straightforward geometric interpretation. The Stratonovich interpretation of the SDE 
\begin{align*}
	dx_t = b^s(x_t)+ \sigma(x_t)\circ dw_t
\end{align*}
is the natural one to consider in a geometric context, when looking at the directions of the drift $b^s$ and volatility vector fields $\sigma_{\cdot i },i=1, \ldots, m$ (i.e., the columns of the volatility matrix field).

Recall from Remark \ref{rem: stratonovich Helmholtz} that the Stratonovich SDE also admits a Helmholtz decomposition $b^s= b^s\rev + b^s\irr$ with $b^s\irr= b\irr$, so that 
\begin{equation}
\label{eq: sumary}
b^s(x) \in \im \sigma(x) \iff b^s\irr (x) \in \im \sigma(x)\iff \bar b^s(x) \in \im \sigma(x) \text{ for any } x\in \operatorname{supp}\mu,
\end{equation}
where $\bar b^s$ is the drift of the time-reversed Stratonovich SDE. In particular, time-reversal is a transformation that sends $b^s$ to $\bar b^s$, $b\irr^s$ to $-b\irr^s$, or, equivalently, adds $-2 b\irr^s$ to the drift.

Our main results can be summarised in a nutshell:
\begin{equation}
\label{eq: main results stratonovich interpretation summary}
\begin{split}
  	\mu\left( \left\{x \in \R^d: b^s(x)\in \im \sigma(x) \right\}\right)= 1 &\Rightarrow e_p = \intr b\irr^s\cdot  D^- b\irr^s \d \mu \quad \text{(see Theorem \ref{thm: epr regular case simple} or \ref{thm: epr regular case general} for details)},\\
	\mu\left( \left\{x \in \R^d: b^s(x)\in \im \sigma(x) \right\}\right)< 1 &\Rightarrow e_p = +\infty \quad \text{(see Theorem \ref{thm: epr singular} for details)}.  
\end{split}
\end{equation}
We derived our main results using the Itô interpretation of an SDE because this allowed us to make more general statements, notably in the context of the general existence and uniqueness theorem of strong solutions to Itô SDEs; it turns out, however, that these results are more naturally interpreted in the Stratonovich context.

Consider the case where there is noise in the direction of the vector field $b^s$, (almost every-) where the process is; in other words, assume that $\mu\left( \left\{x \in \R^d: b^s (x)\in \im \sigma(x) \right\}\right)= 1$. Consider the process at any point $x\in \supp \mu$. In virtue of \eqref{eq: sumary}, the drifts of the forward and time reversed processes both live in $\im \sigma(x)$, the subset of the tangent space that is spanned by the volatility vector fields. Since the driving fluctuations are Gaussian on $\im \sigma(x)$, the time-reversal transformation will be reversed by the random fluctuations with positive probability. Thus, the forward and time-reversed Markov transition kernels (for an infinitesimally small time-step) have the same support---they are mutually equivalent. Under sufficient regularity, made explicit in Theorems \ref{thm: epr regular case simple} or \ref{thm: epr regular case general}, their relative entropy is finite. The $e_p$ is the relative entropy between such Markov kernels on an infinitesimally small time-step (Proposition \ref{prop: epr transition kernels}), so it too will be finite.

On the other hand, if there exists $x \in \supp \mu$ such that there is no noise in the direction of the vector field $b^s$, that is $b^s (x)\not \in \im \sigma(x)$, then the direction of the forward and time-reversed dynamics in an infinitesimal time-step lie on different tangent spaces, $b^s (x) + \im \sigma(x)$ and $  \bar b^s (x) + \im \sigma(x)$, respectively. This means that the forward and time-reversed transition kernels (for an infinitesimally small time-step) are mutually singular and their relative entropy is infinite; thus, the $e_p$ is also infinite.

In particular, it should be straightforward to extend these observations and calculations to diffusions on manifolds.

\section{Addendum: Proofs for Chapter 1}

Here we provide proofs supporting Chapter 1.

\subsection{The $e_p$ of stationary Markov processes}
\label{app: proofs ep stationary Markov processes}

\subsubsection{$e_p$ in terms of path space measures with deterministic initial condition}
\label{app: aggregating local ep}

We prove Proposition \ref{prop: aggregating local ep}:
\begin{proof}
The proof is straightforward
\begin{equation*}
\begin{split}
     e_p 
     &= \frac 1 t \H\left[\p_{[0,t]}\mid\bar \p_{[0,t]}\right] =  \frac 1 t\E_{x_\bullet \sim \p}\left [\log \frac{d \p_{[0,t]}}{d \bar \p_{[0,t]}}(x_\bullet)\right] \\
     &=  \frac 1 t\E_{x \sim \mu}\left[ \E_{x_\bullet \sim \p^{x}_{[0,t]}}\left[\log \frac{d \p_{[0,t]}}{d \bar \p_{[0,t]}}(x_\bullet)\right]\right]\\
     &=  \frac 1 t\E_{x \sim \mu}\left[ \E_{x_\bullet \sim \p^{x}_{[0,t]}}\left[\log \frac{d \p^{x}_{[0,t]}}{d \bar \p^{x}_{[0,t]}}(x_\bullet)\right]\right] =  \frac 1 t\E_{x \sim \mu}\left[\H\left[\p^{x}_{[0,t]}\mid\bar \p^{x}_{[0,t]}\right]\right].
\end{split}
\end{equation*}
\end{proof}

\subsubsection{$e_p$ in terms of transition kernels}
\label{app: epr transition kernels}

We prove Proposition \ref{prop: epr transition kernels}:
\begin{proof}
By Proposition \ref{prop: aggregating local ep},
\begin{equation*}
\begin{split}
     e_p 
     &= \lim_{\e \downarrow 0}\frac 1 \e \E_{x \sim \mu}\left[ \E_{x_\bullet \sim \p^{x}_{[0,\e]}}\left[\log \frac{d \p^{x}_{[0,\e ]}}{d \bar \p^{x}_{[0,\e]}}(x_\bullet)\right]\right]\\
     &= \lim_{\e \downarrow 0} \frac{1}{\e} \E_{x \sim \mu}\left[\E_{y \sim p_\e(\cdot,x)}\left[ \log \frac{d p_\e(\cdot,x)}{d \bar p_\e(\cdot,x)}(y)\right]\right]\\
     &= \lim_{\e \downarrow 0} \frac{1}{\e} \E_{x\sim \mu}\left[\H\left[p_\e(\cdot,x)\mid\bar p_\e(\cdot,x)\right]\right].
\end{split}
\end{equation*}
\end{proof}

\subsection{Time-reversal of stationary diffusions}

\subsubsection{Conditions for the reversibility of the diffusion property}
\label{app: reversibility of the diffusion property}

We prove Lemma \ref{lemma: reversibility of the diffusion property}:
\begin{proof}
Recall the following facts:
\begin{itemize}
    \item $(x_t)_{t\in [0,T]}$ is a Markov diffusion process. Its generator is an unbounded, linear operator given by
\begin{equation}
\label{eq: generator of diffusion process}
\begin{split}
    \L : C_c^\infty(\R^d) &\subset  \dom \L \subset L^p_\mu(\R^d) \to L^p_\mu(\R^d) , \quad 1 \leq p \leq \infty, \quad \L f = b \cdot \nabla f + D \nabla \cdot \nabla f.
\end{split}
\end{equation}
    \item The time-reversal of a Markov process is also a Markov process. Let $\bar \L$ be the generator of the time-reversed process $(\bar x_t)_{t\in [0,T]}$. It is known that $\bar \L$ is the adjoint of $\L$. In other words, we have the identity
\begin{align}
\label{eq: adjoint generator}
    \intr f \L g \:\d \mu = \intr g \bar \L f \: \d \mu, \quad \forall f \in \dom \bar \L, g \in \dom \L,
\end{align}
where
    $\dom \bar \L =\left\{f \in L^1_\mu(\R^d) \mid  \exists h \in L^1_\mu(\R^d),  \forall g\in \dom \L: \intr f \L g \d \mu =\intr h g \d \mu  \right\}$. 
     This follows from the fact that the Markov semigroup of the time-reversed process is the adjoint semigroup \cite[p. 113]{jiangMathematicalTheoryNonequilibrium2004}, 
    and thus the infinitesimal generator is the adjoint generator
    \cite{yosidaFunctionalAnalysis1995,pazySemigroupsLinearOperators2011}, \cite[Thm 4.3.2]{jiangMathematicalTheoryNonequilibrium2004}. 
    \item $L^1\loc$-functions define distributions, and hence admit distributional derivatives (which need not be functions).
\end{itemize}

We identify the generator of the time-reversed process by computing the adjoint of the generator. In the following, all integrals are with respect to the $d$-dimensional Lebesgue measure. Let $f,g \in C_c^\infty(\R^d)$.  Noting that $f\rho b \cdot \nabla g, f\rho D\nabla \cdot \nabla g \in L^1(\R^d)$, 
we have 
\begin{align*}
    \intr f \L g \rho = \intr f \rho b \cdot \nabla g +\intr f\rho D\nabla \cdot \nabla g.
\end{align*}
On the one hand, noting that $f\rho b ,\rho b \in L^1\loc(\R^d, \R^{d})$, we have
\begin{align*}
    \intr f \rho b \cdot \nabla g &= - \intr g \nabla \cdot (f\rho b) 
    =-\intr g \left( \rho b \cdot \nabla f + f  \nabla \cdot (\rho b) \right)\\
    &= -\intr g\left( \rho b \cdot \nabla f + f \nabla \cdot \nabla \cdot (\rho D) \right),
\end{align*}
where the last equality follows from the stationary Fokker-Planck equation. 
(Recall that local boundedness of coefficients $b, \sigma$, and Itô's formula imply that the stationary density $\rho$ satisfies $\nabla \cdot (-b \rho + \nabla \cdot (D\rho))=0$
where the equality is in a distributional sense).

On the other hand, noting that $f\rho D, \rho D \in L^1\loc(\R^d, \R^{d \times d})$
, we have
\begin{align*}
    \intr f\rho D\nabla \cdot \nabla g &= \intr f\rho D\cdot \nabla \cdot \nabla g= \intr g \nabla \cdot \nabla \cdot (f \rho D)\\&= \intr g \nabla \cdot \left(\rho D \nabla f + f\nabla \cdot (\rho D) \right)\\
    &=\intr g \left( 2\nabla \cdot (\rho D)\cdot \nabla f + \rho D \nabla \cdot \nabla f +  f \nabla \cdot\nabla \cdot (\rho D) \right).
\end{align*}

Finally, summing the previous two equations yields:
\begin{align*}
  \intr f \L g \rho =\intr g \left(-\rho b \cdot \nabla f + 2 \nabla \cdot (D\rho )\cdot \nabla f+ \rho D\nabla \cdot \nabla f \right)  = \intr g \bar \L f \rho.
\end{align*}
And thus, the generator of the time-reversed process satisfies $\rho \bar \L f  = -\rho b \cdot \nabla f + 2 \nabla \cdot (D\rho )\cdot \nabla f+ \rho D\nabla \cdot \nabla f$ for all $ f\in C_c^\infty(\R^d)$. The time-reversed process is a diffusion if its generator is a second order differential operator with no constant part. This is the case here, except for the fact that the generator outputs distributions as opposed to functions. For the generator to be a diffusion operator we need to assume that the distributional derivative $\nabla \cdot (D\rho )$ is indeed a function (which is then necessarily in $L^1\loc(\R^d, \R^d)$). Thus, the following are equivalent:
\begin{itemize}
    \item $\nabla \cdot (D \rho) \in L^1\loc(\R^d, \R^d)$,
    \item $\bar \L f \in L^1_\mu(\R^d)$ for any $ f\in C_c^\infty(\R^d)$, where $\bar \L f= -b \cdot \nabla f + 2 \rho^{-1}\nabla \cdot (D\rho )\cdot \nabla f+ D\nabla \cdot \nabla f$,
    \item $(\bar x_t)_{t\in [0,T]}$ is a Markov diffusion process. 
\end{itemize}
\end{proof}

\subsubsection{The time-reversed diffusion}
\label{app: time reversed diffusion}

We prove Theorem \ref{thm: time reversal of diffusions}.

\begin{proof}
Since $(\bar x_t)_{t\in [0,T]}$ is a Markov diffusion process with generator $\bar \L$, we have shown that its drift and diffusion are indeed $\bar b, D$, in the proof of Lemma \ref{lemma: reversibility of the diffusion property}.

To show that any such diffusion process induces the path space measure of the time-reversed process, it suffices to show that the martingale problem associated to $(\bar \L, \rho)$ is well-posed. First note that, by Assumption \ref{ass: coefficients time reversal}, the Itô SDE \eqref{eq: Ito SDE} has a unique strong solution. Therefore it also has a unique weak solution. 
Therefore, $(x_t)_{t\in [0,T]}$ is the unique solution to the martingale problem associated to the generator $\L= b \cdot \nabla + D \nabla \cdot \nabla $ 
\cite[Theorem 1.1]{kurtzEquivalenceStochasticEquations2011}. 
In other words, the martingale problem associated to $(\L, \rho)$ is well-posed. It remains to show that there is a one-to-one correspondence between stationary solutions to the martingale problem associated to $\L$ and $\bar \L$. 

Consider Markov processes $(y_t)_{t\in [0,T]}$, $(\bar y_t)_{t\in [0,T]}$, $\bar y_t= y_{T-t}$ stationary at the density $\rho$. We show that $(\bar y_t)_{t\in [0,T]}$ solves the martingale problem wrt $\bar \L$ if and only if $(y_t)_{t\in [0,T]}$ solves the martingale problem wrt $\L$.

\begin{itemize}
    \item $(\bar y_t)_{t\in [0,T]}$ solves the martingale problem wrt $\bar \L$ if and only if for arbitrary $0\leq s\leq t \leq T$, $f,g \in C_c^\infty(\R^d)$
\begin{equation}
\label{eq: equivalent formulations of martingale problem}
\begin{split}
    &\E\left[f(\bar y_t)- \int_0^t\bar \L f(\bar y_r) \d r \mid \bar y_\theta, 0 \leq \theta \leq s\right] = f(\bar y_s)- \int_0^s \bar \L f(x_r) \d r\\
   \stackrel{\text{Markov}}{\iff} &\E\left[f(\bar y_t)- \int_0^t\bar \L f(\bar y_r) \d r \mid \bar y_s\right] = f(\bar y_s)- \int_0^s \bar \L f(x_r) \d r\\
    \iff & \E\left[f(\bar y_t)-f(\bar y_s)- \int_s^t\bar \L f(\bar y_r) \d r \mid \bar y_s\right]=0\\
    \iff & \E\left[\left(f(\bar y_t)-f(\bar y_s)- \int_s^t\bar \L f(\bar y_r) \d r \right) g(\bar y_s)\right]=0.
\end{split}
\end{equation}
If we make the change of variable $t \leftarrow T-s, s \leftarrow T-t$, so that $0 \leq s \leq t \leq T$, this is equivalent to:
\begin{align*}
    \iff & \E\left[\left(f(y_s)-f( y_t)- \int_{T-t}^{T-s} \bar \L f( y_{T-r}) \d r \right) g(y_t)\right]=0\\
    \iff & \E\left[\left(f(y_s)-f( y_t)- \int_{s}^t \bar \L f(y_r) \d r \right) g(y_t)\right]=0
\end{align*}
\item Repeating the equivalences in \eqref{eq: equivalent formulations of martingale problem}, $(y_t)_{t\in [0,T]}$ solves the martingale problem wrt $\L$ if and only if for arbitrary $0\leq s\leq t \leq T$, $f,g \in C_c^\infty(\R^d)$
\begin{align*}
    \E\left[\left(g(y_t)-g(y_s)- \int_s^t\L g( y_r) \d r \right) f( y_s)\right]=0.
\end{align*}
\item Thus, it suffices to show that the two last expressions are equal, i.e.,
\begin{align*}
    \E\left[\left(f(y_s)-f( y_t)- \int_{s}^t \bar \L f(y_r) \d r \right) g(y_t)\right]= \E\left[\left(g(y_t)-g(y_s)- \int_s^t\L g( y_r) \d r \right) f( y_s)\right]
\end{align*}

By stationarity, we have 
\begin{align*}
    \E\left[\left(f(y_s)-f( y_t)\right) g(y_t)\right]&=\E\left[f(y_s)g(y_t)-f( y_t)g(y_t)\right]\\
    =\E\left[f(y_s)g(y_t)-f( y_s)g(y_s)\right]&=\E\left[\left(g(y_t)-g(y_s)\right) f( y_s)\right].
\end{align*}
Thus, it remains to show that
\begin{align*}
    \E\left[\int_{s}^t g(y_t) \bar \L f(y_r) \d r \right]= \E\left[\int_s^tf( y_s)\L g( y_r) \d r \right]
\end{align*}
We proceed to do this. On the one hand: 
\begin{align*}
    \E\left[\int_s^tf( y_s)\L g( y_r) \d r \right] &=\int_s^t \E\left[f( y_s)\L g( y_r) \right]\d r \\
    &= \int_s^t \E \left[\E\left[f( y_s)\L g( y_r) \mid y_s\right]\right]\d r\\
    &=\int_s^t \E \left[f( y_s)\E\left[\L g( y_r) \mid y_s\right]\right]\d r\\
    &= \int_s^t \E \left[f( y_s)\operatorname{P}_{r-s}\L g(y_s)\right]\d r\\
    &=\int_s^t \intr f( y)\operatorname{P}_{r-s}\L g(y) \rho(y) \d y\d r \quad \text{(stationarity)}\\
    &= \intr f( y)\int_s^t \operatorname{P}_{r-s}\L g(y) \d r \rho(y) \d y \\
    &=\intr f(y) \left(\operatorname{P}_{t-s} -\operatorname{P}_0\right)g(y)\rho(y) \d y \quad \text{($\partial_t \operatorname{P}_{t}= \operatorname{P}_{t} \L
    $)}\\
    &=\intr g(y) \left(\bar{\operatorname{P}}_{t-s} -\bar{\operatorname{P}}_0\right)f(y)\rho(y) \d y
\end{align*}

On the other hand: 
\begin{align*}
    \E\left[\int_s^t g( y_t)\bar \L f( y_r) \d r \right] &= \int_s^t \E\left[g( y_t)\bar \L g( y_r) \right]\d r\\
    &= \int_s^t \E \left[\E\left[g( y_t)\bar \L f( y_r) \mid y_t\right]\right]\d r\\
    &= \int_s^t \E \left[g( y_t)\E\left[\bar \L f( y_r) \mid y_t\right]\right]\d r\\
    &= \int_s^t \E \left[g( y_t)\bar {\operatorname{P}}_{t-r}\bar \L f( y_t)\right]\d r\\
    &= \int_s^t \intr g( y)\bar {\operatorname{P}}_{t-r}\bar \L f( y)\rho(y) \d y\d r \quad \text{(stationarity)}\\
    &=  \intr g( y) \int_s^t\bar {\operatorname{P}}_{t-r}\bar \L f( y) \d r\rho(y) \d y\\
    &= \intr g( y) \left(\bar{\operatorname{P}}_{t-s}-\bar{\operatorname{P}}_0\right)f(y) \rho(y) \d y \quad \text{($\partial_t \bar{\operatorname{P}}_{t}= \bar{\operatorname{P}}_{t} \bar \L
    $)}
\end{align*}
\end{itemize}
This shows the one-to-one correspondence and completes the proof.
\end{proof}

\subsubsection{The Helmholtz decomposition}
\label{app: helmholtz decomposition}

We prove Proposition \ref{prop: helmholtz decomposition}.


\begin{proof}
\begin{itemize}
    \item["$\Rightarrow$"]
    We define the time-reversible and time-irreversible parts of the drift
    \begin{align*}
        b_{\mathrm{rev}} :=\frac {b + \bar b}{2}, \quad  b_{\mathrm{irr}} :=\frac {b - \bar b}{2}.
    \end{align*}
    We now show that the have the predicted functional form. For $x$ such that $\rho(x) =0$, $b_{\mathrm{rev}} =\left(b + \bar b\right)/{2}=0$. For $x$ such that $\rho(x) >0$ 
    \begin{align}
    \label{eq: time reversible drift}
        b_{\mathrm{rev}} &=\frac {b + \bar b}{2}= \rho^{-1}\nabla \cdot (D \rho) = \rho^{-1}D \nabla \rho +\rho^{-1}\rho \nabla \cdot D = D \nabla \log \rho + \nabla \cdot D.
    \end{align}
    For the time-irreversible drift, first note that the stationary density $\rho$ solves the stationary Fokker-Planck equation \cite{pavliotisStochasticProcessesApplications2014,riskenFokkerPlanckEquationMethods1996}
    \begin{equation*}
        \nabla \cdot (-b\rho + \nabla \cdot (D \rho))=0.
    \end{equation*}
    Decomposing the drift into time-reversible and time-irreversible parts, from \eqref{eq: time reversible drift}
    \begin{align*}
        -b\rho + \nabla \cdot (D \rho)= -b_{\mathrm{rev}}\rho- b_{\mathrm{irr}}\rho + \nabla \cdot (D \rho)=- b_{\mathrm{irr}}\rho,
    \end{align*}
    we obtain that the time-irreversible part produces divergence-free (i.e., conservative) flow w.r.t. the steady-state density
    \begin{equation*}
        \nabla \cdot (b_{\mathrm{irr}}\rho)=0.
    \end{equation*}
    \item["$\Leftarrow$"] From \eqref{eq: time reversible drift} the time-reversible part of the drift satisfies the following identity 
    \begin{align}
        b_{\mathrm{rev}}\rho &= \nabla \cdot (D \rho).
    \end{align}
    
    It follows that the density $\rho$ solves the stationary Fokker-Planck equation
    \begin{align*}
        \nabla \cdot (-b\rho + \nabla \cdot (D \rho)) &= \nabla \cdot (-b_{\mathrm{rev}}\rho-b_{\mathrm{irr}}\rho + \nabla \cdot (D \rho))= \nabla \cdot (-b_{\mathrm{irr}}\rho) =0.
    \end{align*}
\end{itemize}
\end{proof}


We prove Proposition \ref{prop: characterisation of irreversible drift}:

\begin{proof}
\begin{itemize}
    \item["$\Rightarrow$"] Recall that any smooth divergence-free vector field is the divergence of a smooth antisymmetric matrix field $A=-A^\top$ \cite{RealAnalysisEvery,yangBivectorialNonequilibriumThermodynamics2021,grahamCovariantFormulationNonequilibrium1977, eyinkHydrodynamicsFluctuationsOutside1996}
    \begin{equation*}
        b_{\mathrm{irr}}\rho = \nabla \cdot A.
    \end{equation*}
    This result holds most generally a consequence of Poincaré duality in de Rham cohomology \cite[Appendix D]{yangBivectorialNonequilibriumThermodynamics2021}. We define a new antisymmetric matrix field $Q:= \rho^{-1}A$. From the product rule for divergences we can rewrite the time-irreversible drift as required
    \begin{equation*}
       b_{\mathrm{irr}} = Q \nabla \log \rho + \nabla \cdot Q.
    \end{equation*}
    \item["$\Leftarrow$"] Conversely, we define the auxiliary antisymmetric matrix field $A := \rho Q$. Using the product rule for divergences it follows that
    \begin{equation*}
        b_{\mathrm{irr}} = \rho^{-1}\nabla \cdot A.
    \end{equation*}
    Finally, 
    \begin{equation*}
       \nabla \cdot( b_{\mathrm{irr}}\rho) = \nabla \cdot (\nabla \cdot A)= 0
    \end{equation*}
    as the matrix field $A$ is smooth and antisymmetric.
\end{itemize}
\end{proof}

\subsubsection{Multiple perspectives on the Helmholtz decomposition}
\label{app: multiple perspectives on Helmholtz}

We prove Proposition \ref{proposition: hypocoercive decomposition of the generator}:

\begin{proof}[Proof of Proposition \ref{proposition: hypocoercive decomposition of the generator}]
The proof is analogous to \cite[Proposition 3]{villaniHypocoercivity2009}. In view of Section \ref{sec: helmholtz decomp inf gen}, we only need to check that: 1) if $\sqrt 2 \operatorname \Sigma f =\sigma^\top \nabla f$, then $\sqrt 2 \operatorname \Sigma^*g=- \nabla \cdot(\sigma g) - \nabla \log \rho \cdot \sigma g$; 2) the symmetric part of the generator factorises as $\operatorname S = - \Sigma^* \Sigma$. 
\begin{enumerate}
    \item[1)] For any $f,g \in C_c^\infty (\R^d):$
    \begin{align*}
        \langle f, \sqrt 2\Sigma^*g \rangle_{L^2_\mu(\R^d)}&=
        \langle \sqrt 2\Sigma f, g \rangle_{L^2_\mu(\R^d)} = \int g \sigma^\top \nabla f \rho(x) \dx \\
        &= \int \sigma g \rho \cdot \nabla f (x) \dx
        =-\int f \nabla \cdot (\sigma g \rho)(x) \dx \\
        &= \int -f \nabla \log \rho \cdot \sigma g \rho(x)-f \nabla \cdot (\sigma g) \rho(x)  \dx \\
        &= \langle f,  - \nabla \log \rho \cdot \sigma g-\nabla \cdot (\sigma g)  \rangle_{L^2_\mu(\R^d)}
    \end{align*}
    This implies $\sqrt 2 \Sigma^* g= - \nabla \log \rho \cdot \sigma g-\nabla \cdot (\sigma g)  $.
    \item[2)] For any $f \in C_c^\infty(\R^d)$:
    \begin{align*}
        -\Sigma^*\Sigma f&= \nabla \log \rho \cdot D \nabla f+\nabla \cdot ( D\nabla f)\\
        &= \nabla \log \rho \cdot D \nabla f + (\nabla \cdot D)\cdot \nabla f + D\nabla \cdot \nabla f\\
        &= b_{\mathrm{rev}} \cdot \nabla f + D\nabla \cdot \nabla f= \operatorname S f
    \end{align*}
    where the penultimate equality follows since $b_{\mathrm{rev}}= D \nabla \log \rho + \nabla \cdot D$, $\mu$-a.e.
\end{enumerate}
\end{proof}

We now prove Proposition \ref{prop: GENERIC decomposition of the Fokker-Planck equation}:

\begin{proof}[Proof of Proposition \ref{prop: GENERIC decomposition of the Fokker-Planck equation}]
\begin{itemize}
    \item We compute the Fréchet derivative of $\H[\cdot \mid \rho]$. First of all, we compute its Gâteaux derivative in the direction of $\eta$.
    \begin{align*}
        \frac{d}{d\e}\operatorname{H}[\rho_t + \e \eta \mid \rho ]&= \frac{d}{d\e}\intr (\rho_t + \e \eta) \log \frac{\rho_t + \e \eta}{\rho}\d x
        = \intr \eta \log \frac{\rho_t + \e \eta}{\rho} +\eta \d x
        = \intr \eta \left( \log \frac{\rho_t + \e \eta}{\rho} +1\right)\d x.
    \end{align*}
    By definition of the Fréchet derivative, we have $\frac{d}{d\e}\left.\operatorname{H}[\rho_t + \e \eta \mid \rho ]\right|_{\e=0}= \langle \d \operatorname{H}[\rho_t \mid \rho ], \eta \rangle$. This implies $ \d \operatorname{H}[\rho_t \mid \rho ]=  \log \frac{\rho_t }{\rho} +1$ by the Riesz representation theorem.
    \item Recall, from Proposition \ref{proposition: hypocoercive decomposition of the generator}, that
 $\sqrt 2\Sigma= \sigma^\top\nabla$. We identify $\Sigma'$. For any $f,g \in C_c^\infty(\R^d)$
 \begin{align*}
     \langle g, \sqrt 2\Sigma f\rangle= \intr g \sigma^\top\nabla f \d x= \intr \sigma g \cdot \nabla f \d x = - \intr f \nabla \cdot (\sigma g) \d x.
 \end{align*}
 This yields $\sqrt 2\Sigma'g =-\nabla \cdot (\sigma g)$. And in particular, $\Sigma'(\rho_t \Sigma \xi)= -\nabla \cdot (\rho_t D\nabla \xi)$.
 \item We define $M_{\rho_t}(\xi) := \Sigma'(\rho_t \Sigma \xi)=-\nabla \cdot (\rho_t D\nabla \xi)$ and verify that this is a symmetric semi-positive definite operator. For any $g,h \in C_c^\infty(\R^d)$:
 \begin{align*}
   \langle \operatorname M_{\rho_t}h, g\rangle= \langle \Sigma h, \Sigma g \rangle = \langle h, \operatorname M_{\rho_t}g\rangle, \quad
   \langle \operatorname M_{\rho_t}g, g\rangle=\langle \Sigma g, \Sigma g \rangle_{\rho_t}\geq 0.
 \end{align*}
 Also, $-M_{\rho_t}\left( \d \operatorname{H}[\rho_t \mid \rho ]\right)=\nabla \cdot(\rho_t D \nabla \log \frac{\rho_t}{\rho} )$ is immediate.
 \item We define $W(\rho_t) = \nabla \cdot (-b\irr \rho_t)$ and verify the orthogonality relation:
 \begin{align*}
    \langle \operatorname W(\rho_t), \d \operatorname{H}[\rho_t \mid \rho ]\rangle &= \intr \left(\log \frac{\rho_t }{\rho} +1\right)\nabla \cdot (-b\irr \rho_t)\d x= \intr b\irr \rho_t \nabla \left(\log \frac{\rho_t }{\rho} +1\right) \d x \\
    &= \intr b\irr \rho_t \frac{\rho }{\rho_t}\nabla \left(\frac{\rho_t }{\rho}\right) \d x 
    = -\intr \nabla\cdot (b\irr \rho)\frac{\rho_t }{\rho}\d x =0, 
 \end{align*}
 where the last equality holds by Proposition \ref{prop: helmholtz decomposition}.
\end{itemize}
\end{proof}

\subsection{The $e_p$ of stationary diffusions}

\subsubsection{Regular case}
\label{app: epr regularity}

We prove Theorem \ref{thm: epr regular case simple}:

\begin{proof}
By Assumption \ref{ass: coefficients time reversal} the Itô SDE \eqref{eq: Ito SDE} has a unique non-explosive strong solution $\left(x_{t}\right)_{t \geq 0}$ with respect to the given Brownian motion $\left(w_{t}\right)_{t \geq 0}$ on a filtered probability space $(\Omega, \mathscr{F},\left\{\mathscr{F}_{t}\right\}_{t \geq 0}, P)$. (Even though Assumption \ref{ass: coefficients time reversal} ensures non-explosiveness of the solution on a finite time-interval, stationarity implies that we may prolong the solution up to arbitrary large times).


By Theorem \ref{thm: time reversal of diffusions} we know that a solution to the following SDE
\begin{align}
\label{eq: reversed SDE}
d\bar x_t= \bar b(\bar x_t) dt + \sigma(\bar x_t) dw_t, \quad \bar x_0 = 
x_0, 
\end{align}
induces the path space measure of the time-reversed process. By Proposition \ref{prop: helmholtz decomposition}, we can rewrite the (forward and time-reversed) drifts as $b = b_{\mathrm{rev}}+ b_{\mathrm{irr}}$ and $\bar b = b_{\mathrm{rev}}- b_{\mathrm{irr}}$.

We define the \textit{localised} coefficients
\begin{align*}
    b\n(x) &:= b\left( \left(1 \land \frac{n}{|x|}\right) x \right)= \begin{cases}
        b(x) \text{ if } |x| \leq n\\
        b\left(n \frac{x}{|x|}\right) \text{ if } |x| > n,
    \end{cases}
\end{align*}
and analogously for $\bar b\n, \sigma\n, b_{\mathrm{rev}}\n,b_{\mathrm{irr}}\n$. Note that the assignment $\cdot\n$ respects sums and products, in particular 
\begin{equation}
\label{eq: localised drift respects helmholtz decomposition}
\begin{split}
    b\n &= (b_{\mathrm{rev}}+ b_{\mathrm{irr}})\n = b_{\mathrm{rev}}\n+ b_{\mathrm{irr}}\n,\\
    \bar b\n &= (b_{\mathrm{rev}}- b_{\mathrm{irr}})\n = b_{\mathrm{rev}}\n- b_{\mathrm{irr}}\n.
\end{split}
\end{equation}

It is easy to see that the localised SDE
\begin{align*}
        dx_t\n = b\n (x_t\n) dt + \sigma\n (x_t\n) dw_t, \quad x_0\n= x_0
\end{align*}
also has a unique strong solution $x\n=(x_t\n)_{t \geq 0}$ with respect to the given Brownian motion $(w_t)_{t\geq 0}$ on the probability space $(\Omega, \mathscr{F},\left\{\mathscr{F}_{t}\right\}_{t \geq 0}, P)$.
This follows from the fact that the localised SDE has locally Lipschitz continuous and bounded coefficients that satisfy the assumptions of Theorem \cite[Theorem 3.1.1]{prevotConciseCourseStochastic2008}.

From assumption \ref{item: epr regular drift image of volatility}, we obtain that for $\rho$-a.e. $x \in \R^d$
\begin{equation}
\label{eq: localised drift in image of the volatility}
\begin{split}
    b_{\mathrm{irr}}(x) \in \im \sigma(x)
    &\Rightarrow b_{\mathrm{irr}}(x) = \sigma \sigma^- b_{\mathrm{irr}}(x)\\
    b_{\mathrm{irr}}\n(x) \in \im \sigma\n(x)
    &\Rightarrow b_{\mathrm{irr}}\n(x) = \sigma\n \sigma^{(n)^-} b_{\mathrm{irr}}\n(x).
\end{split}
\end{equation}
Then, \eqref{eq: localised drift respects helmholtz decomposition} and \eqref{eq: localised drift in image of the volatility} imply that we can rewrite the localised SDE as
\begin{align*}
        dx_t\n = b_{\mathrm{rev}}\n(x_t\n) dt+ \sigma\n(x_t\n)\left[\sigma^{(n)^-}b_{\mathrm{irr}}\n\left(x_t\n\right)dt+dw_t\right], \quad x_0\n= x_0.
\end{align*}
By the definition of Itô's stochastic calculus, $x_t\n$ is an $\mathscr F_t$-adapted process.
By assumption \ref{item: continuity}, $\sigma^{-}b_{\mathrm{irr}}$ is Borel measurable and thus it follows that the localised map $\sigma^{(n)^-}b_{\mathrm{irr}}\n$ is Borel measurable. 
Thus $-2\sigma^{(n)^-}b_{\mathrm{irr}}\n\left(x_t\n\right)$ is also an $\mathscr F_t$-adapted process. In addition, by continuity and localisation, $\sigma^{(n)^-}b_{\mathrm{irr}}\n$ is bounded. Therefore, by \cite[Proposition 10.17 (i)]{pratoStochasticEquationsInfinite2014} 
applied to $-2\sigma^{(n)^-} b_{\mathrm{irr}}\n\left(x_s\n\right)$,
\begin{align*}
    Z_t\n &= \exp \left[-2 \int_0^t \left\langle\sigma^{(n)^-} b_{\mathrm{irr}}\n\left(x_s\n\right), dw_s\right\rangle + \left|\sigma^{(n)^-} b_{\mathrm{irr}}\n\left(x_s\n\right) \right|^2 ds\right], t \geq 0,
\end{align*}
is a martingale on the probability space $\left(\Omega, \mathscr{F},\left\{\mathscr{F}_{t}\right\}_{t \geq 0}, P\right)$.
We define a new probability measure $\bar P_n$ on the sample space $\Omega$ through
\begin{equation}
\left.\frac{d \bar P_n}{d P}\right|_{\mathscr{F}_{t}}=Z_{t}^{(n)}, \quad \forall t \geq 0.
\end{equation}
By Girsanov's theorem \cite[Theorem 10.14]{pratoStochasticEquationsInfinite2014}, 
$x_t\n$ solves the SDE
\begin{align*}
        d x_t\n &= b_{\mathrm{rev}}\n(x_t\n) dt- \sigma\n(x_t\n)\left[\sigma^{(n)^-}b_{\mathrm{irr}}\n\left(x_t\n\right)dt+dw_t\right], \quad  x_0\n=  x_0.
\end{align*}
on the probability space $\left(\Omega, \mathscr{F},\left\{\mathscr{F}_{t}\right\}_{t \geq 0}, \bar P_{n}\right)$.
Using \eqref{eq: localised drift in image of the volatility}, $x_t\n$ solves
\begin{align*}
        d x_t\n &= \bar b\n ( x_t\n) dt + \sigma\n ( x_t\n) dw_t, \quad  x_0\n=  x_0
\end{align*}
on said probability space.

We define a sequence of stopping times $\tau_0=0$ and $\tau_{n}=\inf \left\{t \geq 0:\left|x_{t}\right|>n\right\}$ for $n>0$. Since $\left(x_{t}\right)_{t \geq 0}$ is non-explosive, $P$-a.s. $\lim_{n \to \infty} \tau_n = +\infty$.
As $x_t = x\n_t$ when $t \leq \tau_n$, we have $P$-a.s.
\begin{align*}
    Z_{t \land \tau_n}\n &= \exp \left[-2 \int_0^{t \land \tau_n} \left\langle \sigma^- b_{\mathrm{irr}}(x_s), dw_s\right\rangle + \left|\sigma^- b_{\mathrm{irr}}(x_s) \right|^2 ds\right].
\end{align*}
As $Z_{t}^{(n+1)} \mathds 1_{\left\{t<\tau_{n}\right\}}=Z_{t}^{(n)} \mathds 1_{\left\{t<\tau_{n}\right\}}$, we define the limit as $n \to +\infty$
\begin{align*}
Z_{t} \equiv \sum_{n=1}^{+\infty} Z_{t}^{(n)} \mathds 1_{\left\{\tau_{n-1} \leq t<\tau_{n}\right\}}=\lim _{n \rightarrow+\infty} Z_{t}^{(n)} \mathds1_{\left\{t<\tau_{n}\right\}}=\lim _{n \rightarrow+\infty} Z_{t \wedge \tau_{n}}^{(n)}.
\end{align*}
By definition, $Z_{t}$ is a continuous local martingale on $
\left(\Omega, \mathscr{F},\left\{\mathscr{F}_{t}\right\}_{t \geq 0}, P\right)$. 

We compute $Z_t$. Let's write $-\log Z_{t \wedge \tau_{n}}^{(n)}=M_{t}^{(n)}+Y_{t}^{(n)}$, where
\begin{align*}
    M_{t}^{(n)}=2 \int_0^{t \land \tau_n} \left\langle \sigma^- b_{\mathrm{irr}}(x_s), dw_s\right\rangle, \text{ and, } Y_{t}^{(n)}=2 \int_0^{t \land \tau_n}  \left|\sigma^- b_{\mathrm{irr}}(x_s) \right|^2 ds.
\end{align*}
We also define $M_{t}=2 \int_0^{t} \left\langle \sigma^- b_{\mathrm{irr}}(x_s), dw_s\right\rangle$ and $Y_t= 2\int_0^{t}  \left|\sigma^- b_{\mathrm{irr}}(x_s) \right|^2 ds$.

From assumption \ref{item: epr bounded}, we have
\begin{align*}
    \int_{\R^d} \left|\sigma^- b_{\mathrm{irr}}(x) \right|^2\rho(x) dx = \frac 1 2\int_{\R^d} b_{\mathrm{irr}}^\top D^- b_{\mathrm{irr}} \rho(x)\d x  <+\infty.
\end{align*}
Thus, we obtain that 
\begin{align*}
\E_P\left|M_{t}^{(n)}-M_{t}\right|^{2} &=4 \E_P\left|\int_{0}^{t} \left\langle \sigma^- b_{\mathrm{irr}}(x_s)\mathds 1_{\left\{s>\tau_{n}\right\}}, dw_s\right\rangle  \right|^{2} \\
&=4 \E_P \int_{0}^{t} \left|\sigma^- b_{\mathrm{irr}}(x_s) \right|^2  \mathds 1_{\left\{s>\tau_{n}\right\}} d s \quad \text{(Itô's isometry)}\\
&\xrightarrow{ n \to+\infty} 0, \\
\E_P\left|Y_{t}^{(n)}-Y_{t}\right| &=2 \E_P \int_{0}^{t} \left|\sigma^- b_{\mathrm{irr}}(x_s) \right|^2\mathds 1_{\left\{s>\tau_{n}\right\}} d s \xrightarrow{ n \to+\infty} 0.
\end{align*}
Thus, $-\log Z_{t}=M_{t}+Y_{t}$. By Itô calculus, $M_{t}$ is a martingale on the probability space $\left(\Omega, \mathscr{F},\left\{\mathscr{F}_{t}\right\}_{t \geq 0}, P\right)$, and in particular $\E_P[M_t]=0$.

Let $T> 0$. Let $\left(C([0,T], \R^d), \mathscr{B}\right)$ denote the path space, where $\mathscr{B}$ is the Borel sigma-algebra generated by the sup norm $\|\cdot\|_\infty$. Denote trajectories of the process by $x_\bullet:=(x_t)_{t \in [0,T]}: \Omega \to C([0,T], \R^d)$. By definition of Itô calculus, $Z_{T}$ is measurable with respect to $\langle x_{s}: 0 \leq t \leq T\rangle$, so there exists a positive measurable function $Z_T^C$ on the path space, such that $P$-a.s. $Z^C_T (x_\bullet(\omega))=Z_T(\omega)$ for $\omega \in \Omega$, i.e., the following diagram commutes
\begin{equation*}
  \begin{tikzcd}
    &C([0,T], \R^d) \arrow[dr, "Z_T^C"] &  \\
    \Omega \arrow[ur, "x_\bullet"] \arrow[rr, "Z_T"']&&\R_{>0}.
    \end{tikzcd}
\end{equation*}

Note that the path space $\left(C([0,T], \R^d), \mathscr{B}\right)$ admits a canonical filtration $(\mathscr{B}_{t})_{t \in [0,T]}$, where
\begin{align*}
    \mathscr{B}_{t} = \left\{A \subset C([0,T], \R^d) : \left.A\right|_{[0,t]} \in \text{Borel sigma-algebra on} \left(C([0,t], \R^d), \|\cdot \|_\infty\right)\right\}.
\end{align*}
For any path $\gamma \in C([0,T], \R^d)$ and $n \in \mathbb N$, we define the hitting time $t_n(\gamma) =\inf \left\{t \geq 0:\left|\gamma_{t}\right|>n\right\}$.

\begin{claim}
\label{claim: hitting time is stopping time}
These hitting times are stopping times wrt to the canonical filtration, i.e., $\{\gamma \in C([0,T], \R^d): t_n(\gamma) \leq t\} \in \mathscr{B}_{t}$ for any $t\in [0,T]$.
\end{claim}
\begin{proofclaim}
Let $A:=\{\gamma \in C([0,T], \R^d): t_n(\gamma) \leq t\}$. Then, 
\begin{align*}
    \left.A\right|_{[0,t]} &= \{\gamma \in C([0,t], \R^d): t_n(\gamma) \leq t\}\\
    &= \left\{\gamma \in C([0,t], \R^d): \min\{s\geq 0: \left|\gamma_{s}\right| > n\} \leq t\right\} \text{ (continuity of $\gamma$)}\\
    &=\left\{\gamma \in C([0,t], \R^d): \|\gamma \|_\infty > n \right\},
\end{align*}
which is clearly a Borel set in $\left(C([0,t], \R^d), \|\cdot \|_\infty\right)$.
\end{proofclaim}

Thus we can define stopping time sigma-algebras in the usual way
\begin{align*}
    \mathscr{B}_{T\land t_n} = \left\{A \in \mathscr{B}_{T}: A \cap\{\gamma \in C([0,T], \R^d): T\land t_n(\gamma) \leq t \} \in \mathscr{B}_{t}, \forall t \in [0,T]\right\}.
\end{align*}

We showed above that, under $P, \bar P_n$, the distributions of $x$ restricted to $(C([0,T], \R^d),\mathscr{B}_{T\land t_n})$ are $\p_{\left[0, T \wedge t_{n}\right]}:=\left.\p_{\left[0, T \right]}\right|_{\mathscr{B}_{T\land t_n}}$ and $ \bar{\p}_{\left[0, T \wedge t_{n}\right]}:=\left.\bar{\p}_{\left[0, T \right]}\right|_{\mathscr{B}_{T\land t_n}}$, respectively.


By inspection, we have, for any $t \geq 0$,
\begin{align*}
    \{T <\tau_n \}\cap \{T \land \tau_n \leq t\} = \begin{cases}
        \{T <\tau_n\} \in \mathscr F_T \subset \mathscr F_t \text{ if } T\leq t\\
        \emptyset \subset \mathscr F_t \text{ if } T> t.
    \end{cases}
\end{align*}
Setting $t= T \land \tau_n$, we have $\{T <\tau_n\} \in \mathscr F_{T \land \tau_n}$, which also yields $\{T \geq \tau_n\} \in \mathscr F_{T \land \tau_n}$ and $\{\tau_{n-1} \leq T < \tau_{n}\} \in \mathcal F_{T \land \tau_n}$.

Fix $i \geq 0$ and $A \in \mathscr{B}_{T \wedge t_{i}}$. Then $x^{-1}_\bullet A \in \mathscr F$ as $x_\bullet$ is measurable. Thus $x^{-1}_\bullet A \cap \{T <\tau_i \} \subset \mathscr F_{T \land \tau_i}$ and $x^{-1}_\bullet A \cap \{\tau_{n-1} \leq T < \tau_{n}\} \subset \mathscr F_{T \land \tau_n}$ for any $n >i$. Finally,
\begin{align*}
    &\E_{\p}\left[Z_{T}^{C}  \mathds 1_{A}\right]=\E_{P}\left[ Z_{T} \mathds 1_{x^{-1}_\bullet A}\right]\\
&=\E_{P} \left[Z_{T}^{(i)} \mathds 1_{x^{-1}_\bullet A \cap\left\{T<\tau_{i}\right\}}\right]+\sum_{n=i+1}^{+\infty} \E_{P}\left[ Z_{T}^{(n)} \mathds 1_{x^{-1}_\bullet A \cap\left\{\tau_{n-1} \leq T<\tau_{n}\right\}}\right]\\
&=\E_{\bar{P}_{i}}\left[ \mathds 1_{x^{-1}_\bullet A \cap\left\{T<\tau_{i}\right\}}\right]+\sum_{n=i+1}^{+\infty} \E_{\bar{P}_{n}}\left[\mathds 1_{x^{-1}_\bullet A \cap\left\{\tau_{n-1} \leq T<\tau_{n}\right\}}\right]\\
&=\E_{\bar{\p}} \left[\mathds 1_{A \cap\left\{T<t_{i}\right\}}\right]+\sum_{n=i+1}^{+\infty} \E_{\bar{\p}} \left[\mathds 1_{A \cap\left\{t_{n-1} \leq T<t_{n}\right\}}\right]=\E_{\bar{\p}} \left[\mathds 1_{A}\right].
\end{align*}
From the arbitrariness of $i$ and that $\lim_{n \to \infty} \tau_n = +\infty$ $P$-a.s., it follows that
\begin{align*}
    \frac{d \bar{\p}_{[0, T]}}{d \p_{[0, T]}}=Z_{T}^{C}.
\end{align*}

Finally, we compute the relative entropy between the forward and backward path-space measures
\begin{align*}
\H\left[\p_{[0, T]}, \bar \p_{[0, T]}\right] &=\E_{\p}\left[\log \frac{d \p_{[0, T]}}{d \bar \p_{[0, T]}}(\gamma)\right] \\
&=\E_{P}\left[\log \frac{d \p_{[0, T]}}{d \bar \p_{[0, T]}}(x(\omega))\right] \\
&= \E_{P}\left[-\log Z_T^C(x(\omega))\right]=\E_{P}\left[-\log Z_{T}\right] = \overbrace{\E_{P}\left[M_T\right]}^{=0} + \E_P[Y_T]\\
&=\E_{P}\left[2\int_0^{T}  \left|\sigma^- b_{\mathrm{irr}}(x_s) \right|^2 ds\right] =2T \int_{\R^d}  \left|\sigma^- b_{\mathrm{irr}}(x) \right|^2\rho(x) dx\\
&= T\int_{\R^d} b_{\mathrm{irr}}^\top D^- b_{\mathrm{irr}} \rho(x) dx,
\end{align*}
where we used Tonelli's theorem and stationarity for the penultimate equality. By Theorem \ref{thm: ep is a rate over time}, we obtain the entropy production rate
\begin{align*}
    e_p = \frac 1 T \H\left[\p_{[0, T]}, \bar \p_{[0, T]}\right] = \int_{\R^d} b_{\mathrm{irr}}^\top D^- b_{\mathrm{irr}} \rho(x) \dx. 
\end{align*}
\end{proof}

We now prove Theorem \ref{thm: epr regular case general}:

\begin{proof}
By assumption \ref{item: epr regular general drift image of volatility}, for $\rho$-a.e. $x \in \R^d$
\begin{equation}
\label{eq: localised drift in image of the volatility general}
\begin{split}
    b_{\mathrm{irr}}(x) \in \im \sigma(x)
    &\Rightarrow b_{\mathrm{irr}}(x) = \sigma \sigma^- b_{\mathrm{irr}}(x).
\end{split}
\end{equation}

Then, \eqref{eq: localised drift in image of the volatility general} implies that we can rewrite the SDE \eqref{eq: Ito SDE} as
\begin{align*}
        dx_t = b_{\mathrm{rev}}(x_t) dt+ \sigma(x_t)\left[\sigma^{-}b_{\mathrm{irr}}\left(x_t\right)dt+dw_t\right], \quad x_0 \sim \rho.
\end{align*}

By assumptions \ref{item: epr regular general adapted}, \ref{item: epr regular general expectation}, we may define a new probability measure $\bar P$ on the sample space $\Omega$ through the relation
\begin{equation}
\left.\frac{d \bar P}{d P}\right|_{\mathscr{F}_{T}}=Z_{T},
\end{equation}
and it follows by Girsanov's theorem \cite[Theorem 10.14]{pratoStochasticEquationsInfinite2014} 
applied to the process $-2\sigma^{-} b_{\mathrm{irr}}\left(x_t\right)$, that $\left(x_{t}\right)_{t \in [0,T]}$ solves the SDE
\begin{align*}
        d x_t &= b_{\mathrm{rev}}(x_t) dt- \sigma(x_t)\left[\sigma^{-}b_{\mathrm{irr}}\left(x_t\right)dt+dw_t\right], \quad  x_0\sim \rho.
\end{align*}
on the probability space $(\Omega, \mathscr{F},\left\{\mathscr{F}_{t}\right\}_{t \geq 0}, \bar P)$.
Using \eqref{eq: localised drift in image of the volatility general}, $\left(x_{t}\right)_{t \in [0,T]}$ solves
\begin{align*}
        d x_t &= \bar b ( x_t) dt + \sigma ( x_t) dw_t, \quad  x_0\sim \rho
\end{align*}
on said probability space. By Theorem \ref{thm: time reversal of diffusions} we know that under $\bar P$, $\left(x_{t}\right)_{t \in [0,T]}$ induces the path space measure of the time-reversed process.

Let $\left(C([0,T], \R^d), \mathscr{B}\right)$ denote the path space, where $\mathscr{B}$ is the Borel sigma-algebra generated by the sup norm $\|\cdot\|_\infty$. Denote trajectories of the process by $x_\bullet:=(x_t)_{t \in [0,T]}: \Omega \to C([0,T], \R^d)$.
In summary, we showed that, under $P, \bar P$, the distribution of $x_\bullet$ on $(C([0,T], \R^d),\mathscr{B})$ are $\p_{\left[0, T\right]}$ and $ \bar{\p}_{\left[0, T\right]}$, respectively.

By definition of Itô calculus, $Z_{T}$ is measurable with respect to $\langle x_{s}: 0 \leq t \leq T\rangle$, so there exists a positive measurable function $Z_T^C$ on the path space, such that $P$-a.s. $Z^C_T (x_\bullet(\omega))=Z_T(\omega)$ for $\omega \in \Omega$, i.e., the following diagram commutes
\begin{equation*}
  \begin{tikzcd}
    &C([0,T], \R^d) \arrow[dr, "Z_T^C"] &  \\
    \Omega \arrow[ur, "x_\bullet"] \arrow[rr, "Z_T"']&&\R_{>0}.
    \end{tikzcd}
\end{equation*}
Fix $A \in \mathscr{B}$. Then $x_\bullet^{-1}A \in \mathscr F$ as $x_\bullet$ is measurable. Obviously,
\begin{align*}
    \E_{\p}\left[Z_{T}^{C}  \mathds 1_{A}\right]=\E_{P}\left[ Z_{T} \mathds 1_{x_\bullet^{-1} A}\right]=\E_{\bar{P}} \left[\mathds 1_{x_\bullet^{-1}A}\right]=\E_{\bar{\p}} \left[\mathds 1_{A}\right].
\end{align*}
It follows that
\begin{align*}
    \frac{d \bar{\p}_{[0, T]}}{d \p_{[0, T]}}=Z_{T}^{C}.
\end{align*}

Through this, we obtain the relative entropy between the forward and backward path-space measures
\begin{align*}
\H\left[\p_{[0, T]}, \bar \p_{[0, T]}\right] &=\E_{\p}\left[\log \frac{d \p_{[0, T]}}{d \bar \p_{[0, T]}}(\gamma)\right] \\
&=\E_{P}\left[\log \frac{d \p_{[0, T]}}{d \bar \p_{[0, T]}}(x(\omega))\right] \\
&= \E_{P}\left[-\log Z_T^C(x(\omega))\right]=\E_{P}\left[-\log Z_{T}\right] \\
&= \E_{P}\left[2 \int_0^T \langle \sigma^- b_{\mathrm{irr}}(x_t), dw_t \rangle\right] + \E_P\left[2 \int_0^T |\sigma^- b_{\mathrm{irr}}(x_t) |^2 dt\right]\\
&=2T \int_{\R^d}  \left|\sigma^- b_{\mathrm{irr}}(x) \right|^2\rho(x) dx\\
&= T\int_{\R^d} b_{\mathrm{irr}}^\top D^- b_{\mathrm{irr}} \rho(x) dx.
\end{align*}
The penultimate equality follows from the fact that Itô stochastic integrals are martingales (and hence vanish in expectation), Tonelli's theorem and stationarity. Finally, by Theorem \ref{thm: ep is a rate over time}, we obtain the entropy production rate
\begin{align*}
    e_p = \frac 1 T \H\left[\p_{[0, T]}, \bar \p_{[0, T]}\right] = \int_{\R^d} b_{\mathrm{irr}}^\top D^- b_{\mathrm{irr}} \rho(x) dx. 
\end{align*}
\end{proof}

We prove Proposition \ref{prop: suff conds exp cond}:

\begin{proof}
\begin{enumerate}
    \item[\textit{\ref{item: suff cond martingale}.}] Follows as $Z_0 =1$ and by definition of a martingale. 
\end{enumerate}
We define the $\mathscr F_t$-adapted process $\psi_t =-2 \sigma^- b_{\mathrm{irr}}(x_t)$.
\begin{enumerate}
\setcounter{enumi}{1}
    \item[\textit{\ref{item: suff cond novikov}} $\Rightarrow$ \textit{\ref{item: suff cond martingale}.}] This follows from \cite[Theorem 1]{rufNewProofConditions2013}.
    \item[\textit{\ref{item: suff cond delta exp condition}.}] By \cite[Proposition 10.17 (ii)]{pratoStochasticEquationsInfinite2014}, a sufficient condition for \eqref{eq: exponential condition} is the existence of a $\delta>0$ such that $\sup _{t \in[0, T]} \mathbb{E}_P\left(e^{\delta|\psi_t|^{2}}\right)<+\infty$. By stationarity of $x_t$ at $\rho$, $\mathbb{E}_P\left(e^{\delta|\psi_t|^{2}}\right)= \mathbb{E}_{\rho}\left(e^{4\delta|\sigma^- b_{\mathrm{irr}}(x)|^{2}}\right)$, and so the result follows.
    \item[\textit{\ref{item: suff cond Kazamaki}} $\Rightarrow$ \textit{\ref{item: suff cond martingale}}.] Follows from \cite[Theorem 1]{rufNewProofConditions2013}.
    \item[\textit{\ref{item: suff cond Dellacherie meyer}} $\Rightarrow$ \textit{\ref{item: suff cond novikov}}.] $A_t:=\frac 1 2\int_0^t\left|\psi_s\right|^2 ds= 2\int_0^t\left|\sigma^-b_{\mathrm{irr}}(x_s)\right|^2 ds$ is a non-decreasing, $\mathcal F_t$-adapted process. By assumption, $\E_P\left[A_{T}-A_{t} \mid \mathcal{F}_{t}\right] \leq K$ for all $t \in [0,T]$. By \cite[Theorem 105 (b)]{dellacherieProbabilitiesPotentialTheory1982} $\E_P\left[\exp(A_{T})\right]<+\infty$.
    \item[\textit{\ref{item: suff cond tail}} $\Rightarrow$ \textit{\ref{item: suff cond delta exp condition}}.]
    Let $\delta \in (0 ,c) $. The first equality follows a standard fact about non-negative random variables:
    \begin{align*}
        \E_\rho\left [e^{\delta |\sigma^- b_{\mathrm{irr}}(x)|^{2}}\right] &= \int_0^\infty   P\left(e^{ \delta |\sigma^- b_{\mathrm{irr}}(x)|^{2}}>r\right)dr\\
        &\leq 1 + \int_{e^{cR}}^\infty P\left(e^{ \delta |\sigma^- b_{\mathrm{irr}}(x)|^{2}}>r\right) dr\\
        &=1 + \int_{e^{cR}}^\infty P\left(|\sigma^- b_{\mathrm{irr}}(x)|^{2}> \delta^{-1}\log r\right)dr\\
        &\leq 1 + C \int_{e^{cR}}^\infty r^{- c \delta^{-1}} dr  \quad \text{(by \ref{item: suff cond tail} as $\delta^{-1}\log r> R$)} \\
        &<+\infty.
    \end{align*}
\end{enumerate}
\end{proof}

\subsubsection{Singularity}
\label{app: epr singular}

We prove Theorem \ref{thm: epr singular}:

\begin{proof}
Under the assumption that $b_{\mathrm{irr}}(x) \in \im \sigma(x)$ does not hold for $\rho$-a.e. $x \in \R^d$, we will show that there are paths taken by the forward process that are not taken by the backward process---and vice-versa---resulting in the mutual singularity of forward and backward path space measures.


Recall from Theorem \ref{thm: time reversal of diffusions} that any solution to the following SDE
\begin{align}
d\bar x_t= \bar b(\bar x_t) dt + \sigma(\bar x_t) dw_t, \quad \bar x_0 \sim \rho,
\end{align}
induces the path space measure of the time-reversed process.

We rewrite the forward and backward SDEs into their equivalent Stratonovich SDEs \cite[eq. 3.31]{pavliotisStochasticProcessesApplications2014}
\begin{align*}
    dx_t &= b^s(x_t) dt + \sigma(x_t) \circ dw_t, \quad  x_0 \sim \rho,\\
    d\bar x_t&= \bar b^s(\bar x_t) dt + \sigma(\bar x_t)\circ  dw_t, \quad \bar x_0 \sim \rho,
\end{align*}
By Remark \ref{rem: stratonovich Helmholtz}, time-reversal and the transformation from Itô to Stratonovich commute so $\bar b^s$ is unambiguous, and $b_{\mathrm{irr}} = b^s(x)-\bar b^s(x)$. The volatility and Stratonovich drifts are locally Lipschitz as $\sigma \in C^2$.

Consider an initial condition $x= x_0= \bar x_0$ to the trajectories, with $\rho(x)>0$. Consider trajectories in the Cameron-Martin space 
\begin{align*}
    \gamma \in \mathcal C:=\{\gamma \in A C\left([0, T] ; \mathbb{R}^m\right) \: \mid \: \gamma(0)=0 \text { and } \dot{\gamma} \in L^2\left([0, T] ; \mathbb{R}^m\right)\}
\end{align*}

Given such a trajectory, the approximating ODEs
\begin{align*}
    dx_t &= b^s(x_t) dt + \sigma(x_t) d\gamma_t, \quad  x_0 =x,\\
    d\bar x_t&= \bar b^s(\bar x_t) dt + \sigma(\bar x_t) d\gamma_t, \quad \bar x_0 =x,
\end{align*}
have a unique solution in $[0,T]$, with $T>0$ uniform in $\gamma$. 

We can thus apply the Stroock-Varadhan support theorem \cite[Theorem 3.10]{prokopTopologicalSupportSolutions2016} 
to state 
the possible paths under the forward and backward protocols. These are as follows
\begin{align*}
   \supp \p^{x}_{[0,T]}&=\overline{\left\{ x_t^\gamma = x + \int_0^t b^s(x_s^\gamma) ds + \int_0^t \sigma(x_s^\gamma) \gamma'(s) ds , t \in [0,T]\mid  \gamma \in \mathcal C \right\}}, \\
   \supp \bar \p^{x}_{[0,T]}&=\overline{\left\{  \bar x_t^\gamma = x + \int_0^t \bar b^s(\bar x_s^\gamma) ds + \int_0^t \sigma(\bar x_s^\gamma) \gamma'(s) ds, t \in [0,T] \mid \gamma \in \mathcal C \right\}}.
\end{align*}
where the closure is with respect to the supremum norm on $C\left([0, T] ; \mathbb{R}^d\right)$. The time derivatives of these paths at $t=0$ are
\begin{align*}
   \partial_{t}\supp \p^{x}_{[0,T]}\mid_{t=0}&=\overline{\{ \partial_t x_t^\gamma |_{t=0} \in \R^d \mid \gamma \in \mathcal C \}} = \{b^s(x) + \sigma(x) v \mid  v\in \R^d\}, \\
  \partial_{t}\supp \bar \p^{x}_{[0,T]}\mid_{t=0}&=\overline{\{ \partial_t \bar x_t^\gamma |_{t=0} \in \R^d \mid \gamma \in \mathcal C \}} = \{\bar b^s(x) + \sigma(x) v \mid  v\in \R^d\}.
\end{align*}
where the closure is with respect to the sup norm on $\R^d$.

Consider an initial condition $x$ with $b_{\mathrm{irr}}(x)\not \in \im \sigma(x)$. This implies that the forward and backward path space measures are mutually singular because the time derivatives of the possible paths differ at the origin 
\begin{align*}
    & 2b_{\mathrm{irr}}(x) = b^s(x)-\bar b^s(x) \not \in \im \sigma(x) \\
     \iff& b^s(x) + \im \sigma(x) \neq \bar b^s(x) + \im \sigma(x)\\
    \iff &\partial_{t}\supp \p^{x}_{[0,T]}\mid_{t=0}\not \subset \partial_{t}\supp \bar \p^{x}_{[0,T]}\mid_{t=0} \text{and vice-versa} \\
    \Rightarrow& \supp \p^{x}_{[0,T]} \not \subset \supp \bar \p^{x}_{[0,T]} \text{ and vice-versa}\\
    \Rightarrow& \p^{x}_{[0,T]} \perp \bar \p^{x}_{[0,T]}
\end{align*}

Finally, from Proposition \ref{prop: aggregating local ep}
\begin{align*}
     e_p 
     &= \E_{x \sim \rho}\left[\H\left[\p^{x}_{[0,T]}\mid \bar \p^{x}_{[0,T]}\right]\right]\\
     &\geq \rho\left(\{x \in \R^d \mid  \p^{x}_{[0,T]} \perp \bar \p^{x}_{[0,T]} \}\right)\cdot \infty\\
     &\geq \underbrace{\rho\left(\{x \in \R^d  : b_{\mathrm{irr}}(x) \not \in \im \sigma(x) \}\right)}_{>0} \cdot \infty\\
     &= +\infty.
\end{align*}
\end{proof}

\subsection{Entropy production rate of the linear diffusion process}
\label{app: exact simul OU process}

We require the following Lemma, which can be proved by adjusting the derivation of the relative entropy between non-degenerate Gaussian distributions, cf. \cite[Section 9]{duchiDerivationsLinearAlgebra}.
\begin{lemma}
\label{lemma: KL Gaussian}
On $\R^d$
    \begin{align*}
       2 \H[\mathcal N(\mu_0, \Sigma_0) \mid\mathcal N(\mu_1, \Sigma_1)] &= \operatorname{tr}\left(\Sigma_{1}^{-} \Sigma_{0}\right)-\rank \Sigma_0
       +\log \left(\frac{\operatorname{det}^* \Sigma_{1}}{\operatorname{det}^* \Sigma_{0}}\right) \\
       &+\left(\mu_{1}-\mu_{0}\right)^{\top} \Sigma_{1}^{-}\left(\mu_{1}-\mu_{0}\right).
    \end{align*}
where $\cdot^-$ is the Moore-Penrose pseudo-inverse and $\operatorname{det}^*$ is the pseudo-determinant. 
\end{lemma}

\begin{proof}[Proof of Lemma \ref{lemma: limit ep formula OU process}]
We insert the definitions of the transition kernels \eqref{eq: OU transition kernels} into Lemma \ref{lemma: KL Gaussian}.
\begin{equation*}
\begin{split}
    &\E_{x\sim \rho}[2\H[p_\e(\cdot,x)\mid \bar p_\e(\cdot,x)]] \\
    &=\tr(\bar S_\e^{-}S_\e)-\rank \sigma + \log \frac{\det^*(\bar S_\e)}{\det^*(S_\e)} \\
&+ \E_{x \sim\rho} \left[ x^\top(e^{-\e C}-e^{-\e B})^\top \bar S_\e^{-}(e^{-\e C}-e^{-\e B})x \right] \\
&=\tr(\bar S_\e^{-}S_\e)-\rank \sigma + \log \frac{\det^*(\bar S_\e)}{\det^*(S_\e)}\\
    &+\tr( \Pi^{-1} (e^{-\e C}-e^{-\e B})^\top \bar S_\e^{-}(e^{-\e C}-e^{-\e B}))
\end{split}
\end{equation*}
To obtain the last line, we used the trace trick for computing Gaussian integrals of bilinear forms. The proof follows by Proposition \ref{prop: epr transition kernels}.
\end{proof}

\Xchapter{Bayesian mechanics}{Bayesian mechanics for stationary processes}{By Lancelot Da Costa, Karl Friston, Conor Heins, Grigorios A. Pavliotis\blfootnote{\normalsize \textbf{Adapted from:} L Da Costa, K Friston, C Heins, GA Pavliotis. Bayesian mechanics for stationary processes. \textit{Proceedings of the
Royal Society A}. 2021.}}



\newpage

\section{Abstract}
This chapter develops a Bayesian mechanics for adaptive systems.

Firstly, we model the interface between a system and its environment with a Markov blanket. This affords conditions under which states internal to the blanket encode information about external states.

Second, we introduce dynamics and represent adaptive systems as Markov blankets at steady-state. This allows us to identify a wide class of systems whose internal states appear to infer external states, consistent with variational inference in Bayesian statistics and theoretical neuroscience.

Finally, we partition the blanket into sensory and active states. It follows that active states can be seen as performing active inference and well-known forms of stochastic control (such as PID control), which are prominent formulations of adaptive behaviour in theoretical biology and engineering.

\textbf{Keywords: }Markov blanket, variational Bayesian inference, active inference, non-equilibrium steady-state, predictive processing, free-energy principle

\section{Introduction}
Any object of study must be, implicitly or explicitly, separated from its environment. This implies a boundary that separates it from its surroundings, and which persists for at least as long as the system exists.

In this article, we explore the consequences of a boundary mediating interactions between states internal and external to a system. This provides a useful metaphor to think about biological systems, which comprise spatially bounded, interacting components, nested at several spatial scales \cite{hespMultiscaleViewEmergent2019,kirchhoffMarkovBlanketsLife2018}: for example, the membrane of a cell acts as a boundary through which the cell communicates with its environment, and the same can be said of the sensory receptors and muscles that bound the nervous system.

By examining the dynamics of persistent, bounded systems, we identify a wide class of systems wherein the states internal to a boundary appear to infer those states outside the boundary---a description which we refer to as Bayesian mechanics. Moreover, if we assume that the boundary comprises sensory and active states, we can identify the dynamics of active states with well-known descriptions of adaptive behaviour from theoretical biology and stochastic control. 

In what follows, we link a purely mathematical formulation of interfaces and dynamics with descriptions of belief updating and behaviour found in the biological sciences and engineering. Altogether, this can be seen as a model of adaptive agents, as these interface with their environment through sensory and active states and furthermore behave so as to preserve a target steady-state.

\subsection{Outline of chapter}

This chapter has three parts, each of which introduces a simple, but fundamental, move.
\begin{enumerate}
    \item The first is to partition the world into internal and external states whose boundary is modelled with a Markov blanket \cite{pearlGraphicalModelsProbabilistic1998,bishopPatternRecognitionMachine2006}. This allows us to identify conditions under which internal states encode information about external states.
    \item The second move is to equip this partition with stochastic dynamics. The key consequence of this is that internal states can be seen as continuously inferring external states, consistent with variational inference in Bayesian statistics and with predictive processing accounts of biological neural networks in theoretical neuroscience.
    \item The third move is to partition the boundary into sensory and active states. It follows that active states can be seen as performing active inference and stochastic control, which are prominent descriptions of adaptive behaviour in biological agents, machine learning and robotics.
\end{enumerate}

\subsection{Related work}

The emergence and sustaining of complex (dissipative) structures have been subjects of long-standing research starting from the work of Prigogine \cite{nicolisSelforganizationNonequilibriumSystems1977,goldbeterDissipativeStructuresBiological2018}, followed notably by Haken’s synergetics \cite{hakenSynergeticsIntroductionNonequilibrium1978}, and in recent years, the statistical physics of adaptation \cite{perunovStatisticalPhysicsAdaptation2016}. A central theme of these works is that complex systems can only emerge and sustain themselves far from equilibrium \cite{jefferyStatisticalMechanicsLife2019,englandStatisticalPhysicsSelfreplication2013,skinnerImprovedBoundsEntropy2021}.

Information processing has long been recognised as a hallmark of cognition in biological systems. In light of this, theoretical physicists have identified basic instances of information processing in systems far from equilibrium using tools from information theory, such as how a drive for metabolic efficiency can lead a system to become predictive \cite{dunnLearningInferenceNonequilibrium2013,stillThermodynamicCostBenefit2020,stillThermodynamicsPrediction2012,ueltzhofferThermodynamicsPredictionDissipative2020}.

A fundamental aspect of biological systems is a self-organisation of various interacting components at several spatial scales \cite{hespMultiscaleViewEmergent2019,kirchhoffMarkovBlanketsLife2018}. Much research currently focuses on multipartite processes---modelling interactions between various sub-components that form biological systems---and how their interactions constrain the thermodynamics of the whole \cite{kardesThermodynamicUncertaintyRelations2021,wolpertMinimalEntropyProduction2020,wolpertUncertaintyRelationsFluctuation2020,crooksMarginalConditionalSecond2019,horowitzThermodynamicsContinuousInformation2014}.

At the confluence of these efforts, researchers have sought to explain cognition in biological systems. Since the advent of the 20th century, Bayesian inference has been used to describe various cognitive processes in the brain \cite{pougetInferenceComputationPopulation2003,knillBayesianBrainRole2004,fristonFreeenergyPrincipleUnified2010,raoPredictiveCodingVisual1999,fristonActionBehaviorFreeenergy2010}. In particular, the free energy principle \cite{fristonFreeenergyPrincipleUnified2010}, a prominent theory of self-organisation from the neurosciences, postulates that Bayesian inference can be used to describe the dynamics of multipartite, persistent systems modelled as Markov blankets at non-equilibrium steady-state \cite{fristonParcelsParticlesMarkov2020,fristonLifeWeKnow2013,parrMarkovBlanketsInformation2020,fristonFreeEnergyPrinciple2019a,fristonStochasticChaosMarkov2021}.

This chapter connects and develops some of the key themes from this literature. Starting from fundamental considerations about adaptive systems, we develop a physics of things that hold beliefs about other things--consistently with Bayesian inference--and explore how it relates to known descriptions of action and behaviour from the neurosciences and engineering. Our contribution is theoretical: from a biophysicist's perspective, this chapter describes how Bayesian descriptions of biological cognition and behaviour can emerge from standard accounts of physics. From an engineer's perspective this chapter contextualises some of the most common stochastic control methods and reminds us how these can be extended to suit more sophisticated control problems.

\subsection{Notation}

Let $\Pi \in \R^{d \times d}$ be a square matrix with real coefficients. Let $\eta, b, \mu$ denote a partition of the states $[\![1, d]\!]$, so that
\begin{align*}
    \Pi = \begin{bmatrix} \Pi_{\eta} & \Pi_{\eta b} & \Pi_{\eta \mu}\\
 \Pi_{b \eta} &\Pi_{b}&\Pi_{b \mu} \\
 \Pi_{\mu \eta}& \Pi_{\mu b}& \Pi_{\mu} \end{bmatrix}.
\end{align*}
We denote principal submatrices with one index only (i.e., we use $\Pi_{\eta}$ instead of $\Pi_{\eta\eta}$). Similarly, principal submatrices involving various indices are denoted with a colon

\begin{align*}
    \Pi_{\eta :b } &:=\begin{bmatrix} \Pi_{\eta} & \Pi_{\eta b} \\
 \Pi_{b \eta} &\Pi_{b} \end{bmatrix}.
\end{align*}

When a square matrix $\Pi$ is symmetric positive-definite we write $\Pi \succ 0$. $\ker$, $\operatorname{Im}$ and $\cdot^-$ respectively denote the kernel, image and Moore-Penrose pseudo-inverse of a linear map or matrix, e.g., a non-necessarily square matrix such as $\Pi_{\mu b}$. In our notation, indexing takes precedence over (pseudo) inversion, for example,
\begin{align*}
    \Pi_{\mu b}^{-} := \left(\Pi_{\mu b}\right)^- \neq (\Pi^-)_{\mu b}.
\end{align*}

\section{Markov blankets}
\label{sec: markov blanket}

The section formalises the notion of boundary between a system and its environment as a Markov blanket \cite{pearlGraphicalModelsProbabilistic1998,bishopPatternRecognitionMachine2006}, depicted graphically in Figure \ref{fig: 3 way MB}. Intuitive examples of a Markov blanket are that of a cell membrane, mediating all interactions between the inside and the outside of the cell, or that of sensory receptors and muscles that bound the nervous system.

\begin{figure}[!h]
\centering\includegraphics[width=0.6\textwidth]{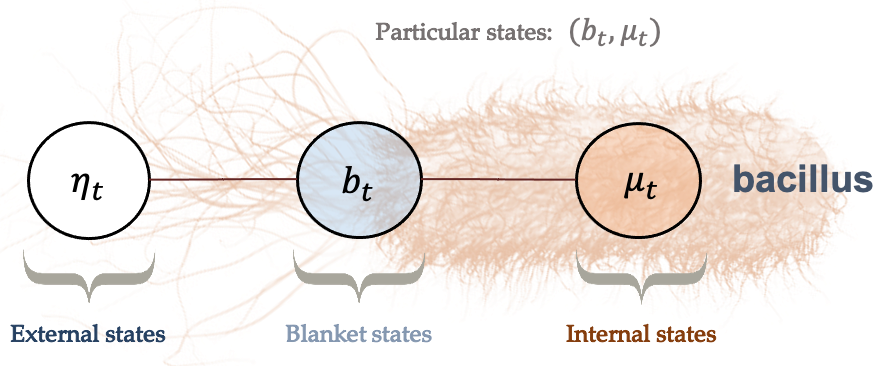}
\caption[Markov blanket]{\textbf{Markov blanket} depicted graphically as an undirected graphical model, also known as a Markov random field \cite{wainwrightGraphicalModelsExponential2007,bishopPatternRecognitionMachine2006}. (A Markov random field is a Bayesian network whose directed arrows are replaced by undirected arrows). The circles represent random variables. The lines represent conditional dependencies between random variables. The Markov blanket condition means that there is no line between $\mu$ and $\eta$. This means that, $\mu$ and $\eta$ are conditionally independent given $b$. In other words, knowing the internal state $\mu$, does not afford additional information about the external state $\eta$ when the blanket state $b$ is known. Thus blanket states act as an informational boundary between internal and external states.}
\label{fig: 3 way MB}
\end{figure}

To formalise this intuition, we model the world's state as a random variable $x$ with corresponding probability distribution $p$ over a state-space $\mathcal X= \R^d$. We partition the state-space of $x$ into \textit{external}, \textit{blanket} and \textit{internal} states:
\begin{align*}
x&= (\eta, b, \mu)\\
\mathcal X&= \mathcal E \times \mathcal B \times \mathcal I.
\end{align*}
External, blanket and internal state-spaces ($\mathcal E, \mathcal B, \mathcal I$) are taken to be Euclidean spaces for simplicity.

A Markov blanket is a statement of conditional independence between internal and external states given blanket states.

\begin{definition}[Markov blanket]
A Markov blanket is defined as

\begin{equation}
\label{eq: def Markov blanket}\tag{M.B.}
          \eta \perp \mu \mid b
\end{equation}
That is, blanket states are a Markov blanket separating $\mu, \eta$ \cite{pearlGraphicalModelsProbabilistic1998,bishopPatternRecognitionMachine2006}.
\end{definition}

The existence of a Markov blanket can be expressed in several equivalent ways
\begin{align}
\label{eq: Markov blanket densities factorising}
  \eqref{eq: def Markov blanket}  \iff p(\eta, \mu | b) = p(\eta | b)p(\mu | b)
  \iff  p(\eta | b, \mu) = p(\eta | b)   \iff p(\mu | b, \eta) = p(\mu | b).
\end{align}

For now, we will consider a (non-degenerate) Gaussian distribution $p$ encoding the distribution of states of the world
\begin{align*}
    p(x)&:= \mathcal N(x; 0,\Pi^{-1}), \quad \Pi \succ 0,
\end{align*}
with associated precision (i.e., inverse covariance) matrix $\Pi$. Throughout, we will denote the (positive definite) covariance by $\Sigma := \Pi^{-1}$. Unpacking \eqref{eq: Markov blanket densities factorising} in terms of Gaussian densities, we find that a Markov blanket is equivalent to a sparsity in the precision matrix

\begin{align}
\label{eq: MB is sparsity in precision matrix}
    \eqref{eq: def Markov blanket} \iff \Pi_{\eta \mu}= \Pi_{\mu \eta}=0.
\end{align}

\begin{example}
For example,
\begin{align*}
 &\Pi = \begin{bmatrix} 2 & 1 &0\\
 1&2&1 \\
 0& 1& 2\end{bmatrix} \Rightarrow \Sigma_{\eta:b}^{-1} = \begin{bmatrix} 2 & 1 \\
 1&1.5\end{bmatrix}, \Sigma_{b: \mu}^{-1} = \begin{bmatrix} 1.5 & 1 \\
 1& 2\end{bmatrix}
\end{align*}
Then,
\begin{align*}
    p(\eta, \mu |b) &\propto p(\eta, \mu , b)\propto \exp\left( -\frac 1 2 x\cdot \Pi x\right)\\
    &\propto \exp\left( -\frac 1 2 \begin{bmatrix} \eta , b \end{bmatrix}\Sigma_{\eta:b}^{-1} \begin{bmatrix} \eta \\ b \end{bmatrix} - \frac 1 2 \begin{bmatrix} b , \mu\end{bmatrix}\Sigma_{b: \mu}^{-1} \begin{bmatrix}  b\\\mu \end{bmatrix}\right)  \propto p(\eta, b)p( b, \mu ) \propto p(\eta|b)p(\mu |b).
\end{align*}
Thus, the Markov blanket condition \eqref{eq: Markov blanket densities factorising} holds.
\end{example}

\subsection{Expected internal and external states}

Blanket states act as an information boundary between external and internal states. Given a blanket state, we can express the conditional probability densities over external and internal states (using \eqref{eq: Markov blanket densities factorising} and \cite[Prop. 3.13]{eatonMultivariateStatisticsVector2007a})\footnote{Note that $\Pi_{\eta},\Pi_{\mu}$ are invertible as principal submatrices of a positive definite matrix.}

\begin{equation}
\label{eq: posterior beliefs}
\begin{split}
    p(\eta|b)&= \mathcal N(\eta; \Sigma_{\eta b} \Sigma_{b}^{-1} b,\: \Pi_{\eta}^{-1}), \\ 
    p(\mu|b)&= \mathcal N(\mu; \Sigma_{\mu b} \Sigma_{b}^{-1} b,\: \Pi_{\mu}^{-1}). 
\end{split}
\end{equation}

This enables us to associate to any blanket state its corresponding expected external and expected internal states:
    
    \begin{align*}
    \boldsymbol \eta(b) &= \E [\eta \mid b]= \E_{p(\eta|b)}[\eta] = \Sigma_{\eta b} \Sigma_{b}^{-1} b \in \mathcal E\\
    \boldsymbol \mu(b) &= \E [\mu \mid b] =\E_{p(\mu|b)}[\mu] =\Sigma_{\mu b} \Sigma_{b}^{-1} b \in \mathcal I.
    \end{align*}

Pursuing the example of the nervous system, each sensory impression on the retina and oculomotor orientation (blanket state) is associated with an expected scene that caused sensory input (expected external state) and an expected pattern of neural activity in the visual cortex (expected internal state) \cite{parrComputationalNeurologyActive2019}.

\subsection{Synchronisation map}

A central question is whether and how expected internal states encode information about expected external states. For this, we need to characterise a synchronisation function $\sigma$, mapping the expected internal state to the expected external state, given a blanket state $\sigma(\boldsymbol \mu(b))=\boldsymbol \eta(b)$. This is summarised in the following commutative diagram:

\begin{equation*}
  \begin{tikzcd}
     & b \in \mathcal B \arrow[dl, "\boldsymbol \eta"']\arrow[dr, "\boldsymbol \mu"]  &  \\
    \operatorname{Image}(\boldsymbol \eta) \arrow[rr, leftarrow, dotted, "\sigma"'] &  & \operatorname{Image}(\boldsymbol \mu)
    \end{tikzcd}
\end{equation*}

The existence of $\sigma$ is guaranteed, for instance, if the expected internal state completely determines the blanket state---that is, when no information is lost in the mapping $b \mapsto \boldsymbol \mu(b)$ in virtue of it being one-to-one. In general, however, many blanket states may correspond to an unique expected internal state. Intuitively, consider the various neural pathways that compress the signal arriving from retinal photoreceptors \cite{meisterNeuralCodeRetina1999}, thus many different (hopefully similar) retinal impressions lead to the same signal arriving in the visual cortex.

\subsubsection{Existence}

The key for the existence of a function $\sigma$ mapping expected internal states to expected external states given blanket states, is that for any two blanket states associated with the same expected internal state, these be associated with the same expected external state. This non-degeneracy means that the internal states (e.g., patterns of activity in the visual cortex) have enough capacity to represent all possible expected external states (e.g., 3D scenes of the environment). We formalise this in the following Lemma:

\begin{lemma}
\label{lemma: equiv sigma well defined}
The following are equivalent:
\begin{enumerate}
    \item There exists a function $\sigma : \operatorname{Image}( \boldsymbol \mu)   \to \operatorname{Image} (\boldsymbol \eta)$ such that for any blanket state $b \in \mathcal B$
    \begin{equation*}
        \sigma(\boldsymbol \mu(b))=\boldsymbol \eta(b).
    \end{equation*}
    \item For any two blanket states $b_1, b_2 \in \B$
        \begin{equation*}
        \boldsymbol \mu(b_1)= \boldsymbol \mu(b_2) \Rightarrow \boldsymbol \eta(b_1)= \boldsymbol \eta(b_2).
        \end{equation*}
    \item $\ker \Sigma_{\mu b} \subset \ker \Sigma_{\eta b}$.
    \item $\ker \Pi_{\mu b} \subset \ker \Pi_{\eta b}$.
\end{enumerate}
\end{lemma}

See Section \ref{app: sync map existence proof} for a proof of Lemma \ref{lemma: equiv sigma well defined}.

\begin{example}
\begin{itemize}
    \item When external, blanket and internal states are one-dimensional, the existence of a synchronisation map is equivalent to $\Pi_{\mu b} \neq 0$ or $\Pi_{\mu b}= \Pi_{\eta b} =0$.
    \item If $\Pi_{\mu b}$ is chosen at random--its entries sampled from a non-degenerate Gaussian or uniform distribution--then $\Pi_{\mu b}$ has full rank with probability $1$. If furthermore the blanket state-space $\mathcal B$ has lower or equal dimensionality than the internal state-space $\mathcal I$, we obtain that $\Pi_{\mu b}$ is one-to-one (i.e., $\ker \Pi_{\mu b}= 0$) with probability $1$. Thus, in this case, the conditions of Lemma \ref{lemma: equiv sigma well defined} are fulfilled with probability $1$.
\end{itemize}
\end{example}

\subsubsection{Construction}

The key idea to map an expected internal state $\boldsymbol \mu(b)$ to an expected external state $\boldsymbol \eta(b)$ is to: 1) find a blanket state that maps to this expected internal state (i.e., by inverting $\boldsymbol \mu$) and 2) from this blanket state, find the corresponding expected external state (i.e., by applying $\boldsymbol \eta$):

\begin{equation*}
  \begin{tikzcd}
    & & b \in \mathcal B \arrow[ddll, "\boldsymbol \eta"']\arrow[ddrr, "\boldsymbol \mu"]  &  &\\
    &&&& \\
    \operatorname{Image}(\boldsymbol \eta) \arrow[rrrr, leftarrow, "\sigma = \boldsymbol \eta \circ \boldsymbol \mu^-"'] &  & && \operatorname{Image}(\boldsymbol \mu) \arrow[uull, bend left, rightarrow, "\boldsymbol \mu^-" near end]
    \end{tikzcd}
\end{equation*}

We now proceed to solving this problem. Given an internal state $\mu$, we study the set of blanket states $b$ such that $\boldsymbol \mu(b)= \mu$
\begin{equation}
\label{eq: inverse problem blanket states}
    \boldsymbol \mu(b)= \Sigma_{\mu b} \Sigma_{b}^{-1} b= \mu \iff b \in \boldsymbol \mu^{-1} (\mu)= \Sigma_{b}\Sigma_{\mu b}^{-1}\mu.
\end{equation}
Here the inverse on the right hand side of \eqref{eq: inverse problem blanket states} is understood as the preimage of a linear map. We know that this system of linear equations has a vector space of solutions given by \cite{jamesGeneralisedInverse1978} 

\begin{equation}
\label{eq: solution system of equations}
    \boldsymbol \mu^{-1} (\mu)=\left\{\Sigma_{b}\Sigma_{\mu b}^-\mu + \left(\id - \Sigma_{b}\Sigma_{\mu b}^-\Sigma_{\mu b} \Sigma_{b}^{-1}\right)b : b \in \mathcal B\right\}.
\end{equation}
Among these, we choose
\begin{equation*}
    \boldsymbol \mu^- (\mu)= \Sigma_{b}\Sigma_{\mu b}^-\mu.
\end{equation*}

\begin{definition}[Synchronisation map]
We define a synchronisation function that maps to an internal state a corresponding most likely internal state\footnote{This mapping was derived independently of our work in \cite[Section 3.2]{aguileraHowParticularPhysics2022a}.}\footnote{Replacing $\boldsymbol \mu^- (\mu)$ by any other element of \eqref{eq: solution system of equations} would lead to the same synchronisation map provided that the conditions of Lemma \ref{lemma: equiv sigma well defined} are satisfied.}
\begin{equation*}
\label{eq: sync map def}
\begin{split}
    \sigma &: \operatorname{Im} \boldsymbol \mu   \to \operatorname{Im} \boldsymbol \eta \\
    \mu &\mapsto \boldsymbol{\eta}(\boldsymbol \mu^- (\mu)) =  \Sigma_{\eta b}\Sigma_{\mu b}^-\mu = \Pi^{-1}_\eta \Pi_{\eta b} \Pi_{\mu b}^- \Pi_\mu \mu.
\end{split}
\end{equation*}
The expression in terms of the precision matrix is a byproduct of Section \ref{app: sync map existence proof}.
\end{definition}

Note that we can always define such $\sigma$, however, it is only when the conditions of Lemma \ref{lemma: equiv sigma well defined} are fulfilled that $\sigma$ maps expected internal states to expected external states $\sigma (\boldsymbol \mu (b)) = \boldsymbol \eta (b)$. When this is not the case, the internal states do not fully represent external states, which leads to a partly degenerate type of representation, see Figure \ref{fig: sigma example non-example} for a numerical illustration obtained by sampling from a Gaussian distribution, in the non-degenerate (left) and degenerate cases (right), respectively.

\begin{figure}
    \centering
    \includegraphics[width= 0.47\textwidth]{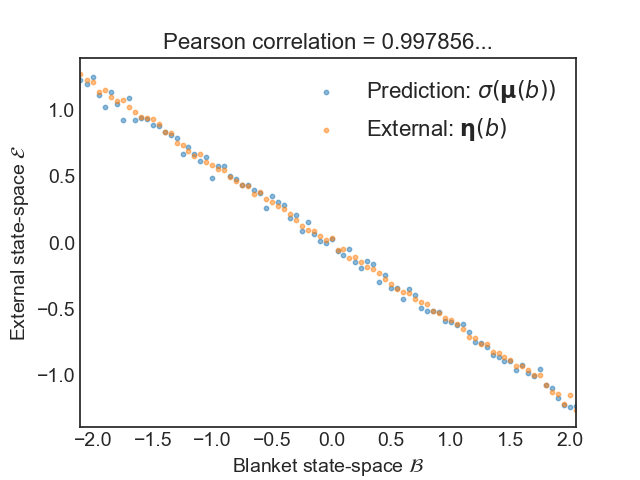}
    \includegraphics[width= 0.47\textwidth]{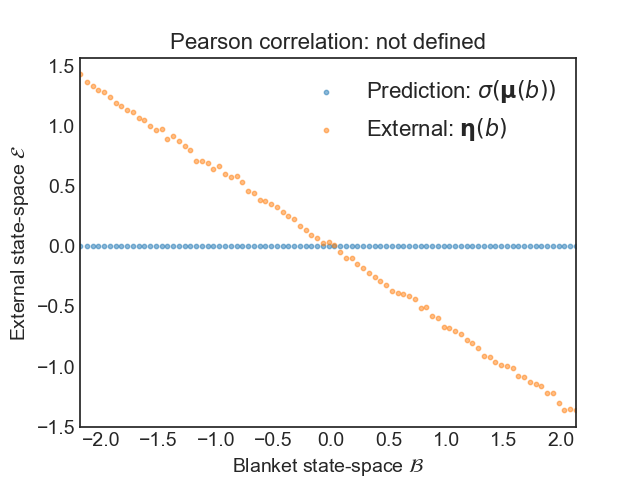}
    \caption[Synchronisation map: example and non-example]{\textbf{Synchronisation map: example and non-example.} This figure plots expected external states given blanket states $\boldsymbol \eta (b)$ (in orange), and the corresponding prediction encoded by internal states $\sigma (\boldsymbol \mu (b))$ (in blue). In this example, external, blanket and internal state-spaces are taken to be one dimensional. We show the correspondence under the conditions of Lemma \ref{lemma: equiv sigma well defined} (left panel) and when these are not satisfied (right panel). To generate these data, 1) we drew $10^6$ samples from a Gaussian distribution with a Markov blanket, 2) we partitioned the blanket state-space into several bins, 3) we obtained the expected external and internal states given blanket states empirically by averaging samples from each bin, and finally, 4) we applied the synchronisation map to the (empirical) expected internal states given blanket states.}
    \label{fig: sigma example non-example}
\end{figure}

\section{Bayesian mechanics}

In order to study the time-evolution of systems with a Markov blanket, we introduce dynamics into the external, blanket and internal states. Henceforth, we assume a synchronisation map under the conditions of Lemma \ref{lemma: equiv sigma well defined}.

\subsection{Processes at a Gaussian steady-state}

We consider stochastic processes at a Gaussian steady-state $p$ with a Markov blanket. The steady-state assumption means that the system's overall configuration persists over time (e.g., it does not dissipate). In other words, we have a Gaussian density $p = \mathcal N(0, \Pi^{-1})$ with a Markov blanket \eqref{eq: MB is sparsity in precision matrix} and a stochastic process distributed according to $p$ at every point in time
\begin{equation*}
    x_t \sim p = \mathcal N(0, \Pi^{-1}) \quad \text{ for any } t.
\end{equation*}
Recalling our partition into external, blanket and internal states, this affords a Markov blanket that persists over time, see Figure \ref{fig: 3 way dynamical MB}
\begin{equation}
\label{eq: MB over time}
   x_t= (\eta_t,b_t, \mu_t) \sim p \Rightarrow \eta_t \perp \mu_t \mid b_t.
\end{equation}

\begin{figure}
    \centering
    \includegraphics[width= 0.7\textwidth]{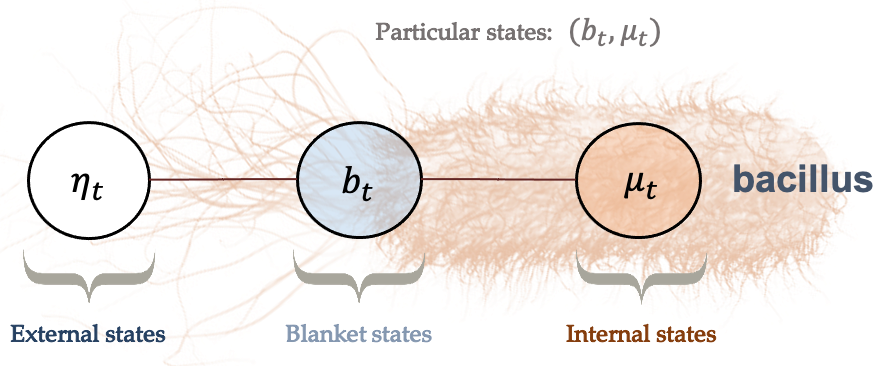}
    \caption[Markov blanket evolving in time]{\textbf{Markov blanket evolving in time.} We use a bacillus to depict an intuitive example of a Markov blanket that persists over time. Here, the blanket states represent the membrane and actin filaments of the cytoskeleton, which mediate all interactions between internal states and the external medium (external states).}
    \label{fig: 3 way dynamical MB}
\end{figure}

Note that we do not require $x_t$ to be independent samples from the steady-state distribution $p$. On the contrary, $x_t$ may be generated by extremely complex, non-linear, and possibly stochastic equations of motion. See Example \ref{eg: processes at a gaussian steady-state} and Figure \ref{fig: stationary process with a Markov blanket} for details.

\begin{figure}
    \centering
    \includegraphics[width= 0.47\textwidth]{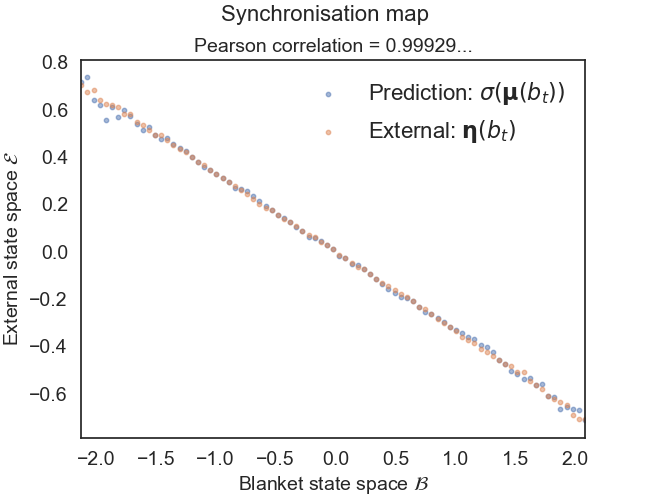}
    \includegraphics[width= 0.47\textwidth]{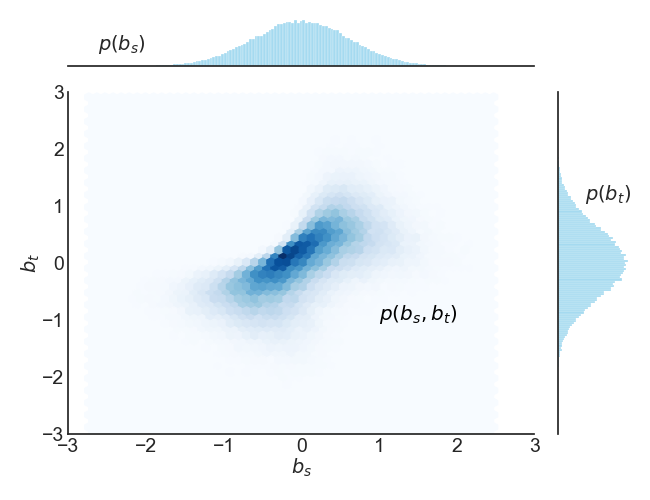}
    \caption[Processes at a Gaussian steady-state]{\textbf{Processes at a Gaussian steady-state}. This figure illustrates the synchronisation map and transition probabilities of processes at a Gaussian steady-state. \textit{Left:} We plot the synchronisation map as in Figure \ref{fig: sigma example non-example}, only, here, the samples are drawn from trajectories of a diffusion process \eqref{eq: Gaussian preserving diffusions} with a Markov blanket. Although this is not the case here, one might obtain a slightly noisier correspondence between predictions $\sigma(\boldsymbol \mu(b_t))$ and expected external states $\boldsymbol \eta(b_t)$---compared to Figure \ref{fig: sigma example non-example}---in numerical discretisations of a diffusion process. This is because the steady-state of a numerical discretisation usually differs slightly from the steady-state of the continuous-time process \cite{mattinglyConvergenceNumericalTimeAveraging2010}.
    \textit{Right:} This panel plots the transition probabilities of the same diffusion process  \eqref{eq: Gaussian preserving diffusions}, for the blanket state at two different times. The joint distribution (depicted as a heat map) is not Gaussian but its marginals---the steady-state density---are Gaussian. This shows that in general, processes at a Gaussian steady-state are not Gaussian processes. In fact, the Ornstein-Uhlenbeck process is the only stationary diffusion process \eqref{eq: Gaussian preserving diffusions} that is a Gaussian process, so the transition probabilities of non-linear diffusion processes \eqref{eq: Gaussian preserving diffusions} are never multivariate Gaussians.}
    \label{fig: stationary process with a Markov blanket}
\end{figure}

\begin{example}
\label{eg: processes at a gaussian steady-state}
The dynamics of $x_t$ are described by a stochastic process at a Gaussian steady-state $p = \mathcal N(0, \Pi^{-1})$. There is a large class of such processes, for example:

\begin{itemize}
\item Stationary diffusion processes, with initial condition $x_0 \sim p$. Their time-evolution is given by an Itô stochastic differential equation (see Section \ref{sec: helmholtz decomposition}):
\begin{equation}
\label{eq: Gaussian preserving diffusions}
\begin{split}
    dx_t &= (\Gamma+Q)(x_t) \nabla \log p( x_t) dt + \nabla \cdot (\Gamma+Q)(x_t) dt+ \varsigma(x_t) dW_t, \quad  \\
    &=-(\Gamma+Q)(x_t) \Pi x_t dt +\nabla \cdot (\Gamma+Q)(x_t) dt+ \varsigma (x_t) dW_t\\
    \Gamma&:= \varsigma \varsigma^\top/2 ,\quad  Q= -Q^\top.
\end{split}
\end{equation}
Here, $W_t$ is a standard Brownian motion (a.k.a., Wiener process) \cite{rogersDiffusionsMarkovProcesses2000a, pavliotisStochasticProcessesApplications2014} and $\varsigma,\Gamma,Q$ are sufficiently well-behaved matrix fields (see Section \ref{sec: helmholtz decomposition}). Namely, $\Gamma$ is the diffusion tensor (half the covariance of random fluctuations), which drives dissipative flow; $Q$ is an arbitrary antisymmetric matrix field which drives conservative (i.e., solenoidal) flow. We emphasise that there are no non-degeneracy conditions on the matrix field $\varsigma$---in particular, the process is allowed to be non-ergodic or even completely deterministic (i.e., $\varsigma \equiv 0$). Also, $\nabla \cdot $ denotes the divergence of a matrix field defined as $\left(\nabla \cdot (\Gamma + Q) \right)_i := \sum_{j} \frac{\partial}{\partial x_j} (\Gamma + Q)_{ij}$.
\item More generally, $x_t$ could be generated by any Markov process at steady-state $p$, such as the zig-zag process or the bouncy particle sampler \cite{bierkensZigZagProcessSuperefficient2019, bierkensPiecewiseDeterministicScaling2017,bouchard-coteBouncyParticleSampler2018}, by any mean-zero Gaussian process at steady-state $p$ \cite{rasmussenGaussianProcessesMachine2004} or by any random dynamical system at steady-state $p$ \cite{arnoldRandomDynamicalSystems1998}.
\end{itemize}
\end{example}

\begin{remark}
When the dynamics are given by an Itô stochastic differential equation \eqref{eq: Gaussian preserving diffusions}, a Markov blanket of the steady-state density \eqref{eq: MB is sparsity in precision matrix} does not preclude reciprocal influences between internal and external states \cite{biehlTechnicalCritiqueParts2021,fristonInterestingObservationsFree2021}. For example,
\begin{align*}
 &\Pi = \begin{bmatrix} 2 & 1 &0\\
 1&2&1 \\
 0& 1& 2\end{bmatrix}, \quad Q \equiv \begin{bmatrix} 0 & 0 &1\\
 0&0&0 \\
 -1& 0& 0\end{bmatrix},\quad  \varsigma \equiv \operatorname{Id}_3 \\
  &\Rightarrow d\begin{bmatrix} \eta_t \\b_t \\ \mu_t \end{bmatrix}=- \begin{bmatrix} 1 & 1.5 &2 \\
  0.5 & 1 &0.5 \\
  -2 &-0.5&1\end{bmatrix} \begin{bmatrix} \eta_t \\b_t \\ \mu_t \end{bmatrix} dt + \varsigma dW_t.
\end{align*}
Conversely, the absence of reciprocal coupling between two states in the drift in some instances, though not always, leads to conditional independence \cite{biehlTechnicalCritiqueParts2021,aguileraHowParticularPhysics2022a,fristonStochasticChaosMarkov2021}.
\end{remark}

\subsection{Maximum a posteriori estimation}

The Markov blanket \eqref{eq: MB over time} allows us to harness the construction of Section \ref{sec: markov blanket} to determine expected external and internal states given blanket states
    \begin{align*}
         \boldsymbol \eta_t := \boldsymbol \eta(b_t) \qquad \boldsymbol \mu_t:= \boldsymbol \mu(b_t).
    \end{align*}
Note that $\boldsymbol \eta, \boldsymbol \mu$ are linear functions of blanket states; since $b_t$ generally exhibits rough sample paths, $\boldsymbol \eta_t,\boldsymbol \mu_t$ will also exhibit very rough sample paths.

We can view the steady-state density $p$ as specifying the relationship between external states ($\eta$, causes) and particular states ($b, \mu$, consequences). In statistics, this corresponds to a generative model, a probabilistic specification of how (external) causes generate (particular) consequences.

By construction, the expected internal states encode expected external states via the synchronisation map
\begin{align*}
    \sigma( \boldsymbol \mu_t)&=\boldsymbol \eta_t,
\end{align*}
which manifests a form of generalised synchrony across the Markov blanket \cite{jafriGeneralizedSynchronyCoupled2016,cuminGeneralisingKuramotoModel2007,palaciosEmergenceSynchronyNetworks2019}. Moreover, the expected internal state $\boldsymbol \mu_t$ effectively follows the most likely cause of its sensations
    \begin{align*}
       \sigma( \boldsymbol \mu_t)&= \arg \max p(\eta_t \mid b_t) \quad \text{for any } t.
    \end{align*}
This has an interesting statistical interpretation as expected internal states perform maximum a posteriori (MAP) inference over external states.

\subsection{Predictive processing}

We can go further and associate to each internal state $\mu$ a probability distribution over external states, such that each internal state encodes beliefs about external states
    \begin{align}
        q_\mu(\eta)&:= \mathcal N(\eta; \sigma(\mu), \Pi_{\eta}^{-1}). \label{eq: def approx posterior}
    \end{align}
We will call $q_\mu$ the approximate posterior belief associated with the internal state $\mu$ due to the forecoming connection to inference. Under this specification, the mean of the approximate posterior depends upon the internal state, while its covariance equals that of the true posterior w.r.t. external states \eqref{eq: posterior beliefs}. It follows that the approximate posterior equals the true posterior when the internal state $\mu$ equals the expected internal state $\boldsymbol{\mu}(b)$ (given blanket states):

\begin{equation}
    q_{\mu}(\eta)=p(\eta|b) \iff \mu = \boldsymbol{\mu}(b). \label{eq: approx posterior equals true posterior}
\end{equation}
  
Note a potential connection with epistemic accounts of quantum mechanics; namely, a world governed by classical mechanics ($\sigma \equiv 0$ in \eqref{eq: Gaussian preserving diffusions}) in which each agent encodes Gaussian beliefs about external states could appear to the agents as reproducing many features of quantum mechanics \cite{bartlettReconstructionGaussianQuantum2012}.

Under this specification \eqref{eq: approx posterior equals true posterior}, expected internal states are the unique minimiser of a Kullback-Leibler divergence \cite{kullbackInformationSufficiency1951} 
\begin{align*}
   \boldsymbol \mu_t = \arg \min_\mu \dkl[q_{\mu}(\eta)\| p(\eta|b)]
\end{align*}
that measures the discrepancy between beliefs about the external world $q_\mu(\eta)$ and the posterior distribution over external variables. Computing the KL divergence (see Section \ref{app: free energy}), we obtain

\begin{align}
\label{eq: precision weighted prediction error}
  \boldsymbol \mu_t = \arg \min_\mu (\sigma(\mu)-\boldsymbol \eta_t)\Pi_\eta(\sigma(\mu)-\boldsymbol \eta_t)
\end{align}

In the neurosciences, the right hand side of \eqref{eq: precision weighted prediction error} is commonly known as a (squared) precision-weighted prediction error: the discrepancy between the prediction and the (expected) state of the environment is weighted with a precision matrix \cite{bogaczTutorialFreeenergyFramework2017,raoPredictiveCodingVisual1999,fristonPredictiveCodingFreeenergy2009} that derives from the steady-state density. This equation is formally similar to that found in predictive coding formulations of biological function \cite{chaoLargeScaleCorticalNetworks2018,iglesiasHierarchicalPredictionErrors2013,dawModelBasedInfluencesHumans2011,raoPredictiveCodingVisual1999}, which stipulate that organisms minimise prediction errors, and in doing so optimise their beliefs to match the distribution of external states.

\subsection{Variational Bayesian inference}

\begin{figure}
    \centering
    \includegraphics[width= 0.47\textwidth]{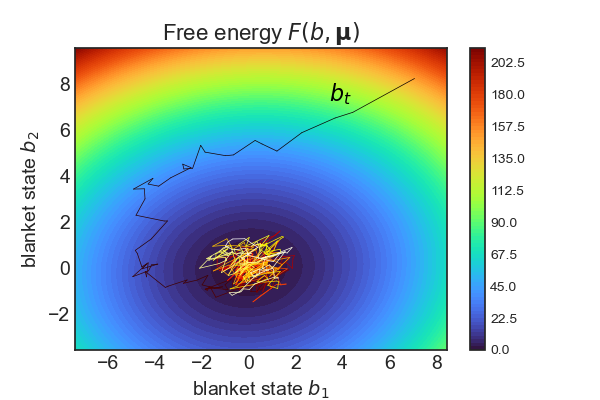}
    \includegraphics[width= 0.47\textwidth]{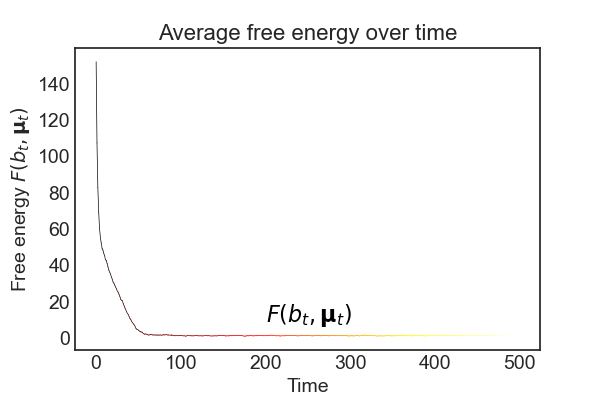}
    \includegraphics[width= 0.47\textwidth]{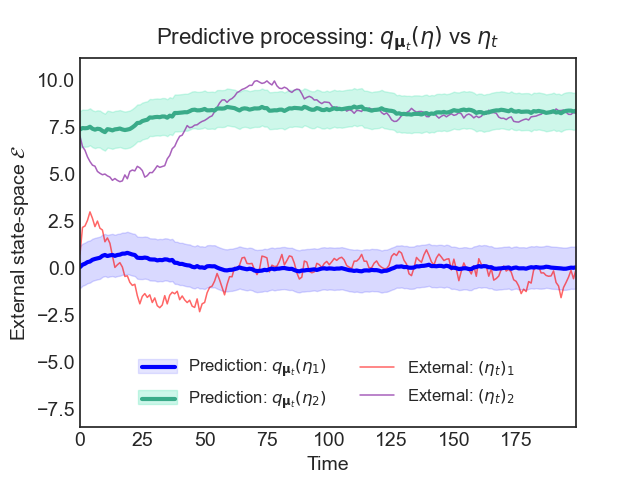}
    \includegraphics[width= 0.47\textwidth]{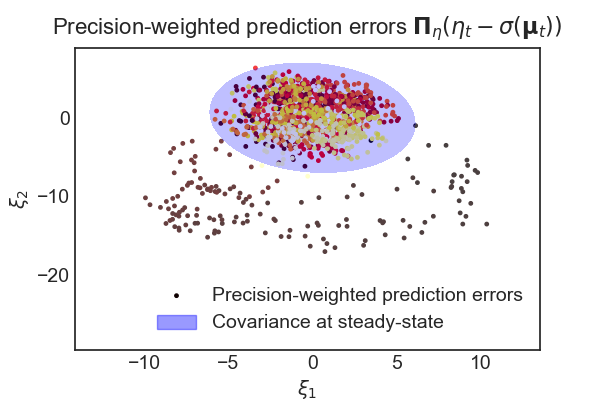}
    \caption[Variational inference and predictive processing, averaging internal variables for any blanket state]{\textbf{Variational inference and predictive processing, averaging internal variables for any blanket state}. This figure illustrates a system's behaviour after experiencing a surprising blanket state, averaging internal variables for any blanket state. This is a multidimensional Ornstein-Uhlenbeck process, with two external, blanket and internal variables, initialised at the steady-state density conditioned upon an improbable blanket state $p(x_0 |b_0)$. \textit{Upper left:} we plot a sample trajectory of the blanket states as these relax to steady-state over a contour plot of the free energy (up to a constant). \textit{Upper right:} this plots the free energy (up to a constant) over time, averaged over multiple trajectories. In this example, the rare fluctuations that climb the free energy landscape vanish on average, so that the average free energy decreases monotonically. This need not always be the case: conservative systems (i.e., $\varsigma \equiv 0$ in \eqref{eq: Gaussian preserving diffusions}) are deterministic flows along the contours of the steady-state density (see Section \ref{sec: helmholtz decomposition}). Since these contours do not generally coincide with those of $F(b , \boldsymbol{\mu})$ it follows that the free energy oscillates between its maximum and minimum value over the system's periodic trajectory. Luckily, conservative systems are not representative of dissipative, living systems. Yet, it follows that the average free energy of expected internal variables may increase, albeit only momentarily, in dissipative systems \eqref{eq: Gaussian preserving diffusions} whose solenoidal flow dominates dissipative flow. \textit{Lower left}: we illustrate the accuracy of predictions over external states of the sample path from the upper left panel. At steady-state (from timestep $\sim 100$), the predictions become accurate. The prediction of the second component is offset by four units for greater visibility, as can be seen from the longtime behaviour converging to four instead of zero. \textit{Lower right:} We show the evolution of precision-weighted prediction errors $\xi_t := \boldsymbol{\Pi}_{\eta}(\eta_t - \sigma(\boldsymbol{\mu}_t))$ over time. These are normally distributed with zero mean at steady-state.}
    \label{fig: Bayesian mechanics 6d}
\end{figure}

\begin{figure}
    \centering
    \includegraphics[width= 0.47\textwidth]{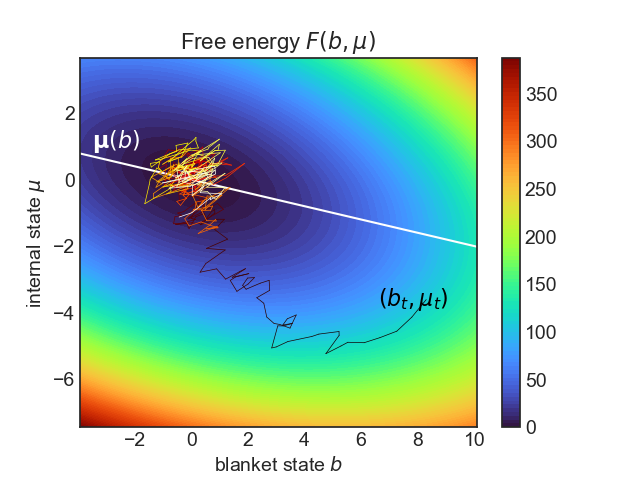}
    \includegraphics[width= 0.47\textwidth]{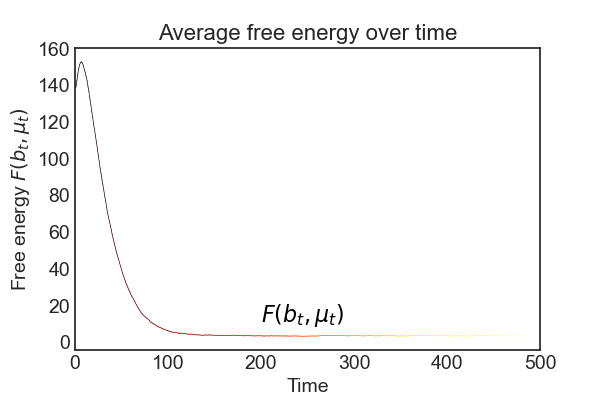}
    \includegraphics[width= 0.47\textwidth]{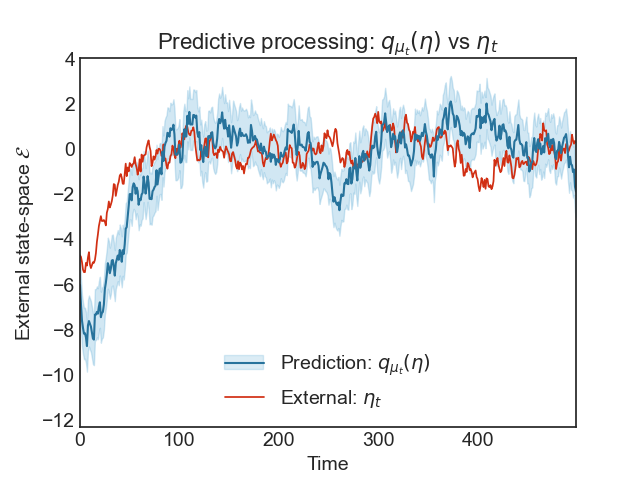}
    \includegraphics[width= 0.47\textwidth]{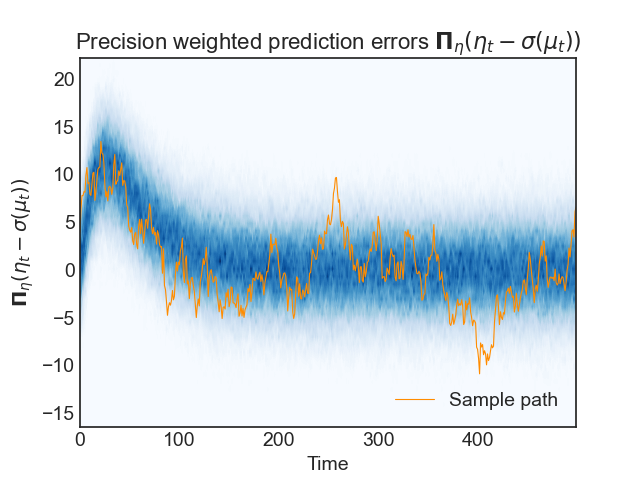}
    \caption[Variational inference and predictive processing]{\textbf{Variational inference and predictive processing}. This figure illustrates a system's behaviour after experiencing a surprising blanket state. This is a multidimensional Ornstein-Uhlenbeck process, with one external, blanket and internal variable, initialised at the steady-state density conditioned upon an improbable blanket state $p(x_0 |b_0)$.
    \textit{Upper left:} this plots a sample trajectory of particular states as these relax to steady-state over a contour plot of the free energy. The white line shows the expected internal state given blanket states, at which point inference is exact. After starting close to this line, the process is driven by solenoidal flow to regions where inference is inaccurate. Yet, solenoidal flow makes the system converge faster to steady-state \cite{ottobreMarkovChainMonte2016,rey-belletIrreversibleLangevinSamplers2015} at which point inference becomes accurate again. \textit{Upper right:} this plots the free energy (up to a constant) over time, averaged over multiple trajectories. \textit{Lower left:} we illustrate the accuracy of predictions over external states of the sample path from the upper left panel. These predictions are accurate at steady-state (from timestep $\sim 100$). \textit{Lower right:} we illustrate the (precision weighted) prediction errors over time. In orange we plot the prediction error corresponding to the sample path in the upper left panel; the other sample paths are summarised as a heat map in blue.}
    \label{fig: Bayesian mechanics 3d}
\end{figure}

We can go further and associate expected internal states to the solution to the classical variational inference problem from statistical machine learning \cite{bleiVariationalInferenceReview2017} and theoretical neurobiology \cite{bogaczTutorialFreeenergyFramework2017,buckleyFreeEnergyPrinciple2017}. Expected internal states are the unique minimiser of a free energy functional (i.e., an evidence bound \cite{bealVariationalAlgorithmsApproximate2003,bleiVariationalInferenceReview2017})
    \begin{equation}
    \label{eq: free energy def}
    \begin{split}
        F(b_t, \mu_t) &\geq F(b_t, \boldsymbol \mu_t)\\
        F(b, \mu) &= \dkl[q_{\mu}(\eta)\| p(\eta|b)] -\log p(b, \mu)\\
        &= \underbrace{\E_{q_{\mu}(\eta)}[-\log p(x) ]}_{\text{Energy}}- \underbrace{\H[q_\mu]}_{\text{Entropy}}.
    \end{split}
    \end{equation}
The last line expresses the free energy as a difference between energy and entropy: energy or accuracy measures to what extent predicted external states are close to the true external states, while entropy penalises beliefs that are overly precise.

At first sight, variational inference and predictive processing are solely useful to characterise the average internal state given blanket states at steady-state. It is then surprising to see that the free energy says a great deal about a system's expected trajectories as it relaxes to steady-state. Figure \ref{fig: Bayesian mechanics 6d} and \ref{fig: Bayesian mechanics 3d} illustrate the time-evolution of the free energy and prediction errors after exposure to a surprising stimulus. In particular, Figure \ref{fig: Bayesian mechanics 6d} averages internal variables for any blanket state: In the neurosciences, perhaps the closest analogy is the event-triggered averaging protocol, where neurophysiological responses are averaged following a fixed perturbation, such a predictable neural input or an experimentally-controlled sensory stimulus (e.g., spike-triggered averaging, event-related potentials) \cite{schwartzSpiketriggeredNeuralCharacterization2006,sayerTimeCourseAmplitude1990,luckIntroductionEventRelatedPotential2014}.

The most striking observation is the nearly monotonic decrease of the free energy as the system relaxes to steady-state. This simply follows from the fact that regions of high density under the steady-state distribution have a low free energy. This \textit{overall decrease} in free energy is the essence of the free-energy principle, which describes self-organisation at non-equilibrium steady-state \cite{fristonFreeenergyPrincipleUnified2010,fristonFreeEnergyPrinciple2019a,parrMarkovBlanketsInformation2020}. Note that the free energy, even after averaging internal variables, may decrease non-monotonically. See the explanation in Figure \ref{fig: Bayesian mechanics 6d}.

\section{Active inference and stochastic control}
\label{sec: active inference and stoch control}

In order to model agents that interact with their environment, we now partition blanket states into sensory and active states

    \begin{align*}
        b_t &= (s_t,a_t)\\
        x_t &= (\eta_t, s_t, a_t, \mu_t).
    \end{align*}
Intuitively, sensory states are the sensory receptors of the system (e.g., olfactory or visual receptors) while active states correspond to actuators through which the system influences the environment (e.g., muscles). See Figure \ref{fig: 4 way dynamical MB}. The goal of this section is to explain how autonomous states (i.e., active and internal states) respond adaptively to sensory perturbations in order to maintain the steady-state, which we interpret as the agent's preferences or goal. This allows us to relate the dynamics of autonomous states to active inference and stochastic control, which are well-known formulations of adaptive behaviour in theoretical biology and engineering.

\begin{figure}[!h]
\centering\includegraphics[width=0.7\textwidth]{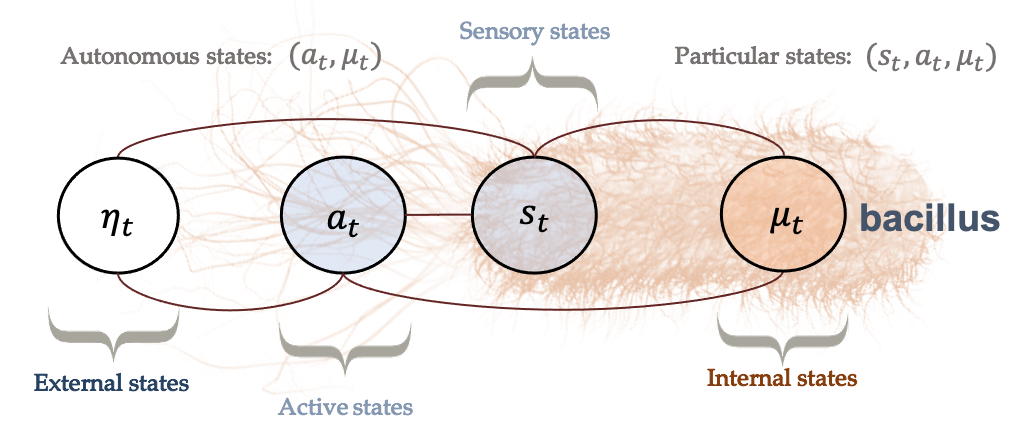}
\caption[Markov blanket evolving in time comprising sensory and active states]{\textbf{Markov blanket evolving in time comprising sensory and active states.} We continue the intuitive example from Figure \ref{fig: 3 way dynamical MB} of the bacillus as representing a Markov blanket that persists over time. The only difference is that we partition blanket states into sensory and active states. In this example, the sensory states can be seen as the bacillus' membrane, while the active states correspond to the actin filaments of the cytoskeleton.}
\label{fig: 4 way dynamical MB}
\end{figure}

\subsection{Active inference}

We now proceed to characterise autonomous states, given sensory states, using the free energy. Unpacking blanket states, the free energy \eqref{eq: free energy def} reads
    \begin{align*}
    F(s, a, \mu) &=\dkl[q_{\mu}(\eta)\| p(\eta|s,a)]-\log p( \mu |s,a ) -\log p(a|s ) -\log p(s ).
    \end{align*}

Crucially, it follows that the expected autonomous states minimise free energy
    \begin{align*}
        F(s_t, a_t, \mu_t) &\geq F( s_t, \boldsymbol a_t,\boldsymbol \mu_t),\\
    \boldsymbol a_t := \boldsymbol a(s_t) &:=  \E_{p(a_t|s_t)}[a_t] =\Sigma_{a s} \Sigma_{s}^{-1} s_t,
    \end{align*}
where $\boldsymbol a_t$ denotes the expected active states given sensory states, which is the mean of $p(a_t|s_t)$. This result forms the basis of active inference, a well-known framework to describe and generate adaptive behaviour in neuroscience, machine learning and robotics \cite{buckleyFreeEnergyPrinciple2017,ueltzhofferDeepActiveInference2018,millidgeDeepActiveInference2020,heinsDeepActiveInference2020,lanillosRobotSelfOther2020,verbelenActiveInferenceFirst2020,adamsPredictionsNotCommands2013,pezzatoNovelAdaptiveController2020,oliverEmpiricalStudyActive2021,fristonActionBehaviorFreeenergy2010}. See Figure \ref{fig: active inference}.

\begin{figure}
    \centering
    \includegraphics[width= 0.47\textwidth]{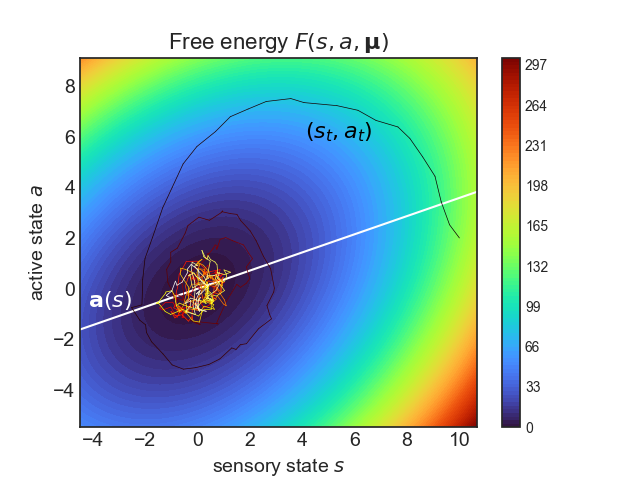}
    \includegraphics[width= 0.47\textwidth]{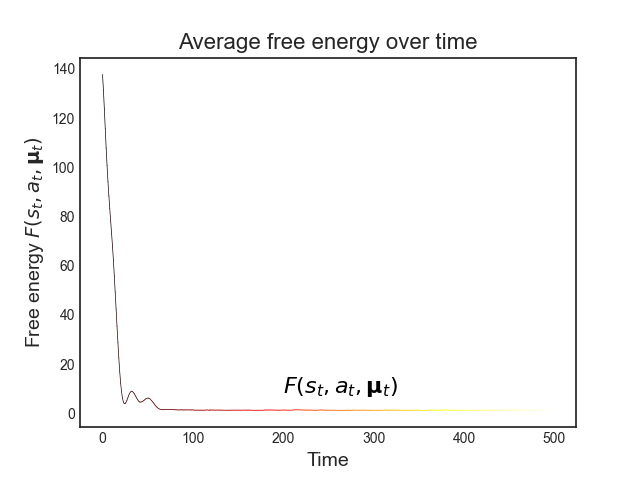}
    \caption[Active inference]{\textbf{Active inference}. 
    This figure illustrates a system's behaviour after experiencing a surprising sensory state, averaging internal variables for any blanket state. We simulated an Ornstein-Uhlenbeck process with two external, one sensory, one active and two internal variables, initialised at the steady-state density conditioned upon an improbable sensory state $p(x_0 |s_0)$. \textit{Left:} The white line shows the expected active state given sensory states: this is the action that performs active inference and optimal stochastic control. As the process experiences a surprising sensory state, it initially relaxes to steady-state in a winding manner due to the presence of solenoidal flow. Even though solenoidal flow drives the actions away from the optimal action initially, it allows the process to converge faster to steady-state \cite{rey-belletIrreversibleLangevinSamplers2015,ottobreMarkovChainMonte2016,lelievreOptimalNonreversibleLinear2013} where the actions are again close to the optimal action from optimal control. \textit{Right:} We plot the free energy of the expected internal state, averaged over multiple trajectories. In this example, the average free energy does not decrease monotonically---see Figure \ref{fig: Bayesian mechanics 6d} for an explanation.}
    \label{fig: active inference}
\end{figure}

\subsection{Multivariate control}

Active inference is used in various domains to simulate control \cite{koudahlWorkedExampleFokkerPlanckBased2020,verbelenActiveInferenceFirst2020,oliverEmpiricalStudyActive2021,ueltzhofferDeepActiveInference2018,fristonWhatOptimalMotor2011,sancaktarEndtoEndPixelBasedDeep2020,baltieriPIDControlProcess2019,pezzatoNovelAdaptiveController2020}, thus, it is natural that we can relate the dynamics of active states to well-known forms of stochastic control.

By computing the free energy explicitly (see Section \ref{app: free energy}), we obtain that

\begin{align}
   (\boldsymbol a_t, \boldsymbol \mu_t) \quad &\text{minimises}\quad  (a, \mu) \mapsto \begin{bmatrix}s_t,a, \mu  \end{bmatrix} K \begin{bmatrix}s_t\\a\\ \mu  \end{bmatrix} \label{eq: multivariate control}\\
   K &:= \Sigma_{b:\mu}^{-1} \nonumber
\end{align}
where we denoted by $K$ the concentration (i.e., precision) matrix of $p(s,a, \mu)$. We may interpret $(\boldsymbol a, \boldsymbol \mu)$ as controlling how far particular states $\left[s,a, \mu  \right]$ are from their target set-point of $\left[0,0, 0  \right]$, where the error is weighted by the precision matrix $K$. See Figure \ref{fig: stochastic control}. (Note that we could choose any other set-point by translating the frame of reference or equivalently choosing a Gaussian steady-state centred away from zero). In other words, there is a cost associated to how far away $s, a, \mu$ are from the origin and this cost is weighed by the precision matrix, which derives from the stationary covariance of the steady-state. In summary, the expected internal and active states can be seen as performing multivariate stochastic control, where the matrix $K$ encodes control gains. From a biologist’s perspective, this corresponds to a simple instance of homeostatic regulation: maintaining physiological variables within their preferred range.

\begin{figure}
    \centering
    \includegraphics[width= 0.47\textwidth]{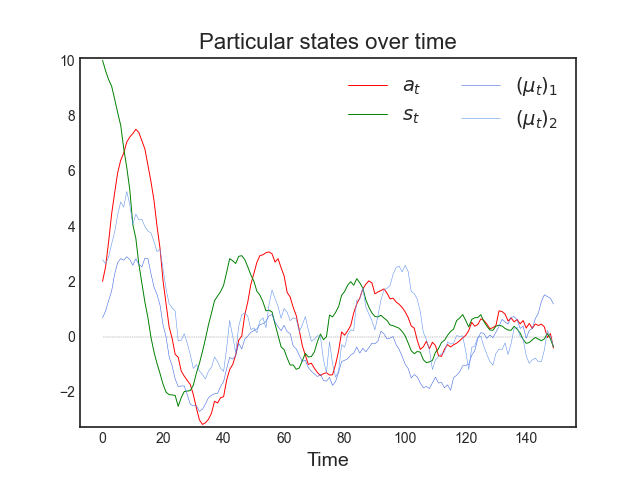}
    \caption[Stochastic control]{\textbf{Stochastic control}. This figure plots a sample path of the system's particular states after it experiences a surprising sensory state. This is the same sample path as shown in Figure \ref{fig: active inference} (left panel), however, here the link with stochastic control is easier to see. Indeed, it looks as if active states (in red) are actively compensating for sensory states (in green): rises in the active state-space lead to plunges in the sensory state-space and vice-versa. Notice the initial rise in active states to compensate for the perturbation in the sensory states. Active states follow a similar trajectory as sensory states, with a slight delay, which can be interpreted as a reaction time \cite{kosinskiLiteratureReviewReaction2012}. In fact, the correspondence between sensory and active states is a consequence of the solenoidal flow--see Figure \ref{fig: active inference} (left panel). The damped oscillations as the particular states approach their target value of $0$ (in grey) is analogous to that found in basic implementations of stochastic control, e.g., \cite[Figure 4.9]{roskillyMarineSystemsIdentification2015}.}
    \label{fig: stochastic control}
\end{figure}

\subsection{Stochastic control in an extended state-space}

More sophisticated control methods, such as PID (Proportional-Integral-Derivative) control \cite{astromPidControllers1995,baltieriPIDControlProcess2019}, involve controlling a process and its higher orders of motion (e.g., integral or derivative terms). So how can we relate the dynamics of autonomous states to these more sophisticated control methods? The basic idea involves extending the sensory state-space to replace the sensory process $s_t$ by its various orders of motion $\tilde s_t = \left(s^{(0)}_t, \ldots, s^{(n)}_t \right)$ (integral, position, velocity, jerk etc, up to order $n$). To find these orders of motion, one must solve the stochastic realisation problem.

\subsubsection{The stochastic realisation problem}
\label{sec: stochastic realisation}

Recall that the sensory process $s_t$ is a stationary stochastic process (with a Gaussian steady-state). The following is a central problem in stochastic systems theory: Given a stationary stochastic process $s_t$, find a Markov process $\tilde s_t$, called the state process, and a function $f$ such that 
\begin{align}
\label{eq: stochastic realisation problem}
    s_t=f(\tilde s_t) \quad \text{for all}\quad  t.
\end{align}
Moreover, find an Itô stochastic differential equation whose unique solution is the state process $\tilde s_t$. The problem of characterising the family of all such representations is known as the stochastic realisation problem \cite{mitterTheoryNonlinearStochastic1981}.

What kind of processes $s_t$ can be expressed as a function of a Markov process \eqref{eq: stochastic realisation problem}?

There is a rather comprehensive theory of stochastic realisation for the case where $s_t$ is a Gaussian process (which occurs, for example, when $x_t$ is a Gaussian process). This theory expresses $s_t$ as a linear map of an Ornstein-Uhlenbeck process \cite{lindquistLinearStochasticSystems2015,lindquistRealizationTheoryMultivariate1985,pavliotisStochasticProcessesApplications2014}. The idea is as follows: as a mean-zero Gaussian process, $s_t$ is completely determined by its autocovariance function $C(t-r)=\E\left[s_{t} \otimes s_{r}\right]$, which by stationarity only depends on $|t-r|$. It is well known that any mean-zero stationary Gaussian process with exponentially decaying autocovariance function is an Ornstein-Uhlenbeck process (a result sometimes known as Doob's theorem) \cite{doobBrownianMovementStochastic1942,wangTheoryBrownianMotion2014,rey-belletOpenClassicalSystems2006,pavliotisStochasticProcessesApplications2014}. Thus if $C$ equals a finite sum of exponentially decaying functions, we can express $s_t$ as a linear function of several nested Ornstein-Uhlenbeck processes, i.e., as an integrator chain from control theory \cite{kryachkovFinitetimeStabilizationIntegrator2010,zimenkoFinitetimeFixedtimeStabilization2018}

\begin{equation}
\label{eq: integrator chain}
\begin{split}
    s_t &= f(s^{(0)}_t)\\
    ds^{(0)}_t &= f_0(s^{(0)}_t, s^{(1)}_t) dt + \varsigma_0 dW^{(0)}_t \\
    ds^{(1)}_t &= f_1(s^{(1)}_t, s^{(2)}_t) dt + \varsigma_1 dW^{(1)}_t \\
    \vdots\quad  & \qquad \vdots\qquad \vdots \\
    ds^{(n-1)}_t &= f_{n-1}(s^{(n-1)}_t, s^{(n)}_t) dt + \varsigma_{n-1} dW^{(n-1)}_t \\
    ds^{(n)}_t &= f_n(s^{(n)}_t) dt + \varsigma_n dW^{(n)}_t. \\
\end{split}
\end{equation}

In this example, $f,f_i$ are suitably chosen linear functions, $\varsigma_i$ are matrices and $W^{(i)}$ are standard Brownian motions. Thus, we can see $s_t$ as the output of a continuous-time hidden Markov model, whose (hidden) states $s^{(i)}_t$ encode its various orders of motion: position, velocity, jerk etc. These are known as generalised coordinates of motion in the Bayesian filtering literature \cite{fristonVariationalFiltering2008a,fristonVariationalTreatmentDynamic2008,fristonGeneralisedFiltering2010}. See Figure \ref{fig: HMM}. 

\begin{figure}
    \centering
    \includegraphics[width=0.5\textwidth]{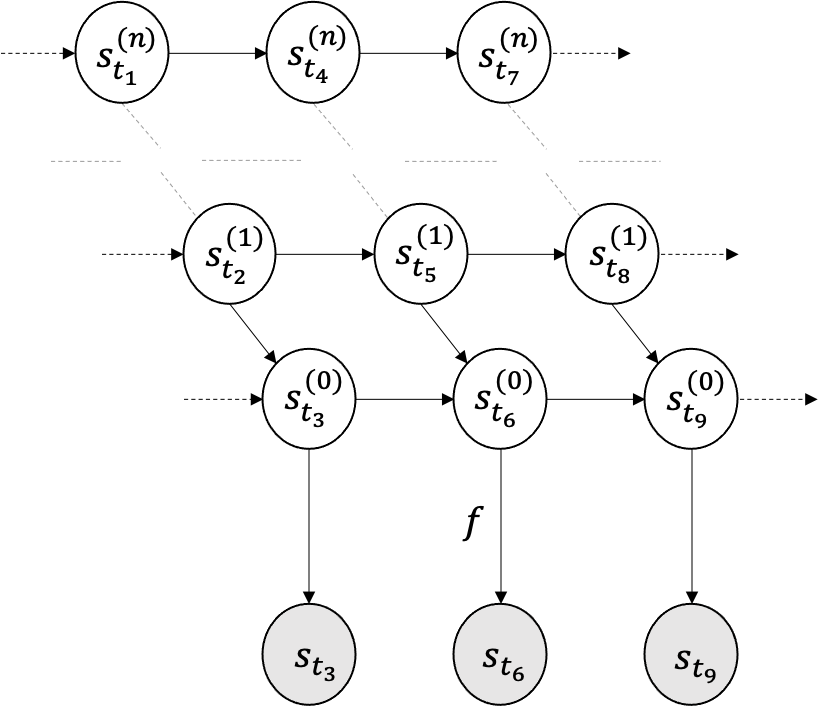}
    \caption[Continuous-time Hidden Markov model]{\textbf{Continuous-time Hidden Markov model}. This figure depicts \eqref{eq: integrator chain} in a graphical format, as a Bayesian network \cite{wainwrightGraphicalModelsExponential2007,pearlGraphicalModelsProbabilistic1998}. The encircled variables are random variables---the processes indexed at an arbitrary sequence of subsequent times $t_1<t_2<\ldots <t_9$. The arrows represent relationships of causality. In this hidden Markov model, the (hidden) state process $\tilde s_t$ is given by an integrator chain---i.e., nested stochastic differential equations $s^{(0)}_t, s^{(1)}_t, \ldots, s^{(n)}_t$. These processes $s^{(i)}_t, i\geq 0$, can respectively be seen as encoding the position, velocity, jerk etc, of the process $s_t$.}
    \label{fig: HMM}
\end{figure}

More generally, the state process $\tilde s_t$ and the function $f$ need not be linear, which enables to realise non-linear, non-Gaussian processes $s_t$ \cite{fristonVariationalFiltering2008a, parrComputationalNeurologyMovement2021a,gomesMeanFieldLimits2020}. Technically, this follows as Ornstein-Uhlenbeck processes are the only stationary Gaussian Markov processes. Note that stochastic realisation theory is not as well developed in this general case \cite{tayorNonlinearStochasticRealization1989,mitterTheoryNonlinearStochastic1981,frazhoStochasticRealizationTheory1982,gomesMeanFieldLimits2020,fristonVariationalFiltering2008a}.

\subsubsection{Stochastic control of integrator chains}

Henceforth, we assume that we can express $s_t$ as a function of a Markov process $\tilde s_t$ \eqref{eq: stochastic realisation problem}. Inserting \eqref{eq: stochastic realisation problem} into \eqref{eq: multivariate control}, we now see that the expected autonomous states minimise how far themselves and $f(\tilde s_t)$ are from their target value of zero

\begin{align}
\label{eq: stoch control of integrator chains}
 (\boldsymbol a_t, \boldsymbol \mu_t) \quad \text{minimises}\quad (a, \mu)\mapsto  \begin{bmatrix}f(\tilde s_t),a, \mu  \end{bmatrix}  K \begin{bmatrix}f(\tilde s_t)\\a\\ \mu  \end{bmatrix}.
\end{align}

Furthermore, if the state process $\tilde s_t$ can be expressed as an integrator chain, as in \eqref{eq: integrator chain}, then we can interpret expected active and internal states as controlling each order of motion $s^{(i)}_t$. For example, if $f$ is linear, these processes control each order of motion $s^{(i)}_t$ towards its target value of zero.

\subsubsection{PID-like control}

Proportional-integral-derivative (PID) control is a well-known control method in engineering \cite{astromPidControllers1995,baltieriPIDControlProcess2019}. More than 90\% of controllers in engineered systems implement either PID or PI (no derivative) control. The goal of PID control is to control a signal $s^{(1)}_t$, its integral $s^{(0)}_t$, and its derivative $s^{(2)}_t$ close to a pre-specified target value \cite{baltieriPIDControlProcess2019}.

This turns out to be exactly what happens here when we consider the stochastic control of an integrator chain \eqref{eq: stoch control of integrator chains} with three orders of motion $(n=2)$. When $f$ is linear, expected autonomous states control integral, proportional and derivative processes $s^{(0)}_t,s^{(1)}_t,s^{(2)}_t$ towards their target value of zero. Furthermore, from $f$ and $K$ one can derive integral, proportional and derivative gains, which penalise deviations of $s^{(0)}_t,s^{(1)}_t,s^{(2)}_t$, respectively, from their target value of zero. Crucially, these control gains are simple by-products of the steady-state density and the stochastic realisation problem.

Why restrict ourselves to PID control when stochastic control of integrator chains is available? It turns out that when sensory states $s_t$ are expressed as a function of an integrator chain \eqref{eq: integrator chain}, one may get away by controlling an approximation of the true (sensory) process, obtained by truncating high orders of motion as these have less effect on the dynamics, though knowing when this is warranted is a problem in approximation theory. This may explain why integral feedback control ($n=0$), PI control ($n=1$) and PID control ($n=2$) are the most ubiquitous control methods in engineering applications. However, when simulating biological control---usually with highly non-linear dynamics---it is not uncommon to consider generalised motion to fourth ($n=4$) or sixth ($n=6$) order \cite{fristonGraphicalBrainBelief2017,parrComputationalNeurologyMovement2021a}.

It is worth mentioning that PID control has been shown to be implemented in simple molecular systems and is becoming a popular mechanistic explanation of behaviours such as bacterial chemotaxis and robust homeostatic algorithms in biochemical networks \cite{baltieriPIDControlProcess2019,chevalierDesignAnalysisProportionalIntegralDerivative2019,yiRobustPerfectAdaptation2000}. We suggest that this kind of behaviour emerges in Markov blankets at non-equilibrium steady-state. Indeed, stationarity means that autonomous states will look as if they respond adaptively to external perturbations to preserve the steady-state, and we can identify these dynamics as implementations of various forms of stochastic control (including PID-like control).

\section{Discussion}

In this chapter, we considered the consequences of a boundary mediating interactions between states internal and external to a system. On unpacking this notion, we found that the states internal to a Markov blanket look as if they perform variational Bayesian inference, optimising beliefs about their external counterparts. When subdividing the blanket into sensory and active states, we found that autonomous states perform active inference and various forms of stochastic control (i.e., generalisations of PID control). 

\textbf{Interacting Markov blankets:}
The sort of inference we have described could be nuanced by partitioning the external state-space into several systems that are themselves Markov blankets (such as Markov blankets nested at several different scales \cite{hespMultiscaleViewEmergent2019}). From the perspective of internal states, this leads to a more interesting inference problem, with a more complex generative model. It may be that the distinction between the sorts of systems we generally think of as engaging in cognitive, inferential, dynamics \cite{fristonHierarchicalModelsBrain2008} and simpler systems rest upon the level of structure of the generative models (i.e., steady-state densities) that describe their inferential dynamics.

\textbf{Temporally deep inference:} This distinction may speak to a straightforward extension of the treatment on offer, from simply inferring an external state to inferring the trajectories of external states. This may be achieved by representing the external process in terms of its higher orders of motion by solving the stochastic realisation problem. By repeating the analysis above, internal states may be seen as inferring the position, velocity, jerk, etc of the external process, consistently with temporally deep inference in the sense of a Bayesian filter \cite{fristonGeneralisedFiltering2010} (a special case of which is an extended Kalman–Bucy filter \cite{kalmanNewApproachLinear1960}).

\textbf{Bayesian mechanics in non-Gaussian steady-states:} The treatment from this chapter extends easily to non-Gaussian steady-states, in which internal states appear to perform approximate Bayesian inference over external states. Indeed, any arbitrary (smooth) steady-state density may be approximated by a Gaussian density at one of its modes using a so-called Laplace approximation. This Gaussian density affords one with a synchronisation map in closed form\footnote{Another option is to empirically fit a synchronisation map to data \cite{fristonLifeWeKnow2013}.} that maps the expected internal state to an approximation of the expected external state. It follows that the system can be seen as performing approximate Bayesian inference over external states---precisely, an inferential scheme known as variational Laplace \cite{fristonVariationalFreeEnergy2007}. We refer the interested reader to a worked-out example involving two sparsely coupled Lorenz systems \cite{fristonStochasticChaosMarkov2021}. Note that variational Laplace has been proposed as an implementation of various cognitive processes in biological systems \cite{buckleyFreeEnergyPrinciple2017,bogaczTutorialFreeenergyFramework2017, fristonActionBehaviorFreeenergy2010} accounting for several features of the brain’s functional anatomy and neural message passing \cite{fristonTheoryCorticalResponses2005,fristonHierarchicalModelsBrain2008,fristonPredictiveCodingFreeenergy2009,adamsPredictionsNotCommands2013,pezzuloActiveInferenceView2012}.

\textbf{Modelling real systems:} The simulations presented here are as simple as possible and are intended to illustrate general principles that apply to all stationary processes with a Markov blanket \eqref{eq: MB over time}. These principles have been used to account for synthetic data arising in more refined (and more specific) simulations of an interacting particle system \cite{fristonLifeWeKnow2013} and synchronisation between two sparsely coupled stochastic Lorenz systems \cite{fristonStochasticChaosMarkov2021}. Clearly, an outstanding challenge is to account for empirical data arising from more interesting and complex structures. To do this, one would have to collect time-series from an organism's internal states (e.g., neural activity), its surrounding external states, and its interface, including sensory receptors and actuators. Then, one could test for conditional independence between internal, external and blanket states \eqref{eq: MB over time} \cite{pelletUsingMarkovBlankets2008}. One might then test for the existence of a synchronisation map (using Lemma \ref{lemma: equiv sigma well defined}). This speaks to modelling systemic dynamics using stochastic processes with a Markov blanket. For example, one could learn the volatility, solenoidal flow and steady-state density in a stochastic differential equation \eqref{eq: Gaussian preserving diffusions} from data, using supervised learning \cite{tzenNeuralStochasticDifferential2019}.

\section{Conclusion}

This chapter outlines some of the key relationships between stationary processes, inference and control. These relationships rest upon partitioning the world into those things that are internal or external to a (statistical) boundary, known as a Markov blanket. When equipped with dynamics, the expected internal states appear to engage in variational inference, while the expected active states appear to be performing active inference and various forms of stochastic control.

The rationale behind these findings is rather simple: if a Markov blanket derives from a steady-state density, the states of the system will look as if they are responding adaptively to external perturbations in order to recover the steady-state. Conversely, well-known methods used to build adaptive systems implement the same kind of dynamics, implicitly so that the system maintains a steady-state with its environment.

\section{Addendum: Proofs for Chapter 3}

\subsection{Existence of synchronisation map: proof}
\label{app: sync map existence proof}

We prove Lemma \ref{lemma: equiv sigma well defined}.

\begin{proof}
$(i) \iff (ii)$ follows by definition of a function.

$(ii) \iff (iii)$ is as follows
\begin{equation*}
\begin{split}
   & \quad \forall b_1, b_2 \in \B: \boldsymbol \mu(b_1)= \boldsymbol \mu(b_2) \Rightarrow \boldsymbol \eta(b_1)= \boldsymbol \eta(b_2)\\
   &\iff \left(\forall b_1, b_2 \in \B: \Sigma_{\mu b} \Sigma_{b}^{-1} b_1 = \Sigma_{\mu b} \Sigma_{b}^{-1} b_2 \Rightarrow \Sigma_{\eta b} \Sigma_{b}^{-1} b_1 = \Sigma_{\eta b} \Sigma_{b}^{-1} b_2\right) \\
    &\iff \left(\forall b \in \B:  \Sigma_{\mu b} \Sigma_{b}^{-1} b=0 \Rightarrow \Sigma_{\eta b} \Sigma_{b}^{-1} b =0 \right) \\
    &\iff  \ker \Sigma_{\mu b} \subset  \ker \Sigma_{\eta b}\\
\end{split}
\end{equation*}

$(iii) \iff (iv)$ From \cite[Section 0.7.3]{hornMatrixAnalysisSecond2012}, using the Markov blanket condition \eqref{eq: MB is sparsity in precision matrix}, we can verify that
\begin{align*}
   \Pi_{\mu} \Sigma_{\mu b} &=- \Pi_{\mu b} \Sigma_b \\
   \Pi_{\eta}\Sigma_{\eta b} &=-  \Pi_{\eta b} \Sigma_b.
\end{align*}
Since $ \Pi_{\mu}, \Pi_{\eta}, \Sigma_b$ are invertible, we deduce
\begin{equation*}
\begin{split}
& \quad \ker \Sigma_{\mu b} \subset  \ker \Sigma_{\eta b}\\
&\iff \ker \Pi_{\mu} \Sigma_{\mu b} \subset  \ker \Pi_{\eta}\Sigma_{\eta b} \\
&\iff \ker -\Pi_{\mu b} \Sigma_b \subset  \ker -\Pi_{\eta b} \Sigma_b \\
&\iff \ker \Pi_{\mu b} \subset  \ker \Pi_{\eta b}.
\end{split}
\end{equation*}
\end{proof}

\subsection{Free energy computations}
\label{app: free energy}

The free energy reads \eqref{eq: free energy def}
\begin{align*}
    F(b, \mu) &=\dkl[q_{\mu}(\eta)\| p(\eta|b)]-\log p(b, \mu ).
\end{align*}
Recalling from \eqref{eq: posterior beliefs}, \eqref{eq: def approx posterior} that $q_{\mu}(\eta)$ and $ p(\eta|b)$ are Gaussian, the KL divergence between multivariate Gaussians is well-known
\begin{align*}
    &q_{\mu}(\eta) = \mathcal N(\eta ; \sigma(\mu), \Pi_\eta^{-1}), \quad p(\eta|b) = \mathcal N(\eta ; \boldsymbol \eta(b), \Pi_\eta^{-1}), \\
    &\Rightarrow \dkl[q_{\mu}(\eta)\| p(\eta|b)] = \frac 1 2 (\sigma(\mu)-\boldsymbol \eta(b))\Pi_\eta(\sigma(\mu)-\boldsymbol \eta(b)).
\end{align*}
Furthermore, we can compute the log partition
\begin{align*}
    -\log p(b, \mu) = \frac 1 2 \begin{bmatrix}b, \mu  \end{bmatrix} \Sigma_{b:\mu}^{-1} \begin{bmatrix}b \\ \mu  \end{bmatrix} \quad \text{(up to a constant)}.
\end{align*}
Note that $\Sigma_{b:\mu}^{-1}$ is the inverse of a principal submatrix of $\Sigma$, which in general differs from $\Pi_{b:\mu}$, a principal submatrix of $\Pi$. Finally,
\begin{align*}
    F(b, \mu) &= \frac 1 2 (\sigma(\mu)-\boldsymbol \eta(b))\Pi_\eta(\sigma(\mu)-\boldsymbol \eta(b)) +\frac 1 2 \begin{bmatrix}b, \mu  \end{bmatrix} \Sigma_{b:\mu}^{-1} \begin{bmatrix}b \\ \mu  \end{bmatrix}  \quad \text{(up to a constant)}.
\end{align*}

\Xchapter{The free-energy principle}{The free-energy principle made simpler but not too simple}{By Karl Friston, Lancelot Da Costa, Noor Sajid, Conor Heins, Kai Ueltzhoeffer, Grigorios A. Pavliotis, Thomas Parr\blfootnote{\normalsize \textbf{Adapted from:} K Friston, L Da Costa, N Sajid, C Heins, K Ueltzhoeffer, GA Pavliotis, T Parr. The free energy principle made simpler but not too simple. \textit{Physics Reports}. 2023}}
\newpage
\section{Abstract}

This chapter provides a concise description of the free energy principle, starting from a formulation of random dynamical systems in terms of a Langevin equation and ending with a Bayesian mechanics that can be read as a physics of sentience.\footnote{Sentience here is meant for a thing as having sensory states through which it is coupled with its external states, and being described as behaving adaptively in accordance to the sensory states.} It rehearses the key steps using standard results from statistical physics. These steps entail (i) establishing a particular partition of states based upon conditional independencies that inherit from sparsely coupled dynamics, (ii) unpacking the implications of this partition in terms of Bayesian inference and (iii) describing the paths of particular states with a variational principle of least action. Teleologically, the free energy principle offers a normative account of self-organisation in terms of optimal Bayesian design and decision-making, in the sense of maximising marginal likelihood or Bayesian model evidence. In summary, starting from a description of the world in terms of random dynamical systems, we end up with a description of self-organisation as sentient behaviour that can be interpreted as self-evidencing; namely, self-assembly, autopoiesis or active inference.

\textbf{Keywords}: self-organisation, nonequilibrium, variational inference, Bayesian, Markov blanket.

\section{Introduction}
\label{sec: intro}

It is said that the free energy principle is difficult to understand. This is ironic on three counts. First, the free energy principle (FEP) is so simple that it is (almost) tautological. Indeed, philosophical accounts compare its explanandum to a desert landscape, in the sense of Quine~\cite{clarkWhateverNextPredictive2013}. Second, a tenet of the FEP is that everything must provide an accurate account of things that is as simple as possible—including itself. Finally, the FEP rests on straightforward results from statistical physics. This review tries to present the free energy principle as simply as possible but without sacrificing too much technical detail. It steps through the formal arguments that lead from a description of the world as a random dynamical system~\cite{crauelAttractorsRandomDynamical1994,arnoldRandomDynamicalSystems1998} to a description of self-organisation in terms of active inference and self-evidencing~\cite{hohwySelfEvidencingBrain2016}. The evidence in question is Bayesian model evidence, which speaks to the Bayesian mechanics on offer~\cite{fristonFreeEnergyPrinciple2019a}. These mechanics have the same starting point as = statistical and classical mechanics. The only difference is that careful attention is paid to the way that the internal states of something couple to its external states.

To make the following account accessible, we use a conversational style, explaining the meaning of key mathematical expressions intuitively. Accordingly, simplifying notation and assumptions are used to foreground the basic ideas. Before starting, it might help to clarify what the free energy principle is—and why it is useful. Many theories in the biological sciences are answers to the question: “what must things do, in order to exist?” The FEP turns this question on its head and asks: “if things exist, what must they do?” More formally, if we can define what it means to be something, can we identify the physics or dynamics that a thing must possess? To answer this question, the FEP calls on some mathematical truisms that follow from each other. Much like Hamilton's principle of least action\footnote{Perhaps a better analogy would be Noether's theorem (Beren Millidge – personal communication)~\cite{noetherInvariantenBeliebigerDifferentialausdrucke1918}.}, it is not a falsifiable theory about the way ‘things’ behave—it is a general description of ‘things’ that are defined in a particular way. As such, the FEP is not falsifiable as a mathematical statement, but it may as well be falsifiable to the extent that its postulates refer to a specific class of empirical phenomena that the principle aims to describe.

Is such a description useful? In itself, the answer is probably no—in the sense that the principle of least action does not tell you how to throw a ball. However, the principle of least action furnishes everything we need to know to simulate the trajectory of a ball in a particular instance. In the same sense, the FEP allows one to simulate and predict the sentient behaviour of a particle, person, artefact or agent (i.e., some ‘thing’). This allows one to build sentient artefacts or use simulations as observation models of particles (or people). These simulations rest upon specifying a \textit{generative model} that is apt to describe the behaviour of the particle (or person) at hand. At this point, committing to a specific generative model can be taken as a commitment to a specific—and falsifiable—theory. Later, we will see some examples of these simulations.

The remaining sections describe the FEP. Each section focuses on an equation—or set of equations—used in subsequent sections. The ensuing narrative is meant to be concise, taking us from the beginning to the end as succinctly as possible. To avoid disrupting the narrative, we use footnotes to address questions that are commonly asked at each step. We also use figure legends to supplement the narrative with examples from neurobiology. Most of the following can be found in the literature~\cite{fristonFreeEnergyPrinciple2019a,dacostaBayesianMechanicsStationary2021a,fristonStochasticChaosMarkov2021}; however, there are a few simplifications that replace earlier accounts. 

\section{Systems, states and fluctuations}

We start by describing the world with a stochastic differential equation~\cite{pavliotisStochasticProcessesApplications2014}. So why start here? The principal reason is that we want a description that is consistent with physics. This follows because things like the fluctuation theorems in statistical mechanics and the Lagrangian formulation of classical mechanics can all be derived from this starting point~\cite{seifertStochasticThermodynamicsFluctuation2012}. In short, if one wants a physics of sentience, this is a good place to start. 

We are interested in systems that have characteristic states. Technically, this means the system has a pullback attractor; namely, a set of states a system will come to occupy from any initial state~\cite{crauelAttractorsRandomDynamical1994,crauelGlobalRandomAttractors1999}. Such systems can be described with stochastic differential equations, such as the Langevin equation describing the rate of change of some states $x(\tau)$, in terms of their flow $f(x)$, and random fluctuations $\omega(\tau)$. The fluctuations are usually assumed to be a normally distributed (white noise) process, with a covariance of $2\Gamma$:
\begin{equation}
\label{eq: system}
\begin{split}
\dot{x}(\tau) &=f(x)+\omega(\tau) \\
p(\omega \mid x) &=\mathcal{N}(\omega; 0,2 \Gamma) \Rightarrow p(\dot{x} \mid x)=\mathcal{N}(\dot{x}; f, 2 \Gamma) \\
p(x) &=?
\end{split}
\end{equation}
The dot notation denotes a derivative with respect to time\footnote{\textbf{Question}: why is the flow in \eqref{eq: system} not a function of time? Many treatments of stochastic thermodynamics allow for time-dependent flows when coupling one system (e.g., an idealised gas) to another (e.g., a heat reservoir), where it is assumed that the other system changes very slowly, e.g.,~\cite{jarzynskiNonequilibriumEqualityFree1997,seifertStochasticThermodynamicsFluctuation2012}. However, the ambition of the FEP is to describe this coupling under a partition of states. In this setting, separation of temporal scales is an emergent property, where \eqref{eq: system} holds at any given temporal scale. See~\cite{fristonFreeEnergyPrinciple2019a} for a treatment using the apparatus of the renormalisation group.}. This means that time and causality are baked into everything that follows, in the sense that states cause their motion. The Langevin equation is itself an approximation to a simpler mapping from some variables to changes in those variables with time. This follows from the separation into states and random fluctuations implicit in \eqref{eq: system}, where states change slowly in relation to fast fluctuations. This (adiabatic) approximation is ubiquitous in physics~\cite{carrApplicationsCentreManifold1982,hakenSynergeticsIntroductionNonequilibrium1978,koidePerturbativeExpansionIrreversible2017}. In brief, it means we can ignore temporal correlations in the fast fluctuations and assume—by the central limit theorem—that they have a Gaussian distribution. This equips the fluctuations with a probability density, which means we know their statistical behaviour but not their trajectory or path, which itself is a random variable~\cite{crauelAttractorsRandomDynamical1994,arnoldRandomDynamicalSystems1998,pavliotisStochasticProcessesApplications2014}.

The next step, shared by all physics, is to ask whether anything can be said about the probability density over the states—the ‘?’ in \eqref{eq: system}. A lot can be said about this probability density, which can be expressed in two complementary ways; namely, as \textit{density dynamics} using the Fokker-Planck equation (a.k.a. the forward Kolmogorov equation) or in terms of the probability of a path through state-space using the \textit{path-integral formulation}. The Fokker-Planck equation describes the change in the density due to random fluctuations and the flow of states through state-space~\cite{riskenFokkerPlanckEquationMethods1996,pavliotisStochasticProcessesApplications2014}:

\begin{equation}
\label{eq: FP eq}
\dot{p}(x, \tau)=\nabla \cdot(\Gamma \nabla-f(x)) p(x, \tau)
\end{equation}

The Fokker-Planck equation describes our stochastic process in terms of deterministic density dynamics—instead of specific realisations—where the density in question is over \emph{states}
$x(\tau)=x_{\tau}$. Conversely, the path-integral formulation considers the probability of a trajectory or \textit{path} $
x[\tau]\triangleq [x(t): 0 \leq t \leq \tau]
$ in terms of its \emph{action} $\mathcal{A}$ (omitting additive constants here and throughout)\footnote{\textbf{Question}: where does the divergence in the third equality come from? This term arises from the implicit use of Stratonovich path integrals~\cite{seifertStochasticThermodynamicsFluctuation2012}. Note that we have assumed that the amplitude of random fluctuations is state—and therefore path—independent in \eqref{eq: system}, which means we can place it outside the integral in the second equality.}:

\begin{equation}
\label{eq: action}
\begin{aligned}
\mathcal{A}(x[\tau]) &=-\ln p\left(x[\tau] \mid x_{0}\right) \\
&=\frac{\tau}{2} \ln |(4 \pi)^{n}\Gamma|+\int_{0}^{\tau} d t \mathcal{L}(x, \dot{x}) \\
\mathcal{L}(x, \dot{x}) &=\frac{1}{2}\left[(\dot{x}-f) \cdot \frac{1}{2 \Gamma}(\dot{x}-f)+\nabla \cdot f\right]
\end{aligned}
\end{equation}
 
Both the Fokker-Planck and path-integral formulations inherit their functional form from assumptions about the statistics of random fluctuations in \eqref{eq: system}. For example, the most likely path—or path of least action—is the path taken when the fluctuations take their most likely value of zero. This means that variations away from this path always increase the action. This is expressed mathematically by saying that its variation is zero when the action is minimised.
\footnote{Omitting the contribution of the divergence term in the Lagrangian to obtain the expression for the path of least action for simplicity, cf. \cite{durrOnsagerMachlupFunctionLagrangian1978}. Taking this simplification at face value means that we are either: 1) considering a description on a short time-scale as the flow can be approximated by a linear function with impunity (e.g., linear response theory, see \cite{pavliotisStochasticProcessesApplications2014}); or 2) we are considering the limit where random fluctuations have vanishingly small amplitude (e.g., precise particles, see Sections \ref{sec: precise} and \ref{sec: curious}).} 

\begin{equation}
\label{eq: path of least action}
\begin{aligned}
\mathbf{x}[\tau] &=\arg \min _{x[\tau]} \mathcal{A}(x[\tau]) \\
& \Leftrightarrow \delta_{x} \mathcal{A}(\mathbf{x}[\tau])=0 \\
& \Leftrightarrow \dot{\mathbf{x}}(\tau)=f(\mathbf{x})
\end{aligned}
\end{equation}
 								
In short, the motion on the path of least action is just the flow without random fluctuations. Paths of least action will figure prominently below; especially, when considering systems that behave in a precise or predictable way. We will denote the most likely states and paths with a bold typeface.

Although equivalent, the Fokker-Planck and path-integral formalisms provide complementary perspectives on dynamics. The former deals with time-dependent probability densities over \textit{states}, while the latter considers time-independent densities over \textit{paths}. The density over n states at any particular time is the time-marginal of the density over trajectories. These probabilities can be conveniently quantified in terms of their negative logarithms (or potentials) leading to surprisal and action, respectively (omitting the divergence of the flow in the last line for simplicity):

\begin{equation}
\label{eq: 5}
\begin{aligned}
\Im(x, \tau) &\triangleq -\ln p(x, \tau) \\
\mathcal{A}(x[\tau]) &\triangleq -\ln p\left(x[\tau] \mid x_{0}\right) \\
\H[p(x, \tau)] &=\mathbb{E}[\Im(x, \tau)] \\
\H\left[p\left(x[\tau] \mid x_{0}\right)\right] &=\mathbb{E}[\mathcal{A}(x[\tau])] \\
&=\frac{\tau}{2} \ln \left[(4 \pi)^{n}|\Gamma|\right]+\int_{0}^{\tau} d t \frac{1}{2} \mathbb{E}_{p(\omega)}\left[\omega(t) \cdot \frac{1}{2 \Gamma} \omega(t)\right]=\frac{\tau}{2} \ln \left[(4 \pi e)^{n}|\Gamma|\right]
\end{aligned}
\end{equation}

The second set of equalities shows that the uncertainty (or entropy) about states and their paths is the expected surprisal and action, respectively. Perhaps counterintuitively, the entropy of paths is easier to specify than the entropy of states. This follows because the only source of uncertainty about paths—given an initial state—are the random fluctuations~\cite{seifertStochasticThermodynamicsFluctuation2012,pavliotisStochasticProcessesApplications2014}, whose probability density does not change with time. The last pair of equalities in \eqref{eq: 5} show that the amplitude of random fluctuations determines the entropy of paths. Intuitively, if the fluctuations are large, then many distinct paths become equally plausible, and the entropy of paths increases\footnote{From a thermodynamic perspective, uncertainty about paths increases with temperature. For example, the Einstein-Smoluchowski relation relates the amplitude of random fluctuations to a mobility coefficient times the temperature $\Gamma=\mu_{m} k_{B} T$.}. 

\section{Solutions, steady-states and nonequilibria}

So far, we have equations that describe the relationship between the dynamics of a system and probability densities over fluctuations, states and their paths. This is sufficient to elaborate much of physics. For example,
we could focus on systems that comprise statistical ensembles of similar states to derive stochastic and statistical mechanics in terms of fluctuation theorems~\cite{seifertStochasticThermodynamicsFluctuation2012}. Finally, we could consider large systems—in which the fluctuations are averaged away—to derive classical mechanics such as electromagnetism.
All of these mechanics require some boundary conditions: for example, 
a heat bath or reservoir in statistical mechanics and a classical potential for Lagrangian mechanics. At this point, the FEP steps back and asks, where do these boundary conditions come from? Indeed, this was implicit in Schrodinger's question:

\say{\textit{How can the events in space and time which take place within the spatial boundary of a living organism be accounted for by physics and chemistry?}}~\cite{schrodingerWhatLifeMind2012}.

We read a boundary in a statistical sense as a Markov boundary~\cite{pearlGraphicalModelsProbabilistic1998}\footnote{A Markov boundary is a subset of states of the system that renders the states of a ‘thing’ or particle conditionally independent from all other states~\cite{pearlCausality2009}.}. Why? Because the only thing we have at hand is a probabilistic description of the system. And the only way to separate the states of something from its boundary states is in terms of probabilistic independencies—in this instance, conditional independencies\footnote{Noting that if two subsets of states were independent, as opposed to being conditionally independent, we would be describing two separate systems.}. This means we need to identify a partition of states that assigns a subset to a ‘thing’ or particle and another subset to the boundary that separates the thing from some ‘thing’ else. In short, one has to define ‘thingness’ in terms of conditional independencies.

However, if things are defined in terms of conditional independencies and conditional independencies are attributes of a probability density, where does the density come from? The Fokker-Planck equation shows that the density over states depends upon time, even if the flow does not. This means that if we predicate ‘thingness’ on a probability density, it may only exist for a vanishingly small amount of time. This simple observation compels us to consider probability densities that do not change with time, namely: (i) steady-state solutions to the Fokker-Planck equation or (ii) the density over paths. We will start with the (slightly more delicate) treatment of steady-state solutions and then show that the (slightly more straightforward) treatment of densities over paths leads to the same notion of ‘thingness’.

The existence of things over a particular timescale implies the density in \eqref{eq: FP eq} does not change over that timescale. This is what is meant by a steady-state solution to the Fokker-Planck equation. The ensuing density is known as a \textit{steady-state density} and, in random dynamical systems, implies the existence of a pullback attractor~\cite{crauelAttractorsRandomDynamical1994,arnoldRandomDynamicalSystems1998}. The notion of an attractor is helpful here, in the sense that it comprises a set of characteristic states, to which the system is attracted over time\footnote{More precisely, the time-dependent solutions to the Fokker-Planck equation will tend towards the stationary solution, or steady-state. In other words, the steady-state density becomes a point attractor in the space of probability densities.}. In short, to talk about ‘things’, we are implicitly talking about a partition of states in a random dynamical system that has an attracting set—i.e., a steady-state solution to the Fokker-Planck equation. In short, we consider systems that self-organise towards a steady-state density\footnote{At this point, the formalism applies equally to steady-states with a high or low entropy, as we have not committed to a particular form of the steady-state density. Later, we will specialise to steady-states with a low entropy to characterise the sort of self-organisation that describes biological systems, e.g., swarming or flocking~\cite{hakenSynergeticsIntroductionNonequilibrium1978,nicolisSelforganizationNonequilibriumSystems1977}}. This solution is also known as a nonequilibrium steady-state (NESS) density, where the ‘nonequilibrium’ aspect rests upon solenoidal flow, as we will see next.

The existence of a solution to the Fokker-Planck equation—i.e., the existence of something—means that we can express the flow of states in terms of the steady-state density (or corresponding surprisal) using a generalisation of the Helmholtz decomposition. This decomposes the flow into conservative (rotational, divergence-free) and dissipative (irrotational, curl-free) components—with respect to the steady-state density—referred to as \textit{solenoidal} and \textit{gradient} flows, respectively~\cite{grahamCovariantFormulationNonequilibrium1977,eyinkHydrodynamicsFluctuationsOutside1996,shiRelationNewInterpretation2012,maCompleteRecipeStochastic2015,barpUnifyingCanonicalDescription2021,dacostaEntropyProductionStationary2022,pavliotisStochasticProcessesApplications2014}:

\begin{equation}
\label{eq: 6}
\begin{aligned}
&\dot{p}(x)=0 \Leftrightarrow f(x)=\Omega(x) \nabla \Im(x)-\Lambda(x)=\underbrace{Q(x) \nabla \Im(x)}_{\text {Solenoidal flow }}-\underbrace{\Gamma \nabla \Im(x)}_{\text {Gradient flow }}-\Lambda(x) \\
&\Im(x)=-\ln p(x), \quad Q=-Q^{T}, \quad 
\Lambda_{i}\triangleq\sum_{j} \frac{\partial \Omega_{i j}}{\partial x_{j}} = \sum_j \frac{\partial Q_{ij}}{\partial x_{j}}.
\end{aligned}
\end{equation}
  
This can be understood intuitively as a decomposition of the flow into two parts. The first (conservative) part of the flow is a \textit{solenoidal} circulation on the isocontours of the steady-state density (or surprisal). This component breaks \textit{detailed balance} and renders the steady-state density a \textit{nonequilibrium} steady-state density~\cite{aoPotentialStochasticDifferential2004,yuanPotentialFunctionDynamical2011}. The second (dissipative) part performs a (natural) gradient descent on the steady-state surprisal and depends upon the amplitude of random fluctuations~\cite{girolamiRiemannManifoldLangevin2011,amariNaturalGradientWorks1998}. The final term, $\Lambda$, can be regarded as a correction term, which is neither curl-free nor divergence-free, and which ensures that the probability density remains constant over time~\cite{dacostaEntropyProductionStationary2022}.

\subsection*{Summary}

We now have a probabilistic description of a system in terms of a (NESS) density that admits conditional independencies among states. These conditional independencies are necessary to separate the states of things from their boundaries. In the next step, we will see how conditional independencies inherit from sparse coupling among states—and how they are used to establish a particular partition of states.

\section{Particles, partitions and things}
\label{sec: particles, partitions and things}

In associating some (stochastic differential) equations of motion with a unique (NESS) density, we have a somewhat special setup, in which the influences entailed by the equations of motion place constraints on the conditional independencies of the NESS density
. These conditional independencies can be used to identify a particular partition of states into \textit{external}, \textit{sensory}, \textit{active} and \textit{internal} states as summarised below. This is an important move because it separates the states of a \textit{particle} (i.e., internal states and their sensory and active states) from the remaining (i.e., \textit{external}) states. However, to do this we have to establish how the causal dynamics in \eqref{eq: system} underwrite conditional independencies. This can be done simply by using the curvature (Hessian) of surprisal as follows:
\begin{equation}
\label{eq: 7}
\begin{aligned}
\left(x_{u} \perp x_{v}\right) \mid b & \Leftrightarrow p(x)=p\left(x_{u} \mid b\right) p\left(x_{v} \mid b\right) p(b) \\
& \Leftrightarrow \Im(x)=\Im\left(x_{u} \mid b\right)+\Im\left(x_{v} \mid b\right)+\Im(b) \Leftrightarrow \frac{\partial^{2} \Im}{\partial x_{u} \partial x_{v}}=\mathbf{H}_{u v}=\mathbf{H}_{v u}=0.
\end{aligned}
\end{equation}
This says that if the $u$-th state is conditionally independent of the $v$-th state, given the remaining states $b$, then the corresponding element of the curvature—or Hessian matrix—of surprisal must be zero. Conversely, a zero entry in the Hessian implies conditional independence. In sum, any two states are conditionally independent if, and only if, the change of surprisal with one state does not depend on the other. We can now use the Helmholtz decomposition \eqref{eq: 6} to express the Jacobian—i.e., the (linear) coupling—of the flow in terms of the Hessian—that entails conditional independencies (with a slight abuse of the dot product notation):
\begin{equation}
\label{eq: 8}
\begin{aligned}
f(x) &=\Omega \nabla \Im-\Lambda \\
& \Rightarrow \mathbf{J}=\Omega \mathbf{H}+\nabla \Omega \cdot \nabla \Im-\nabla \Lambda \\
& \Rightarrow \mathbf{J}_{u v}=\frac{\partial f_{u}}{\partial x_{v}}=\sum_{i} \Omega_{u i} \mathbf{H}_{i v}+\sum_{i} \frac{\partial \Omega_{u i}}{\partial x_{v}} \frac{\partial \Im}{\partial x_{i}}-\sum_{i} \frac{\partial^{2} \Omega_{u i}}{\partial x_{i} \partial x_{v}}.
\end{aligned}
\end{equation}
We can now define \textit{sparse coupling} as a solution to this equation, in which all the terms are identically zero\footnote{This implicitly precludes edge cases, in which some non-zero terms cancel.}:
\begin{equation}
\label{eq: 9}
\left.\begin{array}{l}
Q_{u i} \mathbf{H}_{i v} \\
\Gamma_{u} \mathbf{H}_{u v} \\
\partial \Omega_{u i} / \partial x_{v}
\end{array}\right\}=0: \forall i \Rightarrow \mathbf{J}(x)_{u v}=0.
\end{equation}
Sparse coupling means that the Jacobian coupling states $u$ and $v$ is zero, i.e., an absence of coupling from one state to another. This definition precludes solenoidal coupling with $ u$ that depends on $v$. Because
$\mathbf{H}(x)_{v v}$ and $\Gamma_{u}$ are positive definite, sparse coupling requires associated elements of the solenoidal operator and Hessian to vanish at every point in state-space, which in turn, implies conditional independence:
\begin{equation}
\label{eq: 10}
\begin{aligned}
Q_{u v} \mathbf{H}_{v v} &=0 \Rightarrow Q_{u v}=-Q_{v u}=0 \\
\Gamma_{u} \mathbf{H}_{u v} &=0 \Rightarrow \mathbf{H}_{u v}=\mathbf{H}_{v u}=0 \Leftrightarrow\left(x_{u} \perp x_{v}\right) \mid b.
\end{aligned}
\end{equation}

In short, sparse coupling means that any two states are conditionally independent \textit{if one state does not influence the other}. This is an important observation; namely, that sparse coupling implies a NESS density with conditional independencies. In turn, this means any dynamical influence graph with absent or directed edges admits a \textit{Markov blanket} (the states $b$ above). These independencies can now be used to build a particular partition as follows:
\begin{itemize}
    \item The Markov boundary $a\subset x$  of a set of internal states $\mu \subset x$ is the minimal set of states for which there exists a nonzero Hessian submatrix: $\mathbf{H}_{a \mu} \neq 0$. In other words, the internal states are independent of the remaining states, when conditioned upon their Markov boundary, called \textit{active states}. The combination of active and internal states will be referred to as \textit{autonomous states}: $\alpha=(a, \mu)$.
    \item The Markov boundary $s \subset x$ of autonomous states is the minimal set of states for which there exists a nonzero Hessian submatrix: $\mathbf{H}_{s \alpha} \neq 0$. In other words, the autonomous states are independent of the remaining states, when conditioned upon their Markov boundary, called \textit{sensory states}. The combination of active and sensory (i.e., boundary) states constitute \textit{blanket states}: $b=(s, a)$. The internal and blanket states will be referred to as \textit{particular states}: $\pi=(s, \alpha)=(b, \mu)$.
    \item The remaining states constitute \textit{external states}: $x=(\eta, \pi)$. 
\end{itemize}
The names of active and sensory (i.e., blanket) states inherit from the literature, where they are often associated with biotic systems that act on—and sense—their external milieu\footnote{\textbf{Question}: why does a particular partition comprises four sets of states? In other words, why does a particular partition consider two Markov boundaries; namely, sensory and active states? The reason is that this is the minimal partition that allows for directed coupling with blanket states. For example, sensory states can influence internal states—and active states can influence external states—without destroying the conditional independencies of the particular partition (these directed influences are illustrated in the upper panel of Figure \ref{fig: 1} as dotted arrows).}. In this setting, one can regard external states as influencing internal states via sensory states (directly or through active states). And internal states influence external states via active states (directly or through sensory states\footnote{\textbf{Question}: does this mean that I can act on my world through my sense organs? Yes: much of biotic action is mediated by (active) motile cytoskeletal filaments, muscles and secretory organs that lie beneath (sensory) epithelia, such as receptors on the skin or a cell surface.}). We will see later how this implies a synchronisation between internal and external states, in the sense that internal states can be seen as actively inferring external states~\cite{dacostaBayesianMechanicsStationary2021a,fristonStochasticChaosMarkov2021}. The ensuing conditional independencies implied by a particular partition can be summarised as follows:
\begin{equation}
\label{eq: 11}
\begin{aligned}
\mathbf{J}_{\mu \eta} &=0 \Rightarrow \mathbf{H}_{\mu \eta}=0 \Leftrightarrow(\mu \perp \eta) \mid b \\
\mathbf{J}_{a \eta} &=0 \Rightarrow \mathbf{H}_{a \eta}=0 \Leftrightarrow(a \perp \eta) \mid s, \mu \\
\mathbf{J}_{s \mu} &=0 \Rightarrow \mathbf{H}_{s \mu}=0 \Leftrightarrow(s \perp \mu) \mid a, \eta
\end{aligned}
\end{equation}
 					
A normal form for the flow and Jacobian of a particular partition—with sparse coupling—can be expressed as follows, where  $\alpha=(a, \mu)$ and $\beta=(\eta, s)$:
\begin{equation}
\label{eq: 12}
\begin{split}
f(x)&=\Omega \nabla \Im-\Lambda\\
\left[\begin{array}{l}
f_{\eta}(\eta, b) \\
f_{s}(\eta, b) \\
f_{a}(b, \mu) \\
f_{\mu}(b, \mu)
\end{array}\right]&=\left[\begin{array}{cccc}
Q_{\eta \eta}-\Gamma_{\eta} & Q_{\eta s} & & \\
-Q_{\eta s}^{T} & Q_{s s}-\Gamma_{s} & & \\
& & Q_{a a}-\Gamma_{a} & Q_{a \mu} \\
& & -Q_{a \mu}^{T} & Q_{\mu \mu}-\Gamma_{\mu}
\end{array}\right]\left[\begin{array}{c}
\nabla_{\eta} \Im(\eta \mid b) \\
\nabla_{s} \Im(b \mid \eta) \\
\nabla_{a} \Im(b \mid \mu) \\
\nabla_{\mu} \Im(\mu \mid b)
\end{array}\right]-\Lambda\\
\mathbf{J}(x)&=\Omega \mathbf{H}+\nabla \Omega \cdot \nabla \Im-\nabla \Lambda\\
\left[\begin{array}{lllll}
\mathbf{J}_{\eta \eta \eta} & \mathbf{J}_{\eta s} & \mathbf{J}_{\eta a} & \\
\mathbf{J}_{s \eta} & \mathbf{J}_{s s} & \mathbf{J}_{s a} & \\
& \mathbf{J}_{a s} & \mathbf{J}_{a a} & \mathbf{J}_{a \mu} \\
& \mathbf{J}_{\mu s} & \mathbf{J}_{\mu a} & \mathbf{J}_{\mu \mu}
\end{array}\right]&=\left[\begin{array}{cccc}
Q_{\eta \eta}-\Gamma_{\eta} & Q_{\eta s} & & \\
-Q_{\eta s}^{T} & Q_{s s}-\Gamma_{s} & & \\
& & Q_{a a}-\Gamma_{a} & Q_{a \mu} \\
& & -Q_{a \mu}^{T} & Q_{\mu \mu}-\Gamma_{\mu}
\end{array}\right]\left[\begin{array}{cccc}
\mathbf{H}_{\eta \eta} & \mathbf{H}_{\eta s} & & \\
\mathbf{H}_{\eta s}^{T} & \mathbf{H}_{s s} & \mathbf{H}_{s a} & \\
& \mathbf{H}_{s a}^{T} & \mathbf{H}_{a a} & \mathbf{H}_{a \mu} \\
& & \mathbf{H}_{a \mu}^{T} & \mathbf{H}_{\mu \mu}
\end{array}\right]\\
&+\nabla \Omega \cdot \nabla \Im-\nabla \Lambda\\
\nabla_{\eta} \Omega_{\alpha \alpha}&=0, \quad \nabla_{\mu} \Omega_{\beta \beta}=0, \quad \nabla_{\eta} \Lambda_{\alpha}=0, \quad \nabla_{\mu} \Lambda_{\beta}=0
\end{split}
\end{equation}

This normal form means that particular partitions can be defined in terms of sparse coupling. Perhaps the simplest definition—that guarantees a Markov blanket\footnote{In the absence of solenoidal coupling between autonomous and non-autonomous states, and constraints on the partial derivatives of the solenoidal coupling in \eqref{eq: 12}; i.e., solenoidal coupling among autonomous states does not depend upon external states. Similarly, for non-autonomous and internal states.}—is as follows: \textit{external states only influence sensory states and internal states only influence active states}. This means that sensory states are not influenced by internal states and active states are not influenced by external states,
\begin{equation}
\label{eq: 13}
\left[\begin{array}{c}
\dot{\eta}(\tau) \\
\dot{s}(\tau) \\
\dot{a}(\tau) \\
\dot{\mu}(\tau)
\end{array}\right]=\left[\begin{array}{l}
f_{\eta}(\eta, s, a)+\omega_{\eta}(\tau) \\
f_{s}(\eta, s, a)+\omega_{s}(\tau) \\
f_{a}(s, a, \mu)+\omega_{a}(\tau) \\
f_{\mu}(s, a, \mu)+\omega_{\mu}(\tau)
\end{array}\right]
\end{equation}
and the noise processes $\omega_{i}(\tau), i \in \{\eta, s, a, \mu \}$ are independent. Under this sparse coupling, it is simple to show that not only are internal and external states conditionally independent, but their paths are conditionally independent, given initial states, using the path integral formulation.

The uncertainty (i.e., entropy) over paths derives from random fluctuations. This means that if we knew all the influences on the flow at every point in time, we can evaluate the entropy of external and internal paths from \eqref{eq: 5}:
\begin{equation}
\label{eq: 14}
\begin{aligned}
\H\left[p\left(\eta[\tau] \mid b[\tau], x_{0}\right)\right] &=\frac{\tau}{2} \ln \left[(4 \pi e)^{n_\eta}\left|\Gamma_{\eta}\right|\right] \\
\H\left[p\left(\mu[\tau] \mid b[\tau], x_{0}\right)\right] &=\frac{\tau}{2} \ln \left[(4 \pi e)^{n_\mu}\left|\Gamma_{\mu}\right|\right] \\
& \Rightarrow \\
\H\left[p\left(\eta[\tau] \mid b[\tau], x_{0}\right)\right] &=\H\left[p\left(\eta[\tau] \mid \mu[\tau], b[\tau], x_{0}\right)\right] \Rightarrow(\mu[\tau] \perp \eta[\tau]) \mid b[\tau], x_{0} \\
\H\left[p\left(\mu[\tau] \mid b[\tau], x_{0}\right)\right] &=\H\left[p\left(\mu[\tau] \mid \eta[\tau], b[\tau], x_{0}\right)\right] \Rightarrow(\mu[\tau] \perp \eta[\tau]) \mid b[\tau], x_{0}
\end{aligned}
\end{equation}

The final equalities say that the uncertainty about external (resp., internal) paths does not change when we know the internal (resp., external) path because external (resp., internal) states do not influence internal (resp., external) flow. This means the external and internal paths do not share any mutual information and are therefore independent when conditioned on blanket paths (and initial states). From \eqref{eq: 11}, the initial external and internal states are themselves independent, when conditioned on blanket states. 

Note that the conditional independence of paths inherits directly from the sparse coupling, without any reference to the NESS density or Helmholtz decomposition. This can be seen clearly by replacing the partial derivatives in \eqref{eq: 7} with functional derivatives and noting, from \eqref{eq: 12}, that there are no flows that depend on both internal and external states:
\begin{equation}
\label{eq: 15}
\begin{gathered}
\frac{\partial^{2} f}{\partial \eta \partial \mu}=0 \Rightarrow \frac{\partial^{2} \mathcal{L}}{\partial \eta \partial \mu}=\frac{\partial^{2} \mathcal{L}}{\partial \eta \partial \dot{\mu}}=\frac{\partial^{2} \mathcal{L}}{\partial \dot{\eta} \partial \mu}=\frac{\partial^{2} \mathcal{L}}{\partial \dot{\eta} \partial \dot{\mu}}=0 \\
\Rightarrow \frac{\delta^{2} \mathcal{A}(x[\tau])}{\delta \eta[t] \delta \mu[t]}=0 \Leftrightarrow(\mu[\tau] \perp \eta[\tau]) \mid b[\tau], x_{0} \\
\frac{\delta \mathcal{A}(x[\tau])}{\delta \mu[t]}=\int d t^{\prime}\left(\frac{\partial \mathcal{L}}{\partial \mu\left[t^{\prime}\right]} \frac{\delta \mu\left[t^{\prime}\right]}{\delta \mu[t]}+\frac{\partial \mathcal{L}}{\partial \dot{\mu}\left[t^{\prime}\right]} \frac{d}{d t^{\prime}} \frac{\delta \mu\left[t^{\prime}\right]}{\delta \mu[t]}+\ldots\right) \\
\frac{\delta^{2} \mathcal{A}(x[\tau])}{\delta \eta[t] \delta \mu[t]}=\int d t^{\prime} d t^{\prime \prime}\left(\frac{\delta \eta\left[t^{\prime \prime}\right]}{\delta \eta[t]} \frac{\partial^{2} \mathcal{L}}{\partial \eta \partial \mu} \frac{\delta \mu\left[t^{\prime}\right]}{\delta \mu[t]}+\frac{\delta \eta\left[t^{\prime \prime}\right]}{\delta \eta[t]} \frac{\partial^{2} \mathcal{L}}{\partial \eta \partial \dot{\mu}} \frac{d}{d t^{\prime}} \frac{\delta \mu\left[t^{\prime}\right]}{\delta \mu[t]}+\ldots\right)
\end{gathered}
\end{equation}
These expressions mean that the probability of an internal path, given a blanket path (and initial states), does not depend on the external path and \textit{vice versa}.

\subsection*{Summary}

In summary, the internal dynamics (i.e., paths) of some ‘thing’ are conditionally independent of external paths if, and only if, the flow of internal states does not depend on external states and \textit{vice versa} (given initial states). We take this as a necessary and sufficient condition for something to exist, in the sense that it can be distinguished from everything else. When the initial states are sampled from the NESS density, the internal states are conditionally independent of external states (given blanket states), under certain constraints on solenoidal flow. Figure \ref{fig: 1} illustrates the ensuing particular partition. Note that the edges in this graph represent the influence of one state on another, as opposed to conditional dependencies. This is important because directed influences admit conditional independence. These conditional independencies are manifest as zero entries in the Hessian matrices, which inherit from the sparse, directed coupling of the dynamics.

\begin{figure}
    \centering
    \includegraphics[width=0.7\textwidth]{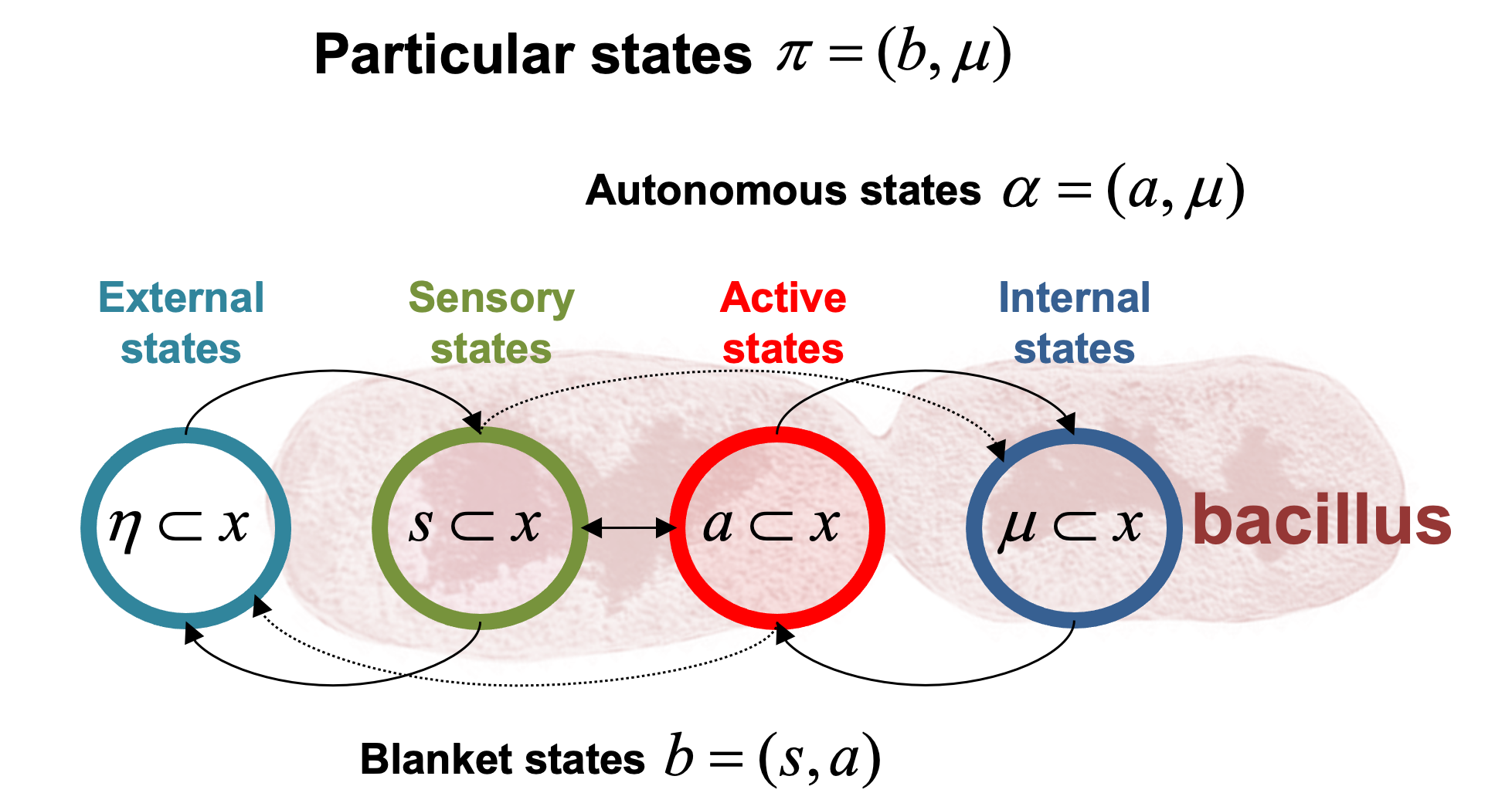}
    \caption[Markov blankets]{\textbf{Markov blankets}. This influence diagram illustrates a particular partition of states into internal states (blue) and external states (cyan) that are separated by a Markov blanket comprising sensory (green) and active states (red). The edges in this graph represent the influence of one state on another, as opposed to conditional dependencies. The diagram shows this partition as it would be applied to a single-cell organism, where internal states are associated with intracellular states, the sensory states become the surface states or cell membrane overlying active states (e.g., the actin filaments of the cytoskeleton). The dotted lines indicate allowable directed influences from sensory (resp., active) to internal (resp., external) states. Particular states constitute a particle; namely, autonomous and sensory states—or blanket and internal states.}
    \label{fig: 1}
\end{figure}

\section{From self-organisation to self-evidencing}

Equipped with a particular partition, we can now talk about things in terms of their internal states and Markov boundary; namely autonomous states. And we can talk about autonomous states and their Markov boundary; namely, particular states—the states of a particle. The next step is to characterise the flow of the autonomous states (of a particle, plant or person) in relation to external states. In other words, we consider the nature of the coupling between the outside and inside of a particle, across its Markov blanket. It is at this point that we move towards a (Bayesian) mechanics that is the special provenance of systems with particular partitions.

The existence of a particular partition means that—given sensory states—one can stipulatively define the conditional density over external states as being parameterised by the most likely internal state~\cite{dacostaBayesianMechanicsStationary2021a}\footnote{In other words, the internal mode supplies the sufficient statistics of the conditional density over external states.}. We will call this a \textit{variational density} parameterised by the \textit{internal} mode $\boldsymbol{\mu}(\tau)$\footnote{\textbf{Question}: what if the conditional densities are not well-behaved, e.g., what if there are no unique modes? The answer is that well-behaved densities are generally guaranteed when increasing the dimensionality of state-spaces using generalised coordinates of motion \cite{kerrGeneralizedPhaseSpace2000,fristonGeneralisedFiltering2010,fristonPathIntegralsParticular2023}. In other words, instead of just dealing with states, we consider states and their generalised motion to arbitrarily high order. We will see examples of this later.}:
\begin{equation}
\label{eq: 16}
\begin{aligned}
q_{\boldsymbol \mu}(\eta) & \triangleq p(\eta \mid s) \\
\boldsymbol{\alpha}(\tau) &=(\mathbf{a}(\tau), \boldsymbol{\mu}(\tau)) \\
\boldsymbol{\alpha}(\tau) &=\arg \min _{\alpha} \Im(\alpha(\tau) \mid s(\tau)) \Rightarrow \\
\boldsymbol{\alpha}[\tau] &=\arg \min _{\alpha} \mathcal{A}(\alpha[\tau] \mid s[\tau]) \Rightarrow \\
\dot{\boldsymbol{\alpha}}(\tau) &=f_{\alpha}(s, \boldsymbol{\alpha})
\end{aligned}
\end{equation} 						
As with the paths of least action, we will use bold typeface to denote a mode or most likely state, given all the states necessary to specify its likelihood. For autonomous states, we only need the sensory states, because the autonomous states are conditionally independent of external states. 

Inducing the variational density is an important move. It means that for every sensory state there is a corresponding active mode and an internal mode (or an autonomous mode in the joint space of active and internal states). The active $\mathbf{a}(\tau)$, internal $\boldsymbol \mu(\tau)$ and autonomous $\boldsymbol \alpha(\tau)$ modes evolve on \textit{active}, \textit{internal} and \textit{autonomous} \textit{manifolds}\footnote{A manifold is a topological (state-) space where each state has a neighbourhood that is homeomorphic to a portion of an Euclidean space of the same dimension~\cite{leeIntroductionTopologicalManifolds2011}. Intuitively, it is a curved space, such as a smooth surface, in a possibly large but finite number of dimensions. In this instance, the states are conditional modes.}, respectively, whose dimensionality is the same as the sensory states\footnote{The dimensionality of the active, internal and autonomous manifolds corresponds to the number of sensory states. This means that both the number of active and internal states must be greater than the number of sensory states. In turn, this limits the straightforward application of the free energy principle to particular partitions where the number of active states—and the number of internal states—exceeds the number of sensory states. In other words, the FEP applies to large particles with a nontrivial internal dynamics.}. We will see later that these manifolds play the role of \textit{centre manifolds}; namely, manifolds on which dynamics do not diverge (or converge) exponentially fast~\cite{carrApplicationsCentreManifold1982}.

Crucially, the internal manifold is also a \textit{statistical manifold} because its states are sufficient statistics for the variational density. In turn, this means that it is equipped with a metric and implicit information geometry~\cite{parrMarkovBlanketsInformation2020,ayInformationGeometry2017,amariInformationGeometryIts2016}. 
Indeed, the Fisher information metric tensor, which measures changes in the Kullback-Leibler (KL) divergence resulting from infinitesimal changes in the internal mode, is a Riemannian metric that yields an information distance~\cite[Appendix B]{dacostaNeuralDynamicsActive2021}. This means we can interpret dynamics on the internal manifold as updating Bayesian beliefs \textit{about} external states. This interpretation can be unpacked in terms of Bayesian inference as follows.

Equation \eqref{eq: 16} means that for every sensory state there is a conditional density over external states and a corresponding internal mode with the smallest surprisal. This mode specifies the variational density, where—by definition—the KL divergence between the variational density and the conditional density over external states is zero\footnote{Since the variational and conditional densities over external states are equal, any divergence between them will vanish, see \cite[Section 3.2]{amariMethodsInformationGeometry2007}.}. This means we can express the autonomous flow as a gradient flow on a free energy functional of the variational density\footnote{A functional is a function of a function, here, the free energy is a function of a conditional density parameterised by the internal mode.}. From \eqref{eq: 12}
 		\begin{equation}
 		\label{eq: 17}
\left[\begin{array}{l}
f_{\eta}(x) \\
f_{s}(x) \\
f_{a}(\pi) \\
f_{\mu}(\pi)
\end{array}\right]=\Omega\left[\begin{array}{c}
\nabla_{\eta} \Im(x) \\
\nabla_{s} \Im(x) \\
\nabla_{a} F(\pi) \\
\nabla_{\mu} F(\pi)
\end{array}\right]-\Lambda,
\end{equation}							
where the free energy in question is (an upper bound on) the surprisal of particular states:
 \begin{equation}
 \label{eq: 18}
\begin{aligned}
F(\pi(\tau)) &=\mathbb{E}_{q}[\ln q(\eta(\tau))-\ln p(\eta(\tau))-\ln p(\pi(\tau) \mid \eta(\tau))]=\Im(\pi(\tau)) \\
&=\underbrace{\mathbb{E}_{q}[\Im(\eta(\tau), \pi(\tau))]}_{\text {Expected energy }}-\underbrace{\H[q(\eta(\tau))]}_{\text {Entropy }} \\
&=\underbrace{\mathbb{E}_{q}[\Im(\pi(\tau) \mid \eta(\tau))]}_{ \text {-ve Accuracy }}+\underbrace{D[q(\eta(\tau)) \| p(\eta(\tau))]}_{\text {Complexity }} \\
&=\underbrace{D[q(\eta(\tau)) \| p(\eta(\tau) \mid \pi(\tau))]}_{=0}+\Im(\pi(\tau)) \\
q &=q_{\boldsymbol \mu}(\eta)=p(\eta \mid s)=p(\eta \mid \pi) \\
\mathbb{E}[F(\pi)] &=\mathbb{E}[\Im(\pi)]=\H[p(\pi)]
\end{aligned}
\end{equation}				
This variational free energy\footnote{\textbf{Question:} why is this functional called \textit{variational} free energy? More generally (for instance in engineering applications where the free energy in question is also called an evidence lower bound \cite{bishopPatternRecognitionMachine2006}) the free energy is a functional of an approximate posterior density $q$ that is an approximation to the Bayesian posterior, as follows:
\begin{equation}
\label{eq: free energy upper bound}
\begin{split}
q_{\boldsymbol \mu}(\eta)\approx p(\eta \mid \pi) \Rightarrow 
F[q] &=\underbrace{D[q(\eta(\tau)) \| p(\eta(\tau) \mid \pi(\tau))]}_{\geq 0}+\Im(\pi(\tau)) \\
\end{split}
\end{equation}
The variational density considered in this article is the minimiser of \eqref{eq: free energy upper bound}, and the free energy evaluated at the variational density is the variational free energy. The term 'variational' inherits from the use of the calculus of variations in variational Bayes (a.k.a., approximate Bayesian inference), applied in the context of a mean field approximation or factorised form of the variational density. The term 'free energy' inherits from Richard Feynman's path integral formulation, in the setting of quantum electrodynamics.} can be rearranged in several ways. First, it can be expressed as expected \textit{energy} minus the \textit{entropy} of the variational density, which licences the name \textit{free energy}\footnote{\textbf{Question}: is variational free energy the same kind of free energy found in statistical mechanics? The answer is no: the entropy term in the variational free energy is the entropy of a variational density—over external states—parameterised by internal states. This entropy is distinct from the entropy of \textit{internal states}. Minimising variational free energy \textit{increases} the entropy of the variational density and, usually, \textit{reduces} the entropy of internal states (see \cite{ueltzhofferVariationalFreeEnergy2021} for an example). 
Mathematically, we can express the different kind of entropies as $\H[q(\eta(\tau))]\neq \H[p(\mu(\tau))]$.}. In this decomposition, minimising variational free energy corresponds to the maximum entropy principle, under the constraint that the expected energy is minimised~\cite{jaynesInformationTheoryStatistical1957,lasotaChaosFractalsNoise1994}. The expected energy is a functional of the NESS density that plays the role of a \textit{generative model}; namely, a joint distribution over causes (external states) and their consequences (particular states)\footnote{\textbf{Question}: in practical applications, variational free energy is usually a function of data or observed (sensory) states. So, why is variational free energy a function of particular states? Later, we will see that practical applications correspond to Bayesian filtering, under the assumption that particular dynamics are very precise. This means that there is no uncertainty about autonomous paths given sensory paths, and the action of a particular path is the action of a sensory path. In generalised coordinates of motion—used in Bayesian filtering—the action of a path becomes the surprisal of a state. In this setting, the variational free energy of particular states is the same as the variational free energy of sensory states.}.

Second, variational free energy can be decomposed into the (negative) log likelihood of particular states (i.e., negative \textit{accuracy}) and the KL divergence between posterior and prior densities (i.e., \textit{complexity}). Finally, it can be written as the self-information associated with particular states (i.e., \textit{surprisal}) plus the KL divergence between the variational and conditional (i.e., posterior) density, which—by construction—is zero. In variational Bayesian inference~\cite{bealVariationalAlgorithmsApproximate2003}, negative surprisal is read as a log marginal likelihood or model evidence, having marginalised over external states. In this setting, negative free energy is an \textit{evidence lower bound} or ELBO~\cite{winnVariationalMessagePassing2005,bishopPatternRecognitionMachine2006}.

So, in what sense can we interpret \eqref{eq: 17} in terms of inference? Let us start by considering the response of autonomous states to some sensory perturbation: that is, the path of autonomous states conditioned upon sensory states. If sensory states change slowly, then the autonomous states will flow towards their most likely value (i.e., their conditional mode) and stay there\footnote{Or, at least in the vicinity, if there are random fluctuations on its motion.}. However, if sensory states are changing, the autonomous states will look as if they are trying to hit a moving target. One can formulate this along the lines of the centre manifold theorem~\cite{carrApplicationsCentreManifold1982,langVoiceRecognitionAphasic2009}, where we have a (fast) flow \textit{off} the centre manifold and a (slow) flow of the autonomous mode \textit{on} the manifold.
 	\begin{equation}
 	\label{eq: 19}
 \begin{aligned}
\alpha(\tau)&= \underbrace{\varepsilon(\tau)}_{\text{Off manifold}} + \underbrace{\boldsymbol{\alpha}(\tau)}_{\text{On manifold}}\\ \varepsilon(\tau) &\triangleq\alpha(\tau)-\boldsymbol{\alpha}(\tau)
\end{aligned}
\end{equation}	

In effect, this is a decomposition in a frame of reference that moves with the autonomous mode, whose path lies on the centre manifold. 
We further describe the off manifold flow using a Taylor expansion around the (time-varying) autonomous mode\footnote{Note that we are performing a Taylor expansion of a (generally rough) stochastic process $\varepsilon$, see \cite[Chapter 5]{kloedenNumericalSolutionStochastic1992}. Alternatively, it may be possible to instead consider motion in generalised coordinates to introduce smooth random fluctuations (see next Section), so that $\varepsilon$ becomes smooth and the usual Taylor expansion applies.} 

\begin{equation}
\label{eq: 19 bis}
    \begin{aligned}
\dot{\varepsilon}(\tau) &=\dot{\alpha}(\tau)-\dot{\boldsymbol{\alpha}}(\tau)=f_{\varepsilon}(0)+\frac{\partial f_{\varepsilon}}{\partial \varepsilon} \cdot \varepsilon+\ldots=\frac{\partial f_{\alpha}}{\partial \alpha} \cdot \varepsilon+\ldots \\
& \Rightarrow \\
\underbrace{\dot{\alpha}(\tau) -\dot{\boldsymbol{\alpha}}(\tau)}_{\text{Off manifold flow}} &=\mathbf{J}_{\alpha} \cdot(\alpha-\boldsymbol{\alpha})+\ldots \\
&=-\underbrace{\left(\Gamma_{\alpha} \nabla_{\alpha \alpha} F\right) \cdot(\alpha-\boldsymbol{\alpha})}_{\text {Flow to centre manifold }}+\underbrace{\left(Q_{\alpha \alpha} \nabla_{\alpha \alpha} F\right) \cdot(\alpha-\boldsymbol{\alpha})}_{\text {Flow parallel to the manifold}}+\ldots
    \end{aligned}
\end{equation}

This means that the flow at the expansion point is zero, leaving the second term of the expansion as the first non-vanishing term. This is the Jacobian of the autonomous flow times the displacement of the current autonomous state from its corresponding mode. The second-order derivatives of the free energy arise from the Jacobian of the flow, i.e., substituting \eqref{eq: 17} into \eqref{eq: 8}. Therefore, the off manifold flow has a component that flows \textit{towards} the centre manifold,\footnote{We know that the flow must be towards the centre manifold because the covariance of random fluctuations is positive definite, and the curvature of the free energy is positive definite at its minima: i.e., around the expansion point.} afforded by the gradient flow, and a component that is \textit{parallel} to the manifold, afforded by the solenoidal flow, cf. \eqref{eq: 6}. Taken together, this means that the autonomous states flow in ever-decreasing circles towards the centre manifold, as illustrated in Figure \ref{fig: 2}.

But what about the flow \textit{on} the centre manifold? We know from \eqref{eq: 17} that the flow of the autonomous mode can be expressed in terms of free energy gradients:
\begin{equation}
\label{eq: 20}
\begin{aligned}
\dot{\boldsymbol{\alpha}}(\tau) &=\left(Q_{\alpha \alpha}-\Gamma_{\alpha}\right) \nabla_{\alpha} F(s, \boldsymbol{\alpha})+\ldots \\
&=\left(Q_{\alpha \alpha}-\Gamma_{\alpha}\right) \underbrace{\nabla_{\alpha} \mathbb{E}_{q}[\Im(s, \boldsymbol{\alpha} \mid \eta)]}_{\text {-ve Accuracy }}+\left(Q_{\alpha \alpha}-\Gamma_{\alpha}\right) \underbrace{\nabla_{\alpha} D\left[q_{\boldsymbol \mu}(\eta) \| p(\eta)\right]}_{\text {Complexity }}+\ldots
\end{aligned}
\end{equation}
 			
This expression unpacks the centre manifold flow in terms of the accuracy and complexity parts of free energy, where the accuracy part depends upon the sensory states, while the complexity part is a function of, and only of, autonomous states. In short, the flow on the centre manifold will look as if it is trying to maximise the accuracy of its predictions, while complying with prior (Bayesian) beliefs.\footnote{The covariance of random fluctuations $\Gamma_\alpha$ is positive definite and the solenoidal matrix field $Q_{\alpha\alpha}$ is skew-symmetric, therefore the flow in \eqref{eq: 20} will seek to minimise complexity minus accuracy.} Here, predictions are read as the expected sensory states, under posterior (Bayesian) beliefs about their causes afforded by the variational density over external states. 

 \begin{figure}
    \centering
    \includegraphics[width=\textwidth]{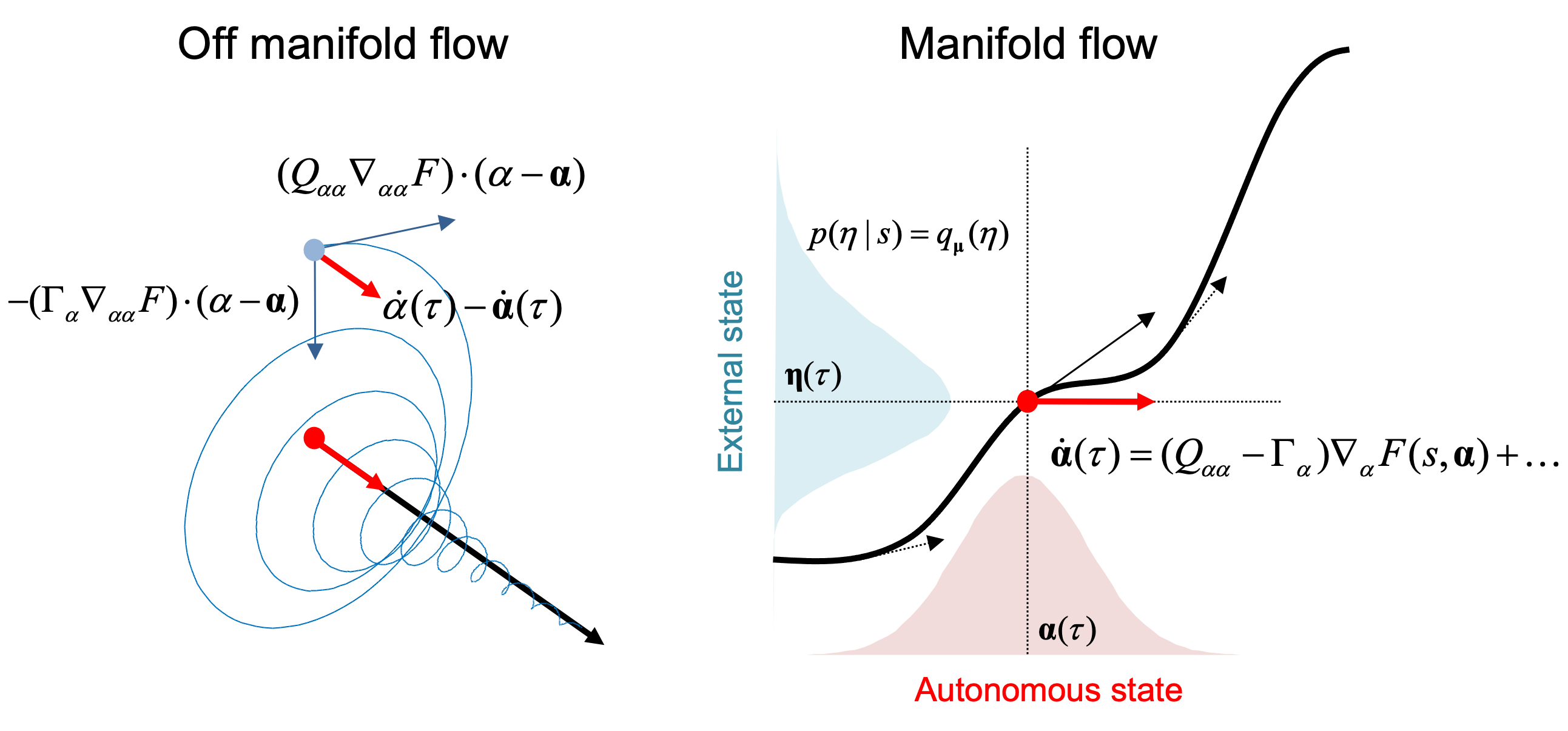}
    \caption[Autonomous flows and Bayesian filters]{\textbf{Autonomous flows and Bayesian filters}. This figure shows two components of the autonomous flow; namely, a (fast) flow $\dot \alpha(\tau)-\dot {\boldsymbol \alpha}(\tau)$ \textit{off} the (centre) manifold, and a (slow) flow $\dot {\boldsymbol \alpha}(\tau)$ \textit{on} the manifold. The manifold here is the set of autonomous modes $\boldsymbol \alpha(\tau)$ given sensory states $s(\tau)$ for all time $\tau$, see \eqref{eq: 16}. The decomposition into fast and slow flows means that the manifold can be thought of as a centre manifold. The left panel shows two components of the fast flow off the manifold; namely, a flow towards the centre manifold and a flow parallel to the manifold, see \eqref{eq: 19 bis}. This decomposition rests upon a first-order Taylor expansion of the off manifold flow. The right panel plots the external mode as a function of the autonomous mode---what is known as a \textit{synchronisation} manifold---as a black curvilinear line. The Gaussian (blue and red) distributions show possible variations in (external and autonomous) conditional modes due to variations in the sensory states. The arrows represent the centre manifold flow $\dot {\boldsymbol \alpha}(\tau)$ in the context of this synchronisation manifold, where the tangent vectors represent possible directions of the flow.
    }
    \label{fig: 2}
\end{figure}

\subsection*{Summary}
In summary, a particular partition of a nonequilibrium steady-state density implies that autonomous dynamics can be interpreted as performing a particular kind of inference. This entails a fast flow towards an autonomous centre manifold and a slow flow on the centre manifold. The centre manifold flow can be interpreted as Bayesian belief updating, where posterior (Bayesian) beliefs are encoded by points on an internal (statistical) manifold. In other words, for every point on the statistical manifold, there is a corresponding variational density or Bayesian belief over external states. We are now in a position to express this belief updating as a variational principle of least action:

\begin{equation}
\label{eq: 21}
\begin{aligned}
\boldsymbol{\alpha}[\tau] &=\arg \min _{\alpha[\tau]} \mathcal{A}(\alpha[\tau] \mid s[\tau]) \\
& \Leftrightarrow \delta_{\alpha} \mathcal{A}(\boldsymbol{\alpha}[\tau] \mid s[\tau])=0 \\
& \Leftrightarrow \\
\dot{\mathbf{a}}(\tau) &=f_{a}(s, \boldsymbol{\alpha})=\left(Q_{a a}-\Gamma_{a}\right) \nabla_{a} F(s, \boldsymbol{\alpha})+\ldots \\
\dot{\boldsymbol{\mu}}(\tau) &=f_{\mu}(s, \boldsymbol{\alpha})=\left(Q_{\mu \mu}-\Gamma_{\mu}\right) \nabla_{\mu} F(s, \boldsymbol{\alpha})+\ldots
\end{aligned}
\end{equation}
 							
This is a basis of the free energy principle. Put simply, it means that the internal states of a particular partition can be cast as encoding conditional or posterior Bayesian beliefs \textit{about} external states. Equivalently, the autonomous path of least action can be expressed as a gradient flow on a variational free energy that can be read as log evidence. This licences a somewhat poetic description of self-organisation as self-evidencing~\cite{hohwySelfEvidencingBrain2016}, in the sense that the surprisal or self-information is known as log model evidence or marginal likelihood in Bayesian statistics\footnote{\textbf{Question}: this Bayesian mechanics seems apt for inference but what about learning over time? We have been dealing with states in a generic sense. However, one can have states that change over different timescales. One can read slowly changing states as special states that play the role of parameters; either parameters of the flow or, implicitly, the generative model. In mathematical and numerical analyses, states and parameters are usually treated identically; i.e., as minimising variational free energy. Indeed, in practical applications of Bayesian filtering schemes that learn, the parameters are treated as slowly changing states. See \cite{schiffKalmanFilterControl2008,fristonGeneralisedFiltering2010} for worked examples.}.

Interestingly, because of the symmetric setup of the Markov blanket, it would be possible to repeat everything above but switch the labels of internal and external states—and active and sensory states—and tell the same story about external states tracking internal states. This evinces a form of generalised synchrony~\cite{huntDifferentiableGeneralizedSynchronization1997,jafriGeneralizedSynchronyCoupled2016,buendiaBroadEdgeSynchronization2022,dacostaBayesianMechanicsStationary2021a}, where internal and external states track each other. Technically, if we consider the (internal and external) manifolds in the joint space of internal and external states, we have something called a synchronisation manifold that offers another perspective on the coupling between the inside and outside~\cite{fristonActiveInferenceCommunication2015,parrMarkovBlanketsInformation2020,dacostaBayesianMechanicsStationary2021a}. 

These teleological interpretations cast particular paths of least action as an optimisation process, where different readings of free energy link nicely to various normative (i.e., optimisation) theories of sentient behaviour. Some cardinal examples are summarised in Figure \ref{fig: 3}; see \cite{fristonFreeenergyPrincipleUnified2010,buckleyFreeEnergyPrinciple2017,fristonFreeEnergyPrinciple2019a,fristonReinforcementLearningActive2009} for some formal accounts of these relationships. Because internal states do not influence sensory (or external) states, they will look as if they are concerned purely with inference, in the sense that they parameterise the variational density over external states. However, active states influence sensory (and external) states and will look as if they play an active role in configuring (and causing) the sensory states that underwrite inference. In the neurosciences, this is known as \textit{active inference}~\cite{fristonActionBehaviorFreeenergy2010,ueltzhofferDeepActiveInference2018,koudahlWorkedExampleFokkerPlanckBased2020}.

The link between optimisation and inference is simply that inference is belief optimisation. However, it is worth unpacking the gradients that ‘drive’ this optimisation. In statistics, variational free energy is used to score the divergence between a variational density and the conditional density over external (i.e., hidden) states, given blanket states~\cite{winnVariationalMessagePassing2005}. Unlike the definition in \eqref{eq: 18}, these densities are not assumed to be equivalent. Variational inference proceeds by optimising the variational density such that it minimises free energy—often using the gradient flows in \eqref{eq: 17}. However, there is a subtle difference between the dynamics of \eqref{eq: 17} and variational inference. In the former, there is no contribution from the KL-divergence as it is stipulated to be zero. In the latter, it is only the divergence term that contributes to free energy gradients. So, is it tenable to interpret gradient flows on variational free energy as variational inference, or is this just teleological window-dressing? The next section addresses this question through the lens of Bayesian filtering. In brief, we will see that the autonomous paths of least action—implied by a particular partition—are the paths of least action of a Bayesian filter. This takes us beyond ‘as if’ arguments by establishing a formal connection between particular dynamics and variational inference.

 \begin{figure}
    \centering
    \includegraphics[width=0.8\textwidth]{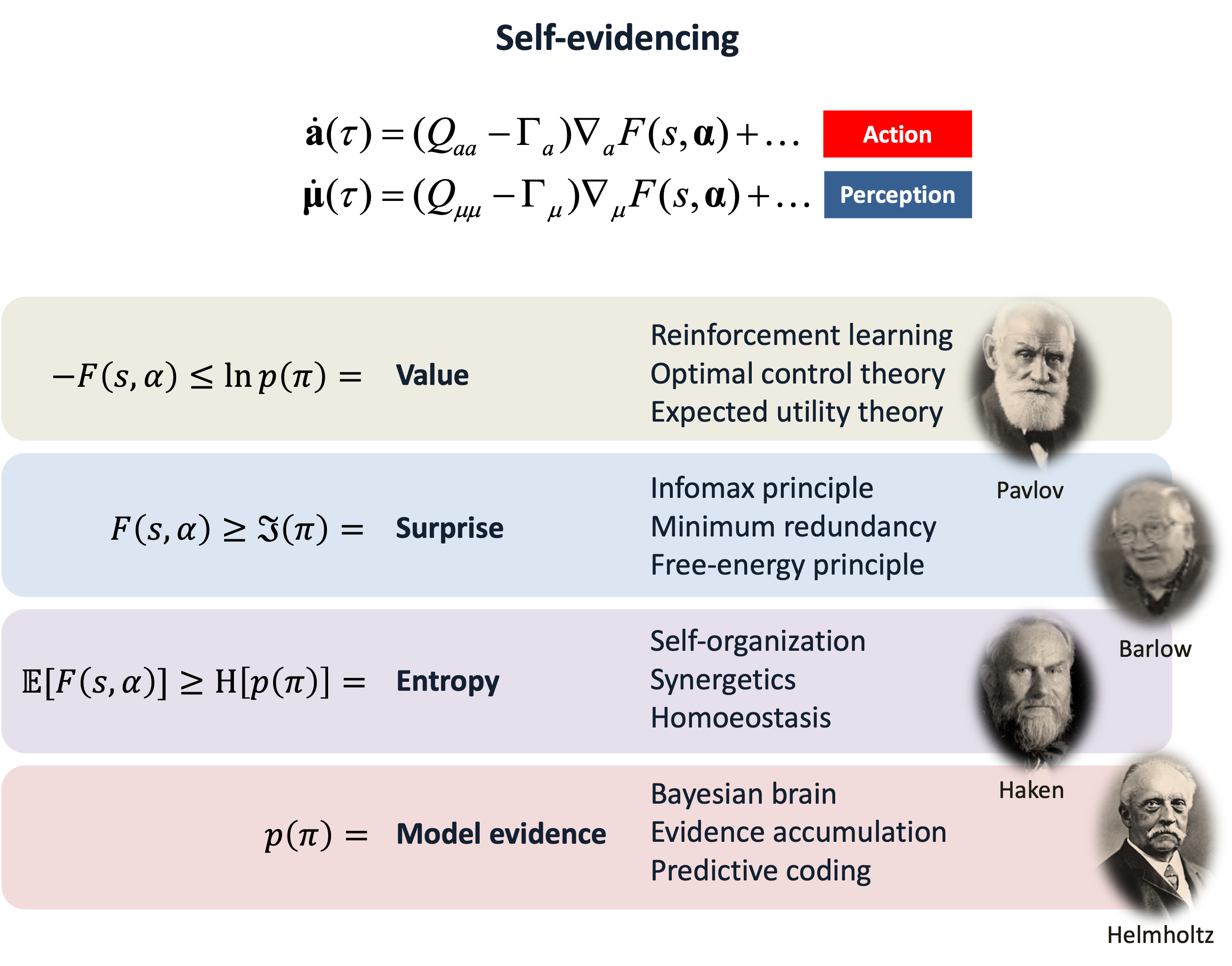}
    \caption[Markov blankets and self-evidencing]{\textbf{Markov blankets and self-evidencing}. This schematic illustrates various points of contact between minimising variational free energy and other normative theories of optimal behaviour. The existence of a Markov blanket entails a certain lack of influences among internal, blanket and external states. These have an important consequence—internal and active states are not influenced by external states, which means their dynamics (i.e., perception and action) are a function of, and only of, particular states, given by a variational (free energy) bound on surprisal. This has a number of interesting interpretations. Given surprisal is the negative log probability of finding a particle or creature in a particular state, minimising surprise corresponds to maximising the value of that state. This interpretation is licensed by the fact that the states with a high probability are, by definition, characteristic of the particle in question. On this view, one could relate this to dynamics in reinforcement learning~\cite{bartoReinforcementLearningIntroduction1992}, optimal control theory~\cite{todorovOptimalFeedbackControl2002} and, in economics, expected utility theory~\cite{bossaertsBehaviouralEconomicsNeuroeconomics2015a,vonneumannTheoryGamesEconomic1944}. Gradient flows that minimise surprisal (i.e., self-information) lead to a series of influential accounts of neuronal dynamics; including the principle of maximum mutual information~\cite{opticanTemporalEncodingTwodimensional1987a,linskerPerceptualNeuralOrganization1990}, the principles of minimum redundancy and maximum efficiency~\cite{barlowPossiblePrinciplesUnderlying1961} and the free energy principle~\cite{fristonFreeEnergyPrinciple2006}. Crucially, the average or expected surprise (over time of particular states) corresponds to entropy. This means that action and perception look as if they are bounding the entropy of particular states. This links nicely with theories of self-organisation, such as synergetics in physics~\cite{nicolisSelforganizationNonequilibriumSystems1977,hakenSynergeticsIntroductionNonequilibrium1978,kauffmanOriginsOrderSelforganization1993} or homoeostasis in physiology~\cite{ashbyPrinciplesSelfOrganizingDynamic1947,conantEveryGoodRegulator1970,bernardLecturesPhenomenaLife1974}. Finally, the probability of a particular state, is, on a statistical view, model evidence or marginal likelihood~\cite{mackayFreeEnergyMinimisation1995,mackayInformationTheoryInference2003}, marginalising over the causes of particular states (i.e., external states). This means that all the above formulations are internally consistent with things like the Bayesian brain hypothesis, evidence accumulation and predictive coding\cite{fristonFreeenergyPrincipleUnified2010,bogaczTutorialFreeenergyFramework2017,buckleyFreeEnergyPrinciple2017}. Most of these formulations inherit from Helmholtz's motion of unconscious inference~\cite{helmholtzHelmholtzTreatisePhysiological1962}, later unpacked in terms of perception as hypothesis testing in psychology~\cite{gregoryPerceptionsHypotheses1980} and machine learning~\cite{dayanHelmholtzMachine1995}. Although not depicted here, the minimisation of complexity—inherent in the minimisation of free energy—enables thermodynamic and metabolic efficiency via Landauer's principle~\cite{landauerIrreversibilityHeatGeneration1961}.}
    \label{fig: 3}
\end{figure}

\section{Lagrangians, generalised states and Bayesian filtering}

Now, say we wanted to emulate or simulate active inference. Given some equations of motion and statistics of random fluctuations, we could find the stationary solution to the Fokker Planck equation and accompanying Helmholtz decomposition. We could then solve \eqref{eq: 21} for the autonomous paths of least action that characterise the expected behaviour of this kind of particle, and obtain realisations of synchronisation and inference. See \cite{fristonStochasticChaosMarkov2021} for a worked example using a system of coupled Lorentz attractors.


In this section, we take a somewhat pragmatic excursion to suggest a simpler way to recover the paths of least action; namely, as the solution to a generic (Bayesian) filtering scheme that is widely used in the engineering literature.

\subsection{Dynamics in generalised coordinates of motion}

Let us go back to the Langevin equation governing our system 

\begin{equation}
\label{eq: langevin fluctuating covariance}
    \dot x(\tau) = f(x) + \omega(\tau).
\end{equation}
In this section, we assume that the random fluctuations driving the motion have smooth (analytic) sample paths; thus, the Langevin equation considered in the rest of the article can be seen as the limit of \eqref{eq: langevin fluctuating covariance} as the fluctuations become rough \cite{wongRelationOrdinaryStochastic1965}. This setup speaks nicely to the fact that, in biology, fluctuations are often smooth up to a certain order---contrariwise to thermal (white noise) fluctuations---as they are the output of other random dynamical systems. 
As before, we assume that the fluctuations are state-independent, and a stationary Gaussian process, e.g., the smoothing of white noise fluctuations with a Gaussian kernel. Just like in the case of white noise, Gaussianity can be motivated by the central limit theorem---fluctuations should be normally distributed at each point in time.

We denote the autocovariance of fluctuations by 
$\Gamma_{h}= \frac 1 2\mathbb E[\omega (\tau) \otimes \omega (\tau+h)].$
The underlying dynamical systems giving rise to this generic type of smooth noise can be recovered through a procedure known as stochastic realisation \cite{dacostaBayesianMechanicsStationary2021a,lindquistRealizationTheoryMultivariate1985,mitterTheoryNonlinearStochastic1981}. The solution to the Langevin equation 
\eqref{eq: langevin fluctuating covariance} can be approximated, on a suitably small interval of time, by a linear Langevin equation in generalised coordinates of motion $\vec{x}=\left(x, x^{\prime}, x^{\prime \prime}, \ldots\right)$ \cite[Section 4]{balajiBayesianStateEstimation2011}:\footnote{\label{footnote expansion in gen coords}The expansion \eqref{eq: 22} is a linear approximation of \eqref{eq: langevin fluctuating covariance} \cite{biscayLocalLinearizationMethod1996}, obtained by recursively differentiating \eqref{eq: langevin fluctuating covariance} and ignoring the contribution of the derivatives of the flow of order higher than one. In other words, the expansion is exact when the flow is linear, and it is accurate on a short time-scale when the flow is non-linear.}\footnote{The curvature (i.e., second derivative) of the autocovariance $\Gamma^{\prime\prime}_{0}$ is a ubiquitous measure of roughness of a stochastic process \cite{coxTheoryStochasticProcesses1977}. Note that in the limit where the fluctuations $\omega$ are uncorrelated (e.g., white noise fluctuations), $\Gamma^{\prime\prime}_{0}$ (and higher derivatives) become infinitely large.}

\begin{equation}
 \label{eq: 22}
 \begin{split}
     \left.\begin{array}{c}
\dot{x}=x^{\prime}=\nabla f \cdot x+ \omega \\
\dot{x}^{\prime}=x^{\prime \prime}=\nabla f \cdot x^{\prime}+\omega^{\prime} \\
\dot{x}^{\prime \prime}=x^{\prime \prime \prime}=\nabla f \cdot x^{\prime \prime}+\omega^{\prime \prime} \\
\vdots
\end{array}\right\} &\Leftrightarrow \begin{array}{r}
\dot{\vec{x}}=\mathbf f (\vec{x})+\vec{\omega} \\
\mathbf{D} \vec{x}
=\mathbf{J} \vec{x}+\vec{\omega} \\
p(\vec{\omega}(\tau))=\mathcal{N}(\vec{\omega}(\tau); 0,2 \boldsymbol{\Gamma})
\end{array}\\
\mathbf{D}=\left[\begin{array}{llll}
0&1 & & \\
&0& 1 & \\
& & 0& \ddots \\
& & & \ddots
\end{array}\right], \quad \mathbf{J}=&\left[\begin{array}{llll}
\nabla f & & & \\
& \nabla f & & \\
& & \nabla f & \\
& & & \ddots
\end{array}\right], \quad \boldsymbol \Gamma=\left[\begin{array}{cccc}
\Gamma_{0} & & \Gamma^{\prime\prime}_{0} & \\
& -\Gamma^{\prime\prime}_{0} & & -\Gamma^{\prime \prime \prime\prime}_{0} \\
\Gamma^{\prime\prime}_{0} & & \Gamma^{\prime \prime \prime\prime}_{0} & \\
& -\Gamma^{\prime \prime \prime\prime}_{0} & & \ddots
\end{array}\right]
 \end{split}
\end{equation}


Here, the different variables $\vec x=\left(x, x^{\prime}, x^{\prime \prime}, \ldots, x^{(n)}, \ldots \right)$ can be seen as the position, velocity, acceleration, jerk, and higher orders of motion of the process, which are treated as separate (generalised) states that are coupled through the Jacobian $\mathbf{J}$. These are driven by smooth fluctuations $\vec{\omega}$ (i.e., the serial derivatives of $\omega$) whose covariance $2\boldsymbol \Gamma$ can be expressed in terms of the serial derivatives of the autocovariance \cite[Appendix A.5.3]{parrActiveInferenceFree2022}.

The generalised states are the coefficients of a Taylor series expansion of the solution to the Langevin equation \eqref{eq: langevin fluctuating covariance}:

\begin{equation}
\label{eq: analytic solution}
   x(\tau) = x(0) + x^{\prime}(0)\tau + \frac {x^{\prime \prime}(0)}2 \tau^2 + \ldots + \frac{x^{(n)}(0)}{n!}\tau^n +\ldots,
\end{equation}
where \eqref{eq: analytic solution} holds, typically, only on a small time-interval to which we restrict ourselves henceforth. In other words, the generalised states at any time-point determine the system's trajectory, and vice versa; that is, there is an isomorphism between generalised states and paths.




 

This line of reasoning has two advantages. First, it means one can let go of white noise assumptions on the random fluctuations and deal with smooth or analytic fluctuations. Second, the linear expansion in generalised coordinates of motion \eqref{eq: 22} means that the distribution of generalised states has a simple Gaussian form
 \begin{equation}
 \label{eq: 23}
\begin{aligned}
\mathcal{L}(\vec{x}(\tau)) &\triangleq -\ln p (\vec{x}(\tau))=\frac{1}{2} \vec{x}(\tau) \cdot \mathbf{M} \vec{x}(\tau) \\
\mathbf{M} &=(\mathbf{D}-\mathbf{J}) \cdot \frac{1}{2 \boldsymbol \Gamma}(\mathbf{D}-\mathbf{J}).
\end{aligned}
\end{equation}
Here, $\mathbf{M}$ can be read as a mass matrix. This suggests that precise particles, with low amplitude random fluctuations, behave like massive bodies. Furthermore, \eqref{eq: 23} is seen as the Lagrangian in generalised coordinates of motion, due to its formal similarity with \eqref{eq: action}.
Under the isomorphism between points and paths in generalised coordinates, the Lagrangian is equivalent to the action; it scores the likelihood of paths of \eqref{eq: langevin fluctuating covariance}, as a path corresponds to a point in generalised coordinates of motion\footnote{\textbf{Question}: how can a point be a path? The generalised states (i.e., temporal derivatives) approximate the path of the solution to \eqref{eq: langevin fluctuating covariance} on a suitably small time interval because they are the coefficients of a Taylor expansion of the path as a function of time \eqref{eq: analytic solution}.}. We will reason about the trajectories of the system by analysing the Lagrangian of generalised states henceforth.

The path of least action corresponds to the minimiser of the Lagrangian, which can be expressed as follows:
 \begin{equation}
  \label{eq: 25}
\begin{aligned}
\overrightarrow{\mathbf{x}}(\tau) &=\arg \min _{\vec{x}(\tau)} \mathcal{L}(\vec{x}(\tau)) \\
& \Leftrightarrow \nabla \mathcal{L}(\overrightarrow{\mathbf{x}}(\tau))=0  \quad \forall \tau .
\end{aligned}
\end{equation}									

We can recover the path of least action by solving the following equation of motion
 \begin{equation}
 \label{eq: 24}
\begin{aligned}
\dot{\vec{x}}(\tau) &=\mathbf{D} \vec{x}-\nabla \mathcal{L}(\vec{x}) \\
\nabla \cdot \mathbf{D} \vec{x} &=0.
\end{aligned}
\end{equation}	
Indeed, this motion can be interpreted as a gradient descent on the Lagrangian, in a frame of reference that moves with the mode of the distribution of generalised states~\cite{fristonGeneralisedFiltering2010}. Thus, the convexity of the Lagrangian means that any solution to \eqref{eq: 24} converges to the path of least action.
In this setting, the divergence-free flow (i.e., the first term) is known as a \textit{prediction} of the generalised state based upon generalised motion, while the curl-free, gradient flow (i.e., the second term) is called an \textit{update}.

\subsection{Particular partitions in generalised coordinates of motion}

We now reintroduce the distinction between internal, external, sensory and active states $x=(\eta, s, a , \mu)$. Briefly, as before, we assume that the Langevin equation \eqref{eq: langevin fluctuating covariance} is sparsely coupled as in \eqref{eq: 13}. This implies that the trajectories internal and external to the particle are conditionally independent given the trajectories of the blanket \eqref{eq: 15}. The same sparse coupling structure carries through the expansion in generalised coordinates \eqref{eq: 22} so that the motion of generalised states entails trajectories with the same conditional independencies. Since paths correspond to generalised states, this yields conditional independence between generalised states, as follows:
\begin{equation}
    (\vec \mu \perp \vec\eta) \mid \vec b \iff \mathcal{L}(\vec{x}) = \mathcal{L}(\vec{\eta} \mid \vec{b}) + \mathcal{L}(\vec{\mu} \mid \vec{b}) + \mathcal{L}(\vec{b}).
\end{equation}

We can now recover paths of least action of the particle by equating the Lagrangian with the variational free energy of generalised states. This allows us to express the internal path of least action as a gradient flow on variational free energy, which can itself be expressed in terms of generalised prediction errors. 
From \eqref{eq: 24}, we have
\begin{equation}
 \label{eq: 26}
\begin{aligned}
\dot{\vec{\mu}}(\tau) & =\mathbf{D} \vec{\mu}-\nabla_{\vec{\mu}} \mathcal{L}(\vec{x})=\mathbf{D} \vec{\mu}-\nabla_{\vec{\mu}} \mathcal{L}(\vec{\mu} \mid \vec{b}) \\
& =\mathbf{D} \vec{\mu}-\nabla_{\vec{\mu}} \mathcal{L}(\vec{\pi}) =\mathbf{D} \vec{\mu}-\nabla_{\vec{\mu}} F(\vec{\pi}),
\end{aligned}
\end{equation}
where the free energy of generalised states is analogous to \eqref{eq: 18}
\begin{equation}
\label{eq: free energy generalised}
    \begin{aligned}
        F(\vec{s}, \vec{a}, \vec{\mu}) & =\underbrace{\mathbb{E}_q[\mathcal{L}(\vec{\eta}, \vec{\pi})]}_{\text {Expected energy }}-\underbrace{H[q(\vec{\eta})]}_{\text {Entropy }} \\
& =\underbrace{\mathbb{E}_q[\mathcal{L}(\vec{\pi} \mid \vec{\eta})]}_{\text {Accuracy }}+\underbrace{D[q(\vec{\eta}) \| p(\vec{\eta})]}_{\text {Complexity }} \\
& =D[q(\vec{\eta}) \| p(\vec{\eta} \mid \vec{\pi})]+\mathcal{L}(\vec{\pi})\\
q_{\vec{\mu}}(\vec{\eta}) & =\mathcal{N}(\vec{\eta}; \overrightarrow{\boldsymbol{\mu}}, \Sigma(\overrightarrow{\boldsymbol{\mu}}))=p(\vec{\eta} \mid \vec{\pi})=p(\vec{\eta} \mid \vec{b}) \\
\mathcal{L}(\vec{\eta}, \vec{\pi}) & =\varepsilon_{\vec{\eta}} \cdot \frac{1}{4 \boldsymbol \Gamma_\eta} \varepsilon_{\vec{\eta}}+\varepsilon_{\vec{s}} \cdot \frac{1}{4 \boldsymbol \Gamma_s} \varepsilon_{\vec{s}}+\ldots \\
\varepsilon_{\vec{\eta}} & \triangleq\mathbf{D} \vec{\eta}-\mathbf f_{\vec{\eta}}(\vec{\eta}, \vec{s}) \\
\varepsilon_{\vec{s}} & \triangleq\mathbf{D} \vec{s}-\mathbf f_{\vec{s}}(\vec{\eta}, \vec{s}).
    \end{aligned}
\end{equation}
The variational free energy of generalised states is easy to evaluate, given a generative model in the form of a state-space model~\cite{fristonGeneralisedFiltering2010}; that is, the generalised flow of external and sensory states $\mathbf f_{\vec{\eta}},\mathbf f_{\vec{s}}$, and the covariance of their generalised fluctuations 
$\boldsymbol \Gamma_\eta, \boldsymbol \Gamma_s$. Note that the parameterisation of the variational density is very simple: the internal states parameterise the expected external states. Furthermore, the quadratic form of the Lagrangian means that the variational density over the generalised motion of external states is Gaussian\footnote{\textbf{Question}: why is the covariance of the variational density only a function of the internal mode? This follows from the quadratic Lagrangian that furnishes an analytic solution to the free energy minimum. Please see \cite{fristonVariationalFreeEnergy2007} for details.}. This licenses a ubiquitous assumption in variational Bayes called the Laplace assumption. Please see~\cite{fristonVariationalFreeEnergy2007} for a discussion of the simplifications afforded by the Laplace assumption. 

Crucially, in the absence of active states, the dynamic in \eqref{eq: 26} coincides with a generalised Bayesian filter. Generalised filtering is a generic Bayesian filtering scheme for nonlinear state-space models formulated in generalised coordinates of motion~\cite{fristonGeneralisedFiltering2010}; 
special cases include variational filtering \cite{fristonVariationalFiltering2008a}, dynamic expectation maximisation \cite{fristonVariationalTreatmentDynamic2008}, extended Kalman filtering \cite{loeligerLeastSquaresKalman2002}, and generalised predictive coding. 

Furthermore, if the autonomous paths are conditionally independent from external paths, given sensory paths\footnote{This is the case for precise particles, which are defined by particular fluctuations of infinitesimally small amplitude---see next Section and \cite{fristonPathIntegralsParticular2023}.}, the autonomous paths of least action can be recovered from a generalised gradient descent on variational free energy:
\begin{equation}
 \label{eq: gen grad descent autonomous states}
\begin{aligned}
\dot{\vec{\alpha}}(\tau) & =\mathbf{D} \vec{\alpha}-\nabla_{\vec{\alpha}} \mathcal{L}(\vec{x})=\mathbf{D} \vec{\alpha}-\nabla_{\vec{\alpha}} \mathcal{L}(\vec{\alpha} \mid \vec{s}) \\
& =\mathbf{D} \vec{\alpha}-\nabla_{\vec{\alpha}} \mathcal{L}(\vec{\pi}) =\mathbf{D} \vec{\alpha}-\nabla_{\vec{\alpha}} F(\vec{\pi}).
\end{aligned}
\end{equation}
In this case, the most likely paths of both internal and active states can be recovered by a gradient descent on variational free energy, and one can simulate active inference using generalisations of linear quadratic control or model predictive control \cite{kappenPathIntegralsSymmetry2005,todorovGeneralDualityOptimal2008a}:

\begin{equation}
\label{eq: 27}
\left[\begin{array}{c}
\dot{\vec{\eta}}(\tau) \\
\dot{\vec{s}}(\tau) \\
\dot{\vec{ a}}(\tau) \\
\dot{\vec{ \mu}}(\tau)
\end{array}\right]=\left[\begin{array}{c}
\mathbf f_{\vec{\eta}}(\vec{\eta}, \vec{s}, \vec{a})+\vec{\omega}_{\eta}(\tau) \\
\mathbf f_{\vec{s}}(\vec{\eta}, \vec{s}, \vec{a})+\vec{\omega}_{s}(\tau) \\
\mathbf{D} \vec{ a}-\nabla_{\vec{a}} F(\vec{s}, \vec{a}, \vec{\mu}) \\
\mathbf{D} \vec{\mu}-\nabla_{\vec{\mu}} F(\vec{s}, \vec{  a}, \vec{ \mu})
\end{array}\right]
\end{equation}
 								
This is effectively a (generalised) version of the particular dynamics in \eqref{eq: 21}.

\subsection*{Summary}

This section has taken a somewhat pragmatic excursion from the FEP narrative to consider generalised coordinates of motion. This excursion is important because it suggests that the gradient flows in systems with attracting sets are the paths of least action in Bayesian filters used to assimilate data in statistics~\cite{loeligerLeastSquaresKalman2002} and, indeed, control theory~\cite{vandenbroekRiskSensitivePath2010}. 

Working in generalised coordinates of motion is effectively working with paths and the path integral formulation. Practically, this is useful because one can use the density over paths directly to evaluate the requisite free energy gradients, as opposed to solving the Fokker-Planck equation to find the NESS density. Effectively, the generative model becomes a state-space model, specified with flows and the statistics of random fluctuations: see \eqref{eq: free energy generalised}. These are the sufficient statistics of the joint density over external and sensory paths.

Hitherto, we have largely ignored random fluctuations in the motion of particular states to focus on the underlying flows. Are these flows ever realised or does the principle of least action in \eqref{eq: 21} only apply to the most likely autonomous paths? In what follows, we will consider a special class of systems, where we suppress particular fluctuations to recover the behaviour of particles that show a precise or predictable response to external states. For this kind of particle, the particular paths are always the paths of least action.

\section{From statistical to classical particles}
\label{sec: precise}

So far, we have a Bayesian mechanics that would be apt to describe a particle or person with a pullback attractor. But what is the difference between a particle and a person? This question speaks to distinct classes of things to which the free energy principle could apply; e.g., molecular versus biological. Here, we associate biotic self-organisation with the precise and predictable dynamics of large particles. Thanks to the Helmholtz decomposition \eqref{eq: 6}, it is known that when random fluctuations are large, dissipative flow dominates conservative flow, and we have ensembles described by statistical mechanics (i.e., small particles). Conversely, when random fluctuations have a low amplitude, solenoidal flow\footnote{And its accompanying correction term $\Lambda$, see \eqref{eq: 6}.} dominates and we have classical mechanics and deterministic chaos (i.e., of heavenly and $n$-body problems). Here, we consider the distinction between statistical and classical mechanics in the setting of a particular partition.

It is often said that the free energy principle explains why biological systems resist the second law and a natural tendency to dissipation and disorder~\cite{fristonLifeWeKnow2013}. However, this is disingenuous on two counts. First, the second law only applies to closed systems, while the free energy principle describes open systems in which internal states are exposed to—and exchange with—external states through blanket states. Second, there is nothing, so far, to suggest that the entropy of particular states or paths is small. Everything we have done would apply equally to particles with high and low entropy densities. So, what distinguishes between high and low entropy systems (e.g., between candle flames and concierges), respectively? 

One answer can be found in the path-integral formulation: from \eqref{eq: 5}, we can associate the entropy of a path (i.e., history or trajectory of particular states) with the amplitude of random fluctuations. This licences the notion of \textit{precise particles} that are characterised by low or vanishing random fluctuations\footnote{\textbf{Question}: but surely my neurons are noisy? There is a substantial literature that refers to neuronal and synaptic noise: e.g., \cite{toutounjiSpatiotemporalComputationsExcitable2014}. However, the population dynamics of neuronal ensembles or assemblies are virtually noiseless by the central limit theorem (because they comprise thousands of neurons), when averaged over suitable spatial and temporal scales. For example, in electrophysiology, averaging several fluctuating single trial responses yields surprisingly stable and reproducible event-related potentials. From the perspective of the FEP, studying single neurons (or trials) is like studying single molecules to characterise fluid dynamics.}. In essence, precise particles are simply ‘things’ that are subject to the classical laws of nature; i.e., Lagrangian mechanics. In the case of vanishing fluctuations on particular states, every autonomous trajectory is a path of least action. From \eqref{eq: 5} and \eqref{eq: 21} this can be expressed as follows:
 \begin{equation}
 \label{eq: 28}
\begin{aligned}
\Gamma_{\pi} &\equiv 0\\
&\Rightarrow\\
\dot{\alpha}(\tau) &=f_{\alpha}(\pi(\tau)) \Leftrightarrow \delta_{\alpha} \mathcal{A}(\alpha[\tau] \mid s[\tau])=0 \Leftrightarrow \alpha[\tau]=\boldsymbol\alpha[\tau] \\
& \Rightarrow \\
\dot{a}(\tau) &=\dot{\mathbf{a}}(\tau)=\left(Q_{a a}-\Gamma_{a}\right) \nabla_{a} F(\pi)+\ldots \\
\dot{\mu}(\tau) &=\dot{\boldsymbol{\mu}}(\tau)=\left(Q_{\mu \mu}-\Gamma_{\mu}\right) \nabla_{\mu} F(\pi)+\ldots
\end{aligned}
\end{equation}	

 \begin{figure}[t!]
    \centering
    \includegraphics[width=0.9\textwidth]{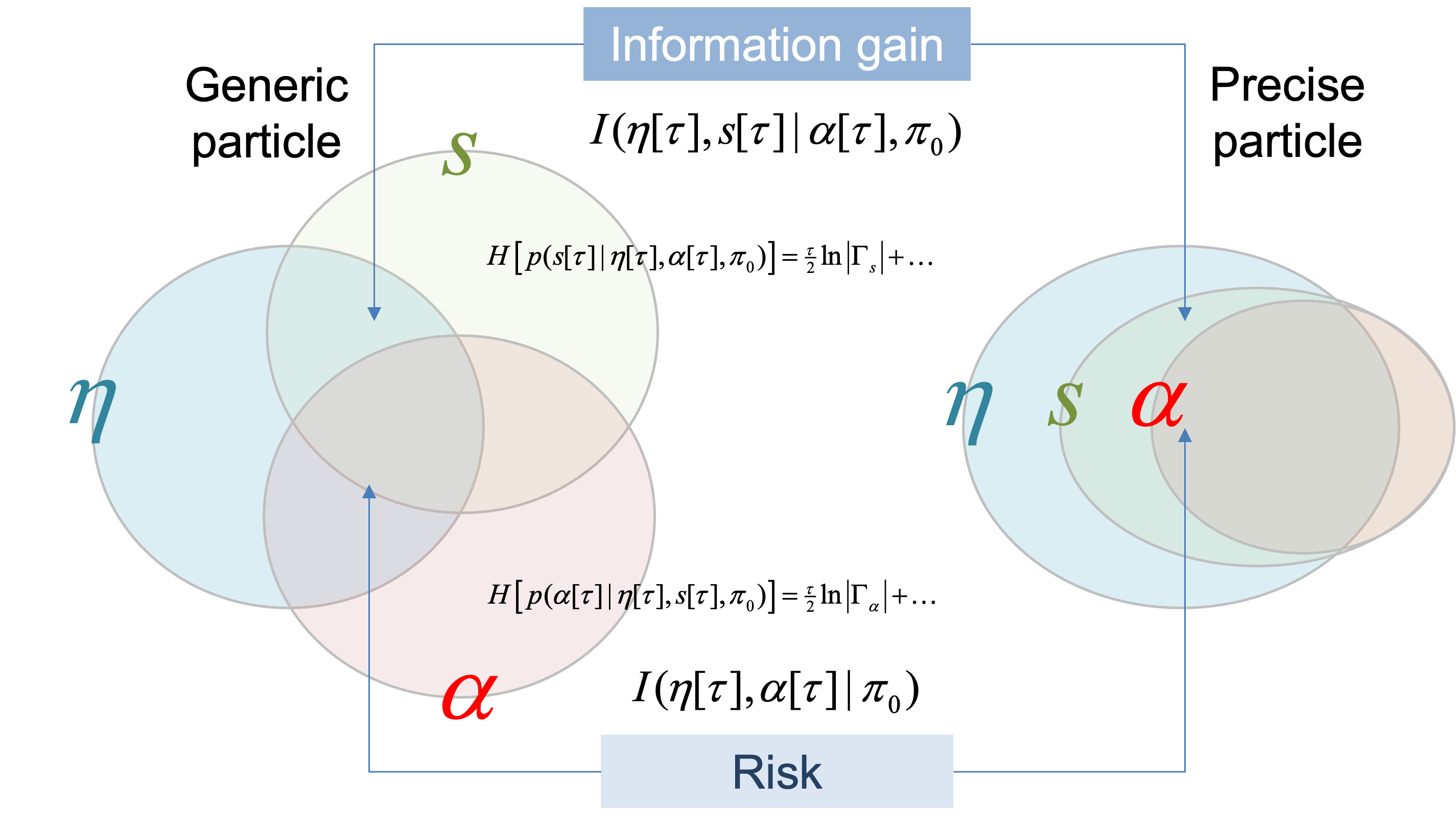}
    \caption[Generic and precise particles]{\textbf{Generic and precise particles}. These information diagrams depict the entropy of external, sensory and autonomous paths, where intersections correspond to shared or mutual information. A conditional entropy corresponds to an area that is outside the variable upon which the entropy is conditioned. The diagram on the left shows the generic case, in which uncertainty about paths inherits from random fluctuations that determine the conditional entropies of paths. When the amplitude of random fluctuations on the motion of particular states is very small, we have precise particles in which there is no uncertainty about autonomous paths, given sensory paths (the right information diagram). Similarly, there is no uncertainty about sensory paths given external and autonomous paths. Note that because we are dealing with continuous states, we are implicitly interpreting the entropies as the limiting density of discrete points (LDDP), which have a lower bound of zero \cite{jaynesInformationTheoryStatistical1957}. (LDDP is an adjustment to differential entropy which ensures that entropy is lower bounded by zero. LDDP equals the negative KL-divergence between the density in question and a uniform density). Two relative entropies (information gain and risk) are highlighted as areas of intersection. These will play an important role later, when decomposing the action (i.e., expected free energy) of autonomous paths.}
    \label{fig: 4}
\end{figure}

This suggests that precise particles—such as you and me—will respond to environmental flows and fluctuations in a precise and predictable fashion. Figure \ref{fig: 4} illustrates the difference between generic and precise particles using an information diagram. Note that for precise particles, there is no uncertainty about autonomous states, given sensory states. This follows because the flow of autonomous states depends only on sensory states and themselves. Is the behaviour of precise particles a sufficient description of sentient behaviour? 

On one reading, perhaps: one can reproduce biological behaviour by numerically integrating \eqref{eq: 21} or \eqref{eq: 27} under a suitable generative (state-space) model specifying the motion of external and sensory states\footnote{A generative model can be specified through the flow of external or sensory states, and the random fluctuations of their motion; that is, the first two lines of \eqref{eq: 27}. Observing that the free energy \eqref{eq: 26} is only a function of these flows and the covariance of fluctuations, it is sufficient to specify those covariances, rather than the whole structure of the fluctuations.}. Figure \ref{fig: 5} illustrates the implicit computational architecture used to simulate sentient behaviour by integrating \eqref{eq: 21}. This scheme allows one to simulate the internal and active states through sensory states caused by external dynamics. Figure \ref{fig: 6} showcases an example from the active inference literature, that integrates \eqref{eq: 27} under a suitably specified generative model, to simulate sentient behaviour that looks like handwriting. The details of the simulation and the details of the generative model are not relevant here but are summarised in the figure legend; what is important is to get a sense of the kind of behaviour that can be reproduced by integrating \eqref{eq: 27}.

 \begin{figure}[t!]
    \centering
    \includegraphics[width=\textwidth]{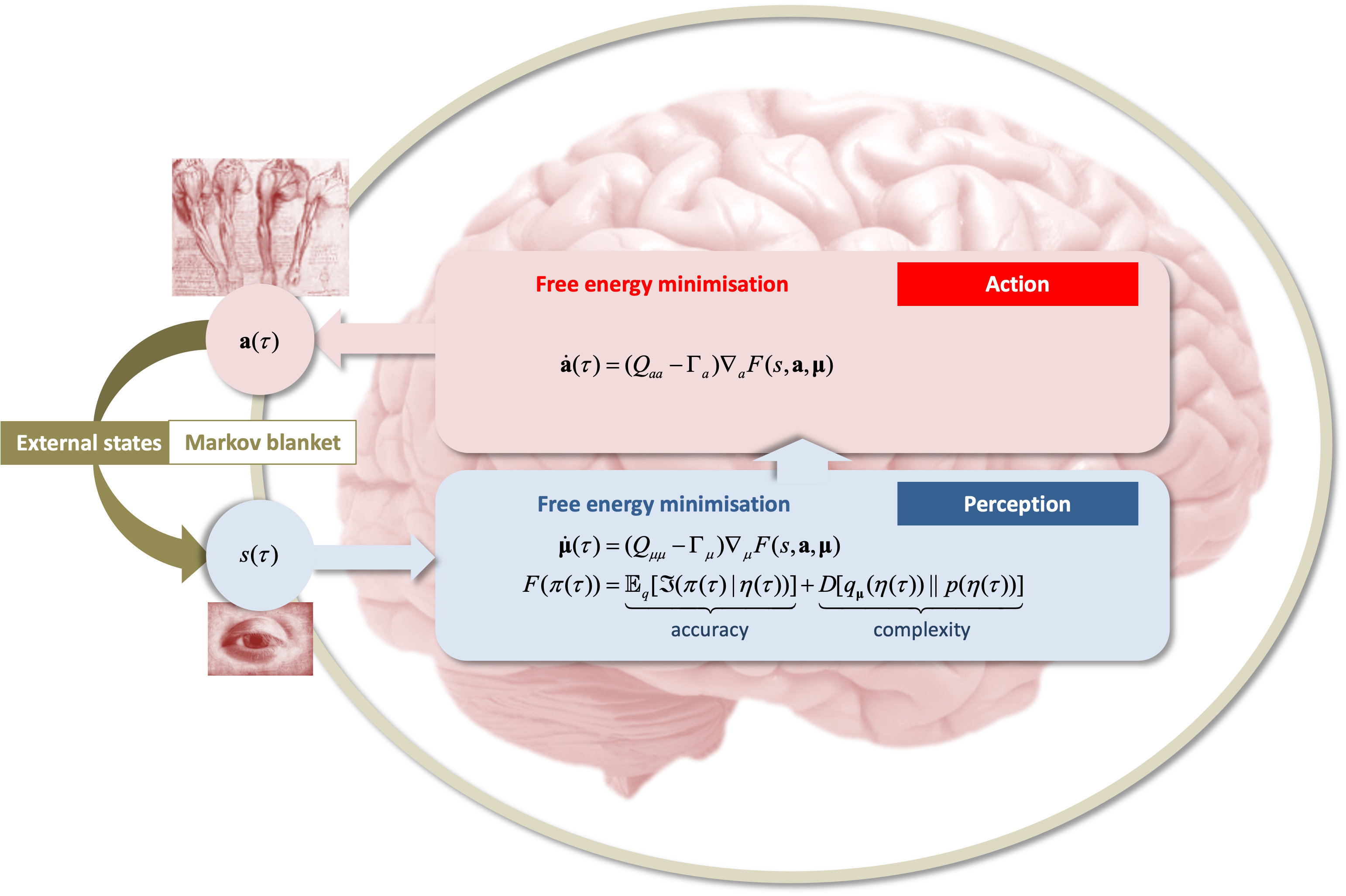}
    \caption[Bayesian mechanics and active inference]{\textbf{Bayesian mechanics and active inference}. This graphic summarises the belief updating implicit in gradient flows on variational free energy. These are the paths taken by a precise particle or the paths of least action of a generic particle. It illustrates a simple form of (active) inference that has been used in a variety of applications and simulations; ranging from handwriting and action observation~\cite{fristonActionUnderstandingActive2011}, through to birdsong and generalised synchrony in communication~\cite{fristonActiveInferenceCommunication2015}. In brief, sensory states furnish free energy gradients (often expressed as prediction errors), under some generative model. Neuronal dynamics are simulated as a flow on the resulting gradients to produce internal states that parameterise posterior beliefs about external states. Similarly, active states are simulated as a flow on free energy gradients that generally play the role of prediction errors. In other words, active states mediate motor or autonomic reflexes~\cite{feldmanNewInsightsActionperception2009,mansellControlPerceptionShould2011}. An example of this kind of active inference is provided in the next figure.}
    \label{fig: 5}
\end{figure}

  \begin{figure}
    \centering
    \includegraphics[width=0.8\textwidth]{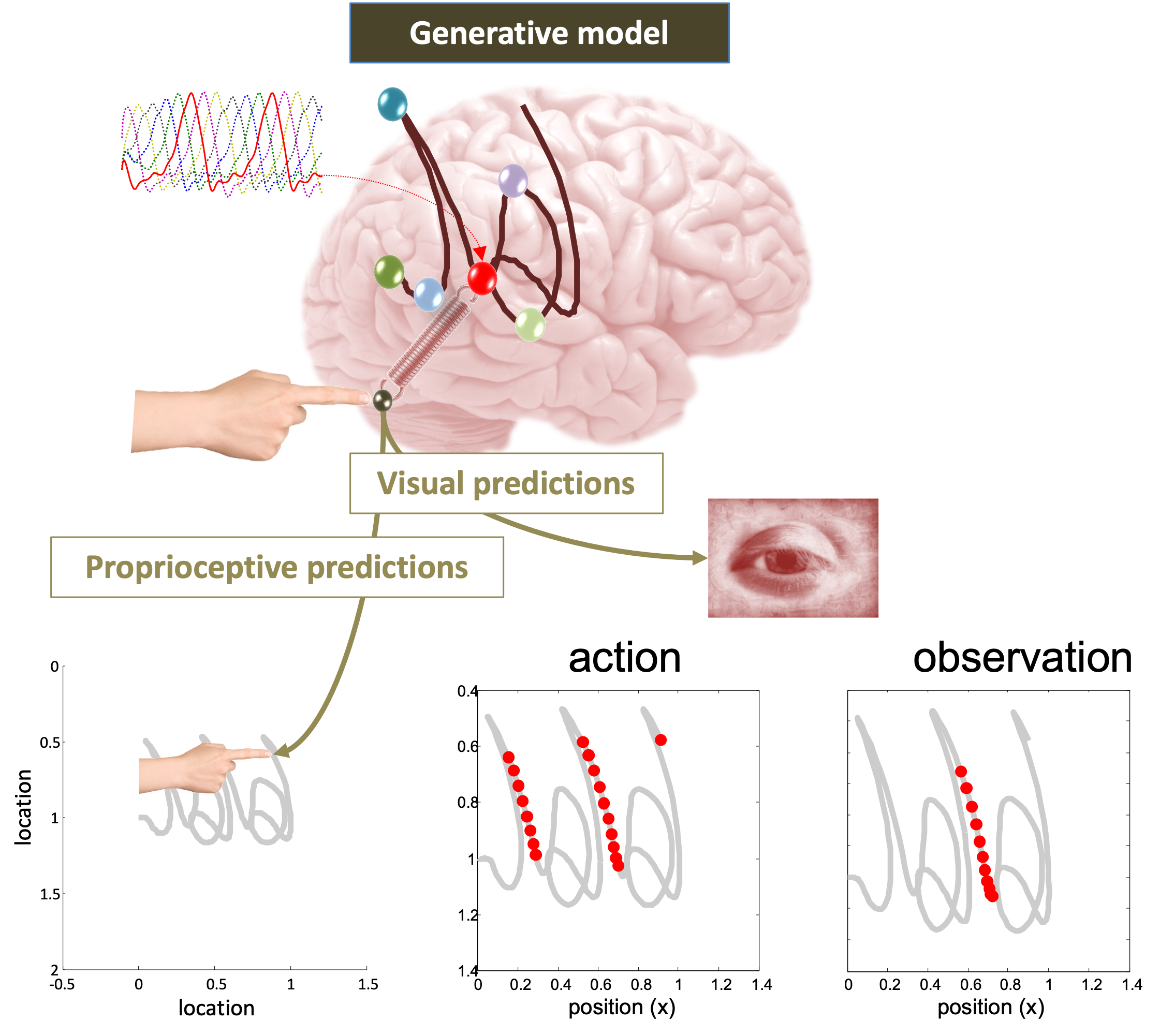}
    \caption[Sentient behaviour and action observation]{\textbf{Sentient behaviour and action observation}. This figure illustrates a simulation of active inference (here, writing) evinced by a precise particle, in terms of inferences about external states of the world, consequent predictions about sensory input, and ensuing action. The autonomous dynamics that underwrite this behaviour rest upon a generative model of sensory states in the form of Lotka-Volterra dynamics; see sample sensory trajectories as (arbitrarily) coloured lines in the upper left inset. The generative model defines the joint density under which internal trajectories can be seen as parameterising external states. This model is not a description of the true external states (which here are simply the positions of the joints in the simulated arm---with dynamics given by simple Newtonian rules). In this generative model, external trajectories are  assumed to follow predator-prey like dynamics such that a succession of peaks are generated for a subset of external states (or coordinates) in turn. Each coordinate is associated with a location in Euclidean space that attracts the agent’s finger (the active states); i.e., with a trajectory towards that attracting point. The resulting attracting point is thus a weighted sum of each possible attracting point weighted by the coordinates following the Lotka-Volterra trajectory. In turn, the internal states supply predictions of what sensory states should register if the agent’s beliefs were true. Active states (i.e., the forces driving changes in the angular velocities of the limb joints) try to suppress the ensuing prediction error by adjusting expected changes in sensed angular velocity, through exerting forces on the agent’s joints (not shown). The subsequent movement of the arm is traced out in the lower-left panel. This trajectory has been plotted in a moving frame of reference so that it looks like handwriting (e.g., a succession of ‘j’ and ‘a’ letters). The lower right panels show the activity of one internal state during distinct phases of ‘action’, and ‘action-observation’. During the action phase, sensory states register the visual and proprioceptive consequences of movement, while under action observation, only visual sensations are available—as if the agent was watching another agent. The red dots correspond to the times during which this internal state exceeded an arbitrary threshold. The key thing to note here is that this internal state responds preferentially when, and only when, the motor trajectory produces a down-stroke, but not an up-stroke—evincing a cardinal feature of neuronal responses, namely, their functional selectivity. Furthermore, with a slight delay, this internal state responds during action and action observation. From a biological perspective, this is interesting because it speaks to an empirical phenomenon known as mirror neuron activity~\cite{galleseMirrorNeuronsSimulation1998,rizzolattiMirrorneuronSystem2004,kilnerPredictiveCodingAccount2007}. Please see~\cite{fristonActionUnderstandingActive2011} for further details.}
    \label{fig: 6}
\end{figure}

The example in Figure \ref{fig: 6} illustrates an application of the free energy principle. Here, instead of describing a system by deriving its NESS density, we have specified some equations of motion (and covariance of random fluctuations) to realise particular dynamics using \eqref{eq: 27} and \eqref{eq: free energy generalised}. In effect, we have simulated self-evidencing, starting from a definition (i.e., state-space generative model) of paths that characterise this kind of particle\footnote{The example in Figure \ref{eq: 6} used \eqref{eq: 27} with generalised coordinates of motion up to fourth order. 
Numerical analyses suggest that simulating generalised motion up to order six (i.e., ignoring all subsequent orders of motions) is sufficient in most circumstances \cite{fristonVariationalTreatmentDynamic2008}.}.

These simulations speak to a key aspect in the applications of the FEP. Hitherto, we have simply defined the variational density as the conditional density over external states given a sensory state. However, when simulating precise particles through a gradient flow on variational free energy, as in \eqref{eq: 21} or \eqref{eq: 27}, the requisite gradients have to be evaluated. In turn, this requires the functional form of the variational density or posterior distribution, which may be difficult to compute exactly\footnote{In Bayesian inference, it is well-known that computing the posterior distribution given data and a generative model $p(\eta \mid \pi) = p(\eta, \pi)/p(\pi)$ is computationally costly as it involves computing a (typically) high-dimensional integral $p(\pi)=\int p(\eta, \pi) d \eta$ (i.e., a partition function).}. In this case, we take a variational density that approximates the true posterior, whence the variational free energy becomes an upper bound on surprisal: see \eqref{eq: free energy upper bound}. From the perspective of Bayesian inference, this takes us from (computationally costly) \textit{exact} Bayesian inference to (computationally cheap) \textit{approximate} Bayesian inference~\cite{bealVariationalAlgorithmsApproximate2003,winnVariationalMessagePassing2005,dauwelsVariationalMessagePassing2007}. On one reading of its inception, this is why variational free energy was introduced \cite{feynmanStatisticalMechanicsSet1998}; namely, to convert a computationally expensive marginalisation problem into a computationally manageable optimisation problem. Note that when using generalised coordinates to realise active inference; i.e., \eqref{eq: 27}, we are generally employing approximate Bayesian inference: the functional form of the variational density inherits directly from Gaussian assumptions about random fluctuations, however the expansion in generalised coordinates on which it is based upon \eqref{eq: 24} is generally an approximation to the underlying dynamic (cf. \ref{footnote expansion in gen coords}).

\subsection*{Summary}
Precise particles, immersed in an imprecise world, respond (almost) deterministically to external fluctuations\footnote{\textbf{Question}: does the absence of random fluctuations preclude dissipative gradient flows? No, because the gradients can increase with the precision of random fluctuations. In the limit of no random fluctuations, the steady-state density tends towards a delta function (i.e., a fixed-point attractor) and the dissipative gradients tend towards infinity.}. This means, given a generative model (i.e., NESS density), one can solve the equations of motion in \eqref{eq: 27} to predict how autonomous states evolve as they pursue their path of least action. So, why might this limiting behaviour be characteristically biological?

Precise particles may be the kind of particles that show lifelike or biotic behaviour, in the sense they respond predictably, given their initial states and the history of external influences. The distinction between imprecise (e.g., statistical) and precise (e.g., classical) particles rests on the relative contribution of dissipative and conservative flow to their path through state-space, where solenoidal flow predominates in the precise setting. This means precise particles exhibit solenoidal behaviour such as oscillatory and (quasi) periodic orbits—and an accompanying loss of detailed balance, i.e., turbulent and time-irreversible dynamics~\cite{andresMotionSpikesTurbulentLike2018,decoTurbulentlikeDynamicsHuman2020,dacostaEntropyProductionStationary2022}. On this view, one might associate precise particles with living systems with characteristic biorhythms~\cite{lopesdasilvaNeuralMechanismsUnderlying1991,kopellNeuronalAssemblyDynamics2011,arnalCorticalOscillationsSensory2012,buzsakiScalingBrainSize2013}; ranging from gamma oscillations in neuronal populations, through slower respiratory and diurnal cycles to, perhaps, lifecycles per se. Turning this on its head, one can argue that living systems are a certain kind of particle that, in virtue of being precise, evince conservative dynamics, biorhythms and time irreversibility.

One might ask if solenoidal flow confounds the gradient flows that underwrite self-evidencing. In fact, solenoidal flow generally augments gradient flows—or at least this is what it looks like. In brief, the mixing afforded by solenoidal flow can render gradient descent more efficient~\cite{ottControllingChaos1990,hwangAcceleratingDiffusions2005,hwangAcceleratingGaussianDiffusions1993,lelievreOptimalNonreversibleLinear2013,aslimaniNewHybridAlgorithm2018}. An intuitive example is stirring sugar into coffee. The mixing afforded by the solenoidal stirring facilitates the dispersion of the sugar molecules down their concentration gradients. On this view, the solenoidal flow can be regarded as circumnavigating the contours of the steady-state density to find a path of steepest descent.

The emerging picture here is that biotic systems feature solenoidal flow, in virtue of being sufficiently large to average away random fluctuations, when coarse-graining their dynamics \cite{fristonFreeEnergyPrinciple2019a}. From the perspective of the information geometry induced by the FEP, this means biological behaviour may be characterised by internal solenoidal flows that do not change variational free energy—or surprisal—and yet move on the internal (statistical) manifold to continually update Bayesian beliefs about external states. Biologically, this may be a description of central pattern generators \cite{arnalCorticalOscillationsSensory2012,grossSpeechRhythmsMultiplexed2013} that underwrite rhythmical activity (e.g., walking and talking) that is characteristic of biological systems \cite{buzsakiNeuronalOscillationsCortical2004}. The example in Figure \ref{fig: 6} was chosen to showcase the role of solenoidal flows in Bayesian mechanics that—in this example—arise from the use of Lotka-Volterra dynamics in the generative model. In psychology, this kind of conservative active inference may be the homologue of being in a ‘flow state’~\cite{csikszentmihalyiFlowPsychologyOptimal2008}.

In short, precise particles may be the kind of particles we associate with living systems. And precise particles have low entropy paths. If so, the question now becomes: what long-term behaviour does this class of particle show? In other words, instead of asking which behaviours \textit{lead} to low entropy dynamics, we can now ask which behaviours \textit{follow from} low entropy dynamics? We will see next that precise particles appear to plan their actions and, perhaps more interestingly, show information and goal-seeking behaviour.

\section{Path integrals, planning and curious particles}
\label{sec: curious}

While the handwriting example in Figure \ref{fig: 6} offers a compelling simulation of self-evidencing—in the sense of an artefact creating its own sensorium—there is something missing as a complete account of sentient behaviour. This is because we have only considered the response of autonomous states to sensory states over limited periods of time. To disclose a deeper Bayesian mechanics, we need to consider the paths of autonomous states over extended periods. This takes us to the final step and back to the path-integral formulation.

In the previous section, we focused on linking dynamics to densities over (generalised) states. In brief, we saw that internal states can be construed as parameterising (Bayesian) beliefs about external states at any point in time. In what follows, we move from densities over \textit{states} to densities over \textit{paths}—to characterise the behaviour of particles in terms of their trajectories.

In what follows, we will be dealing with predictive posterior densities over external and particular paths, given (initial) particular states, which can be expressed in terms of the variational density parameterised by the current (initial) internal state:\footnote{\textbf{Question:} Why is the variational density parameterised by the initial internal state rather than the initial internal mode? The answer is that in precise particles, the absence of fluctuations on particular dynamics means that the internal states always coincides with the internal mode.}
 \begin{equation}
 \label{eq: 29}
\begin{aligned}
q\left(\eta[\tau], \pi[\tau] \mid \pi_{0}\right) & \triangleq \mathbb{E}_{q_{\mu}}\left[p\left(\eta[\tau], \pi[\tau] \mid \eta_{0}, \pi_{0}\right)\right]=p\left(\eta[\tau], \pi[\tau] \mid \pi_{0}\right) \\
q_{\mu}\left(\eta_{0}\right) &=p\left(\eta_{0} \mid \pi_{0}\right).
\end{aligned}
\end{equation}	

All this equation says is that, given the initial particular states, we can evaluate the joint density over external and particular paths, because we know the density over the initial external states, which is parameterised by the initial internal state.

We are interested in characterising autonomous responses to initial particular states. This is given by the action of autonomous paths as a function of particular states. In other words, we seek an expression for the probability of an autonomous path that \textit{(i)} furnishes a teleological description of self-organisation and \textit{(ii)} allows us to simulate the sentient trajectories of particles, given their sensory streams. Getting from the action of particular paths to the action of autonomous paths requires a marginalisation over sensory paths. This is where the precise particle assumption comes in: it allows us to eschew this (computationally costly) marginalisation by expressing the action of particular paths as an \textit{expected free energy}.

Recall that when random fluctuations on the motion of particular states vanish, there is no uncertainty about autonomous paths, given external and sensory paths. And there is no uncertainty about sensory paths given external and autonomous paths. If we interpret entropies as the limiting density of discrete points (see Figure \ref{fig: 4}), then the uncertainty about particular, autonomous and sensory paths, given external paths, become interchangeable:
 \begin{equation}
 \label{eq: 30}
\begin{aligned}
\Gamma_{\pi} &\equiv 0\\
&\Rightarrow\\
\H\left[p\left(\pi[\tau] \mid \eta[\tau], \pi_{0}\right)\right] &=\underbrace{\H\left[p\left(\alpha[\tau] \mid \eta[\tau], s[\tau], \pi_{0}\right)\right]}_{=0}+\H\left[p\left(s[\tau] \mid \eta[\tau], \pi_{0}\right)\right] \\
&=\underbrace{\H\left[p\left(s[\tau] \mid \eta[\tau], \alpha[\tau], \pi_{0}\right)\right]}_{=0}+\H\left[p\left(\alpha[\tau] \mid \eta[\tau], \pi_{0}\right)\right] \\
& \Rightarrow \\
\mathbb{E}_{q}\left[\ln p\left(\pi[\tau] \mid \eta[\tau], \pi_{0}\right)\right] &=\mathbb{E}_{q}\left[\ln p\left(s[\tau] \mid \eta[\tau], \pi_{0}\right)\right]=\mathbb{E}_{q}\left[\ln p\left(\alpha[\tau] \mid \eta[\tau], \pi_{0}\right)\right]
\end{aligned}
\end{equation}
We can leverage this exchangeability to express the action of autonomous paths in terms of an expected free energy. From \eqref{eq: 29} and \eqref{eq: 30}, we have (dropping the conditioning on initial states for clarity):
 \begin{equation}
 \label{eq: 31}
\begin{aligned}
&0=\mathbb{E}_{q}\left[\ln \frac{p(\eta[\tau], \alpha[\tau])}{q(\eta[\tau], \alpha[\tau])}\right]=\mathbb{E}_{q}\left[\ln \frac{p(\alpha[\tau] \mid \eta[\tau]) p(\eta[\tau])}{q(\eta[\tau] \mid \alpha[\tau]) q(\alpha[\tau])}\right]\\
&=\mathbb{E}_{q}\left[\ln \frac{p(s[\tau] \mid \eta[\tau]) p(\eta[\tau])}{q(\eta[\tau] \mid \alpha[\tau])}-\ln q(\alpha[\tau])\right]=\mathbb{E}_{q(\alpha[\tau])}[\mathcal{A}(\alpha[\tau])-G(\alpha[\tau])]\\
&=D\left[q(\alpha[\tau]) \| \mathrm{e}^{-G}\right] \Rightarrow G(\alpha[\tau])=\mathcal{A}(\alpha[\tau])\\
&G(\alpha[\tau])=\mathbb{E}_{q(\eta[\tau], s[\tau] \alpha[\tau])}[\overbrace{\underbrace{\ln q(\eta[\tau] \mid \alpha[\tau])-\ln p(\eta[\tau])}_{\text {Expected complexity }}}^{\text {Risk }} \overbrace{-\underbrace{\ln p(s[\tau] \mid \eta[\tau])}_{\text {Expected accuracy }}}^{\text {Ambiguity }}]
\end{aligned}
\end{equation}	
All we have done here is to exchange the density over autonomous paths, conditioned on external paths, with the corresponding density over sensory paths (in the second line) thanks to the precise particle assumption. By gathering terms into a functional of autonomous paths, we recover autonomous action as an expected free energy.

 \begin{figure}
    \centering
    \includegraphics[width=\textwidth]{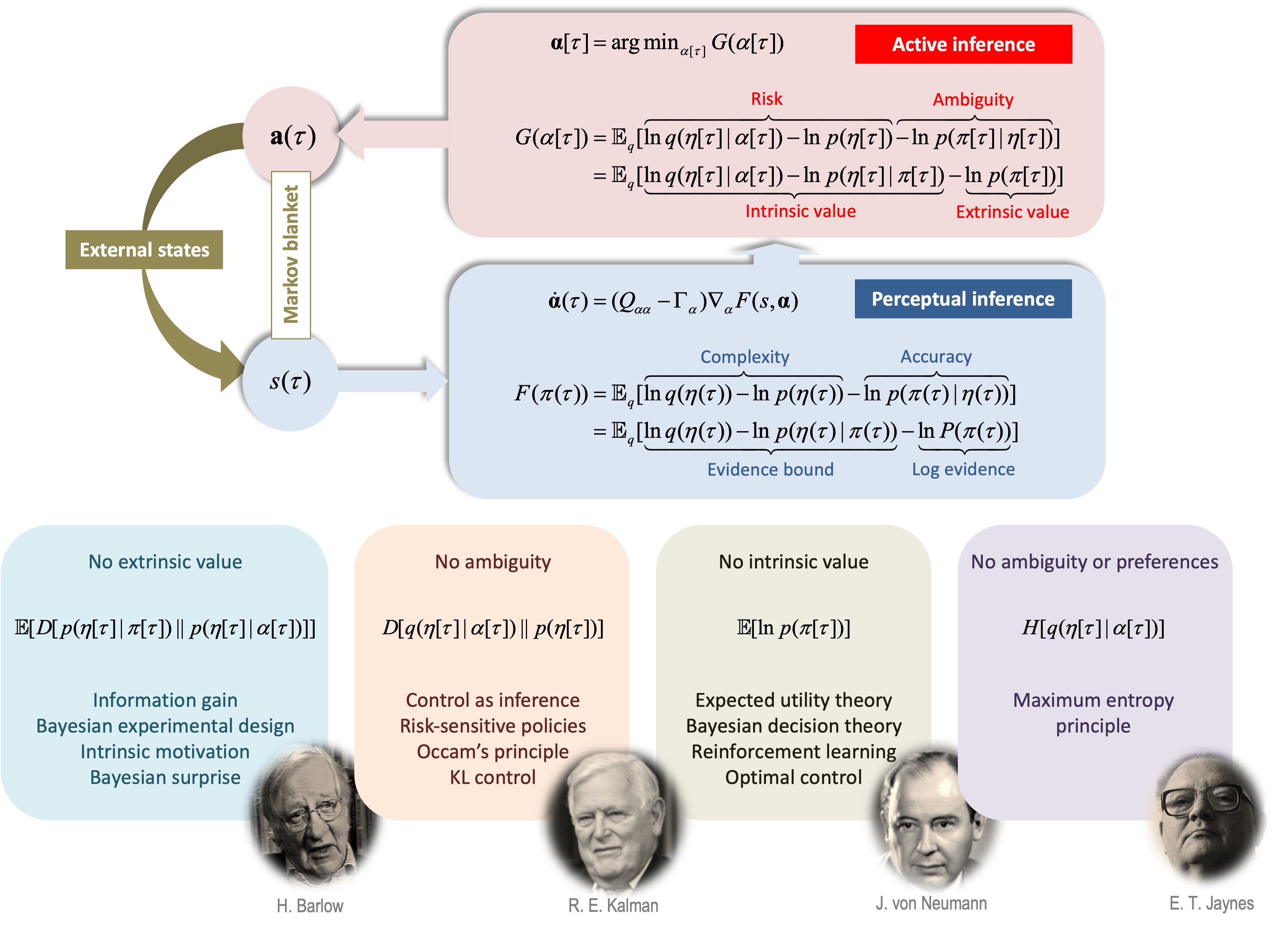}
    \caption[Expected free energy and active inference]{\textbf{Expected free energy and active inference}. This figure illustrates active inference, and highlights various points of contact with other accounts of sentient, purposeful or intelligent behaviour. The upper panel casts action and perception as the minimisation of expected and variational free energy, respectively. Crucially, the path integral formulation of active inference introduces posterior beliefs over autonomous paths (i.e., policies) that entail a description of planning as inference~\cite{attiasPlanningProbabilisticInference2003,botvinickPlanningInference2012,kaplanPlanningNavigationActive2018}. When simulating active inference, posterior beliefs about external paths, under plausible policies, are optimised by a gradient flow on the variational (free energy) bound on log evidence—as in Figure \ref{fig: 3}. These beliefs are then used to evaluate the expected free energy of allowable policies, from which actions can be selected~\cite{fristonActiveInferenceProcess2017,dacostaActiveInferenceDiscrete2020,barpGeometricMethodsSampling2022a}. Crucially, expected free energy contains terms that arise in various formulations of optimal behaviour that predominate in cognitive science, engineering and economics. These terms are disclosed when one removes certain sources of uncertainty.
    For example, if we remove ambiguity, decision-making minimises risk, which corresponds to aligning predictions with preferences about the external course of events. This underwrites prospect theory of human choice behaviour in economics~\cite{kahnemanProspectTheoryAnalysis1979} and modern approaches to control as inference~\cite{levineReinforcementLearningControl2018,rawlikStochasticOptimalControl2013,toussaintRobotTrajectoryOptimization2009}, variously known as Kalman duality~\cite{kalmanNewApproachLinear1960,todorovGeneralDualityOptimal2008a}, KL control~\cite{kappenOptimalControlGraphical2012} and maximum entropy reinforcement learning~\cite{ziebartModelingPurposefulAdaptive2010}.
    If we further remove preferences, decision-making maximises the entropy of external trajectories. This maximum entropy principle~\cite{jaynesInformationTheoryStatistical1957,lasotaChaosFractalsNoise1994} can be interpreted as least committing to a presupposed external trajectory and therefore keeping options open \cite{klyubinKeepYourOptions2008}.
    If we reintroduce ambiguity, but ignore preferences, decision-making maximises intrinsic value or expected information gain~\cite{mackayInformationTheoryInference2003}. This underwrites Bayesian experimental design~\cite{lindleyMeasureInformationProvided1956} and active learning in statistics~\cite{mackayInformationBasedObjectiveFunctions1992}, intrinsic motivation and artificial curiosity in machine learning and robotics~\cite{oudeyerWhatIntrinsicMotivation2007,schmidhuberFormalTheoryCreativity2010,bartoNoveltySurprise2013,sunPlanningBeSurprised2011,deciIntrinsicMotivationSelfDetermination1985}. This is mathematically equivalent to optimising expected Bayesian surprise and mutual information, which underwrites visual search~\cite{ittiBayesianSurpriseAttracts2009,parrGenerativeModelsActive2021} and the organisation of our visual apparatus~\cite{barlowPossiblePrinciplesUnderlying1961,linskerPerceptualNeuralOrganization1990,opticanTemporalEncodingTwodimensional1987a}.
    Lastly, if we remove intrinsic value, we are left with maximising extrinsic value or expected utility, which underwrites expected utility theory~\cite{vonneumannTheoryGamesEconomic1944}, game theory, optimal control~\cite{bellmanDynamicProgramming1957, astromOptimalControlMarkov1965a} and reinforcement learning~\cite{bartoReinforcementLearningIntroduction1992}. Bayesian formulations of maximising expected utility under uncertainty are also known as Bayesian decision theory~\cite{bergerStatisticalDecisionTheory1985}.
    The expressions for variational and expected free energy are arranged to illustrate the relationship between \textit{complexity} and \textit{accuracy}, which become \textit{risk} and \textit{ambiguity} in the path integral formulation. This suggests that risk-averse policies minimise expected complexity or computational cost~\cite{schmidhuberFormalTheoryCreativity2010}.}
    \label{fig: 7}
\end{figure}

By analogy with the expression for variational free energy \eqref{eq: 18}, the expressions for the expected free energy in \eqref{eq: 31} suggest that \textit{accuracy} becomes \textit{ambiguity}, while \textit{complexity} becomes \textit{risk}. So why have we called these terms ambiguity and risk? Ambiguity is just the expected precision or conditional uncertainty about sensory states given external states. A heuristic example of an imprecise likelihood mapping—between external and sensory paths—would be a dark room, where there is no precise information at hand. Indeed, according to \eqref{eq: 31}, sensory paths into dark rooms should be highly unlikely. However, this is not the complete story, in the sense that the risk puts certain constraints on any manifest tendency to minimise ambiguity. 

Here, risk is simply the divergence between external paths given an autonomous path (i.e., policy or plan), relative to external states of affairs. The marginal density over external paths is often referred to in terms of \textit{prior preferences}, because they constitute the priors of the generative model characterising the particle's behaviour \cite{parrGeneralisedFreeEnergy2019}. In short, the expression for expected free energy, suggests that particles will look as if they are (i) minimising the risk of incurring external trajectories that diverge from prior preferences, while (ii) resolving ambiguity in response to external events. In this formulation, autonomous paths play the dual role of registering the influences of external events (via ambiguity), while also causing those events (via risk). 

The autonomous path with the least expected free energy is the most likely path taken by the autonomous states. 
 \begin{equation}
 \label{eq: 32}
\begin{aligned}
G(\alpha[\tau]) &=\mathcal{A}(\alpha[\tau]) \\
& \Rightarrow \boldsymbol \alpha[\tau]=\arg \min _{\alpha[\tau]} G(\alpha[\tau]) \\
& \Rightarrow \delta_{\alpha} G(\boldsymbol \alpha[\tau])=0 \\
\mathbb{E}[G(\alpha[\tau])] &=\mathbb{E}[\mathcal{A}(\alpha[\tau])]=\H[p(\alpha[\tau])]
\end{aligned}
\end{equation}	

In short, expected free energy scores the autonomous action of particles that do not admit noisy dynamics. Expected free energy has a specific form that inherits from the assumption that the amplitude of particular fluctuations is small, which is the case for precise articles by definition. Although variational and expected free energy are formally similar, they are fundamentally different kinds of functionals: variational free energy is a functional of a density over states, while expected free energy is a functional of a density over paths. Variational free energy can also be read as a function of particular states, while expected free energy is a function of an autonomous path. Finally, variational free energy is a bound on surprisal, while expected free energy is not a bound—it is the action of autonomous trajectories.

Expected free energy plays a definitive role in active inference, where it can be regarded as a fairly universal objective function for selecting autonomous paths of least action. Figure \ref{fig: 7} shows that the expected free energy contains terms that arise in various formulations of optimal behaviour; ranging from optimal Bayesian design~\cite{lindleyMeasureInformationProvided1956} through to control as inference~\cite{ziebartModelingPurposefulAdaptive2010,levineReinforcementLearningControl2018}. We refer the reader to ~\cite{dacostaRewardMaximizationDiscrete2023,fristonSophisticatedInference2021,sajidActiveInferenceBayesian2022,hafnerActionPerceptionDivergence2020,millidgeRelationshipActiveInference2020a} for formal investigations of the relationship between these formulations.

Equipped with a specification of the most likely autonomous path—in terms of expected free energy—we can simulate fairly lifelike behaviour, given a suitable generative model. An example is provided in Figure \ref{fig: 9}---relying upon the computational architecture in Figure \ref{fig: 8}---which illustrates the ambiguity resolving part of the expected free energy in a simulation of visual epistemic foraging.

This epistemic aspect of expected free energy can be seen more clearly if we replace the conditional uncertainty about sensory paths with conditional uncertainty about particular paths, noting that they are the same by \eqref{eq: 30}. After rearrangement, we can express expected free energy in terms of \textit{expected value} and \textit{expected information gain}
~\cite{sajidActiveInferenceBayesian2022,barpGeometricMethodsSampling2022a}:
 \begin{equation}
 \label{eq: 33}
\begin{aligned}
&G(\alpha[\tau])=\mathbb{E}_{q(\eta[\tau], s[\tau] \mid  \alpha[\tau])}[\ln q(\eta[\tau] \mid \alpha[\tau])-\ln p(\eta[\tau])-\ln p(\pi[\tau] \mid \eta[\tau])]\\
&=\underbrace{\overbrace{\mathbb{E}_{q(s[\tau] \mid \alpha[\tau]}[\mathcal{A}(\pi[\tau])]}^{\text {Expected value }}}_{\text {Bayes optimal decisions }}-\underbrace{\overbrace{\mathbb{E}_{q(s[\tau] \mid \alpha[\tau])}\left[D[p(\eta[\tau] \mid s[\tau], \alpha[\tau]) \| p(\eta[\tau] \mid \alpha[\tau])]\right.}^{\text {Expected information gain }}}_{\text {Bayes optimal design}}
\end{aligned}
\end{equation}

 \begin{figure}[t!]
    \centering
    \includegraphics[width=\textwidth]{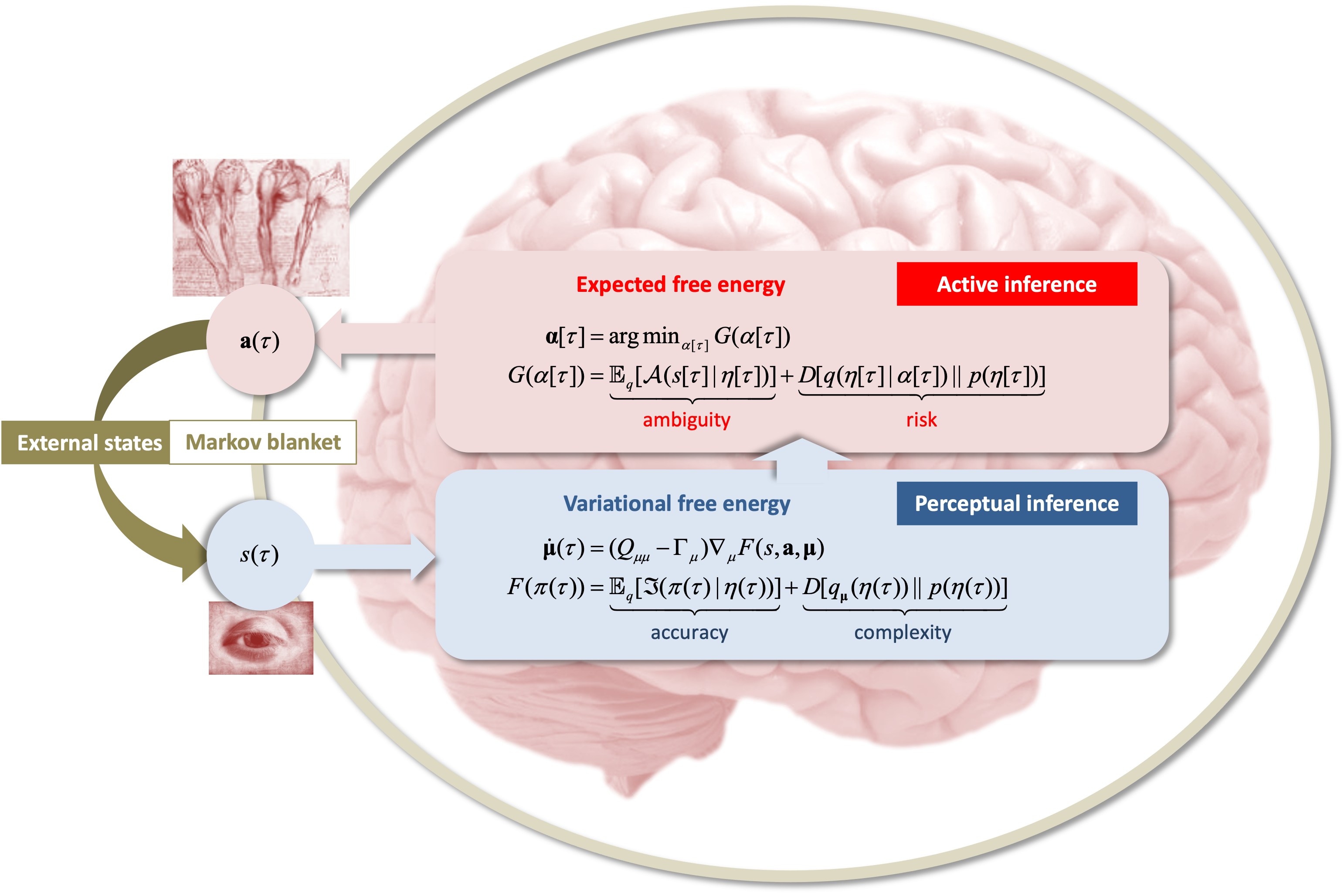}
    \caption[Bayesian mechanics and active inference]{\textbf{Bayesian mechanics and active inference}. This graphic summarises the belief updating implicit in the minimisation of variational and expected free energy. It describes active inference based upon autonomous paths or policies and has been used in a variety of applications and simulations; ranging from games in behavioural economics~\cite{fitzgeraldActiveInferenceEvidence2015} and reinforcement learning~\cite{schwartenbeckEvidenceSurpriseMinimization2015,sajidActiveInferenceDemystified2021} through to language understanding~\cite{fristonDeepTemporalModels2018} and scene construction~\cite{mirzaSceneConstructionVisual2016}. In this setup, actions solicit a sensory outcome that informs approximate posterior beliefs about hidden or external states of the world—via minimisation of variational free energy under a set of plausible policies (i.e., perceptual inference). The approximate posterior beliefs are then used to evaluate expected free energy and subsequent action (i.e., active inference). A key insight from simulations is that the form of the generative model can be quite different from the process by which external states generate sensory states. In effect, this enables agents (i.e., particles) to author their own sensorium in a fashion that has close connections with econiche construction~\cite{bruinebergSelforganizationFreeEnergy2014}. Please see~\cite{fristonGraphicalBrainBelief2017,dacostaActiveInferenceDiscrete2020} for technical details and for a heuristic discussion of how the belief updating could be implemented in the brain.}
    \label{fig: 8}
\end{figure}

This provides a complementary interpretation of expected free energy. The first term can be construed as expected cost in the sense it is the expected action of particular paths. This marginal likelihood scores the plausibility of a particle pursuing this kind of path and is usually interpreted in terms of expected loss (i.e., negative expected reward or utility)~\cite{vonneumannTheoryGamesEconomic1944,bartoReinforcementLearningIntroduction1992}, and pragmatic affordance~\cite{schwartenbeckEvidenceSurpriseMinimization2015,fristonActiveInferenceProcess2017}. The second term corresponds to the expected divergence between posterior beliefs about external paths, given autonomous paths, with and without sensory paths. In other words, it scores the resolution of uncertainty or expected information gain afforded by sensory trajectories arising from a commitment to an autonomous path. In this sense, it is sometimes referred to as epistemic affordance~\cite{fristonGraphicalBrainBelief2017}.

When simulating the kind of planning and active inference afforded by the path integral formulation, one usually works with discrete state-spaces and belief updating over discrete epochs of time~\cite{dacostaActiveInferenceDiscrete2020,fristonActiveInferenceProcess2017}. One can see this as a coarse-graining of continuous space-time into discrete space and time bins, where trajectories of continuous states become sequences of discrete states $x[\tau]=(x_1, \ldots, x_\tau)$. In discrete state-spaces, the generative model is usually formulated as a partially observed Markov decision process \cite{dacostaRewardMaximizationDiscrete2023,dacostaActiveInferenceDiscrete2020,parrActiveInferenceFree2022,smithStepbystepTutorialActive2022}, in which the paths of autonomous states constitute policies, which determine transitions among external states. Plausible policies can then be scored with their expected free energy and the next action is selected from the most likely policy $\alpha=(\alpha_{0}, \ldots, \alpha_{\tau})$\footnote{See \cite{dacostaActiveInferenceDiscrete2020,parrActiveInferenceFree2022} for a derivation of these functional forms in partially observable Markov decision processes.}
 \begin{equation}
 \label{eq: 34}
\begin{aligned}
\mathbf{a} &=\arg \min _{a} G(a, \boldsymbol{\mu}) \\
G &=\mathbb{E}_{Q}\left[\ln Q_{\mu}\left(\eta_{1}, \ldots, \eta_{\tau} \mid \eta_{0}, a\right)-\ln P\left(\eta_{1}, \ldots, \eta_{\tau} \mid \eta_{0}\right)-\ln P\left(s_{1}, \ldots, s_{\tau} \mid \eta_{1}, \ldots, \eta_{\tau}\right)\right] \\
& \approx \sum_{t>0} \underbrace{\mathbb{E}_{Q}\left[\ln Q_{\mu}\left(\eta_{t} \mid a\right)-\ln P\left(\eta_{t} \mid \eta_{0}\right)\right]}_{\text {Risk }} \underbrace{-\mathbb{E}_{Q}\left[\ln P\left(s_{t} \mid \eta_{t}\right)\right]}_{\text {Ambiguity }} \\
\boldsymbol{\mu} &=\arg \min _{\mu} F(s, a, \mu) \\
F&=\sum_{t<\tau} \underbrace{\mathbb{E}_{Q}\left[\ln Q_{ \mu}\left(\eta_{t} \mid a\right)-\ln P\left(\eta_{t+1} \mid \eta_{t}, a\right)-\ln P\left(\eta_{t} \mid \eta_{t-1}, a\right)\right]}_{\text {Complexity }}-\sum_{t \leq 0} \underbrace{\mathbb{E}_{Q}\left[\ln P\left(s_{t} \mid \eta_{t}\right)\right]}_{\text {Accuracy }}.
\end{aligned}
\end{equation}
The conditional independencies among states implicit in partially observed Markov decision processes entail the above functional forms for variational and expected free energies~\cite{dacostaActiveInferenceDiscrete2020,fristonActiveInferenceProcess2017}. Crucially, the posterior over external states uses a mean-field approximation, in which the joint distribution over current and future states factorises into marginal distributions at each point in time [this approximation can be finessed by conditioning on previous states, leading to a different (Bethe) variational free energy~\cite{yedidiaConstructingFreeEnergyApproximations2005,parrNeuronalMessagePassing2019}]. Note that the discrete version of variational free energy is a functional of a distribution over a sequence of states and can be regarded as the discrete homologue of the variational free energy of generalised states in \eqref{eq: free energy generalised}.

The ensuing minimisation of free energy can be formulated as gradient flows following \eqref{eq: 17}—between the discrete arrival of new sensory input—in a way that relates comfortably to neuronal dynamics~\cite{fristonActiveInferenceProcess2017,dacostaActiveInferenceDiscrete2020,dacostaNeuralDynamicsActive2021}. In some simulations, one can mix discrete and continuous state-space models by placing the former on top of the latter, to produce deep generative models that, through active inference, can be used to simulate many known aspects of computational anatomy and physiology in the brain~\cite{fristonGraphicalBrainBelief2017}.

\begin{figure}
    \centering
    \includegraphics[width=0.85\textwidth]{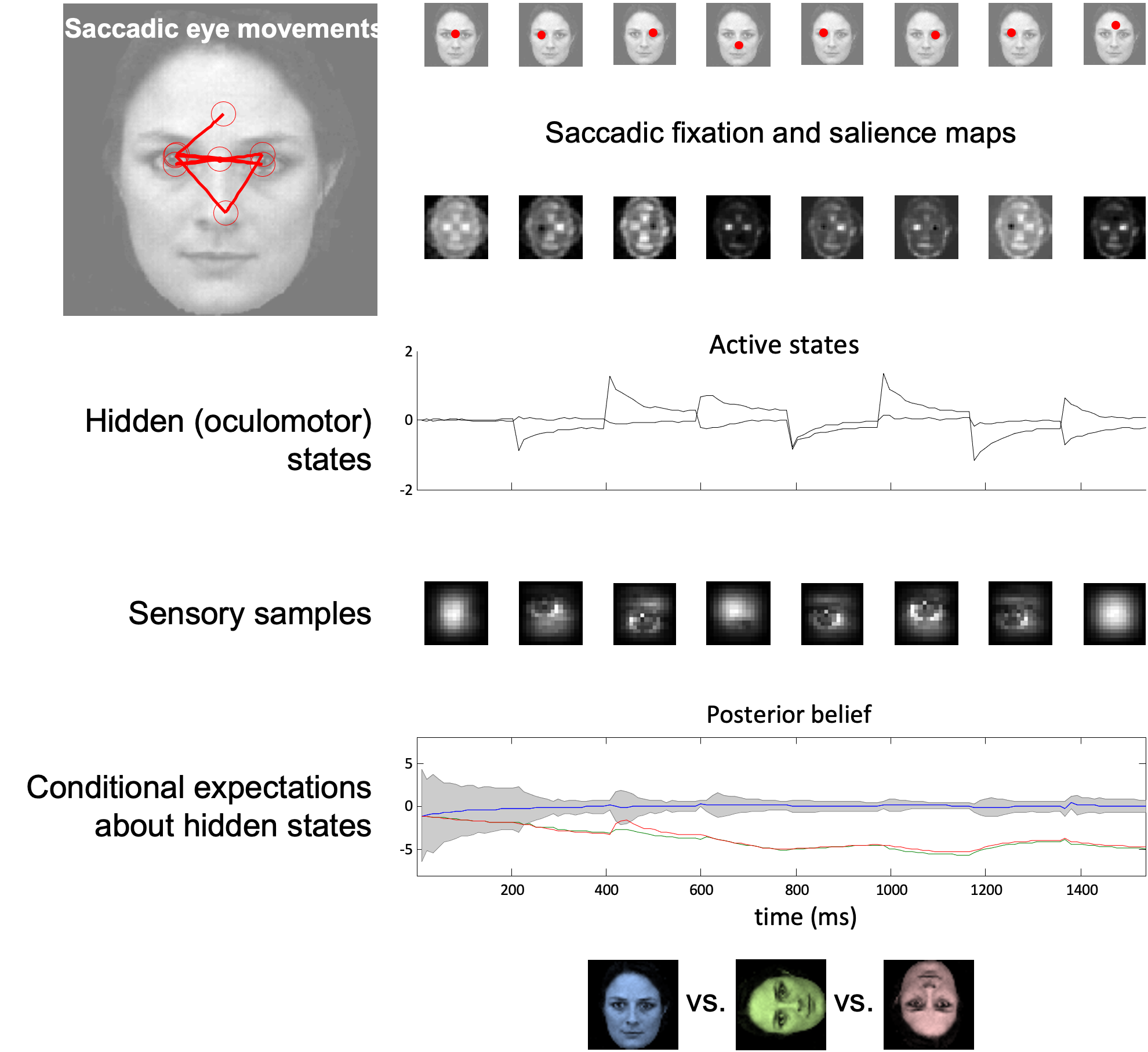}
    \caption[Epistemic foraging]{\textbf{Epistemic foraging}. This figure shows the results of a numerical simulation where a face was presented to an agent, whose responses were obtained by selecting active states that minimised expected free energy following an eye movement. The agent had three internal images or hypotheses (i.e., internal states) about the external state she might sample (an upright face (blue), an inverted face (magenta) and a rotated face (green)---shown at the bottom). The agent was presented with sensory samples of an upright face and her variational posterior over the external state was obtained by descending variational free energy over a 12ms time bin until the next saccade (i.e., action) was emitted. This perception-action cycle was repeated eight times. The agent's eye movements are shown as red dots at the end of each saccade in the upper row. The corresponding sequence of eye movements is shown in the upper-left inset, where the red circles correspond roughly to the proportion of the visual image sampled. These saccades are driven by the salience maps in the second row, which correspond to the expected free energy as a function of the policies; namely, the next saccade or where to look next. As expected free energies are defined in terms of trajectories, it is best to see the locations of on these salience maps as expressing the expected free energy of a trajectory that ends in that location. Note that these maps change with successive saccades as variational posterior beliefs become progressively more confident about the external state. Note also that salience is depleted in locations that were foveated in the previous saccade because these locations no longer have epistemic affordance or expected information gain (i.e., the ability to reduce uncertainty in the expected free energy). In neuroscience, this empirical phenomenon is known as inhibition of return. Oculomotor responses are shown in the third row in terms of the two oculomotor states corresponding to vertical and horizontal eye movements. The associated portions of the image sampled (at the end of each saccade) are shown in the fourth row. The fifth row shows the evolution of variational posterior beliefs about external (a.k.a. hidden) in terms of the log probability they assign to each possible external state (colour coded) and 90\% confidence intervals. The key thing to note is that the credence about the true external state supervenes over alternative expectations and, as a result, confidence about the category increases (and confidence intervals shrink to the mode). This illustrates the nature of evidence accumulation when selecting a hypothesis or percept that best explains sensory states. Please see~\cite{fristonPerceptionsHypothesesSaccades2012} for further details.}
    \label{fig: 9}
\end{figure}

\subsection*{Summary}
In summary, we now have at hand a way of identifying the most likely autonomous trajectory from any initial particular state that can be used to simulate the sentient behaviour of precise particles that we have associated with biotic systems. The expected free energy absorbs two aspects of Bayes optimal behaviour into the same (objective) functional~\cite{sajidActiveInferenceBayesian2022}. On a Bayesian reading, the expected information gain is exactly the same quantity that underwrites the principles of optimal Bayesian design~\cite{lindleyMeasureInformationProvided1956,mackayInformationTheoryInference2003,baliettiOptimalDesignExperiments2021}. In other words, the principles that prescribe the best way to solicit evidence that reduces uncertainty about various hypotheses. The second imperative comes from Bayesian decision theory, where the objective is to minimise some expected cost function expected under a choice or decision~\cite{waldEssentiallyCompleteClass1947,brownCompleteClassTheorem1981,bergerStatisticalDecisionTheory1985}. 

Teleologically, it is worth reflecting upon the differences between the generative models that underwrite state-wise and path-wise descriptions of Bayesian mechanics, respectively. For the state-wise formulation \eqref{eq: 21}, the generative model is just a joint density over external and particular states, supplied by—or supplying—the NESS density. For the path-wise formulation \eqref{eq: 27}, \eqref{eq: 34}, the generative model is a joint distribution over the paths of external and sensory states. In other words, there is an implicit state-space model of dynamics that can be summarised heuristically as modelling the consequences of an action on external and sensory dynamics. Because consequences follow causes, the generative model acquires a temporal depth~\cite{yildizBirdsongHumanSpeech2013,fristonDeepTemporalModels2018}. This depth required to describe any given particle may, of course, be another characteristic that distinguishes different kinds of particles. In short, the path-wise formulation describes \textit{particles that plan}, under a proximal or distal horizon.

\section{Conclusion}

There are many points of contact between the variational formulation above and other normative theories of self-organisation and purposeful behaviour. However, to focus the narrative we have deliberately suppressed demonstrating precedents, variants and special cases. Figure \ref{fig: 3} highlights a few relationships between the free energy principle and various formulations of self-organisation and sentient behaviour. In brief, this casts things like reinforcement learning and optimal control theory as optimising the marginal likelihood of particular states, conditioned upon a generative model supplied by a nonequilibrium steady-state density. It could be argued that the link between the free energy principle and established formulations is most direct for synergetics~\cite{hakenSynergeticsIntroductionNonequilibrium1978,hakenInformationSelforganizationUnifying2016} and related treatments of dissipative structures~\cite{prigogineTimeStructureFluctuations1978}. There is also a formal and direct link to information theoretic formulations and Bayesian statistics. Furthermore, the free energy principle can be regarded as dual to the constrained maximum entropy principle \cite{sakthivadivelEntropyMaximisingDiffusionsSatisfy2023}, where the constraints are supplied by the generative model. Please see~\cite{hafnerActionPerceptionDivergence2020,fristonSophisticatedInference2021} for a treatment of things like empowerment~\cite{klyubinEmpowermentUniversalAgentcentric2005}, information bottleneck~\cite{tishbyInformationBottleneckMethod2000} and predictive information~\cite{stillInformationtheoreticApproachCuriositydriven2012,stillThermodynamicsPrediction2012}.

In a similar vein, there are several accounts of optimal behaviour—in both its epistemic and pragmatic aspects—that are closely related to the path integral formulation of active inference. Some key relationships are highlighted in Figure \ref{fig: 7}, such as intrinsic motivation, artificial curiosity~\cite{oudeyerWhatIntrinsicMotivation2007,schmidhuberFormalTheoryCreativity2010,bartoNoveltySurprise2013} and optimal control~\cite{kappenPathIntegralsSymmetry2005,vandenbroekRiskSensitivePath2010,kappenOptimalControlGraphical2012}. The interesting thing about these other theories is that they are predicated on optimising some objective function that can be recovered from expected free energy by taking various sources of uncertainty off the table. This discloses things like the objective optimised in reinforcement learning and expected utility theory in behavioural psychology and economics, respectively~\cite{botvinickHierarchicallyOrganizedBehavior2009,bossaertsBehaviouralEconomicsNeuroeconomics2015a}. 

This chapter has focused on a single particle and has largely ignored the (external) context that leads to generalised synchrony among internal and external states. This synchronisation goes hand-in-hand with existence \textit{per se} and the Bayesian mechanics supplied by the free energy principle. The very fact that this mechanics rests upon synchronisation may speak to the emergence of synchronisation among formally similar particles; namely, populations or ensembles. In other words, an individual member of an ensemble or ecosystem owes its existence to the ensemble of which it is a member—at the level of multicellular organisation or indeed its conspecifics in evolutionary biology~\cite{manickaModelingSomaticComputation2019}. In a similar vein, the context established by supra- and subordinate scales plays an existential role. In brief, particles at one scale can only exist if there is a nonequilibrium steady-state density at a higher scale that entails Markov blankets of Markov blankets~\cite{palaciosMarkovBlanketsHierarchical2020}. Due to a separation of temporal scales, much of the self-evidencing at one scale is absorbed into the fast, random fluctuations at the scale above. For example, the fast electrophysiological fluctuations of a neuron become, random fluctuations from the point of view of neuronal population dynamics and sensory motor coordination in the brain~\cite{kelsoDynamicPatternsSelfOrganization1995,decoDynamicBrainSpiking2008,kelsoUnifyingLargeSmallScale2021}. This follows in a straightforward way from applying the apparatus of the renormalisation group. Please see~\cite{fristonFreeEnergyPrinciple2019a} for further discussion.

For brevity and focus, we have not considered applications of the free energy principle and active inference in detail. A brief review of the literature in this area will show that that the majority of applications are in the neurosciences~\cite{dacostaActiveInferenceDiscrete2020} with some exceptions: e.g.,~\cite{fristonKnowingOnePlace2015,ramsteadAnsweringSchrodingerQuestion2018}. Recently, there has been an increasing focus on active inference in the setting of machine learning and artificial intelligence~\cite{ueltzhofferDeepActiveInference2018,tschantzScalingActiveInference2019,barpGeometricMethodsSampling2022,fountasDeepActiveInference2020,hafnerActionPerceptionDivergence2020,catalRobotNavigationHierarchical2021,dacostaRewardMaximizationDiscrete2023}. Much of this literature deals with simulation and modelling and, specifically, scaling active inference to real-world problems. These developments speak to the shift in focus from the foundational issues addressed in this article to their applications. It is quite possible that the foundational aspects of the free energy principle may also shift as simpler interpretations and perspectives reveal themselves.

\chapter{Conclusion}
\label{ch:conclusions}

\section{Summary of thesis achievements}

This thesis has focused on three fundamental aspects of biological entities; namely, entropy production, Bayesian mechanics, and the free-energy principle. Its main contributions are threefold: 1) Providing a comprehensive mathematical theory of entropy production for stochastic differential equations that describes a greater number of systems than before, in particular Markovian approximations of stochastic differential equations driven by colored noise. 2) Developing a theory of Bayesian mechanics, with sufficient and necessary conditions for the internal states of a particle or biological entity to synchronize and infer the external states consistently with variational Bayesian inference in statistics and theoretical neuroscience. 3) Refining these conditions to obtain a description that is more exclusive to biological entities --- as opposed to merely physical ones ---called the free-energy principle. Perhaps the main practical outcome here are equations of motions that biological entities must satisfy in virtue of processing a boundary that separates them from their environment. These equations of motion, in terms of minimizing free energy and expected free energy can be used for producing simulations of biotic behavior --- and in artificial intelligence.

\section{Applications}

The main applications of this thesis is providing tools and formulas for computing entropy production in stochastic models of biological systems, and affording equations --- via the free-energy principle --- to simulate practically any kind of biological behavior. More broadly, the mathematical foundation to the free-energy principle presented here unlocks a coherent narrative from stochastic process and statistical physics, to cognitive neuroscience and biophysics, and to artificial intelligence. Indeed the equations provided by the free-energy principle form the basis of a generic framework for generating sentient behavior called \textit{active inference}, which has been used in a plethora of applications that include computational psychiatry, psychology and neurology, all the way to artificial intelligence, robotics and machine learning, passing by simulations of biological populations and ecosystems \cite{dacostaActiveInferenceDiscrete2020,rubinFutureClimatesMarkov2020}. The mathematical foundation of the free-energy principle enables a first principles approach to simulating behavior in all of these fields, which has been taken up by a multitude of research groups worldwide.

\section{Future work}

Perhaps the main strength of the free-energy principle – its genericity – is also its main weakness. By deriving a formal description of biological systems with as few assumptions as possible, we obtain a description of cognition and behavior that may apply to viruses and the human brain alike. As we move forward with this research program, one important goal is to refine the class of biological systems that we are willing to study, by singling out constraints and properties that are fundamental to the human brain and other organisms possessing higher forms of cognition, to derive refined forms of the principle that are more informative about these sentient beings. Another important challenge is to map the local computations done by neurons and neural populations in the brain with the overarching description provided by the free-energy principle, thereby bridging the detailed and coarse-grained descriptions of brain dynamics \cite{dacostaNeuralDynamicsActive2021,isomuraCanonicalNeuralNetworks2022}. A final challenge is to scale the methodologies of active inference to tackle high dimensional complex problems in robotics and artificial intelligence \cite{fountasDeepActiveInference2020,lanillosActiveInferenceRobotics2021}. This challenge will necessitate significant expertise in scalable Bayesian statistics and computer science.

\clearpage

\addcontentsline{toc}{chapter}{Data availability}

\begin{data}
\textbf{Chapter 2:} All data and numerical simulations can be reproduced with code freely available at \url{https://github.com/lancelotdacosta/entropy_production_stationary_diffusions}.

\textbf{Chapter 3:} All data and numerical simulations can be reproduced with code freely available at \url{https://github.com/conorheins/bayesian-mechanics-sdes}.

\textbf{Chapter 4:} All data is contained within the chapter.

\end{data}

\vspace{25pt}

\addcontentsline{toc}{chapter}{Author contributions}

\begin{contributions}
\textbf{Chapter 2:} 
I conceptualized the paper (together with G.A. Pavliotis); I wrote the initial draft; I did the mathematical analysis and the proofs in their entirety; I wrote the accompanying software, performed the simulations and produced the graphics; I reviewed the draft (together with G.A. Pavliotis) and edited for journal submission; I handled journal correspondence and revised the paper for publication.

\textbf{Chapter 3:} 
I conceptualized the paper (together with co-authors); I wrote the initial draft; I did the mathematical analysis and the proofs in their entirety; I wrote the accompanying software, performed the simulations and produced the graphics (all with C. Heins); I reviewed the draft and edited for journal submission (all with co-authors); I handled journal correspondence and revised the paper for publication.

\textbf{Chapter 4:} 
I wrote parts of the draft; I did parts of the mathematical analysis and proofs; I reviewed the draft and edited for journal submission (all with co-authors); I handled journal correspondence and revised the paper for publication.

\end{contributions}

\addcontentsline{toc}{chapter}{Bibliography}
\bibliographystyle{unsrt}
\bibliography{bib}

\end{document}